\pdfoutput=1
\PassOptionsToPackage{table}{xcolor}
\documentclass[12pt,a4paper]{thesis}

\newcommand{\tipo}{Thesis}
\newcommand{\tipoport}{Tese}
\newcommand{\grau}{Doutor}
\newcommand{\graueng}{Doctor}
\newcommand{\abrevi}{D.Sc.}

\usepackage{ifpdf,psfrag,thumbpdf} 

\usepackage{sty/identfirst}    
\usepackage{sty/datetime}      
\usepackage{sty/ulem}          
\usepackage{sty/rotating}      
\usepackage[T1]{fontenc}       
\usepackage[latin1]{inputenc}  
\usepackage{ae}                
\usepackage{lscape}            
\usepackage{footnote}          
\usepackage{verbatim}          
\usepackage{textcomp}          
\usepackage{comment}

\usepackage{graphicx}
\usepackage{stmaryrd}
\usepackage{tikz}
\usepackage{pgfplots}
\usetikzlibrary{shapes,arrows,shapes.geometric,fit,matrix,backgrounds,plotmarks}
\PassOptionsToPackage{table}{xcolor}
\usepackage{graphicx}
\usepackage{caption}
\usepackage{subcaption}
\usepackage{framed}
\usepackage{fancybox}
\usepackage{fancyhdr}
\usepackage{colortbl}
\usepackage{color}
\usepackage{pgfpages}
\usepackage{amsmath,mathtools,amssymb}
\usepackage{epigraph}
\usepackage{lineno}
\usepackage{algorithm}
\usepackage{algorithmicx}
\usepackage[noend]{algpseudocode}
\newtheorem{myex}{Example}

\newtheorem{mydef}{Def.}
\newtheorem{myprob}{Problem}

\newtheorem{mythm}{Theorem}
\newtheorem{proof}{Proof}
\newtheorem{myprop}{Proposition}
\newtheorem{mycor}{Corollary}
\newtheorem{mylemma}{Lemma}
\newtheorem{myremark}{Remark}

\usepackage{setspace}

\usepackage{tocloft}
\usepackage[explicit]{titlesec}
\usepackage[linktocpage=true,breaklinks]{hyperref}

\RequirePackage[left=2cm,%
                right=2cm,%
				top=2.25cm,%
				bottom=2.25cm,%
				headheight=11pt,%
				a4paper]{geometry}%
\RequirePackage[labelfont={bf,sf},%
                labelsep=period,%
                justification=raggedright]{caption}
\RequirePackage{fancyhdr}  
\RequirePackage{lastpage}  

\newcommand{\helv}{%
\fontfamily{phv}\fontseries{b}\fontsize{9}{11}\selectfont}

\fancypagestyle{plain}{%
\fancyhf{} 
\fancyfoot[R]{\helv \thepage} 
 
}
\RequirePackage[explicit]{titlesec}
\titleformat{name=\section}
  {\color{blue!30!black}\large\sffamily\bfseries}
  {\Large}
  {0em}
  {\colorbox{gray!15}{\parbox{\dimexpr\linewidth-2\fboxsep\relax}{\thesection.\ #1}}}
  []  
\titleformat{\subsection}
  {\color{blue!30!black}\sffamily\bfseries}
  {\thesubsection}
  {0.5em}
  {#1}
  []
\titleformat{\subsubsection}
  {\sffamily\small\bfseries}
  {\thesubsubsection}
  {0.5em}
  {#1}
  [] 
\titleformat{\paragraph}[runin]
  {\sffamily\small\bfseries}
  {}
  {0em}
  {#1} 
\titlespacing*{\section}{0pc}{3ex \@plus4pt \@minus3pt}{5pt}
\titlespacing*{\subsection}{0pc}{2.5ex \@plus3pt \@minus2pt}{0pt}
\titlespacing*{\subsubsection}{0pc}{2ex \@plus2.5pt \@minus1.5pt}{0pt}
\titlespacing*{\paragraph}{0pc}{1.5ex \@plus2pt \@minus1pt}{10pt}

\fancypagestyle{plain}
{\fancyhead{} 
} 
\tocloftpagestyle{empty}

\logo{\includegraphics{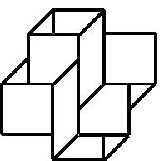}}

\title{\uppercase{M}anaging large-scale scientific hypotheses as uncertain and probabilistic data}

\titleport{{\uppercase{G}}er\^encia de Hip\'oteses Cient\'ificas de Larga-Escala como Dados Incertos e Probabil\'isticos}

\author{Bernardo}{Gon\c{c}alves}

\degree{\graueng{} of Sciences in Computational Modeling} 
       {\abrevi{}}                            

\dept{National Laboratory for Scientific Computing}{}  

\advisor {, D.Sc.}{Fabio}{Porto}  

\readerTwo{, Ph.D.}{Pedro L.}{Dias}
\readerFour {, Ph.D.}{Marco A.}{Casanova}               
\readerFive {, D.Sc.}{Ana Carolina}{Salgado}               

\catalography[Ficha Catalogr\'afica]{\OnePageChapter

\numpages{xvii, 128p.}


\codbib{CDD -- 629.8} \codtese{G635m} }

\abstractport{  \OnePageChapter  
\noindent
Tendo em vista a mudan\c{c}a de paradigma que faz da ci\^encia cada vez mais guiada por dados, nesta tese propomos um m\'etodo para codifica\c{c}\~ao e ger\^encia de hip\'oteses cient\'ificas determin\'isticas de larga escala como dados incertos e probabil\'isticos. 

Na forma de equa\c{c}\~oes matem\'aticas, hip\'oteses relacionam simetricamente aspectos do fen\^omeno de estudo. 
Para computa\c{c}\~ao de predi\c{c}\~oes, no entanto, hip\'oteses determin\'isticas podem ser abstra\'idas como fun\c{c}\~oes. Levamos adiante a no\c{c}\~ao de Simon de equa\c{c}\~oes estruturais para extrair de forma eficiente a ent\~ao chamada ordena\c{c}\~ao causal impl\'icita na estrutura de uma hip\'otese.

Mostramos como processar a estrutura preditiva de uma hip\'otese atra\-v\'es de algoritmos originais para sua codifica\c{c}\~ao como um conjunto de depend\^encias funcionais (df's) e ent\~ao realizamos infer\^encia causal em termos de racioc\'inio ac\'iclico pseudo-transitivo sobre df's. 
Tal racioc\'inio revela importantes depend\^encias causais impl\'icitas nos dados preditivos da hip\'otese, que conduzem nossa s\'intese do banco de dados probabil\'istico. Como na \'area de modelos gr\'aficos (IA), o banco de dados probabil\'istico deve ser normalizado de tal forma que a incerteza oriunda de hip\'oteses alternativas seja decomposta em fatores e propagada propriamente recuperando sua distribui\c{c}\~ao de probabilidade conjunta via jun\c{c}\~ao `lossless.' Isso \'e motivado como um princ\'ipio te\'orico de projeto para ger\^encia e an\'alise de hip\'oteses.

O m\'etodo proposto \'e aplic\'avel a hip\'oteses determin\'isticas quantitativas e qualitativas e \'e demonstrado em casos real\'isticos de ci\^encia computacional.
}

\abstract{  \OnePageChapter
\noindent
In view of the paradigm shift that makes science ever more data-driven, in this thesis we propose a synthesis method for encoding and managing large-scale deterministic scientific hypotheses as uncertain and probabilistic data. 

In the form of mathematical equations, hypotheses symmetrically relate aspects of the studied phenomena. For computing predictions, however, deterministic hypotheses can be abstracted as functions. We build upon Simon's notion of structural equations in order to efficiently extract the (so-called) causal ordering between variables, implicit in a hypothesis structure (set of mathematical equations). 

We show how to process the hypothesis predictive structure effectively through original algorithms for encoding it into a set of functional dependencies (fd's) and then performing causal reasoning in terms of acyclic pseudo-transitive reasoning over fd's. Such reasoning reveals important causal dependencies implicit in the hypothesis predictive data and guide our synthesis of a probabilistic database. Like in the field of graphical models in AI, such a probabilistic database should be normalized so that the uncertainty arisen from competing hypotheses is decomposed into factors and propagated properly onto predictive data by recovering its joint probability distribution through a lossless join. That is motivated as a design-theoretic principle for data-driven hypothesis management and predictive analytics.

The method is applicable to both quantitative and qualitative deterministic hypotheses and demonstrated in realistic use cases from computational science.

}

\inspiration[]{ \OnePageChapter 
\vspace{-85pt}
\begin{spacing}{1.3}
\begin{flushright}
\setlength{\epigraphrule}{0pt}
\setlength{\epigraphwidth}{.6\textwidth}

\epigraph{

``\emph{Originally, there was just experimental science, and then there was theoretical science, with Kepler's Laws, Newton's Laws of Motion, Maxwell's equations, and so on. Then, for many problems, the theoretical models grew too complicated to solve analytically, and people had to start simulating. These simulations have carried us through much of the last half of the last century. At this point, these simulations are generating a whole lot of data, along with a huge increase in data from the experimental sciences.}''
}{\vspace{-6pt}--- Jim Gray, 2007}
\end{flushright}
\end{spacing}
 }

\dedication[Dedicatory]
{
    \OnePageChapter 

    \textbf{Dedicatory}\\
To my parents Tania and Francisco,

and to my special love, Marcelle,

for being an island of certainty in an uncertain world.

}

\acknowledgements{
\OnePageChapter 

\vspace{-15pt}
This thesis work has been supported by LNCC (Graduate Program in Computational Modeling), CNPq (grant 141838/2011-6), FAPERJ (grant `Nota $\!$10' $\!$E-26/100.286/2013) and IBM (Ph.D. Fellowship 2013/2014). 

I would like to express my gratitude to my advisor Fabio Porto for the gift of the challenging topic of this thesis, so special to me. I am grateful to him for being an inspiring advisor, and for nicely influencing me towards database research. I also thank my thesis committee for their attention and time devoted in the assessment of my work. 
I am indebted to Ana Maria Moura for her generous advice and support throughout my Ph.D. research, 
and to Frederico C. Silva and Adolfo Sim\~oes for their support to my research project at DEXL/LNCC.

I would like to thank very much all researchers at LNCC who give lectures in the graduate program, specially Prof. Jos\'e Karam Filho for teaching me generously about the roots of mathematical modeling. They have contributed significantly to my education as a cross-disciplinary thinker and the shaping of my scientific and mathematical skills. I thank my Ph.D. colleagues and friends at LNCC, specially Eduardo Lima, Ramon Costa, Klaus Wehmuth, Raquel Lopes, Karine Guimar\~aes and Diego Paredes for their companionship and joy shared in the pursuit of their Ph.D. theses. I would also like to gratefully recall my earlier professors at UFES, Jos\'e Gon\c{c}alves, Giancarlo Guizzardi, Rosane Caruso and Berilhes Garcia for shaping the most essential building blocks in my education.

Finally, I would like to thank my family for their support: my father Francisco, my example of simplicity and goodness; my mother, Tania, for her tenacity and true love; Leo and Sche, my dear brother and sister, true union for life;  my grandmothers Zeca, Orizontina and Rosita for their love and prayers. I thank also Marcel and Maris for the greatest gift, their daughter and my near future wife, Marcelle -- the best partner one could ever hope for. 
}

\symbols{	\OnePageChapter	
{\footnotesize
\begin{itemize}
\item fd: functional dependency
\item p-DB: probabilistic database
\item p-WSA: probabilistic world set algebra
\item \textsf{MayBMS}: U-relational database management system
\item AI: Artificial Intelligence
\item GM: graphical models (e.g., Bayesian Networks) 
\item ETL: extract, transform, load
\item OLAP: On-Line Analytical Processing
\item DW: Data Warehouse
\item SEM: structural equation model
\item COA: causal ordering algorithm
\item $\mathcal S_k$: structure of hypothesis $k$
\item $\mathcal E$: set of equations in a structure
\item $\mathcal V$: set of variables appearing in the equations in a structure
\item $Vars(f)$: set of variables appearing in equation $f$
\item $\varphi\!:\; \mathcal E \to \mathcal V$: total causal mapping from equations to variables in a structure
\item $C_\varphi$: set of causal dependencies
\item $G_\varphi$: causal graph induced by $\varphi$
\item $H_k$: `big' fact table of hypothesis $k$
\item $\boldsymbol H$: set of hypothesis `big' fact tables
\item $Y_k^\ell$: U-relation synthesized for hypothesis $k$
\item $\boldsymbol Y_k$: set of U-relations synthesized for hypothesis $k$
\item $H_0$: relational `explanation' table
\item $Y_0$: U-relational `explanation' table
\item $\phi$: phenomenon identifier
\item $\upsilon$: hypothesis identifier
\item $\Sigma$, $\Gamma$, $\Delta$, $\Omega$: fd sets
\item $\Sigma^\vartriangleright$: subset of the closure of an fd set derived by pseudo-transitivity
\item $\Sigma^+$: closure of an fd set
\item $\Sigma^\looparrowright$: the folding of an fd set
\item synthesis `4U': synthesis for uncertainty
\item tid: hypothesis trial identifier
\item BCNF: Boyce-Codd Normal Form
\item MathML: Mathematical Markup Language
\item TCM: total causal mapping algorithm
\item lhs, rhs: left-hand side, right-hand side
\item \emph{sch(R)}: data columns of table $R$
\item $\overline{V_i\,D_i}$: condition columns of a table 
\end{itemize}
}
}

\ToCisShort 

\LoFisShort 

\LoTisShort 

\begin{document}


\chapter{Introduction}\label{ch:introduction}

\noindent
In view of the paradigm shift that makes science ever more data-driven \cite{hey2009}, \emph{in this thesis we demonstrate that large deterministic scientific hypotheses can be effectively encoded and managed as a kind of uncertain and probabilistic data}. 

Deterministic hypotheses can be formed as principles or ideas, then expressed mathematically and implemented in a program that is run to give their \emph{decisive} form of data (see Fig. \ref{fig:galileo}). Hypotheses can also be learned in large scale, as exhibited in the \textsf{Eureqa} project \cite{schmidt2009}. 
Examples of `structured deterministic hypotheses' include tentative mathematical models in physics, engineering and economical sciences, or conjectured boolean networks in molecular biology and social sciences. These are important reasoning devices, as they are solved to generate valuable predictive data for decision making in science and increasingly in business as well. 

In fact, we can refer nowadays to a broad, modern context of \emph{data science} \cite{dhar2013} and \emph{big data} \cite{jagadish2014} in which the complexity and scale of so-called `data-driven' problems require proper data management tools for the predicted data to be analyzed effectively. In this thesis, we pay attention to a quite general class of (tentative) \emph{computational science} models,\footnote{`Computational science' is (sic.) ``a rapidly growing multidisciplinary field that uses advanced computing capabilities to understand and solve complex problems'' \cite{benioff2005}. We may refer to non-stochastic, tentative computational science models throughout this text as `structured deterministic hypotheses.'} and we look at them in an original way as a distinguished kind of data source. 

\begin{spacing}{1}
\noindent\begin{minipage}[t]{0.5\columnwidth}%
\footnotesize{\shadowbox{Law of free fall}}\vspace{1pt}\\
\footnotesize{\emph{``If a body falls from rest, its velocity at any point is proportional to the time it has been falling.''}
\begin{center}
\vspace{-2pt}
(i)
\end{center}
}
\vspace{-7pt}
\begin{center}
\begin{footnotesize}
\begin{verbatim}
for k = 0:n;
   t = k * dt; 
   v = -g*t + v_0; 
   s = -(g/2)*t^2 + v_0*t + s_0; 
   t_plot(k) = t; 
   v_plot(k) = v; 
   s_plot(k) = s;
end
\end{verbatim}
\end{footnotesize}
\vspace{3pt}
(iii)
\end{center}
\end{minipage}
\hspace{0.2pt}	
\noindent\begin{minipage}[t]{0.5\columnwidth}%
\vspace{-20pt}
\begin{footnotesize}
\begin{eqnarray*}
a(t) \!\!\!&=&\!\!\! -g\\
\operatorname{v}(t) \!\!\!&=&\!\!\! -g t \,+\, \operatorname{v_0}\\
s(t) \!\!\!&=&\!\!\!  -(g/2)t^2 \,+\, \operatorname{v_0}t \,+\, s_0
\end{eqnarray*}
\begin{center}
\vspace{-2.5pt}
(ii)
\end{center}
\vspace{-7pt}
\begin{center}
\begingroup\setlength{\fboxsep}{1pt}
\colorbox{blue!5}{%
   \begin{tabular}{c|c|c|c}
  \textsf{FALL} & $t$ & $\operatorname{v}$ & $s$\\
      \hline    
  & $0$ & $0$ & $5000$\\
  & $1$ & $-32$ & $4984$\\
  & $2$ & $-64$ & $4936$\\
  & $3$ & $-96$ & $4856$\\
  & $4$ & $-128$ & $4744$\\
  & $\cdots$ & $\cdots$ & $\cdots$\\
   \end{tabular}
}\endgroup\vspace{8pt}\\
(iv)
\end{center}
\end{footnotesize}
\end{minipage}
\vspace{-8pt}
\begin{figure}[hb]
\caption{Multi-fold view of a deterministic scientific hypothesis.}
\label{fig:galileo}
\vspace{-8pt}
\end{figure}
\end{spacing}

It is generally considered that computational science models, interpreted here as hypotheses to explain real-world phenomena, are of strategic relevance \cite{benioff2005}. 
They are usually complex in that they may have hundreds to thousands of intertwined (coupled) variables and be computed along space, time or frequency domains in arbitrarily large scale. It is important to note the distinction between the structure and data levels. Consider, say, Lotka-Volterra's model, which essentially consists in (Eqs. \ref{eq:lotka-volterra-model}) two ordinary differential equations, complemented by seven subsidiary equations $f_1(t),\, f_2(x_0),\, f_3(y_0),\, f_4(b),\, f_5(p),\, f_6(r),\, f_7(d)$ to set the values of its domain variable $t$ and (input) parameters $x_0,\, y_0,\, b,\, p,\, r,\, d$.
\begin{eqnarray}
\left\{ 
  \begin{array}{lll}
\dot{x} &=& x(b - py)\\
\dot{y} &=& y(rx - d)\\
\end{array} \right.
\label{eq:lotka-volterra-model}
\end{eqnarray}
In a sense, it can be said fairly simple, as it is characterized by a set $\mathcal E$ of equations and a set $\mathcal V$ of variables, sized $|\mathcal E| = |\mathcal V| = 9$. 
Yet, at the data level this model (cf. Chapter \ref{ch:vision}) can be made very large just by computing its predictions in a fine time resolution and/or along an extended time window. 

As we shall see shortly, the technical challenges associated with this thesis involve (not only but) majorly the structure level where, e.g., such Lotka-Volterra model can be abstracted as a deterministic structure $\mathcal S(\mathcal E, \mathcal V)$ with $|\mathcal S|=18$.\footnote{The structure length $|\mathcal S|$ is a measure of how dense the hypothesis structure is, comprising the total sum of the number of variables appearing in each equation.} We are really concerned here with models whose structure $\mathcal S$ is in the order of $|\mathcal S| \lesssim 1M$, and whose results (data!) shall be difficult to analyze by handicrafted practice. Note that the data level of a model can be set as large as wanted (set the domain resolution and/or extension accordingly), but it shall be necessarily large when its structure is itself large. By `large-scale hypotheses' then we mean tentative deterministic models that are large at structure level.

Overall, such class of hypotheses can be said to qualify to at least four of the five v's associated to the notion of big data:\footnote{The `v' of \emph{velocity} may appear in connection with machine learning hypotheses, which we discuss in Chapter \ref{ch:applicability}.} \emph{value}, because of their role in advancing science and technology; \emph{volume}, due to the large scale of modern scientific problems; \emph{variety}, because of their structural heterogeneity, even when they refer to the same phenomena; and \emph{veracity}, due to their uncertainty.

The idea of managing hypotheses `as data' may sound intriguing and in fact it raises a number of research questions of both conceptual and technical nature.\footnote{We shall keep record of those questions and revisit them in \S\ref{sec:questions}.} We start by outlining below the conceptual research questions.

\begin{itemize}
\item[\textbf{RQ1}.]\label{rq1}
How to define and encode hypotheses `as data'? What are the sources of uncertainty that may be present and should be considered?

\item[\textbf{RQ2}.]\label{rq2} 
How does hypotheses `as data' relate with observational data or, likewise, phenomena `as data' from a database perspective?

\item[\textbf{RQ3}.]\label{rq3}
Does every piece of simulated data qualify as a scientific hypothesis? What is the difference between managing `simulation' data from managing `hypotheses' as data?

\item[\textbf{RQ4}.]\label{rq4} 
Is there available a proper (machine-readable) data format we can use to automatically extract mathematically-expressed hypotheses from? 
\end{itemize}

\noindent
It has been a challenge of this thesis to provide reasonable answers to these questions, which are brought together into the vision of hypotheses `as data' (we call it the $\Upsilon$-DB vision) and its use case that we present in Chapter \ref{ch:vision}, and experiment with in realistic scenarios in Chapter \ref{ch:applicability}. 

\begin{framed}
\noindent
The $\Upsilon$-DB vision formulates the problem of hypothesis encoding as a problem of \emph{probabilistic database design}. A number of technical questions arise then.
\end{framed}

We introduce now technical context, materials and methods identified and selected in this thesis as a basis to realize the $\Upsilon$-DB vision in terms of probabilistic database design. We shall outline in the sequel the technical research questions to be answered by the core of the thesis.

\section{Problem Space and Specific Goals}\label{sec:goals}

It has been a goal of this thesis to investigate the capabilities of probabilistic databases to enable hypothesis data management as a particular case of simulation data management. In the sequel, we first characterize the use case of hypothesis data management and then formulate it in terms of probabilistic DB design.

\subsection{Simulation data management}\label{subsec:simulation}

Simulation laboratories provide scientists and engineers with very large, possibly huge datasets that reconstruct phenomena of interest in high resolution. Notorious examples are the John Hopkins Turbulance Databases \cite{meneveau2007}, and the Human Brain Project (HBP) neuroscience simulation datasets \cite{markram2006}. A core motivation for the delivery of such data is enabling new insights and discoveries through \emph{hypothesis testing against observations}. 
Nonetheless, while the use case for \emph{exploratory analytics} is currently well understood and many of its challenges have already been coped with so that high-resolution simulation data is increasingly more accessible \cite{ailamaki2010,ahmad2010}, %
only very recently, as part of this thesis work, the use case of hypothesis management has been taken into account for \emph{predictive analytics} \cite{goncalves2014}. 

In fact, there is a pressing call for innovative technology to integrate (observed) data and (simulated) theories in a unified framework \cite{cushing2013,weinberg2010,golub2010}. The point has just been raised by leading neuroscientists in the context of the HBP, who are incisive on the compelling argument that massive simulation databases should be constrained by experimental data in corrective loops to test precise hypotheses \cite[p. 28]{fregnac2014}.
Fig. \ref{fig:flow-chart}$\,$ shows a simplified view of the (data-driven) scientific method life cycle. It distinguishes the phases of exploratory analytics (context of discovery) and predictive analytics (context of justification), and highlights the loop between hypothesis formulation and testing \cite{losee2001}.

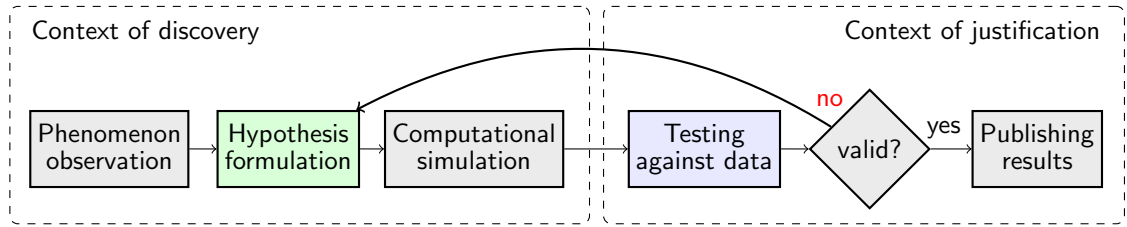
\begin{figure}[t]\footnotesize
\advance\leftskip-0.1cm
\tikzstyle{rect}=[rectangle,
                                    thick,
                                    minimum size=20pt,
                                    minimum width=20pt,
                                    fill=black!8,
                                    draw=black]
\tikzstyle{rectGreen}=[rectangle,
                                    thick,
                                    minimum size=20pt,
                                    fill=green!15,
                                    draw=black]
\tikzstyle{rectBlue}=[rectangle,
                                    thick,
                                    minimum size=20pt,
                                    fill=blue!9,
                                    draw=black]
\tikzstyle{backDisc}=[rectangle,
			       style=dashed,
                                    rounded corners=3pt,
                                    minimum size=82pt,
                                    minimum width=217pt,
                                    draw=black]
\tikzstyle{backJust}=[rectangle,
			       style=dashed,
                                    rounded corners=3pt,
                                    minimum size=82pt,
                                    minimum width=195pt,
                                    draw=black]
\tikzstyle{box}=[rectangle,
                                    fill=none,
                                    draw=none]
\tikzstyle{diamond1}=[diamond,
                                    thick,
                                    minimum size=45pt,
                                    inner sep=0pt,
                                    fill=black!8,
                                    draw=black]
\tikzstyle{edge} = [draw,thick,->, ,bend right] 
\begin{tikzpicture}[scale=1.5]
    \node[backDisc] (backDisc) at (-1.65,4) {};
    \node[backJust] (backJust) at (3.3,4) {};
    \node[box] (disc) at (-3.0,4.73) {\textsf{Context of discovery}};
    \node[box] (just) at (4.25,4.73) {\textsf{Context of justification}};
        
    \node[rect] (ph) at (-3.32,3.7) {
    \begin{tabular}{cc}
	$\!\!\!\!$\textsf{Phenomenon}$\!\!\!\!$\vspace{-2pt}\\
	$\!\!\!\!$\textsf{observation}$\!\!\!\!$\\
    \end{tabular}
};    
    \node[rectGreen] (h) at (-1.75,3.7) {
    \begin{tabular}{cc}
	$\!\!\!\!$\textsf{Hypothesis}$\!\!\!\!$\vspace{-2pt}\\
	$\!\!\!\!$\textsf{formulation}$\!\!\!\!$\\
    \end{tabular}
};    
    \node[rect] (cm) at (-0.12,3.7) {
    \begin{tabular}{cc}
	$\!\!\!\!$\textsf{Computational}$\!\!\!\!$\vspace{-2pt}\\
	$\!\!\!\!$\textsf{simulation}$\!\!\!\!$\\
    \end{tabular}
};    
    \node[rectBlue] (t) at (1.9,3.7) {
    \begin{tabular}{cc}
	$\!\!\!\!$\textsf{Testing}$\!\!\!\!$\vspace{-2pt}\\
	$\!\!\!\!$\textsf{against data}$\!\!\!\!$\\
    \end{tabular}
};    
    \node[rect] (pub) at (4.82,3.7) {
    \begin{tabular}{cc}
	$\!\!\!\!$\textsf{Publishing}$\!\!\!\!$\vspace{-2pt}\\
	$\!\!\!\!$\textsf{results}$\!\!\!\!$\\
    \end{tabular}
};    

    \node[diamond1,rotate=90] (hvalid) at (3.35,3.7) {\rotatebox[origin=c]{-90}{\textsf{valid?}}};
    \node[box] (yes) at (4.0,3.87) {\textsf{yes}};
    \node[box] (no) at (3.0,4.13) {\textcolor{red}{\textsf{no}}};

    \draw[->] (ph) to (h);
    \draw[->] (h) to (cm);
    \draw[->] (cm) to (t);
    \draw[->] (t) to (hvalid);
    \draw[->] (hvalid) to (pub);
    \draw[edge] (hvalid) to (h);
                
\end{tikzpicture}
\caption[A view of the scientific method life cycle]{A view of the scientific method life cycle. It highlights hypothesis formulation and a backward transition to reformulation if predictions `disagree' with observations.}
\label{fig:flow-chart}
\end{figure}

Simulation data, being generated and tuned from a combination of theoretical and empirical principles, has a distinctive feature to be considered when compared to data generated by high-throughput technology in large-scale scientific experiments. It has a pronounced \emph{uncertainty} component that motivates the use case of hypothesis data management for \emph{predictive analytics} \cite{goncalves2014}. Essential aspects of hypothesis data management can be described in contrast to simulation data management as follows --- Table \ref{tab:hypothesis} summarizes our comparison. 

\begin{table}
\caption{Simulation data management vs. hypothesis data management.}
\label{tab:hypothesis}
\begin{spacing}{1.1}
\begingroup\setlength{\fboxsep}{3pt}
\colorbox{gray!5}{%
   \begin{tabular}{p{.5\textwidth}|p{.45\textwidth}}
  \rowcolor{gray!15} \textbf{Simulation data management} & \textbf{Hypothesis data management}\\
      \hline    
   Exploratory analytics & Predictive analytics\\
   Raw data & Sample data\\
   Extremely large (TB, PB) & Very large (MB, GB)\\
   Dimension-centered access pattern & Claim-centered access pattern\\
   Denormalized for faster retrieval & Normalized for uncertainty factors\\
   Batch-, incremental-only data updates & Probability distribution updates\\  
   \end{tabular}
}\endgroup
\end{spacing}
\end{table}

\begin{itemize}
\item \emph{Sample data}. Hypothesis management shall not deal with the same volume of data as in simulation data management for exploratory analytics, but only samples of it. This is aligned, for example, with the architectural design of CERN's particle-physics experiment and simulation ATLAS, where there are four tier/layers of data. The volume of data significantly decreases from (tier-0) the raw data to (tier-3) the data actually used for analyses such as hypothesis testing \cite[p. 71-2]{ailamaki2010}. 
Samples of raw simulation data are to be selected for comparative studies involving competing hypotheses in the presence of evidence (sample observational data). This principle is also aligned with how data is delivered at model repositiories. Since observations are usually less available, only the fragment (sample) of the simulation data that matches in coordinates the (sample) of observations is required out of simulation results for comparative analysis. 
For instance, we show in \S\ref{subsec:case-baroreflex} a predictive analytical study extracted from the Virtual Physiological Rat Project (VPR1001-M) comparing sample simulation data (heart rates) from a baroreflex model with observations on a Dahl SS rat strain.\footnote{\url{http://virtualrat.org/computational-models/vpr1001/}.} 
The simulation is originally set to produce predictions in the time resolution of $t_\Delta=0.01$. But since the observational sample is only as fine as $t_\Delta=0.1$, there is no gain in rendering a predicted sample with $t_\Delta \geq 0.1$ for hypothesis testing. Note that such a `sampling' does not incur in any additional uncertainty as typical of statistical sampling \cite{bolstad2007}.

\item \emph{Claim-centered access pattern}. In simulation data management the access pattern is dimension-centered (e.g., based on selected space-time coordinates) and the data is denormalized for faster retrieval, as typical of Data Warehouses (DW's) and OLAP applications.\footnote{On-Line Analytical Processing, as distinguished from OLTP (On-Line Transaction Processing. The latter is meant for transaction processing of daily queries and updates in operational systems, while the former is for analytical queries in Data Warehouses (DW's) that gather a lot of data collected from different sources for decision making.}
In particular, on account of the so-called `big table' approach, each state of the modeled physical system is recorded in a large, single row of data. This is fairly reasonable for an Extract-Transform-Load (ETL) data ingesture pipeline characterized by \mbox{batch-,} incremental-only updates (see Fig. \ref{fig:etl-pipeline}). Such a setting is in fact fit for exploratory analytics, as entire states of the simulated system shall be accessed at once (e.g., providing data to a visualization system). Altogether, data retrieval is critical and there is no risk of update anomalies. 
Hypothesis management, in contrast, should be centered on claims identified within the hypothesis structure w.r.t. available data dependencies. Since the focus is on resolving uncertainty for decision making (which hypothesis is a best fit?), the data must be normalized based on \emph{uncertainty factors}. 
This is key for the correctness of uncertainty modeling and efficiency of probabilistic reasoning, say, in a probabilistic database \cite[p.30-1]{suciu2011}.

\item \emph{Uncertainty modeling}. In uncertain and probabilistic data management \cite{suciu2011}, the uncertainty may come from two sources: \emph{incompleteness} (missing data), and \emph{multiplicity} (inconsistent data). Hypothesis management on sample simulation data is concerned with the multiplicity of prediction records due to competing hypotheses targeted at the same studied phenomenon. Such a multiplicity naturally gives rise to a probability distribution that may be initially uniform and eventually conditioned on observations. Conditioning is an applied \emph{Bayesian inference} problem that translates into database update for transforming the prior probability distribution into a posterior \cite{goncalves2014}.

\end{itemize}

\begin{figure}[t]
\begin{center}
\tikzstyle{rect2}=[rectangle,
                                    thick,
                                    minimum size=23pt,
                                    fill=black!20,
                                    draw=black]
\tikzstyle{rect3}=[rectangle,
                                    rounded corners=3pt,
                                    minimum size=75pt,
                                    minimum width=110pt,
                                    draw=black]
\tikzstyle{box}=[rectangle,
                                    fill=none,
                                    draw=none]
\tikzstyle{cyl1}=[cylinder,
                                    thick,
                                    minimum size=50pt,
                                    inner sep=0pt,
                                    fill=none,
                                    draw=black]
\tikzstyle{edge} = [draw,thick,->,bend left]
\begin{tikzpicture}[scale=1.1]
    \node[rect3] (back) at (0.1,3) {};
    \node[rect2] (d1) at (-1.0,3.65) {$\mathcal{D}^1$};
    \node[rect2] (d2) at (-0.1,3.1) {$\mathcal{D}^2$};
    \node[box] (dot) at (0.5,2.65) {$\ddots$};
    \node[rect2] (dn) at (1.25,2.4) {$\mathcal{D}^p$};
    \node[cyl1,rotate=90] (h) at (4.8,3) {\rotatebox[origin=c]{-90}{sim}};
    \node[box] (hlabel) at (4.8,2.0) {$\bigcup_{i=1}^{p} R_i$};
    \draw[->] (back) to (h);
    \node[box] (etl) at (3.0,3.3) {\textsf{ETL}};    
\end{tikzpicture}
\caption[The usual data ingesture pipeline of simulation data management]{The usual data ingesture pipeline of simulation data management. Datasets $\bigcup_{i=1}^p\mathcal{D}^i$ generated by simulation trials on (hypothesis) models are loaded each into a `big' table $\bigcup_{i=1}^{p} R$. The uncertainty is then ``buried'' in the database, which lacks a logical organization for enabling data-driven hypothesis management and predictive analytics.}
\label{fig:etl-pipeline}
\end{center}
\end{figure}
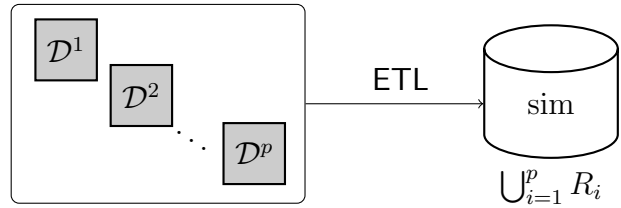

Overall, hypothesis data management is also OLAP-like, yet markedly different from simulation data management.

\begin{framed}
\noindent
A key point that distinguishes hypothesis management is that a fact or unit of data is defined by its \textbf{predictive content}. That is, every clear-cut predicted fact (w.r.t.$\!$ available data dependencies) is a claim. Accordingly, the data should be decomposed and organized for a claim-centered access pattern.
\end{framed}

\begin{figure}[H]
\begin{spacing}{1}
\begin{center}
\tikzstyle{rect1}=[rectangle,
                                    thick,
                                    minimum size=23pt,
                                    draw=black]
\tikzstyle{rect2}=[rectangle,
                                    thick,
                                    minimum size=23pt,
                                    fill=black!20,
                                    draw=black]
\tikzstyle{rect3}=[rectangle,
                                    rounded corners=3pt,
                                    minimum size=82pt,
                                    minimum width=110pt,
                                    draw=black]
\tikzstyle{box}=[rectangle,
                                    fill=none,
                                    draw=none]
\tikzstyle{cyl1}=[cylinder,
                                    thick,
                                    minimum size=30pt,
                                    inner sep=0pt,
                                    fill=none,
                                    draw=black]
\tikzstyle{cyl2}=[cylinder,
                                    thick,
                                    fill=black!20,
                                    minimum size=30pt,
                                    inner sep=0pt,
                                    draw=black]
\tikzstyle{edge} = [draw,thick,->,bend left]
\begin{tikzpicture}[scale=1.1]
    \node[rect3] (back) at (0,3) {};
    \node[rect1] (s) at (0,3.7) {$\mathcal{S}_k$};
    \node[rect2] (d1) at (-1.2,2.25) {$\mathcal{D}_k^1$};
    \node[rect2] (d2) at (-0.225,2.25) {$\mathcal{D}_k^2$};
    \node[box] (dot) at (0.45,2.25) {...};
    \node[rect2] (dn) at (1.15,2.25) {$\mathcal{D}_k^p$};
    \node[cyl1,rotate=90] (h) at (4.2,3) {\rotatebox[origin=c]{-90}{h}};
    \node[box] (hlabel) at (4.2,2.1) {$\bigcup_{k=1}^{n}H_k$};
    \node[cyl2,rotate=90] (y) at (7.3,3) {\rotatebox[origin=c]{-90}{y}};
    \node[box] (ylabel) at (7.3,2.1) {$\bigcup_{k=1}^{n}\bigcup_{\ell=1}^{m}Y_k^\ell$};
    \node[box] (cond) at (7.3,4) {\huge\rotatebox[origin=c]{180}{$\circlearrowleft$}};
    \draw[-] (s) to (d1);
    \draw[-] (s) to (d2);
    \draw[-] (s) to (dn);
    \draw[->] (back) to (h);
    \node[box] (etl) at (2.7,3.3) {\textsf{ETL}};    
    \draw[->] (h) to (y);
    \node[box] (etl) at (5.7,3.3) {\textsf{U-intro}};
    \node[box] (etl) at (7.3,4.5) {\textsf{conditioning}};
\end{tikzpicture}
\caption[Design-theoretic pipeline for processing hypotheses as uncertain and probabilistic data]{Pipeline for processing hypotheses as uncertain and probabilistic data. For each hypothesis $k$, its structure $\mathcal S_k$ is given in a machine-readable format, and all of its sample simulation data trials $\bigcup_{i=1}^p\mathcal{D}_k^i$ are indicated their target phenomenon, say $\phi$, to be loaded into a `big table' $H_k$. Then the synthesis comes into play to read a base of possibly very many hypotheses $\bigcup_{k=1}^n H_k$ and transform them into a probabilistic database where each hypothesis is decom- posed into claim tables $\bigcup_{\ell=1}^{m}Y_k^\ell$. A probability distribution is computed for each pheno\-menon $\phi$, covering all the hypotheses and their trials targeted at $\phi$. This distribution is then updated into a posterior in the presence of observational data.} 
\label{fig:pipeline}
\end{center}
\end{spacing}
\end{figure}
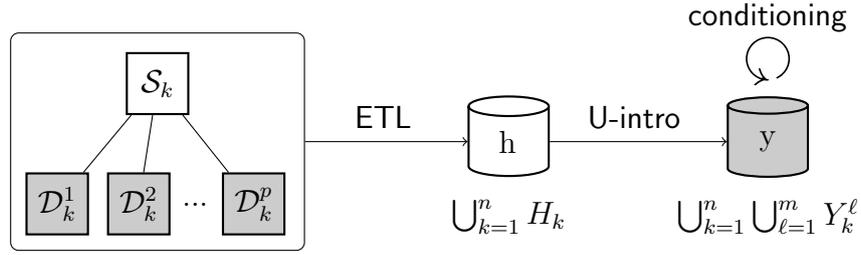

To anticipate Chapter \ref{ch:vision}, the synthesis method we have developed in this thesis work for processing hypotheses as uncertain and probabilistic data comprises a design-theoretic pipeline (see Fig. \ref{fig:pipeline}) that extends the one shown in Fig. \ref{fig:etl-pipeline}.

\subsection{Probabilistic database design}

Probabilistic databases (p-DB's) have evolved into mature technology in the last decade with the emergence of new data models and query processing techniques \cite{suciu2011}. One of the state-of-the-art probabilistic data models is the U-relational representation system with its probabilistic world-set algebra (p-WSA) implemented in \textsf{MayBMS} \cite{koch2009}. That is an elegant extension of the relational model we shall refer to in this thesis for the management of large-scale uncertain and probabilistic data. 

We look at U-relations from the point of view of p-DB design, for which no formal design methodology has yet been proposed. Despite the advanced state of probabilistic data management techniques, 
a lack of methods for the systematic design of p-DBs may prevent wider adoption. The availability of design methods has been considered one of the key success factors for the rapid growth of applications in the field of Graphical Models (GM's) \cite{darwiche2010}, considered to inform research in p-DB's \cite[p. $\!$14]{suciu2011}. Analogously, we have proposed to distinguish methods for p-DB design in three classes \cite{goncalves2014}: (i) \emph{subjective} construction, (ii) \emph{learning} from data, and (iii) \emph{synthesis} from other kind of formal specification. 

The first is the less systematic, as the user has to model for the data and correlations by steering all the p-DB construction process (\textsf{MayBMS}' use cases \cite{koch2009}, e.g., are illustrated that way). The second comprises analytical techniques to extract the data and learn correlations from external sources, possibly unstructured, into a p-DB under some ad-hoc schema. This is the prevalent one up to date, motivated by information extraction and data integration applications \cite[p. 10-3]{suciu2011}. In this thesis we present a methodology of the third kind, as we extract data dependencies from some previously existing formal specification (the hypothesis mathematical structure) to synthesize a p-DB algorithmically. Such a type of construction method has been successful, e.g., for building Bayesian Networks \cite{darwiche2010}. To our knowledge, this thesis is the first synthesis method for p-DB design (cf. \S\ref{sec:related-work-synthesis4u}). 

\begin{framed}
\noindent
We shall develop means to extract the specification of a hypothesis and encode it into a U-relational DB for data-driven hypothesis management and analytics. That is, we shall flatten deterministic hypotheses into U-relations. 
\end{framed}

The synthesis method that we have developed for p-DB's relies on the extraction of \emph{functional dependencies} (fd's; cf. \cite{ullman1988,abiteboul1995,maier1983}) that are basic input to algorithmic synthesis.\footnote{In fact, it has been considered a critical failure in traditional DB design the lack of techniques to obtain important information such as fd's in the real world \cite[p. 62]{badia2011}. 
} For an example of fd, consider relation \textsf{FALL} in Fig. \ref{fig:galileo}. There holds an fd $t \to \operatorname{v} s$, meaning that values of attribute time $t$ functionally determine values of both attributes velocity $\operatorname{v}$ and position $s$. More precisely, let $\mu$ and $\tau$ be any two tuples (rows) in an instance of relation (table) \textsf{FALL}. Then it satisfies fd $t \to \operatorname{v} s$ iff $\mu[t]=\tau[t]$ implies $\mu[\operatorname v s]=\tau[\operatorname v s]$. In our illustrative relation \textsf{FALL}, that fd is, in particular, a \emph{key constraint}, which means that (values of) $t$ play the role of a key to (provide access to the values of) $\operatorname v$ and $s$ in the relation. 

A related concept which is also a major one for us is that of \emph{normalization} \cite{ullman1988,abiteboul1995,maier1983}, viz., to ensure that the DB resulting from a design process bears some desirable properties which are associated with some notion of \emph{normal form} (ibid.). For hypothesis management, the uncertainty has to be modeled and should be normalized so that the uncertainty of one claim may not be undesirably mixed with the uncertainty of another claim. It is expected to involve a processing of the causal dependencies implicit in the given hypothesis structure. We shall introduce in detail such concepts in context when necessary.

\subsection{Structural equations}\label{subsec:structural}
\noindent
The flattening of the user mathematical models into hypothesis p-DB's, nonetheless, is not straightforward. It has been a goal of this thesis to investigate proper abstractions on mathematical models in order to (partly) capture their semantics, viz., to an extent that is tailored for hypothesis management (as opposed to, say, model solving). We shall abstract mathematical models into intermediary artifacts that are amenable to be further encoded into fd's.

In fact, given a system of equations with a set of variables appearing in them, in a seminal article Simon introduced an asymmetrical, functional relation among variables that establishes a (so-called) \emph{causal ordering} \cite{simon1953}. That became known as \emph{structural equation models} (SEM's) or just `structural equations' (cf. also \cite{pearl2000}). Along these lines, our goal is to extract the causal ordering implicit in the structure of a deterministic hypothesis into a set of fd's that guides our synthesis of U-relational DB's. As we shall see throughout this text, 

\begin{framed}
\noindent
the causal ordering we capture and process through fd's provides causal dependencies implicit in the predictive data that are very useful information to decompose uncertainty for the sake of probabilistic modeling and reasoning. 
\end{framed}

\subsection{Uncertainty Model}\label{subsec:uncertainty}
\noindent
In uncertain and probabilistic data management, there are essentially two sources of uncertainty: \emph{incompleteness} (missing data), and \emph{multiplicity} (inconsistent data). 

\begin{framed}
\noindent
The kind of uncertainty that is dealt with in this work is the multiplicity of hypothesis trial records identified to be targeted at the same phenomenon record. That is, the uncertainty arises from the existance of competing hypotheses. If multiple hypotheses and trials are inserted for the same phenomenon, the system interprets it as defining a probability distribution.
\end{framed}

Such a probability distribution (usually uniform) on the multiplicity of competing hypotheses is in accordance with probability theory under possible-worlds semantics \cite[Ch. 1]{suciu2011}. It is modeled into the U-relational data model and its p-WSA operators, and implemented into the \textsf{MayBMS} system as we shall see in \S\ref{sec:u-relations}.\footnote{Our own system of hypothesis management is to be delivered on top of the \textsf{MayBMS} backend.} The \textsf{conf()} aggregate operator, for instance, in spite of the name, performs standard (non-Bayesian) probabilistic inference on such probability distribution. Eventually, however, there is a need to condition the initial probability distribution in the presence of observations. For the conditioning, then, we shall adopt Bayesian inference so that the prior probability distribution can be updated to a posterior.

The informal discussion of this section opens the way for a number of technical research questions that we outline next.

\begin{itemize}
\item[\textbf{RQ5}.]\label{rq5}
Is there an algorithm to, given a SEM, efficiently extract its causal ordering? What are the computational properties of this problem?

\item[\textbf{RQ6}.]\label{rq6} 
What is the connection between SEM's and fd's? Can we devise an encoding scheme to `orient equations' and then effectively transform one into the other with guarantees? Once we do it, what design-theoretic properties have such a set of fd's? 

\item[\textbf{RQ7}.]\label{rq7} 
Is  such fd set ready to be used for p-DB schema synthesis as an encoding of the hypothesis causal structure? If not, what kind of further processing we have to do? Can we perform it efficiently by reasoning directly on the fd's? How does it relate to the SEM's causal ordering? 

\item[\textbf{RQ8}.]\label{rq8} 
Is the uncertainty decomposition required for predictive analytics reducible to the structure level (fd processing), or do we need to process the simulated data to identify additional uncertainty factors? Finally, what properties are desirable for a p-DB schema targeted at hypothesis management? Are they ensured by this synthesis method?

\item[\textbf{RQ9}.]\label{rq9} 
Given all such a design-theoretic machinery to process hypotheses into \mbox{(U-)relational} DB's, what properties can we detect on the hypotheses back at the conceptual level? Do we have now technical means to speak of hypotheses that are ``good'' in terms of principles of the philosophy of science?
\end{itemize}

\noindent
The core of this thesis is devoted to answer these questions, and we shall accomplish it throughout Chapters \ref{ch:encoding}, \ref{ch:reasoning} and \ref{ch:synthesis4u}.

\section{Thesis Statement}\label{sec:statement}
The statement of this thesis is that it is possible to effectively encode and manage large deterministic scientific hypotheses as uncertain and probabilistic data. Its key challenges are of both conceptual and technical nature. Conceptually, we provide core, non-obvious abstractions to define and encode hypotheses as data. Technically, we provide a number of algorithms that compose a design-theoretic pipeline to encode hypotheses as uncertain and probabilistic data, and verify their efficiency and correctness. The applicability and effectiveness of our method is demonstrated in realistic case studies in computational science.

Besides, it is worthwhile highlighting some non-goals of this thesis.
\begin{itemize}
\item[\textbf{N1}.] Although we perform some sort of \emph{information extraction} \cite{chang2006} for the acquisition of hypotheses from some model repositories on the web, it is very basic and ad-hoc in order to obtain a testbed for our method. That is, we are not proposing means for the systematic extraction of hypotheses from available sources. In fact we shall outline it in \S\ref{sec:future-work} as an important direction of future work.

\item[\textbf{N2}.] We do not address \emph{solving} computational models or \emph{numerical analytics} in any sense. In fact we rely on the numerical solvers (implemented into tools that we use) as `transaction processing' systems, load their computed data into a relational `big' fact table and then render it into U-relational tables synthesized by our method. We do not deal with \emph{data visualization} either in any sense.

\item[\textbf{N3}.] The efficiency and scalability of \emph{query processing} in p-DB's, in particular U-relational's \textsf{MayBMS} and its p-WSA (which we rely on) is \emph{not} addressed or evaluated in this thesis. In fact, the performance of U-relations and p-WSA has been extensively evaluated and shown to be effective \cite{antova2008,koch2009}. All performance tests carried out in this thesis comprise our design-theoretic techniques for the encoding and synthesis of U-relational hypothesis databases. 

\item[\textbf{N4}.] In terms of \emph{uncertainty} and \emph{statistical analysis}, we stick to (i) process some well-defined forms of multiplicity in the data which constitute the model of uncertainty dealt with in this work; then (ii) by relying on \textsf{MayBMS} we perform probabilistic inference; and (iii) eventually (at application level) we perform Bayesian inference and so that a posterior probability distribution is propagated through p-DB updates. We do not provide any additional form of \emph{uncertainty management}. Rather, we manage the data extracted into the system (under user control) and process its uncertainty in terms of the specific sources of uncertainty recognized in $\Upsilon$-DB (cf. Chapter $\!$\ref{ch:vision}).
\end{itemize}

\section{Thesis Contributions}\label{sec:contributions}

The contributions of this thesis are outlined as follows.

\subsection{Innovative Contributions}\label{subsec:vision-contributions}

This thesis presents the vision of hypotheses as data (and its use case) so-called $\!\Upsilon\!$-DB vision. It has been published in the vision track of VLDB 2014 \cite{goncalves2014}, for its (sic.) potentially high-impact visionary content. The innovative system of $\!\Upsilon\!$-DB has been described in a `system prototype demonstration' paper \cite{goncalves2015b}.\footnote{Preliminary version available at \href{http://arxiv.org/abs/1411.7419}{CoRR abs/1411.7419}.}

\subsection{Technical Contributions}\label{subsec:technical-contributions}

This thesis presents specific technical developments over the $\!\Upsilon\!$-DB vision. In short, it shows how to encode deterministic hypotheses as uncertain and probabilistic data. Our detailed technical contributions (cf. Chapters \ref{ch:encoding}, \ref{ch:reasoning}, and \ref{ch:synthesis4u}) are formulated into a formal method for the design of hypothesis p-DB's which is described in a technical report \cite{goncalves2015a}.\footnote{Preliminary version available at \href{http://arxiv.org/abs/1411.5196}{CoRR abs/1411.5196}.} 
The method, together with our realistic testbed scenarios and performance evaluation, are yet to be published.

\section{Thesis Outline}\label{sec:outline}

The structure of the remainder of this thesis is outlined for reference.
\vspace{12pt}

\noindent
\textbf{Chapter \ref{ch:vision}}. [$\boldsymbol \Upsilon$\textbf{-DB Vision}]. The research vision of hypotheses as (uncertain and probabilistic) data, the characterization of its use case, key points and technical challenges are presented.
\vspace{12pt}

\noindent
\textbf{Chapter \ref{ch:encoding}}. [\textbf{Encoding}]. The problem of encoding a hypothesis `as data' given its formal specification (set of mathematical equations) is presented and addressed by an encoding scheme that transforms the equations into fd's with guarantees in terms of preserving the hypothesis causal structure.
\vspace{12pt}

\noindent
\textbf{Chapter \ref{ch:reasoning}}. [\textbf{Causal Reasoning}]. It is presented a technique for causal reasonig as acyclic pseudo-transitive reasoning over the encoded fd's. It processes the hypothesis causal ordering to find the `first causes' for each of its predictive variables.
\vspace{-5pt}

\noindent
\textbf{Chapter \ref{ch:synthesis4u}}. [\textbf{p-DB Synthesis}]. It is presented a technique to address the problem of uncertainty introduction and propagation for the transformation of hypotheses into U-relational databases. The synthesized U-database is shown to bear desirable properties for hypothesis management and predictive analytics.
\vspace{12pt}

\noindent
\textbf{Chapter \ref{ch:applicability}}. [\textbf{Applicability}]. A discussion of applicability, the implementation of the proposed techniques into a prototype system for test and demonstration of the vision realization through realistic case studies are presented.
\vspace{12pt}

\noindent
\textbf{Chapter \ref{ch:conclusions}}. [\textbf{Conclusions}]. Research questions are revisited, and the significance and limitations of the thesis with directions to future work and final considerations are discussed.

\chapter{Vision: Hypotheses as Data}\label{ch:vision}

High-throughput technology and large-scale scientific experiments provide scientists with \emph{empirical} data that has to be extracted, transformed and loaded before it is ready for analysis \cite{hey2009}. In this vision we consider \emph{theoretical} data, or data generated by simulation from deterministic scientific hypotheses, which also needs to be pre-processed to be analyzed. %

\textbf{Hypotheses as data}. In view of the age of data-driven science, we consider deterministic scientific hypotheses from a multi-fold point of view: formed as principles or learned in large scale,\footnote{As exhibited, e.g., in the \textsf{Eureqa} project \cite{schmidt2009}.} hypotheses are formulated mathematically and coded in a program that is run to give their \emph{decisive} form of data (see Fig. \ref{fig:galileo}).

\textbf{Uncertain data.} $\!$The semantic structure of relation \textsf{FALL} (Fig. \ref{fig:galileo}), item (iv) can be expressed by the functional dependency (fd) $t \to \operatorname{v}\,s$. $\!$This is typical semantics assigned to empirical data in the design of experiment databases. $\!$A space-time dimension (like time $t$ in our example) is used as a key to observables (like velocity $\operatorname{v}$ and position $s$). In \emph{empirical} uncertainty, it is such ``physical'' dimension keys like $t$ that may be violated, say, by alternative sensor readings.

Hypotheses, however, are tentative explanations of phenomena \cite{losee2001}, which characterizes a different kind of uncertain data. In order to manage such \emph{theoretical} uncertainty, we shall need two special attributes to compose, say, the epistemological dimension of keys to observables: $\phi$, identifying the studied phenomena; and $\upsilon$, identifying the hypotheses aimed at explaining them. That is, we shall leverage the semantics of relations like \textsf{FALL} to $\phi\, \upsilon\, t \to \operatorname v s$. This leap is a core abstraction in this vision of $\Upsilon$-DB.

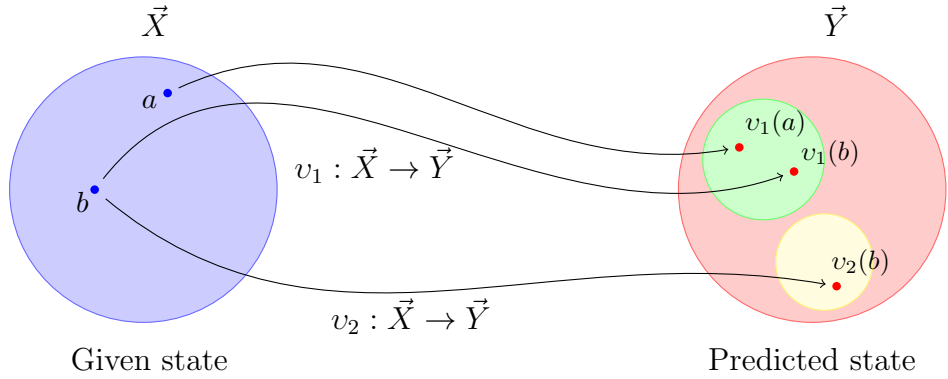
\begin{figure}[t]
\begin{center}
\begin{tikzpicture}[scale=1.6]
    \filldraw[fill=blue!20, draw=blue!60] (-1.5,0) circle (1.1cm);
    \filldraw[fill=red!20, draw=red!60] (4.0,0) circle (1.1cm);
    \filldraw[fill=green!20, draw=green!60] (3.6,0.25) circle (0.5cm);
    \filldraw[fill=yellow!15, draw=yellow!70] (4.1,-0.6) circle (0.4cm);

    \node at (-1.4,1.4) {$\vec{X}$};
    \node at (4.2,1.4) {$\vec{Y}$};
    \node at (-1.45,-1.4) {Given state};
    \node at (4.0,-1.4) {Predicted state};

    \node at (-1.45,0.73) {\small{$a$}};
    \node at (-2.0,-0.1) {\small{$b$}};
    \node at (3.7,0.55) {\footnotesize{$\upsilon_1(a)$}};
    \node at (4.15,0.3) {\footnotesize{$\upsilon_1(b)$}};
    \node at (4.4,-0.6) {\footnotesize{$\upsilon_2(b)$}};
    \node at (0.4,0.2) {$\upsilon_1: \vec{X} \to \vec{Y}$};
    \node at (0.7,-1.05) {$\upsilon_2: \vec{X} \to \vec{Y}$};

    \node (x1) at (-1.3,0.8) {};
    \node (x2) at (-1.9,0.0) {};
    \node (y1) at (3.4,0.35) {};
    \node (y2a) at (3.85,0.15) {};
    \node (y2) at (4.2,-0.8) {};

    \fill[blue] (x1) circle (1pt);
    \fill[blue] (x2) circle (1pt);
    \fill[red] (y1) circle (1pt);
    \fill[red] (y2a) circle (1pt);
    \fill[red] (y2) circle (1pt);

    \draw[->] (x1) to[out=25,in=190] (y1);
    \draw[->] (x2) to[out=52,in=200] (y2a);
    \draw[->] (x2) to[out=-40,in=170] (y2);
\end{tikzpicture}
\caption[Scientific hypotheses seen as alternative functions to predict data.]{Deterministic scientific hypotheses seen as alternative functions to predict data, giving rise to both theoretical and empirical sources of uncertainty.}\label{fig:functions}
\end{center}
\vspace{-10pt}
\end{figure}

\textbf{Predictive data.} Scientific hypotheses are tested by way of their predictions \cite{losee2001}. In the form of mathematical equations, hypotheses symmetrically relate aspects of the studied phenomenon. However, for computing predictions, deterministic hypotheses are applied asymmetrically as \emph{functions} \cite{simon1966}. 
They take a given valuation over input variables (parameters) to produce values of output variables (predictions). By observing that, we shall seek a principled method to transform the (symmetric) mathematical equations of a hypothesis into (asymmetric) fd's. 

By looking at deterministic hypotheses as alternative functions to predict data (see Fig. \ref{fig:functions}), in this vision we shall deal with two sources of uncertainty. Given a well-defined context with a set of alternative hypotheses aimed at explaining (providing predictions for) a selected phenomenon:

\begin{itemize}
\item \emph{Theoretical} uncertainty,\footnote{That is, multiplicity of hypothesis entries associated with a phenomenon.} comprises selecting the best tentative model (function) to produce (the best) data?
\item \emph{Empirical} uncertainty,\footnote{That is, multiplicity of hypothesis trial entries associated with a phenomenon.} comprises, for each candidate model, what is the (parameter) input setting that calibrates it the best way for the selected phenomenon?
\end{itemize}

Note that these two sources of uncertainty are intertwined in that one cannot `clean' one without cleaning the other --- neither theory nor parameters are directly observable, but only their joint results (the predictions) \cite{losee2001}. In this thesis we aim at providing means to support such kind of `integrated' analytics.

\textbf{Applications.} Big computational science research programs such as the Human Brain Project,\footnote{\url{http://www.humanbrainproject.eu/}.} or Cardiovascular Mathematics,\footnote{\url{http://icerm.brown.edu/tw14-1-pdecm}.} are highly-demanding applications challenged by such theoretical (big) data. Users need to analyze results of hundreds to thousands of data-intensive simulation trials.

Besides, recent initiatives on web-based model repositories have been fostering large-scale model integration, sharing and reproducibility in the computational sciences (e.g., \cite{hines2004,chelliah2013,hunter2003}). They are growing reasonably fast on the web, (i) promoting some \textsf{MathML}-based standard for model specification, but (ii) with limited integrity and lack of support for rating/ranking competing models. For those two reasons, they provide a strong use case for our vision of hypothesis management. The Physiome project \cite{hunter2003,bassingthwaighte2000}, e.g., is planned to integrate several large deterministic models of human physiology --- a fairly simple model of the human cardiovascular system, e.g., has about 600+ variables.

Also, there is a pressing call for \emph{deep} predictive analytic tools to support users assessing what-if scenarios in business enterprises \cite{haas2011}. Deep predictive analytics are based on first principles (deterministic hypotheses) and go beyond descriptive analytics or shallow predictive analytics such as statistical forecasting (ibid.).\footnote{The concept of `deep' predictive analytics is from Haas et al. \cite{haas2011}, and is discussed in more detail in \S\ref{subsec:models-and-data}.}

\textbf{U-relations.} 
All that ratifies that \emph{hypothesis management} is a promising class of applications for probabilistic DB's. The vision of \mbox{$\Upsilon$-DB} is currently set to be delivered on top of U-relations and probabilistic world-set algebra (p-WSA) \cite{koch2009}. These were developed in the influential \textsf{MayBMS} project.\footnote{Project website: \url{http://maybms.sourceforge.net/}. \textsf{MayBMS} is as a backend extension of \textsf{PostgreSQL}. It offers all the traditional querying capabilities of the latter in addition to the uncertain and probabilistic's.} $\!\!\!$As implied by some of its design principles, viz., compositionality and the ability to introduce uncertainty, \textsf{MayBMS}' query language fits well to hypothesis management. We shall look at it, as previously mentioned, from the point of view of a synthesis method for p-DB design. 
We shall particularly make use of the \texttt{repair key} operation, which gives rise to alternative worlds as maximal-subset repairs of an argument key.

\textbf{Predictive analytics tool.} In database research (e.g., \cite{rahm2000}), uncertainty is usually seen as an undesirable property that hinders data quality. We shall refer to U-relations and p-WSA as implemented in \textsf{MayBMS}, nonetheless, to show that the ability to introduce \emph{controlled} uncertainty into an (otherwise complete) simulation dataset can be a tool for `deep' predictive analytics on a set of competing or alternative hypotheses. Fig.\ref{fig:study} shows such a scenario of hypotheses `as data' compete to explain a phenomenon `as data.'

\begin{figure}[t]
\centering
\includegraphics[keepaspectratio,width=\textwidth]{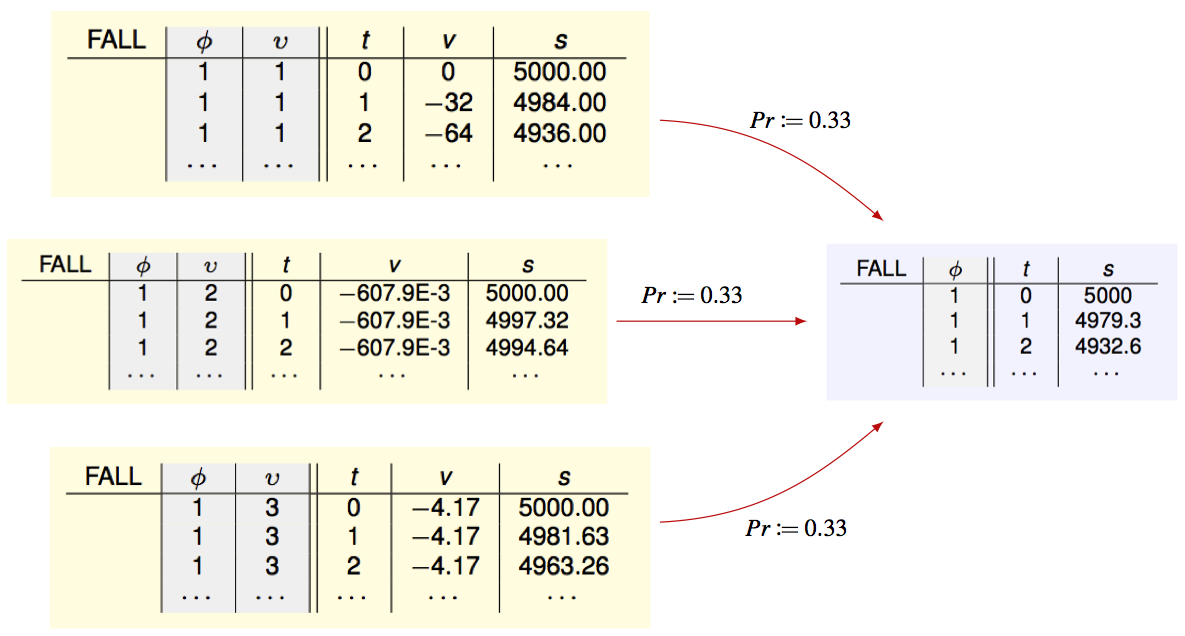}
\caption[Predictive analytics in a data-intensive hypothesis evaluation study.]{Predictive analytics in a data-intensive hypothesis evaluation study: hypotheses (simulated data) compete to explain a phenomenon (observed data).}
\label{fig:study}
\end{figure}

As a roadmap to most of the remainder of this chapter, we claim that if hypotheses can be encoded and identified (see \S \ref{sec:encoding}), and their uncertainty quantified by some probability distribution (see \S \ref{sec:transformation}), then they can be rated/ranked and browsed by the user under selectivity criteria. Furthermore, their probabilities can be conditioned for possibly being re-ranked in the presence of evidence (see \S \ref{sec:analytics}).

\section{Running Example}\label{sec:use-case}

Let us consider Example \ref{ex:fall} for the presentation of the vision.

\begin{myex}\label{ex:fall}
A research is conducted on the effects of gravity on a falling object in the Earth's atmosphere. Scientists are uncertain about the precise object's density and its predominant state as a fluid or a solid. Three hypotheses are then considered as alternative explanations of the fall (see Fig. \ref{fig:descriptive}). Due to parameter uncertainty, six simulation trials are run for $\mathcal H_1$, and four for $\mathcal H_2$ and $\mathcal H_3$ each. $\Box$
\end{myex}

\begin{spacing}{1}
\begin{figure}[H]
   \centering
\begingroup\setlength{\fboxsep}{5pt}
\colorbox{blue!5}{%
   \begin{tabular}{c|c|p{0.48\linewidth}}
  \textsf{PHENOMENON} & $\phi$ & \textsf{Description}\\
      \hline    
   & $1$ & Effects of gravity on an object falling in the Earth's atmosphere.\\
   \end{tabular}
}\endgroup
\vspace{6pt}
\begingroup\setlength{\fboxsep}{5pt}
\colorbox{yellow!15}{%
   \begin{tabular}{c|c|l}
  \textsf{HYPOTHESIS} & $\upsilon$ & \textsf{Name}\\
      \hline    
   & $1$ & Law of free fall\\
   & $2$ & Stokes' law\\
   & $3$ & Velocity-squared law\\
   \end{tabular}
}\endgroup
\caption{Descriptive (textual) data of Example \ref{ex:fall}.}
\label{fig:descriptive}
\end{figure}
\end{spacing}

\begin{spacing}{1}
\begin{figure}[H]
\centering
\begingroup\setlength{\fboxsep}{3pt}
\colorbox{blue!5}{%
   \begin{tabular}{c|>{\columncolor[gray]{0.92}}c||c|c|c|c|c|c|c|c|c}
  $H_1$ & \textsf{tid} & $\phi$ & $\upsilon$ & $t$ & $g$ & $\operatorname{v_0}$ & $s_0$ & $a$ & $\operatorname{v}$ & $s$\\
      \hline    
   & $1$ & $1$ & $1$ & $0$ & $32.0$ & $0$ & $5000$ & $-32.0$ & $0$ & $5000$\\
   & $2$ & $1$ & $1$ & $0$ & $32.0$ & $10$ & $5000$ & $-32.0$ & $10$ & $5000$\\
   & $3$ & $1$ & $1$ & $0$ & $32.0$ & $20$ & $5000$ & $-32.0$ & $20$ & $5000$\\
   & $4$ & $1$ & $1$ & $0$ & $32.2$ & $0$ & $5000$ & $-32.2$ & $0$ & $5000$\\
   & $5$ & $1$ & $1$ & $0$ & $32.2$ & $10$ & $5000$ & $-32.2$ & $10$ & $5000$\\
   & $6$ & $1$ & $1$ & $0$ & $32.2$ & $20$ & $5000$ & $-32.2$ & $20$ & $5000$\\
   \cline{2-11}
   & $1$ & $1$ & $1$ & $0.1$ & $32.0$ & $0$ & $5000$ & $-32.0$ & $-3.2$ & $4999.84$\\
   & $2$ & $1$ & $1$ & $...$ & $32.0$ & $0$ & $5000$ & $-32.0$ & $...$ & $...$\\
   \end{tabular}
}\endgroup
\caption[`Big' fact table $H_1$ loaded with simulation raw data]{`Big' fact table $H_1$ of hypothesis $\upsilon=1$ loaded with simulation raw data: trials on $H_1$ are identified by \textsf{tid}.}
\label{fig:fact-table}
\vspace{-6pt}
\end{figure}
\end{spacing}

\noindent
The construction of $\Upsilon$-DB, a Data Warehouse (DW), requires a simple user description of a research. That is, descriptive records of the phenomena and hypotheses dimensions (see Fig. \ref{fig:descriptive}) are to be inserted first such that basic referential constraints are satisfied by their associated datasets (fact tables). For instance, each one of the six trial datasets for hypothesis $\mathcal H_1$ shall reference its id $\upsilon=1$ as a foreign key from table \textsf{HYPOTHESIS} further in their synthesized relations. 

Fig. \ref{fig:fact-table} shows the `big' fact table $H_1$ for hypothesis $\upsilon\!=\!1$ loaded with its trial datasets for phenomenon $\phi=\!1$. Although table $H_1$ is denormalized for faster data retrieval as usual in DW's, the extraction of the hypothesis equations allows to render it automatically since all variables must appear in some equation. Now we proceed to the hypothesis encoding and start to address research questions RQ1-4.

\section{Hypothesis Encoding}\label{sec:encoding}

We aim at extracting, for each hypothesis, a set of fd's from its mathematical equations. Suppose we are given a set of equations of hypothesis $\mathcal H_1$ below, and let us examine the set $\Sigma_{1}$ of fd's we target at.\footnote{Recall that a rigorous presentation of the method to encode fd set $\Sigma_{1}$ is due by Chapter \ref{ch:encoding}.}

\begin{center}
\begingroup\setlength{\fboxsep}{2pt}
\colorbox{gray!7}{%
\begin{tabular}{l l l}
\advance\leftskip-3cm
$\;\mathcal H_1$.&&$\!\!\!\!\!\!\!\!\!\!$ \textsf{Law of free fall}\\
\hline
$\!a(t)$ &$\!\!\!$=$\!\!\!$& $g$\\
$\!\operatorname{v}(t)$ &$\!\!\!$=$\!\!\!$& $-g t \,+\, \operatorname{v_0}$\\
$\!s(t)$ &$\!\!\!$=$\!\!\!$& $-(g/2)t^2 \,+\, \operatorname{v_0}t \,+\, s_0$
\end{tabular}
}\endgroup
\end{center}

\begin{spacing}{1}
\begin{eqnarray*}\label{sigma1}
\Sigma_{1} = \{\quad 
\phi &\to& g,\\ 
\phi &\to& \operatorname{v_0},\\ 
\phi &\to& s_0,\\ 
g\,\upsilon &\to& a,\\
g\,\operatorname{v_0}\,t\,\upsilon &\to& \operatorname{v},\\
g\,\operatorname{v_0}\,s_0\,t\,\upsilon &\to& s \quad\}.
\end{eqnarray*}
\end{spacing}

\noindent
In order to derive $\Sigma_{1}$ from the equations of $\mathcal H_1$, we focus on their implicit data dependencies and get rid of constants and possibly complex mathematical constructs. Equation $\operatorname{v}(t) \!=\! -gt \!+\! \operatorname{v_0}$, e.g., written this way (roughly speaking), suggests that $\operatorname{v}$ is a prediction variable \emph{functionally dependent} on $t$ (the physical dimension), $g$ and $\operatorname{v_0}$ (the parameters). Yet a dependency like $\,g\operatorname{v_0} t \to \operatorname{v}\,$ may hold for infinitely many equations.\footnote{Think of, say, how many polynomials satisfy that dependency `signature.'} In fact, we need a way to identify $\mathcal H_1$'s mathematical formulation precisely, i.e., an abstraction of its data-level semantics. This is achieved by introducing hypothesis id $\upsilon$ as a special attribute in the fd (see $\Sigma_{1}$). 

\begin{framed}
\noindent
This is a \emph{data representation} of a deterministic scientific hypothesis. It is built into an encoding scheme (see \S\ref{sec:scheme}) that leverages the semantics of structural equations.
\end{framed}

\noindent
The other special attribute, the phenomenon id $\phi$, is supposed to be a key to the values of parameters, i.e., determination of parameters is an empirical, phenomenon-dependent task. The fd $\phi \to g\operatorname{v_0} s_0\,$ is to be (expectedly) violated when the user is uncertain about the values of parameters. 
%
The same rationale applies to derive $\Sigma_{2} \!=\! \Sigma_{3}$ from the equa\-tions of $\mathcal H_2$, $\mathcal H_3$ below. These, n.b., vary in structure w.r.t. $\mathcal H_1$ (e.g., they include parameter $D$, the object's diameter). 

\begin{center}
\begingroup\setlength{\fboxsep}{2pt}
\colorbox{gray!7}{%
\begin{tabular}{l | l}
\advance\leftskip-3cm
$\!\mathcal H_2$. \textsf{Stokes' law} & $\!\mathcal H_3$. \textsf{Velocity-squared law}\\
\hline
$a(t) \,=\, 0$ & $a(t) \,=\, 0$\\
$\operatorname{v}(t) \,=\, -\sqrt{gD/\,4.6\!\times\! 10^{-4}}$ & $\operatorname{v}(t) \,=\, -gD^2/\,3.29\!\times\! 10^{-6}$\\
$s(t) \,=\,  -t\,\sqrt{gD/\,4.6\!\times\! 10^{-4}} \!+\! s_0$ & $s(t) \,=\,  -(gD^2/\,3.29\!\times\! 10^{-6})\,t \!+\! s_0\!\!$\\
\end{tabular}
}\endgroup
\end{center}

\begin{spacing}{1}
 \begin{eqnarray*}
\Sigma_{2} = \Sigma_{3} = \{\quad 
\phi &\to& g,\\ 
\phi &\to& D,\\ 
\phi &\to& s_0,\\ 
\phi &\to& a,\\
g\,D\,\upsilon &\to& \operatorname{v},\\
g\,D\,s_0\,t\,\upsilon &\to& s \quad\}.
\end{eqnarray*}
\end{spacing}

The key point here is that, if the hypothesis structure (set of equations) is given in a machine-readable format for mathematics, then the method to extract the hypothesis fd set from its equations can be carefully designed based on such hypothesis data representation abstraction. In fact, we shall explore W3C's \textsf{MathML} as a format for hypothesis specification.\footnote{\url{http://www.w3.org/Math/}.}

\section{Reasoning over FD's}\label{sec:reasoning}
\noindent
Once each hypothesis fd set has been extracted, some reasoning is to be performed to discover implicit data dependencies. In fact, dependency theory is equipped with a formal system (cf. \S\ref{sec:armstrong}) for reasoning over fd sets like $\Sigma_1$ and derive other fd's in its closure $\Sigma_1^+$. As we elaborate on in Chapter \ref{ch:reasoning}, we shall be particularly concerned with the pseudo-transitivity inference rule. Applied over fd's $\{\,\phi \to g,\;\;g\, \upsilon \to a\,\} \subset \Sigma_1$, for instance, it gives us $\phi\, \upsilon \to a\,$. This inference allows us to observe that $\{g\}$ is a `factor' on the uncertainty of $a$, but $\{\phi\,\upsilon\}$ should be a dimensional \emph{key constraint} for values of $a$. 

In fact, note that derived fd's like $\langle \phi\, \upsilon,\, a\rangle \in \Sigma_1^+$, which should be a constraint on values of $a$ in $H_1$, are (expectedly) violated in the presence of uncertainty: observe in Fig. \ref{fig:fact-table} the multiplicity $\{32.0, 32.2\}$ of values of $a$ under the same pair \mbox{$(\phi\mapsto 1,\, \upsilon\mapsto1)$}, which should functionally determine them in $H_1$. 
For that reason we admit a special attribute `trial id' \textsf{tid} to be overimposed into $H_1$ for a trivial repair, provisionally, until uncertainty can be introduced in a controlled way by synthesis `4U.'
It is meant to identify simulation trials and ``pretend'' certainty not to lose the integrity of the data. It is under this imposed certainty that the raw simulation trial data is safely loaded from files (see Fig. \ref{fig:fact-table}). Note, however, how `certainty' is held at the expense of redundancy and, mostly important, opaqueness for predictive analytics (since \textsf{tid} isolates or hides the inconsistency w.r.t. to the violated constraints). This is until the next stage of the $\Upsilon$-DB construction pipeline, when uncertainty is to be introduced in a controlled manner.

\section{Uncertainty Introduction}\label{sec:transformation}

Before we proceed to the uncertainty introduction procedure, note in relation $H_1$ (Fig. \ref{fig:fact-table}), that the predicted acceleration values $a$ are such that an association between the hypothesis and a target phenomenon, viz., $(\phi\mapsto1, \upsilon\mapsto1)$ is established.
In fact, as of the insertion of each hypothesis trial dataset, the user must set for it a target phenomenon. This may be non-obvious but is quite convenient a design decision for the envisioned system of $\Upsilon$-DB because hypotheses, as (abstract) universal statements \cite{losee2001}, can only be derived predictions from (be empirically grounded) by assigning (callibrating) them onto some real-world phenomenon. 
This assignment is set at data entry time because in fact it only holds at the data level.\footnote{Hypotheses are `universal' by definition \cite{losee2001}. They (must) qualify for a \emph{class} of different situated phenomena, while its predictive datasets must be very specific (for one specific situation).} It is to be recorded in an `explanation' table named $H_0$ by default (see Fig. \ref{fig:explanation}, top), being provided with weights for establishing a prior probability distribution which (by user choice) may or may not be uniform.

\begin{figure}[t]
\advance\leftskip0.5cm
\begingroup\setlength{\fboxsep}{2pt}
\colorbox{blue!8}{%
   \begin{tabular}{c|c||c|c}
  $H_0$ & $\phi$ & $\upsilon$ & \textsf{Conf}\\
      \hline    
   & $1$ & $1$ & 3\\
   & $1$ & $2$ & 1\\
   & $1$ & $3$ & 1\\
   \end{tabular}
}\endgroup
\hspace{3pt}
\begingroup\setlength{\fboxsep}{2pt}
\colorbox{yellow!15}{%
   \begin{tabular}{c|>{\columncolor[gray]{0.92}}c|c||c}
  $Y_0$ & $V \mapsto D$ & $\phi$ & $\upsilon$\\
      \hline    
   & $\textsf{x}_0 \mapsto 1$ & $1$ & $1$\\
   & $\textsf{x}_0 \mapsto 2$ & $1$ & $2$\\
   & $\textsf{x}_0 \mapsto 3$ & $1$ & $3$\\
   \end{tabular}
}\endgroup
\hspace{3pt}
\begingroup\setlength{\fboxsep}{2pt}
\colorbox{yellow!15}{%
   \begin{tabular}{c|>{\columncolor[gray]{0.92}}c||c}
  $W$ & $V \mapsto D$ & \textsf{Pr}\\
      \hline    
   & $\textsf{x}_0 \mapsto 1$ & $.6$\\
   & $\textsf{x}_0 \mapsto 2$ & $.2$\\
   & $\textsf{x}_0 \mapsto 3$ & $.2$\\
   \end{tabular}
}\endgroup
\caption[`Explanation' relational table $H_0$.]{`Explanation' relational table $H_0$ and its associated U-relational table $Y_0$ rendered by application of the \textsf{repair-key} operation.}
\label{fig:explanation}
\end{figure}

The data transformation of `certain' to `uncertain' relations then starts with query $Q_0$, whose result set is materialized into U-relational table $Y_0$ (see Fig. \ref{fig:explanation}). As we introduce in detail in \S\ref{sec:u-relations}, U-relations have in their schema a set of pairs $(V_i, D_i)$ of \emph{condition columns} \cite{koch2009} to map each discrete random variable $\textsf{x}_i$ created by the \textsf{repair-key} operation to one of its possible values (e.g., $\textsf{x}_0 \mapsto 1$). The world table $W$ is internal to \textsf{MayBMS}' and automatically stores their marginal probabilities. The formal semantics the \textsf{repair-key} operation is given in \S\ref{sec:u-relations}.

\begin{framed}
\noindent
$Q_0$. \textsf{\textbf{create table} $Y_0$ \textbf{as select} $\phi,\,\upsilon$ \textbf{from} (\textbf{repair key} $\phi$ \textbf{in} $H_0$ \textbf{weight by} Conf);}
\end{framed}

The possible-world semantics of p-DB's (cf. \S\ref{sec:u-relations}) can be seen as a generalization of \emph{data cleaning}. In the context of p-DB's \cite{suciu2011}, data cleaning does not have to be one-shot --- which is more error-prone \cite{beskales2009}. Rather, it can be carried out gradually, viz., by keeping all mutually inconsistent tuples under a probability distribution (ibid.) that can be updated in face of evidence until the probabilities of some tuples eventually tend to zero to be eliminated. This motivates Remark \ref{rmk:method}.

\begin{myremark}
\label{rmk:method}
Consider U-relational table $Y_0$ (Fig. \ref{fig:explanation}). Note that it abstracts the goal of a data-intensive hypothesis evaluation study, or the scientific method itself \cite{losee2001}, as the repair of each $\phi$ as a key. That is, in \mbox{$\Upsilon$-DB} users can develop their research directly upon data with support of query and update capabilities to rate/rank their hypotheses $\upsilon$ w.r.t. each $\phi$, until the relationship $r(\phi,\, \upsilon)$ is repaired to be a \emph{function} $f\!: \Phi \to \Upsilon$ from each phenomenon $\phi $ to its best explanation $\upsilon$. $\Box$
\end{myremark}

Given a `big' fact table such as $H_1$, we need to identify/group the correlated input attributes under independent uncertainty units, viz., `u-factors,' each one associated with a random variable.\footnote{An attribute can be inferred `input' (viz., a parameter) by means of fd reasoning (cf. \S\ref{sec:scheme}).} 
We illustrate that by means of query $Q_1$, which materializes view $Y_1[g]$ for (let $g=Z_i$) identified u-factor $Z_i \subseteq Z$ in $H_1[\phi,\, Z]$. 

\begin{framed}
\noindent
$Q_1$. \textsf{\textbf{create table} $Y_1[g]$ \textbf{as select} U.$\phi$, U.$g$ \textbf{from} (\textbf{repair key} $\phi$ \textbf{in}}\\ 
\indent$\!\!\!\!\!$ \textsf{(\textbf{select} $\phi,\, g$, \textbf{count(*)} \textbf{as} Fr \textbf{from}
$H_1$ \textbf{group by} $\phi,\, g$)}\\ 
\indent$\!\!\!\!\!$ \textsf{\textbf{weight by} Fr) as U;}
\end{framed}

\noindent
The result set of $Q_1$ is stored in $Y_1[g]$, see Fig. \ref{fig:input}. 
Note that the possible values of $g$ are mapped to random variable $\textsf{x}_1$, and that table $H_1$ is considered source for a joint probability distribution (on the values of $H_1$'s input parameters) which may not be uniform: we count the frequency \textsf{Fr} of each possible value of a u-factor $Z_i \subseteq Z$ (as done for $g$ in $Q_1$) and pass it as argument to the \textsf{weight-by} construct.

\begin{figure}[t]
\begin{spacing}{1}
\advance\leftskip 0.4cm
\begin{subfigure}{0.45\textwidth}
\begingroup\setlength{\fboxsep}{3pt}
\colorbox{blue!5}{%
   \begin{tabular}{c|c|c||c|c|c}
  $H_1[\textsf{tid},\,\phi,\,Z]$ & \textsf{tid} & $\phi$ & $g$ & $\operatorname{v_0}$ & $s_0$\\
      \hline    
   & $1$ & $1$ & $32$ & $0$ & $5000$\\
   & $2$ & $1$ & $32$ & $10$ & $5000$\\
   & $3$ & $1$ & $32$ & $20$ & $5000$\\
   & $4$ & $1$ & $32.2$ & $0$ & $5000$\\
   & $5$ & $1$ & $32.2$ & $10$ & $5000$\\
   & $6$ & $1$ & $32.2$ & $20$ & $5000$\\
   \end{tabular}
}\endgroup\\
\end{subfigure}
\hspace{45pt}
\begin{subfigure}{0.45\textwidth}
\begingroup\setlength{\fboxsep}{5pt}
\colorbox{yellow!12}{%
   \begin{tabular}{c|>{\columncolor[gray]{0.92}}c|c||c}
  $Y_1[g]$ & $V \mapsto D$ & $\phi$ & $g$\\
      \hline    
   & $\textsf{x}_1 \mapsto 1$ & $1$ & $32$\\
   & $\textsf{x}_1 \mapsto 2$ & $1$ & $32.2$\\
   \end{tabular}
}\endgroup
\vspace{7pt}
\begingroup\setlength{\fboxsep}{3pt}
\colorbox{yellow!15}{%
   \begin{tabular}{c|>{\columncolor[gray]{0.92}}c|c}
  $W$ & $V \mapsto D$ & \textsf{Pr}\\
      \hline    
   & $\cdots$ & $\cdots$\\   
\cline{2-3} 
   & $\textsf{x}_1 \mapsto 1$ & $.5$\\
   & $\textsf{x}_1 \mapsto 2$ & $.5$\\
   \end{tabular}
}\endgroup
\end{subfigure}
\caption{Result set of query $Q_1$ on simulation trial dataset for hypothesis $H_1$.}
\label{fig:input}
\end{spacing}
\end{figure}

So far, we have presented informally the procedure of \emph{u-factorization}. Now we proceed to \emph{u-propagation} --- both are presented rigorously in Chapter \ref{ch:synthesis4u}.
We consider $\,g\,\upsilon \to a \in \Sigma_1$ again in order to synthesize predictive U-relation $Y_1[a]$. Since $a$ is functionally determined by $\upsilon$ and $g$ only, and these are independent, we propagate their uncertainty onto $a$ into $Y_1[a]$ by query $Q_2$.

\begin{framed}
\noindent
$Q_2.$ \textsf{\textbf{create table} $Y_1[a]$ \textbf{as select} $H_1.\phi$, $H_1.\upsilon$, $H_1.a$ \textbf{from} $H_1$, $\,Y_0$, $Y_1[g]$ \textbf{as} G }\\
\indent$\!\!\!\!\!\!\!\!$ \textsf{\textbf{where} $H_1.\phi$=$Y_0.\phi$ \textbf{and} $H_1.\upsilon$=$Y_0.\upsilon$ \textbf{and} G.$\phi$=$H_1.\phi$ \textbf{and} G.$g$=$H_1.g$;}
\end{framed}

\noindent
Query $Q_2^\prime$ (not shown) then selects $\phi$, $\upsilon$ and $a$ from \mbox{$Y_i[a]$} for each $i\!=\!1..3$. The result sets of $Q_2$ and $Q_2^\prime$ (resp. $Y_1[a]$ and $Y[a]$) are shown in Fig. \ref{fig:output}.

\begin{figure}[H]
\begin{spacing}{1}
\centering
\begingroup\setlength{\fboxsep}{3pt}
\colorbox{yellow!15}{%
   \begin{tabular}{c|>{\columncolor[gray]{0.92}}c|>{\columncolor[gray]{0.92}}c|c|c||c}
  $Y_1[a]$ & $V_0 \mapsto D_0$ & $V_1 \mapsto D_1$ & $\phi$ & $\upsilon$ & $a$\\
      \hline    
   & $\textsf{x}_0 \mapsto 1$ & $\textsf{x}_1 \mapsto 1$ & $1$ & $1$ & $-32$\\
   & $\textsf{x}_0 \mapsto 1$ & $\textsf{x}_1 \mapsto 2$ & $1$ & $1$ & $-32.2\!\!\!\!$\\
   \end{tabular}
}\endgroup
\vspace{5pt}\\
\begingroup\setlength{\fboxsep}{2pt}
\colorbox{yellow!15}{%
   \begin{tabular}{c|>{\columncolor[gray]{0.92}}c|>{\columncolor[gray]{0.92}}c|c|c||c}
  $Y[a]$ & $V_0 \mapsto D_0$ & $V_1 \mapsto D_1$ & $\phi$ & $\upsilon$ & $a$\\
      \hline    
   & $\textsf{x}_0 \mapsto 1$ & $\textsf{x}_1 \mapsto 1$ & $1$ & $1$ & $-32$\\
   & $\textsf{x}_0 \mapsto 1$ & $\textsf{x}_1 \mapsto 2$ & $1$ & $1$ & $-32.2$\\
   & $\textsf{x}_0 \mapsto 2$ & $-$ & $1$ & $2$ & $0$\\
   & $\textsf{x}_0 \mapsto 3$ & $-$ & $1$ & $3$ & $0$\\
   \end{tabular}
}\endgroup
\caption{U-relational predictive tables rendered by query using the fd's.}
\label{fig:output}
\end{spacing}
\end{figure}

Compare relations $H_1[a]$ and $Y_1[a]$. 
By accounting for the correlations captured in the fd $\,g\,\upsilon \to a$, we could propagate onto $a$ the uncertainty coming from the hypothesis and the only parameter $a$ is sensible to, thus precisely situating tuples of $Y_1[a]$ in the space of possible worlds. The same is done for predictive attributes $\operatorname{v}$ and $s$. In the end, $\Upsilon$-DB shall be ready for predictive analytics, i.e., with all competing predictions as possible alternatives which are mutually inconsistent.

A key point here is that all the synthesis process is amenable to algorithm design. Except for the user `research' description, the $\Upsilon$-DB construction is fully automated based on the hypothesis structure (set of equations) and the raw hypothesis trial data.

\section{Predictive Analytics}\label{sec:analytics}

Users of Example \ref{ex:fall}, has to be able, say, to query \mbox{phenomenon $\phi=1$} w.r.t. predicted position $s$ at specific values of time $t$ by considering all hypotheses $\upsilon$ admitted. That is illustrated by query $Q_3$, which creates integrative table $Y[s]$; and by query $Q_4$, which computes the \emph{confidence} aggregate operation \cite{koch2009} for all $s$ tuples where $t=3$ (Fig. \ref{fig:analytics} shows $Q_4$'s result, apart from column \textsf{Posterior}).

The confidence on each hypothesis for the specific prediction of $Q_4$ is split due to parameter uncertainty such that they sum up back to its total confidence. For $H_2$ and $H_3$, e.g., we have $\{g\,D\,s_0\,t\,\upsilon \to s\} \,\subset\, \Gamma$, where $\Gamma=\Sigma_2 = \Sigma_3$. Since $g$ and $D$ are the parameter uncertainty factors of $s$ ($s_0$ is certain), with 2 possible values (not shown) each, then there are only $2 \times  2 = 4$ possible $s$ tuples for $H_2$ and $H_3$ each. Considering all hypotheses $\upsilon$ for the same phenomenon $\phi$, the confidence values sum up to one in accordance with the laws of probability.

\begin{framed}
$Q_3$. \textsf{\textbf{create table} $Y[s]$ \textbf{as} \textbf{select} U.$\phi$, U.$\upsilon$, U.$t$, U.$s$ \textbf{from}}\\
\indent\indent \textsf{(\textbf{select} $\phi$, $\upsilon$, $t$, $s$ \textbf{from} $Y_1[s]$ \textbf{union all}}\\
\indent\indent \textsf{\textbf{select} $\phi$, $\upsilon$, $t$, $s$ \textbf{from} $Y_2[s]$ \textbf{union all}}\\
\indent\indent \textsf{\textbf{select} $\phi$, $\upsilon$, $t$, $s$ \textbf{from} $Y_3[s]$) \textbf{as} U, $Y_0$}\\
\indent\indent \textsf{\textbf{where} U.$\phi$=$Y_0$.$\phi$ \textbf{and} U.$\upsilon$=$Y_0$.$\upsilon$;}
\end{framed}

\begin{framed}
$Q_4$. \textsf{\textbf{select} $\phi$, $\upsilon$, $s$, \textbf{conf()} \textbf{as} Prior \textbf{from} $Y[s]$ \textbf{where} $t$=3}\\
\indent\indent \textsf{\textbf{group by} $\phi$, $\upsilon$, $s$ \textbf{order by} Prior \textbf{desc};}
\end{framed}

\begin{figure}[H]
\begin{spacing}{1}
\centering
\begingroup\setlength{\fboxsep}{3pt}
\colorbox{yellow!15}{%
   \begin{tabular}{c|c|c|c||>{\columncolor[gray]{0.92}}c|>{\columncolor[gray]{0.92}}c}
  $Y[s]$ & $\phi$ & $\upsilon$ & $s$ & \textsf{Prior} & \textsf{Posterior}\\
      \hline    
   & $1$ & $1$ & $2188.36$ & $.1$ & $.167$\\
   & $1$ & $1$ &  $2205.82$ & $.1$ & $.168$\\
   & $1$ & $1$ &  $2320.51$ & $.1$ & $.167$\\
   & $1$ & $1$ &  $2337.97$ & $.1$ & $.165$\\
   & $1$ & $1$ &  $2452.66$ & $.1$ & $.149$\\
   & $1$ & $1$ &  $2470.12$ & $.1$ & $.145$\\
\cline{2-6}
   & $1$ & $2$ & $2930.59$ & $.05$ & $.020$\\
   & $1$ & $2$ &  $2943.44$ & $.05$ & $.019$\\
   & $1$ & $2$ &  $4991.92$ & $.05$ & $.000$\\
   & $1$ & $2$ &  $4991.97$ & $.05$ & $.000$\\
\cline{2-6}
   & $1$ & $3$ &  $4778.87$ & $.05$ & $.000$\\
   & $1$ & $3$ &  $4779.56$ & $.05$ & $.000$\\
   & $1$ & $3$ &  $4944.72$ & $.05$ & $.000$\\
   & $1$ & $3$ &  $4944.89$ & $.05$ & $.000$ 
   \end{tabular}
}\endgroup
\caption{Analytics on predicted position $s$ conditioned on observation.}\label{fig:analytics}
\end{spacing}
\end{figure}

Users can make informed decisions in light of such confidence aggregates, which are to be eventually conditioned in face of evidence (observed data). Example \ref{ex:bayes} features such kind of Bayesian conditioning for \emph{discrete} random variables mapped to the possible values of predictive attributes (like position $s$) whose domain are \emph{continuous}.

\begin{myex}\label{ex:bayes}
Suppose position $s=2250$ feet is observed at \mbox{$t=3$} secs, with standard deviation $\sigma=20$. Then, by applying Bayes' theorem for normal mean with a discrete prior \cite{bolstad2007}, \textsf{Prior} is updated to \textsf{Posterior} (see Fig. \ref{fig:analytics}). $\Box$
\end{myex}

The procedure uses normal density function (\ref{eq:density}), with (say) \mbox{$\sigma=20$}, to get the likelihood $f(y \,|\, \mu_k)$ of each alternative prediction of $s$ from $Y[s]$ as mean $\mu_k$ given $y$ at observed $s=2250$. Then it applies Bayes' rule (\ref{eq:bayes}) to get the posterior $p(\mu_k \,|\, y)$ \cite{bolstad2007}.
\begin{eqnarray}
f(y \,|\, \mu_k) \!\!&=&\!\! \frac{1}{\sqrt{2\pi \sigma^2}}\, e^{-\frac{1}{2\sigma^2}(y-\mu_k)^2}\label{eq:density}\\
p(\mu_k \,|\, y) \!\!&=&\!\! f(y \,|\,\mu_k)\;p(\mu_k) \;/\;\textstyle\sum_{i=1}^n f(y \,|\,\mu_i)\;p(\mu_i)\label{eq:bayes}
\end{eqnarray}

\noindent
In the general case (cf. examples shown in Chapter \ref{ch:applicability}), we actually have phenomenon `as data:' a sample of independent observed values $y_1,\,...,\,y_n$ (e.g., Brazil's population observed by census over the years). Then, the likelihood $f(y_1,\,...,\,y_n \,|\, \mu_k)$ for each competing trial $\mu_k$, is computed as a product $\textstyle\prod_{j=1}^n f(y_j \,|\, \mu_{kj})$ of the single likelihoods $f(y_j \,|\, \mu_{kj})$ \cite{bolstad2007}. Bayes' rule is then settled by (\ref{eq:sample-bayes}) to compute the posterior $p(\mu_k \,|\, y_1,\,...,\,y_n)$ given prior $p(\mu_k)$. 

\vspace{-5pt}
\begin{eqnarray}
p(\mu_k \,|\, y_1,\,\hdots,\,y_n) \!&=&\! \frac{\textstyle\prod_{j=1}^n f(y_j \,|\, \mu_{kj})\;p(\mu_k)}{\displaystyle\sum_{i=1}^m \displaystyle\prod_{j=1}^n f(y_j \,|\, \mu_{ij})\;p(\mu_i)}\label{eq:sample-bayes}
\end{eqnarray}

As a result, the prior probability distribution assigned to u-factors via \mbox{\textsf{repair key}} is to be eventually conditioned on observed data. 
This is an applied \emph{Bayesian inference} problem that translates into a \emph{p-$\!$DB update} one to induce effects of posteriors back to table $W$. In a first prototype of the \mbox{$\Upsilon$-DB} system (cf. \ref{sec:demo}), we accomplish it by performing Bayesian inference at application level and then applying p-WSA's \textsf{update} (a variant of SQL's \textsf{update}) into \textsf{MayBMS}. This solution is good enough to let us complete use case demonstrations of $\Upsilon$-DB.\footnote{Cf. Chapter \ref{ch:applicability}, and \S\ref{sec:demo} in particular.}

\section{Related Work}\label{sec:related-work}

The vision of managing hypotheses as data has some roots in Porto and Spaccapietra \cite{porto2011}, who motivated a conceptual data model to support (the so-called) \emph{in silico} science by means of a scientific model management system. We discuss now the work we understand to be mostly related to our vision of data-driven hypothesis management and analytics.

\subsection{Models-and-data}\label{subsec:models-and-data}
Haas et al. \cite{haas2011} provide an original long-term perspective on the evolution of database technology. They characterize the data typically managed by traditional DB systems as a record about the past, not a conclusion or an insight or a solution (ibid.). In the context of scientific databases, e.g., their position is suggestive that DB technology has been designed for empirical data, not the theoretical data generated by simulation from domain-specific principles or scientific hypotheses. 

They recognize current DB technology to have raised the art of scalable `descriptive' analytics to a very high level; but point out, however, that nowadays (sic.) what enterprises really need is `prescriptive' analytics to identify optimal business, policy, investment, and engineering decisions in the face of uncertainty. Such analytics, in turn, shall rest on \emph{deep} `predictive' analytics that go beyond mere statistical forecasting and are imbued with an understanding of the fundamental mechanisms that govern a system's behavior, allowing what-if analyses \cite{haas2011}.
In sum, there is a pressing call for \emph{deep} predictive analytic tools in business enterprises as much as in science's. 

In comparison with the $\Upsilon$-DB vision, Haas et al. are proposing a long-term \emph{models-and-data} research program to pursue data management technology for deep predictive analytics. They discuss strategies to extend query engines for model execution within a \mbox{(p-)DB}. Along these lines, query optimization is understood as a more general problem with connections to algebraic solvers.

Our framework in turn essentially comprises an abstraction and technique for the encoding of \emph{hypotheses as data}. It can be understood (in comparison) as putting models strictly into a flattened data perspective. For that reason it has been directly applicable by building upon recent work on p-DBs \cite{suciu2011}. In principle, it can be integrated into, say, the OLAP layer of the models-and-data project.

\subsection{Scientific simulation data}\label{subsec:simulation-data}
\noindent
As previsouly mentioned, science's ETL is distinguished by its unfrequent, incremen\-tal-only updates and by having large raw files as data sources \cite{ailamaki2010}. 
Challenges for enabling an efficient access to high-resolution, raw simulation data have been documented from both supercomputing,\cite{meneveau2007} and database research viewpoints;\cite{ailamaki2013} and pointed as key to the use case of \emph{exploratory analytics}. The extreme scale of the raw data has motivated such non-conventional approaches for data exploration, viz., the `immersive' query processing (move the program to the data) \cite{meneveau2007,kanov2011}, or `in situ' query processing in the raw files \cite{ailamaki2011,ailamaki2012}. Both exploit the spatial structure of the data in their indexing schemes.

That line of research is motivated for equipping scientist end-users for an immediate interaction with their very large simulation datasets.\footnote{Sometimes phrased `here is my files, here is my queries, where are my results?' \cite{ailamaki2011}.} The NoDB approach, in particular, argues to eliminate such ETL phase (viz., the loading) for a direct access to data `in situ' in the raw data files \cite{ailamaki2012}. In fact, data exploration is a fundamental use case of data-driven science.

Nonetheless, being generated from first principles or learned deterministic hypotheses, simulation data has a pronounced uncertainty component that motivates a another use case, viz., the case of hypothesis management and \emph{predictive} analytics \cite{haas2011,goncalves2014}. As we have motivated in \S\ref{sec:goals}, the latter requires probabilistic DB design for enabling uncertainty decomposition (factorization). 

Hypothesis management shall not deal with the same volume of data as in simulation data management for exploratory analytics, but samples of it (cf. Table \ref{tab:hypothesis} for a comparison). 
For instance, in CERN's particle-physics experiment ATLAS there are four tier/layers of data management. The volume of data significantly decreases from the (tier-0) raw data to the (tier-3) data actually used for analyses such as hypothesis testing \cite[p. 71-2]{ailamaki2010}. 

Overall, the overhead incurred in loading samples of raw simulation trial datasets into a p-DB is justified for enabling a principled hypothesis evaluation and rating/ranking according to the scientific method.

\subsection{Hypothesis encoding}\label{subsec:hypothesis-encoding}

Our framework is comparable with Bioinformatics' initiatives that address hypothesis encoding into the RDF data model \cite{soldatova2011}: (i) the Robot Scientist \cite{king2009} is a knowledge-base system (KBS) for automated generation and testing of hypotheses about what genes encode enzymes in the yeast organism; (ii) HyBrow \cite{racunas2004} is a KBS for scientists to test their hypotheses about events of the galactose metabolism also of the yeast organism; and (iii) SWAN \cite{gao2006} is a KBS for scientists to share hypotheses on possible causes of the Alzheimer disease.

The Robot Scientist relies on rule-based logic programming analytics to automatically generate and test RDF-encoded hypotheses of the kind `gene G has function A' against RDF-encoded empirical data \cite{king2009}. HyBrow is likewise, but hypotheses are formulated by the user about biological events \cite{racunas2004}. SWAN in turn disfavors analytic techniques for hypothesis evaluation and focus on descriptive aspects: hypotheses are high-level natural language statements retrieved from publications. Each `hypothesis' is associated with lower-level  `claims' (both RDF-encoded) that are meant to support it on the basis of some empirical evidence (RDF-encoded gene/protein data). 
In particular, SWAN \cite{gao2006} differs from the former in that each hypothesis is unstructured, being then more related to efforts on the retrieval of textual claims from the narrative fabric of scientific reports \cite{deWaard2009}. 

All of them though, consist in some ad-hoc RDF encoding of sequence and genome analysis hypotheses under varying levels of structure (viz., from `gene G has function A' statements to free text). Our framework in turn consists in the U-relational encoding of hypotheses from mathematical equations, which is (to our knowledge) the first work on hypothesis relational encoding. 

Finally, as for hypothesis evaluation and comparison analytics, the \mbox{$\Upsilon$-DB} vision is distinguished in terms of its Bayesian inference approach. The latter has been pointed out as a major direction for the improvement of the Bioinformatics' initiatives just mentioned (cf. \cite[p. 13]{soldatova2011}), and is in fact an influential model of decision making for hypothesis evaluation \cite[p. 220]{losee2001}.

\section{Summary: Key Points}\label{sec:challenges}
\noindent
We outline some key points in the $\Upsilon$-DB vision:

\begin{itemize}
\item `Structured deterministic hypotheses' are encoded as theoretical data and distinguished from empirical data by the introduction of an \emph{epistemological dimension} into their semantic structure.

\item Two \emph{sources of uncertainty} are considered: theoretical uncertainty, originating from competing hypotheses; and empirical uncertainty, derived from alternative simulation trials on each hypothesis for the same phenomenon.

\item A method to extract the structure of a hypothesis can be carefully designed based on a \emph{hypothesis data representation} and shall be reducible in terms of  machine-readable format for mathematical modeling, viz., W3C's \textsf{MathML}, which we shall adopt as a standard for hypothesis specification.

\item We have seen that the controlled \emph{introduction of uncertainty} into simulation data is amenable to algorithm design and then reducible to a design-theoretic synthesis method for the construction of U-relational DB's.

\item Simulation data can be modeled as hypothesis data whenever it is associated with a target phenomenon. 
As the same phenomenon may happen to be associated with many such hypotheses, the research activity can be modeled as a \emph{data cleaning problem} in p-DB's. 
\end{itemize}

Essentially, the vision of $\Upsilon$-DB comprises a design-theoretic pipeline (Fig. \ref{fig:pipeline}). 
For the insertion of a hypothesis $k$, we shall be given a \textsf{MathML}-compliant structure $\mathcal S_k$ together with its simulation trial datasets $\mathcal D_k^\ell$ in raw files (e.g., .mat, .csv). 
Then we apply an Extract-Transform-Load (ETL) automatic procedure to generate the hypothesis `big' fact table $H_k$ under the trial id's.

The extracted equations are firstly 
encoded into fd's. Then, at any time, as many hypotheses may have been inserted into the system, the  uncertainty introduction (U-intro) procedure can be applied to process the encoded fd's and synthesize the `uncertain' U-relations that are to be eventually conditioned on observations.

Note, in Fig. \ref{fig:pipeline}, that the \textsf{ETL} procedure is operated in a `local' view for each hypothesis $k$, while the \textsf{U-intro} procedure and the conditioning are operated in the `global' view of all available hypotheses $k=1..n$. 
The pipeline opens up four main tracks of technical research challenges from the ETL stage on, viz., (i) \emph{hypothesis encoding} and (ii) \emph{causal reasoning over fd's}, (iii) \emph{p-DB synthesis} and (iv) \emph{conditioning}. We address in the sequel the three first track of challenges in depth. The problem of conditioning is outlined for further work in \S\ref{sec:future-work}.


\chapter{Hypothesis Encoding}\label{ch:encoding}

In this chapter we address the problem of hypothesis encoding. In \S\ref{sec:sems} we introduce notation and basic concepts of structural equations and the problem of causal ordering. In \S\ref{sec:coa} we study the problem of extracting the causal ordering implicit in the structure of a deterministic hypothesis and show that Simon's classical approach \cite{simon1953,druzdzel2008} is intractable. In \S\ref{sec:tcm} then we build upon a less notorious approach of Nayak's \cite{nayak1994} and borrow an efficient algorithm for it that fits very well our use case for hypothesis encoding. In \S\ref{sec:scheme} we develop an encoding scheme that builds upon the idea of structural equations through an original abstraction of hypotheses `as data.' In \S\ref{sec:tcm-experiments} we present experiments that attest how the encoding scheme works in practice for large hypotheses. In \S\ref{sec:related-work-encoding} we discuss related work. In \S\ref{sec:encoding-conclusions} we summarize the results of this chapter.

\section{Preliminaries: Structural Equations}\label{sec:sems}
\noindent
Given a system of mathematical equations involving a set of variables, to build a \emph{structural equation model} (SEM) is, essentially, to establish a one-to-one mapping between equations and variables \cite{simon1953}. That shall enable further detecting the hidden asymmetry between variables, i.e., their causal ordering. For instance, Einstein's famous equation $E=m\,c^2$ states the equivalence of mass and energy, summarizing a theory that can be imputed two different asymmetries (for different applications), say, given a fixed amount of mass $m=m_0$ (and recall $c$ is a constant), predict the particle's relativistic rest energy $E$; or given the particle's rest energy, predict its mass or potential for nuclear fission.

To stress the point, consider Newton's second law $F = m\,a\,$ in such a scalar setting. The modeler can either use it to compute (predict), say, acceleration values given an amount of mass and different force intensities, or to predict force intensities given a fixed acceleration (e.g., for testing an engineered dynamometer). The point here is that Newton's equation is not enough to derive predictions. That is, it has a number of variables $|\mathcal V|=3$, which is larger than $|\mathcal E|=1$. It must be completed with two more equations in order to qualify as an (applied) hypothesis `as data.' Although usually it is interpreted an asymmetry towards $a$, technically, there is nothing in its semantics to suggest so.\footnote{As the equality construct `=' is used as a predicate, not an assignment operator.} Compare the two systems given in Fig. \ref{fig:newton-graphs}.\footnote{We shall introduce the notion of `directed causal graphs' shortly.} In sum, the causal ordering of any system of equations is not to be guessed, as it can be inferred. In this chapter we rely on previous work (mostly AI's work, viz., \cite{simon1953,nayak1994,druzdzel2008}) and adapt it for the encoding of hypotheses into fd's.

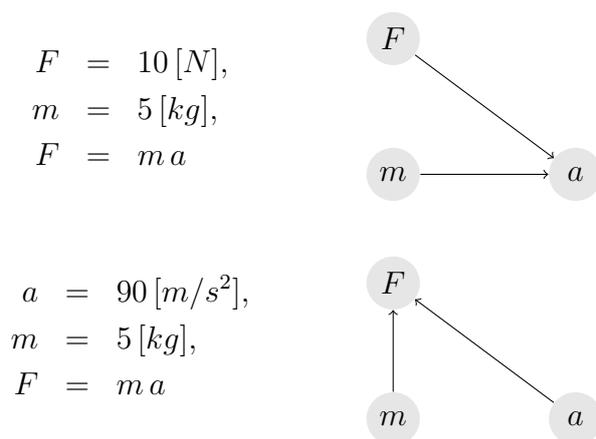
\begin{figure}[t]
\advance\leftskip0.2cm
\begin{subfigure}{0.475\textwidth}
\vspace{-15pt}
\begin{eqnarray*}
F &=& 10\,[N],\\ 
m &=& 5\,[kg],\\ 
F &=& m\,a
\end{eqnarray*}
\vspace{-16pt}\\
\begin{eqnarray*}
a &=& 90\,[m/s^2],\\ 
m &=& 5\,[kg],\\ 
F &=& m\,a
\end{eqnarray*}
\end{subfigure}
\hspace{-20pt}
\begin{subfigure}{0.475\textwidth}
\vspace{8pt}
\tikzstyle{circ}=[circle,
                                    thick,
                                    minimum size=0.3cm,
                                    draw=black]
\tikzstyle{vertex}=[circle,fill=black!10,minimum size=20pt,inner sep=0pt]
\tikzstyle{selected vertex} = [vertex, fill=red!24]
\tikzstyle{edge} = [draw,thick,->,bend left]
\tikzstyle{weight} = [font=\small]
\tikzstyle{selected edge} = [draw,line width=5pt,-,red!50]
\tikzstyle{ignored edge} = [draw,line width=5pt,-,black!20]
\begin{tikzpicture}[scale=0.6]
    \foreach \pos/\name in {{(0,5)/F}, {(0,2)/m}, {(4,2)/a}}
        \node[vertex] (\name) at \pos {$\name$};
    \draw[->] (F) to (a);
    \draw[->] (m) to (a);
\end{tikzpicture}\vspace{20pt}\\
\tikzstyle{circ}=[circle,
                                    thick,
                                    minimum size=0.3cm,
                                    draw=black]
\tikzstyle{vertex}=[circle,fill=black!10,minimum size=20pt,inner sep=0pt]
\tikzstyle{selected vertex} = [vertex, fill=red!24]
\tikzstyle{edge} = [draw,thick,->,bend left]
\tikzstyle{weight} = [font=\small]
\tikzstyle{selected edge} = [draw,line width=5pt,-,red!50]
\tikzstyle{ignored edge} = [draw,line width=5pt,-,black!20]
\begin{tikzpicture}[scale=0.6]
    \foreach \pos/\name in {{(0,5)/F}, {(0,2)/m}, {(4,2)/a}}
        \node[vertex] (\name) at \pos {$\name$};
    \draw[->] (a) to (F);
    \draw[->] (m) to (F);
\end{tikzpicture}
\end{subfigure}
\vspace{10pt}
\caption{``Directed causal graphs'' associated with the two systems.}\label{fig:newton-graphs}
\end{figure}

\begin{mydef}\label{def:structure}
A \textbf{structure} is a pair $\mathcal S(\mathcal E, \mathcal V)$, where $\mathcal E$ is a set of equations over set $\mathcal V\!$ of variables, $|\mathcal E| \leq |\mathcal V|$, such that:
\begin{itemize}
\item[(a)] In any subset of $k$ equations of the structure, at least $k$ different variables appear;

\item[(b)] In any subset of $k$ equations in which $r$ variables appear, $k \leq r$, if the values of any $(r - k)$ variables are chosen arbitrarily, then the values of the remaining $k$ variables can be determined uniquely --- finding these unique values is a matter of solving the equations.
\end{itemize}
\end{mydef}

\begin{mydef}\label{def:complete}
Let $\mathcal S(\mathcal E, \mathcal V)$ be a structure. We say that $\mathcal S$ is self-contained or \textbf{complete} if $|\mathcal E|=|\mathcal V|$.
\end{mydef}
\noindent
In short, we are interested in systems of equations that are `structural' (Def. \ref{def:structure}) and `complete' (Def. \ref{def:complete}), viz., that has as many equations as variables and no subset of equations has fewer variables than equations.\footnote{Also, we expect the systems of equations given as input to be `independent' in the sense of Linear Algebra. In our context, that means systems that can only have non-redundant equations. In that case, if some subset of equations has fewer variables than equations, then the system must be `overconstrained.'}

Complete structures can be solved for unique sets of values of their variables. In this work, however, we are not concerned with solving sets of mathematical equations at all, but with processing their causal ordering in view of U-relational DB design. Simon's concept of causal ordering has its roots in econometrics studies (cf. \cite{simon1953}) and has been taken further in AI with a flavor of Graphical Models (GMs) \cite{druzdzel1993,pearl2000,druzdzel2008}. In this thesis we translate the problem of causal ordering into the language of data dependencies, viz., into fd's.

\begin{mydef}
Let $\mathcal S$ be a structure. We say that $\mathcal S$ is \textbf{minimal} if it is complete and there is no complete structure $\mathcal S^\prime \!\subset \mathcal S$. 
\end{mydef}

\begin{mydef}
The \textbf{structure matrix} $A_S$ of a structure $\mathcal S(\mathcal E, \mathcal V)$, with $f_1,\, f_2, \hdots, f_n \in \mathcal E$ and $x_1,\, x_2, \hdots, x_m \in \mathcal V$, is a $n \times m$ matrix of 1's and 0's in which entry $a_{ij}$ is non-zero if variable $x_j$ appears in equation $f_i$, and zero otherwise.
\end{mydef}

Elementary row operations (e.g., row multiplication by a constant) on the structure matrix may hinder the structure's causal ordering and then are not valid in general \cite{simon1953}. This also emphasizes that the problem of causal ordering is not about solving the system of mathematical equations of a structure, but identifying its hidden asymmetries.

\begin{mydef}\label{def:tcm}
$\!$Let $\mathcal S(\mathcal E, \mathcal V)\!$ be a complete structure. $\!\!$Then a \textbf{total causal mapping} over $\mathcal S$ is a bijection $\varphi\!: \mathcal E \to \mathcal V$ such that, for all $f \!\in \mathcal E$, if $\varphi(f)=x$ then $x \!\in Vars(f)$.
\end{mydef}

\vspace{-3pt}
\noindent
Simon has informally described an algorithm (cf. \cite{simon1953}) that, given a complete structure $\mathcal S(\mathcal E, \mathcal V)$, can be used to compute a \emph{partial} causal mapping $\varphi_p$ from partitions on the set of equations to same-cardinality partitions on the set of variables. As shown by Dash and Druzdzel \cite{druzdzel2008}, the causal mapping returned by Simon's (so-called) Causal Ordering Algorithm (\textsf{COA}) is not \emph{total} when $\mathcal S$ has variables that are \emph{strongly coupled} (because they can only be determined simultaneously). They also have shown that any total mapping $\varphi$ over $\mathcal S$ must be consistent with \textsf{COA}'s partial mapping $\varphi_p$ \cite{druzdzel2008}. The latter is made partial by design (merge strongly coupled variables into partitions or clusters) in order to force its induced causal graph $G_{\varphi_{p}}$ to be acyclic. Algorithm \ref{alg:coat}, \textsf{COA}$_t$, is a variant of Simon's $\!$\textsf{COA} adapted to illustrate our use case. It returns a total causal mapping $\varphi$, instead of a partial causal mapping. We illustrate it through Example \ref{ex:struct-matrix} and Fig. \ref{fig:coa}.

\begin{spacing}{1.2}
\begin{figure}[t]
\tikzset{node style ge/.style={circle,inner sep=0pt,minimum size=17pt}}
\begin{subfigure}{0.475\textwidth}
\hspace{12pt}
\begin{tikzpicture}[baseline=(A.center)]
\matrix (A) [matrix of math nodes, nodes = {node style ge},column sep=01.0 mm] {
 & \node (x1) {x_1}; & \node (x2) {x_2}; & \node (x3) {x_3}; & \node (x4) {x_4}; & \node (x5) {x_5}; & \node (x6) {x_6}; & \node (x7) {x_7};\\
\node (f1) {f_1}; & \node (a11) {1}; & \node (a12) {0}; & \node (a13) {0}; & \node (a14) {0}; & \node (a15) {0}; & \node (a16) {0}; & \node (a17) {0};\\
\node (f2) {f_2}; & \node (a21) {0}; & \node (a22) {1}; & \node (a23) {0}; & \node (a24) {0}; & \node (a25) {0}; & \node (a26) {0}; & \node (a27) {0};\\
\node (f3) {f_3}; & \node (a31) {0}; & \node (a32) {0}; & \node (a33) {1}; & \node (a34) {0}; & \node (a35) {0}; & \node (a36) {0}; & \node (a37) {0};\\
\node (f4) {f_4}; & \node (a41) {1}; & \node (a42) {1}; & \node (a43) {1}; & \node (a44) {1}; & \node (a45) {1}; & \node (a46) {0}; & \node (a47) {0};\\
\node (f5) {f_5}; & \node (a51) {1}; & \node (a52) {0}; & \node (a53) {1}; & \node (a54) {1}; & \node (a55) {1}; & \node (a56) {0}; & \node (a57) {0};\\
\node (f6) {f_6}; & \node (a61) {0}; & \node (a62) {0}; & \node (a63) {0}; & \node (a64) {1}; & \node (a65) {0}; & \node (a66) {1}; & \node (a67) {0};\\
\node (f7) {f_7}; & \node (a71) {0}; & \node (a72) {0}; & \node (a73) {0}; & \node (a74) {0}; & \node (a75) {1}; & \node (a76) {0}; & \node (a77) {1};\\
};
\end{tikzpicture}
\caption{Structure matrix as given.}\label{fig:coa-a}
\end{subfigure}
\hspace{-6pt}
$\to$
\tikzstyle{background0}=[rectangle,
                                                fill=gray!05,
                                                inner sep=0.025cm,
                                                rounded corners=1mm]
\tikzstyle{background1}=[rectangle,
                                                fill=gray!20,
                                                inner sep=0.025cm,
                                                rounded corners=1mm]
\tikzstyle{background2}=[rectangle,
                                                fill=gray!40,
                                                inner sep=0.025cm,
                                                rounded corners=1mm]
\begin{subfigure}{0.475\textwidth}
\vspace{-5pt}
\hspace{6pt}
\begin{tikzpicture}[baseline=(A.center)]
  \tikzset{BarreStyle/.style =   {opacity=.35,line width=3.25 mm,line cap=round,color=#1}}
\matrix (A) [matrix of math nodes, nodes = {node style ge},column sep=1.0 mm] {
 & \node (x1) {x_1}; & \node (x2) {x_2}; & \node (x3) {x_3}; & \node (x4) {x_4}; & \node (x5) {x_5}; & \node (x6) {x_6}; & \node (x7) {x_7};\\
\node (f1) {f_1}; & \node (a11) {1}; & \node (a12) {0}; & \node (a13) {0}; & \node (a14) {0}; & \node (a15) {0}; & \node (a16) {0}; & \node (a17) {0};\\
\node (f2) {f_2}; & \node (a21) {0}; & \node (a22) {1}; & \node (a23) {0}; & \node (a24) {0}; & \node (a25) {0}; & \node (a26) {0}; & \node (a27) {0};\\
\node (f3) {f_3}; & \node (a31) {0}; & \node (a32) {0}; & \node (a33) {1}; & \node (a34) {0}; & \node (a35) {0}; & \node (a36) {0}; & \node (a37) {0};\\
\node (f4) {f_4}; & \node (a41) {1}; & \node (a42) {1}; & \node (a43) {1}; & \node (a44) {1}; & \node (a45) {1}; & \node (a46) {0}; & \node (a47) {0};\\
\node (f5) {f_5}; & \node (a51) {1}; & \node (a52) {0}; & \node (a53) {1}; & \node (a54) {1}; & \node (a55) {1}; & \node (a56) {0}; & \node (a57) {0};\\
\node (f6) {f_6}; & \node (a61) {0}; & \node (a62) {0}; & \node (a63) {0}; & \node (a64) {1}; & \node (a65) {0}; & \node (a66) {1}; & \node (a67) {0};\\
\node (f7) {f_7}; & \node (a71) {0}; & \node (a72) {0}; & \node (a73) {0}; & \node (a74) {0}; & \node (a75) {1}; & \node (a76) {0}; & \node (a77) {1};\\
};
 \draw [BarreStyle=blue] (a11.north west) to (a11.south east); 
 \draw [BarreStyle=red] (a22.north west) to (a22.south east);
 \draw [BarreStyle=green] (a33.north west) to (a33.south east);
 \draw [BarreStyle=blue] (a44.north west) to (a55.south east); 
 \draw [BarreStyle=red] (a66.north west) to (a66.south east); 
 \draw [BarreStyle=green] (a77.north west) to (a77.south east); 
     \begin{pgfonlayer}{background}
        \node [background0,
                    fit=(a11) (a12) (a13) (a14) (a15) (a16) (a17) (a21) (a22) (a23) (a24) (a25) (a26) (a27) (a31) (a32) (a33) (a34) (a35) (a36) (a37) (a41) (a42) (a43) (a44) (a45) (a46) (a47) (a51) (a52) (a53) (a54) (a55) (a56) (a57) (a61) (a62) (a63) (a64) (a65) (a66) (a67) (a71) (a72) (a73) (a74) (a75) (a76) (a77) ]
                    {};
        \node [background1,
                    fit=(a44) (a45) (a46) (a47) (a54) (a55) (a56) (a57) (a64) (a65) (a66) (a67) 
                    (a74) (a75) (a76) (a77) ]
                    {};
        \node [background2,
                    fit=(a66) (a67) (a76) (a77) ]
                    {};
    \end{pgfonlayer}
\end{tikzpicture}
\vspace{2pt}
\caption{\textsf{COA} execution in 3 recursive steps.}\label{fig:coa-b}
\end{subfigure}
\caption[Running Simon's Causal Ordering Algorithm (\textsf{COA})]{Running Simon's Causal Ordering Algorithm (\textsf{COA}) on a given structure (Fig. \ref{fig:coa-a}). Minimal subsets detected in a recursive step $k$, highlighted in different shades of gray, have their diagonal elements colored (Fig. \ref{fig:coa-b}).}
\label{fig:coa}
\end{figure}
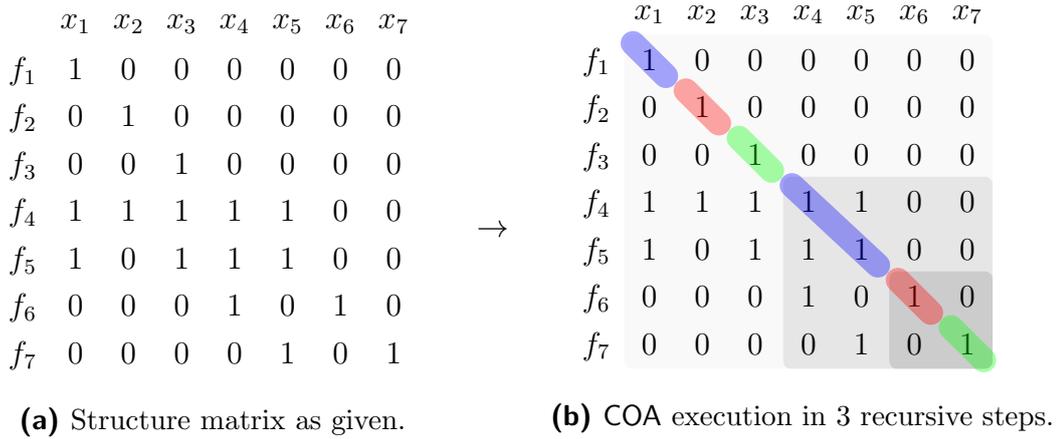
\end{spacing}

\begin{spacing}{1.1}
\begin{algorithm}[t]
\caption{COA$_t$ as a variant of Simon's COA.}
\label{alg:coat}
\begin{algorithmic}[1]
\Procedure{\textsf{COA}$_t$}{${\mathcal S\!:\, \text{structure over}\; \mathcal E \;\text{and}\; \mathcal V}$}
\Require $\mathcal S$ given is complete, i.e., $|\mathcal E|=|\mathcal V|$
\Ensure Returns total causal mapping $\varphi: \mathcal E \to \mathcal V$
\State $\varphi \gets \varnothing$, $\mathcal{S}_c \gets \varnothing$ 
\ForAll{minimal $\mathcal S^\prime \subset \mathcal S$}\vspace{1pt}
\State $\mathcal{S}_k \gets \mathcal{S}_k \cup \mathcal S^\prime$ \Comment{store minimal structures $\mathcal S^\prime$ found in $\mathcal S$}\vspace{1pt}
\State $\mathcal V^\prime \gets \mathcal S^\prime(\mathcal V)$
\ForAll{$f \in \mathcal S^\prime(\mathcal E)$}
\State $x \gets \text{any} \;x_a \in \mathcal V^\prime$\vspace{1pt}
\State $\varphi \gets \varphi \cup \langle f,\, x \rangle$\vspace{1pt}
\State $\mathcal V^\prime \gets \mathcal V^\prime \setminus \{x\}$
\EndFor
\EndFor
\State $\mathcal T \gets \mathcal S \setminus \bigcup_{\mathcal S^\prime \in \mathcal{S}_k} \mathcal S^\prime$\vspace{1pt}
\If{$\mathcal T \neq \varnothing$}
\State \Return $\varphi \;\cup\;$\textsf{COA}$_t(\mathcal T)$
\EndIf
\State \Return $\varphi$
\EndProcedure
\end{algorithmic}
\end{algorithm}
\end{spacing}

\begin{myex}
Consider structure $\mathcal S(\mathcal E, \mathcal V)$ whose matrix is shown in Fig. $\!$\ref{fig:coa-a}. $\!$Note that $\mathcal S$ is complete, since $|\mathcal E|\!=\!|\mathcal V|\!=\!7$, but not minimal. The set of all minimal subsets $\mathcal S^\prime\! \subset \mathcal S$ is $\mathcal{S}_c \!=\! \{\, \{f_1\},\, \{f_2\},\, \{f_3\} \,\}$. $\!$By eliminating the variables identified at recursive step $k$, a smaller structure $\mathcal T \subset \mathcal S$ is derived. Compare the partial causal mapping eventually returned by \textsf{COA}, $\varphi_p \!=\! \{\, \langle\{f_1\}, \{x_1\}\rangle,\, \langle\{f_2\}, \{x_2\}\rangle,\, \langle\{f_3\}, \{x_3\}\rangle,$ $\langle\{f_4, f_5\},\, \{x_4, x_5\}\rangle,\,\langle\{f_6\}, \{x_6\}\rangle,\, \langle\{f_7\}, \{x_7\}\rangle \,\}$, to the total causal mapping returned by \textsf{COA}$_t$, $\,\varphi =\! \{\langle f_1, x_1\rangle,\, \langle f_2, x_2\rangle,\, \langle f_3, x_3 \rangle,\, \langle f_4, x_4 \rangle,\, \langle f_5, x_5 \rangle,\, \langle f_6, x_6\rangle,\, \langle f_7, x_7\rangle\}$. $\!$Since $x_4$ and $x_5$ are strongly coupled (see Fig.\ref{fig:coa-b}), \textsf{COA}$_t$ maps them arbitrarily (e.g., it could be $f_4 \mapsto x_5,\, f_5 \mapsto x_4$ instead). Such total mapping $\varphi$ renders a cycle in the directed causal graph $G_{\varphi}$ induced by $\varphi$ (see Fig.\ref{fig:causal-graph}). $\Box$
\label{ex:struct-matrix}
\end{myex}

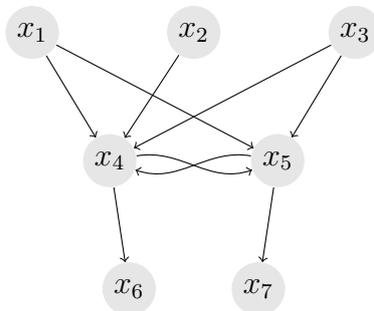
\begin{figure}[H]
\begin{center}
\tikzstyle{rect}=[rectangle,
                                    thick,
                                    minimum size=0.3cm,
                                    draw=black]
\tikzstyle{circ}=[circle,
                                    thick,
                                    minimum size=0.3cm,
                                    draw=black]
\tikzstyle{vertex}=[circle,fill=black!10,minimum size=20pt,inner sep=0pt]
\tikzstyle{selected vertex} = [vertex, fill=red!24]
\tikzstyle{edge} = [draw,thick,->,bend left]
\tikzstyle{weight} = [font=\small]
\tikzstyle{selected edge} = [draw,line width=5pt,-,red!50]
\tikzstyle{ignored edge} = [draw,line width=5pt,-,black!20]
\begin{tikzpicture}[scale=0.85]
    \foreach \pos/\name in {{(-0.5,2)/x_1}, {(2,2)/x_2}, {(4.5,2)/x_3}, 
                            		 {(0.7,0)/x_4}, {(3.3,0)/x_5},
		 			 {(1,-2)/x_6}, {(3,-2)/x_7}}
        \node[vertex] (\name) at \pos {$\name$};
    \draw[->] (x_1) to (x_4);
    \draw[->] (x_1) to (x_5);
    \draw[->] (x_2) to (x_4);
    \draw[->] (x_3) to (x_4);
    \draw[->] (x_3) to (x_5);
    \draw[->] (x_4) to[out=10,in=200] (x_5);
    \draw[->] (x_5) to[out=170,in=-20] (x_4);
    \draw[->] (x_4) to (x_6);
    \draw[->] (x_5) to (x_7);
\end{tikzpicture}
\end{center}
\caption[Directed causal graph $G_{\varphi}$ induced by mapping $\varphi$ for structure $\mathcal S$]{Directed causal graph $G_{\varphi}$ induced by mapping $\varphi$ for structure $\mathcal S$. An edge connects a node $x_i$ towards a node $x_j$, with $x_i, x_j \in \mathcal V$, iff $x_i$ appears in the equation $f \in \mathcal E$ such that $\varphi(f)=x_j$.}
\label{fig:causal-graph}
\end{figure}

\section{The Problem of Causal Ordering}\label{sec:coa}

The serious issue with Alg. \ref{alg:coat}, \textsf{COA}($_t$), is that finding all minimal structures in a given structure (cf. line 3) is a hard problem that can only be addressed heuristically as a problem of co-clustering (also called biclustering \cite{madeira2004,dhillon2003}) in Boolean matrices. Simon's approach, however, as we shall see next, is not the only way to cope with the problem of causal ordering.

In fact, in order to study the computational properties of SEM's and the problem of causal ordering, we observe that any structure $\mathcal S(\mathcal E, \mathcal V)$ satisfying Def. \ref{def:structure} can be modeled straightforwardly as a bipartite graph $G=(V_1 \cup V_2, E)$, where the set $\mathcal E$ of equations and the set $\mathcal V$ of variables are the disjoint vertex sets, i.e., $V_1 \mapsto \mathcal E$, $V_2 \mapsto \mathcal V$, and $E \mapsto \mathcal S$ is the edge set connecting equations to the variables appearing in them. Fig. \ref{fig:bipartite} shows the bipartite graph $G$ corresponding to the structure given in Example \ref{ex:struct-matrix} --- for a comprehensive text on graph concepts and its related algorithmic problems, cf. Even \cite{even2011}. 

\vspace{5pt}
\begin{figure}[H]
\begin{center}
\tikzstyle{rect}=[rectangle,
                                    thick,
                                    minimum size=0.3cm,
                                    draw=black]
\tikzstyle{circ}=[circle,
                                    thick,
                                    minimum size=0.3cm,
                                    draw=black]
\tikzstyle{vertex}=[circle,fill=black!10,minimum size=20pt,inner sep=0pt]
\tikzstyle{selected vertex} = [vertex, fill=red!24]
\tikzstyle{edge} = [draw,thick,->,bend left]
\tikzstyle{weight} = [font=\small]
\tikzstyle{selected edge} = [draw,line width=5pt,-,red!50]
\tikzstyle{ignored edge} = [draw,line width=5pt,-,black!20]
\begin{tikzpicture}[scale=0.85]
    \foreach \pos/\name in {{(0,10)/f_1}, {(6,10)/x_1}, {(0,9)/f_2}, {(6,9)/x_2},
					{(0,8)/f_3}, {(6,8)/x_3}, {(0,7)/f_4}, {(6,7)/x_4},  
					{(0,6)/f_5}, {(6,6)/x_5}, {(0,5)/f_6}, {(6,5)/x_6}, 
		 			 {(0,4)/f_7}, {(6,4)/x_7}}
        \node[vertex] (\name) at \pos {$\name$};
    \draw[-] (f_1) to (x_1);
    \draw[-] (f_2) to (x_2);
    \draw[-] (f_3) to (x_3);
    \draw[-] (f_4) to (x_1);
    \draw[-] (f_4) to (x_2);
    \draw[-] (f_4) to (x_3);
    \draw[-] (f_4) to (x_4);
    \draw[-] (f_4) to (x_5);
    \draw[-] (f_5) to (x_1);
    \draw[-] (f_5) to (x_3);
    \draw[-] (f_5) to (x_4);
    \draw[-] (f_5) to (x_5);
    \draw[-] (f_6) to (x_4);
    \draw[-] (f_6) to (x_6);
    \draw[-] (f_7) to (x_5);
    \draw[-] (f_7) to (x_7);
\end{tikzpicture}
\end{center}
\vspace{-8pt}
\caption{Bipartite graph $G$ of structure $\mathcal S$ from Example \ref{ex:struct-matrix}.}\label{fig:bipartite}
\end{figure}
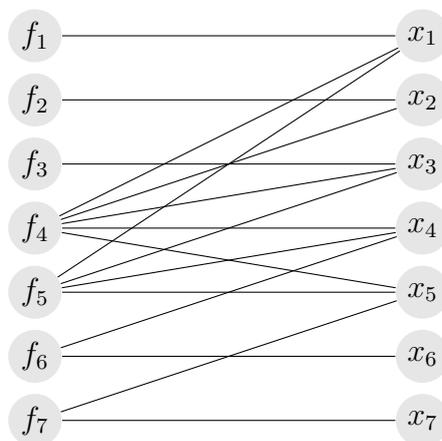

A \emph{biclique} (or complete bipartite graph) is a bipartite graph $G=(A \cup B,\, E)$ such that for every two vertices $a \in A$, $b \in B$, we have $(a,\, b) \in E$ \cite{even2011}. Note that for \emph{balanced} bicliques, i.e., when $|A|=|B|=K$, the degree \emph{deg}$(u)$ of any vertex $u \in A \cup B$ must be \emph{deg}$(u)=|A|=|B|=K$.

Recent approaches to co-clustering problems (e.g., \cite{uno2010}) have come with the notion of \emph{pseudo-biclique} (also called `quasi-biclique'), which is a relaxation of the biclique concept to allow some less rigid notion of connectivity than the `complete connectivity' required in a biclique. 
Now, recall that Simon's \textsf{COA}($_t$) needs to find, at each recursive step, all minimal subsets $\mathcal S^\prime \subseteq \mathcal S$. Theorem \ref{thm:coa-hardness} situates this particular computational task in terms of its complexity, which for $|\mathcal E^\prime| = |\mathcal V^\prime| \geq 2$ is equivalent to find, at each recursive step, the minimal-size `pseudo-bicliques' (i.e., with the least $K \geq 2$) in its corresponding bipartite graph (e.g., see Fig. \ref{fig:bipartite}). Here we take, in Def. \ref{def:pseudo-biclique}, a specific notion of pseudo-biclique. 

\begin{mydef}\label{def:pseudo-biclique}
Let $G=(A \cup B,\, E)$ be a bipartite graph. We say that $G$ is a \textbf{$K$-balanced pseudo-biclique} if $|A|=|B|=K$ with $|E| \geq 2K$ and, for all vertices $u \in A \cup B$, $deg(u)\geq 2$.
\end{mydef}
Now we state (originally) the \emph{balanced pseudo-biclique problem} (BPBP) as a decision problem as follows. 

\begin{framed}
\noindent
(BPBP). Given a bipartite graph $G=(V_1 \,\cup\, V_2,\, E)$ and a positive integer $K \geq 2$, does $G$ contain a $K$-balanced pseudo-biclique?
\end{framed}

\begin{mylemma}\label{lemma:pseudo-biclique}
The \emph{balanced pseudo-biclique problem} (BPBP) is NP-Complete.
\end{mylemma}
\begin{proof}
We show (by restriction \cite{garey1979}) that the BPBP is a generalization of the balanced biclique problem (BBP), referred `balanced complete bipartite subgraph' problem \cite[GT24, p. 196]{garey1979}, which is shown to be NP-Complete by means of a transformation from `clique' \cite[p. 446]{johnson1987}. The restriction from BPBP to BBP (special case) is made by requiring (cf. Def. \ref{def:pseudo-biclique}) either (a) $|E|=K^2$ or (b) $deg(u)=K$, for $K\geq 2$,\footnote{Note that clearly, for any positive integer $K\geq 2$, we have for $K^2 \geq 2K$.} which are equivalent ways of enforcing the inquired $K$-balanced pseudo-biclique to be a $K$-balanced biclique. 
$\Box$
\end{proof}
\noindent
We introduce another hypothesis structure (see Fig. \ref{fig:hardness}) to illustrate the correspondence between the pseudo-biclique property and \textsf{COA}'s algorithmic approach as elaborated in the proof of Theorem \ref{thm:coa-hardness}.

\vspace{-7pt}
\begin{spacing}{1.2}
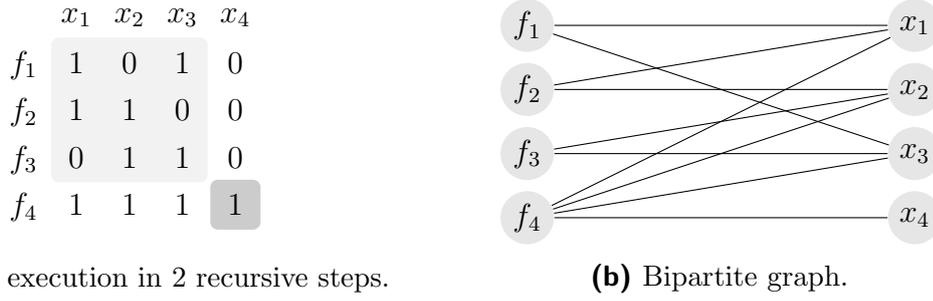
\begin{figure}[t]
\tikzset{node style ge/.style={circle,inner sep=0pt,minimum size=17pt}}
\tikzstyle{background0}=[rectangle,
                                                fill=gray!10,
                                                inner sep=0.025cm,
                                                rounded corners=1mm]
\tikzstyle{background1}=[rectangle,
                                                fill=gray!40,
                                                inner sep=0.025cm,
                                                rounded corners=1mm]
\begin{subfigure}{0.475\textwidth}
\vspace{-5pt}
\hspace{40pt}
\begin{tikzpicture}[baseline=(A.center)]
  \tikzset{BarreStyle/.style =   {opacity=.35,line width=3.25 mm,line cap=round,color=#1}}
\matrix (A) [matrix of math nodes, nodes = {node style ge},column sep=1.0 mm] {
 & \node (x1) {x_1}; & \node (x2) {x_2}; & \node (x3) {x_3}; & \node (x4) {x_4};\\
\node (f1) {f_1}; & \node (a11) {1}; & \node (a12) {0}; & \node (a13) {1}; & \node (a14) {0};\\
\node (f2) {f_2}; & \node (a21) {1}; & \node (a22) {1}; & \node (a23) {0}; & \node (a24) {0};\\
\node (f3) {f_3}; & \node (a31) {0}; & \node (a32) {1}; & \node (a33) {1}; & \node (a34) {0};\\
\node (f4) {f_4}; & \node (a41) {1}; & \node (a42) {1}; & \node (a43) {1}; & \node (a44) {1};\\
};
     \begin{pgfonlayer}{background}
        \node [background0,
                    fit=(a11) (a12) (a13) (a21) (a22) (a23) (a31) (a32) (a33) ]
                    {};
        \node [background1,
                    fit=(a44) ]
                    {};
    \end{pgfonlayer}
\end{tikzpicture}
\vspace{2pt}
\caption{\textsf{COA} execution in 2 recursive steps.}
\end{subfigure}
\begin{subfigure}{0.55\textwidth}
\hspace{-2pt}
\tikzstyle{rect}=[rectangle,
                                    thick,
                                    minimum size=0.3cm,
                                    draw=black]
\tikzstyle{circ}=[circle,
                                    thick,
                                    minimum size=0.3cm,
                                    draw=black]
\tikzstyle{vertex}=[circle,fill=black!10,minimum size=20pt,inner sep=0pt]
\tikzstyle{selected vertex} = [vertex, fill=red!24]
\tikzstyle{edge} = [draw,thick,->,bend left]
\tikzstyle{weight} = [font=\small]
\tikzstyle{selected edge} = [draw,line width=5pt,-,red!50]
\tikzstyle{ignored edge} = [draw,line width=5pt,-,black!20]
\begin{tikzpicture}[scale=0.85]
    \foreach \pos/\name in {{(0,10)/f_1}, {(6,10)/x_1}, {(0,9)/f_2}, {(6,9)/x_2},
					{(0,8)/f_3}, {(6,8)/x_3}, {(0,7)/f_4}, {(6,7)/x_4}} 
        \node[vertex] (\name) at \pos {$\name$};
    \draw[-] (f_1) to (x_1);
    \draw[-] (f_1) to (x_3);
    \draw[-] (f_2) to (x_1);
    \draw[-] (f_2) to (x_2);
    \draw[-] (f_3) to (x_2);
    \draw[-] (f_3) to (x_3);
    \draw[-] (f_4) to (x_1);
    \draw[-] (f_4) to (x_2);
    \draw[-] (f_4) to (x_3);
    \draw[-] (f_4) to (x_4);
\end{tikzpicture}
\caption{Bipartite graph.}
\end{subfigure}
\caption{Another hypothesis structure example.}
\label{fig:hardness}
\end{figure}
\end{spacing}

\begin{mythm}\label{thm:coa-hardness}
Let $\mathcal S(\mathcal E, \mathcal V)$ be a complete structure. Then the extraction of its causal ordering by Simon's $\textsf{COA}(\mathcal S)$ is NP-Hard.
\end{mythm}
\begin{proof}
We show that, at each recursive step $k$ of \textsf{COA}, to find all non-trivial minimal subsets (i.e., $|\mathcal E^\prime| \geq 2$) translates into an optimization problem associated with the decision problem BPBP, which we know by Lemma \ref{lemma:pseudo-biclique} to be NP-Complete. See Appendix \S\ref{a:coa-hardness}. 
$\Box$
\end{proof}

Nonetheless, the problem of causal ordering can be solved efficiently by means of a different, less notorious approach due to Nayak \cite{nayak1994}, which we introduce and build upon next.

\section{Total Causal Mappings}\label{sec:tcm}

The problem of causal ordering can be solved in polynomial time by (i) finding any total causal mapping $\varphi\!:\, \mathcal E \to \mathcal V$ over structure $\mathcal S$ given (cf. Def. \ref{def:tcm}); and then (ii) by computing the transitive closure $C^+_{\varphi}$ of the set $C_{\varphi}$ (cf. Eq. \ref{eq:direct-causal-dependencies}) of \emph{direct causal dependencies} induced by $\varphi$. 
\begin{eqnarray}\label{eq:direct-causal-dependencies}
\!\!\!C_{\varphi} \!= \{\, (x_a, x_b) \,| \;\text{there exists}\; f \!\in \mathcal E \;\text{such that}\; \varphi(f) = x_b \;\text{and}\; x_a \in Vars(f)  \,\}
\end{eqnarray}
\vspace{-30pt}
\begin{mydef}\label{def:causal-dependency}
Let $\mathcal S(\mathcal E, \mathcal V)$ be a structure with variables $x_a, x_b \in \mathcal V$, and $\varphi$ a total causal mapping over $\mathcal S$ inducing set of direct causal dependencies $C_{\varphi}$ and its transitive closure $C^+_{\varphi}$. We say that $(x_b, x_a)$ is a \textbf{direct causal dependency} in $\mathcal S$ if $(x_b, x_a) \in C_{\varphi}$, and that $(x_b, x_a)$ is a \textbf{causal dependency} in $\mathcal S$ if $(x_b, x_a) \in C^+_{\varphi}$.
\end{mydef}

In other words, $(x_a, x_b)$ is in $C_{\varphi}$ iff $x_b$ direct and causally depends on $x_a$, given the causal asymmetries induced by $\varphi$. 
Those notions open up an approach to causal reasoning that fits very well to our use case, which is aimed at encoding hypothesis structures into fd sets and then performing (symbolic) causal reasoning in terms of acyclic pseudo-transitive reasoning over fd's (cf. Chapter \ref{ch:reasoning}).\footnote{Note that it differs from AI research (e.g., \cite{druzdzel2008}) geared for reasoning over GM's.}
For it to be effective, nonetheless, we shall need to ensure some properties of total causal mappings first. 

For a given structure $\mathcal S$, there may be multiple total causal mappings over $\mathcal S$ (recall Example \ref{ex:struct-matrix}). But the causal ordering of $\mathcal S$ must be unique (see Fig. \ref{fig:causal-graph}). Therefore, a question that arises is whether the transitive closure $C^+_{\varphi}$ is the same for any total causal mapping $\varphi$ over $\mathcal S$. Proposition \ref{prop:causal-ordering}, originally from Nayak \cite{nayak1994}, ensures that is the case.

\begin{myprop}\label{prop:causal-ordering}
Let $\mathcal S(\mathcal E, \mathcal V)$ be a structure, and $\varphi_1\!:\, \mathcal E \to \mathcal V$ and $\varphi_2\!:\, \mathcal E \to \mathcal V$ be any two total causal mappings over $\mathcal S$. Then $C^+_{\varphi_1}$ = $C^+_{\varphi_2}$. 
\end{myprop}
\begin{proof}
The proof is based on an argument from Nayak \cite{nayak1994}, which we present in arguably much clearer way (see Appendix, \S\ref{a:causal-ordering}). Intuitively, it shows that if $\varphi_1$ and $\varphi_2$ differ on the variable an equation $f$ is mapped to, then such variables, viz., $\varphi_1(f)$ and $\varphi_2(f)$, must be causally dependent on each other (strongly coupled). 
$\Box$
\end{proof}

Another issue is concerned with the precise conditions under which total causal mappings exist (i.e., whether or not all variables in the equations can be causally determined). In fact, by Proposition \ref{prop:mapping-existence}, based on Nayak \cite{nayak1994} apud. Hall \cite[p. 135-7]{even2011}, we know that the existence condition holds iff the given structure is complete. 

Before proceeding to it, let us refer to Even \cite{even2011} to briefly introduce the additional graph-theoretic concepts which are necessary here. A \emph{matching} in a graph is a subset of edges such that no two edges in the matching share a common node. A matching is said \emph{maximum} if no edge can be added to the matching (without hindering the matching property). Finally, a matching in a graph is said `perfect' if every vertex is an end-point of some edge in the matching --- in a bipartite graph, a perfect matching is said a \emph{complete} matching.

\begin{myprop}\label{prop:mapping-existence}
Let $\mathcal S(\mathcal E, \mathcal V)$ be a structure. Then a total causal mapping $\varphi\!:\, \mathcal E \to \mathcal V$ over $\mathcal S$ exists iff $\mathcal S$ is complete. 
\end{myprop}
\begin{proof}
We observe that a total causal mapping $\varphi\!:\, \mathcal E \to \mathcal V$ over $\mathcal S$ corresponds exactly to a complete matching $M$ in a bipartite graph $B = (V_1 \cup V_2, E)$, where $V_1 \mapsto \mathcal E$, $V_2 \mapsto \mathcal V$, and $E \mapsto \mathcal S$.
In fact, by Even apud. Hall's theorem (cf. \cite[135-7]{even2011}), we know that $B$ has a complete matching iff (a) for every subset of vertices $F \subseteq V_1$, we have $|F| \leq |E(F)|$, where $E(F)$ is the set of all vertices connected to the vertices in $F$ by edges in $E$; and (b) $|V_1|=|V_2|$.  
By Def. \ref{def:structure} (no subset of equations has fewer variables than equations), and Def. \ref{def:complete} (number of equations is the same as number of variables), it is easy to see that conditions (a) and (b) above hold iff $\mathcal S$ is a complete structure. 
$\Box$
\end{proof}

The problem of finding a maximum matching is a well-studied algorithmic problem. In this thesis we adopt the Hopcroft-Karp algorithm \cite{karp1973}, which is known to be polynomial-time, bounded by $O(\sqrt{|V_1|+|V_2|}\,|E|)$.\footnote{The Hopcroft-Karp algorithm solves maximum matching in a bipartite graph efficiently as a problem of finding maximum flow in a network (cf. \cite[p. 135-7]{even2011}, or \cite[p. 664-9]{cormen2009}).} That is, we handle the problem of total causal mapping by (see Alg. \ref{alg:tcm}) translating it to the problem of maximum matching in a bipartite graph (in linear time) and then applying the Hopcroft-Karp algorithm to get the matching and finally translate it back to the total causal mapping, as suggested by the proof of Proposition \ref{prop:mapping-existence}.

\begin{spacing}{1.1}
\begin{algorithm}[H]
\caption{Find a total causal mapping for a given structure.}
\label{alg:tcm}
\begin{algorithmic}[1]
\Procedure{TCM}{$\mathcal S\!:\, \text{structure over}\; \mathcal E \;\text{and}\; \mathcal V$}
\Require $\mathcal S$ given is a complete structure, i.e., $|\mathcal E|=|\mathcal V|$
\Ensure Returns a total causal mapping $\varphi$
\State $B(V_1 \cup V_2, E) \gets \varnothing$
\State $\varphi \gets \varnothing$
\ForAll{$\langle f, X\rangle \in \mathcal S$}\Comment{translates the structure $\mathcal S$ to a bipartite graph $B$}\vspace{1pt}
\State $V_1 \gets V_1 \cup \{f\}$
\ForAll{$x \in X$}\vspace{1pt}
\State $V_2 \gets V_2 \cup \{x\}$
\State $E \gets E \cup \{(f, x)\}$
\EndFor
\EndFor
\State $M \gets \textsf{Hopcroft-Karp}(B)$\Comment{solves the maximum matching problem}
\ForAll{$(f, x) \in M$}\vspace{1pt}\Comment{translates the matching to a total causal mapping}
\State $\varphi \gets \varphi \cup \{\langle f, x\rangle\}$
\EndFor
\State \Return $\varphi$\vspace{1pt}
\EndProcedure
\end{algorithmic}
\end{algorithm}
\end{spacing}

\noindent
Fig. \ref{fig:matching} shows the complete matching found by the \textsf{Hopcroft-Karp} algorithm for the structure given in Example \ref{ex:struct-matrix}.

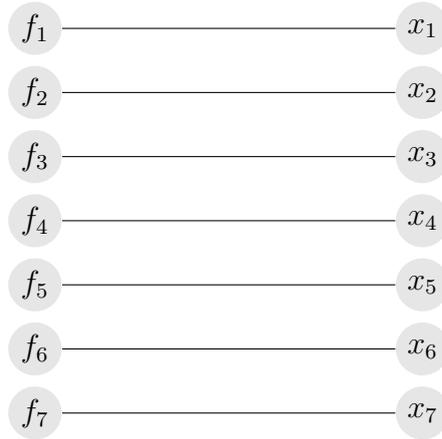
\begin{figure}[H]
\begin{center}
\tikzstyle{rect}=[rectangle,
                                    thick,
                                    minimum size=0.3cm,
                                    draw=black]
\tikzstyle{circ}=[circle,
                                    thick,
                                    minimum size=0.3cm,
                                    draw=black]
\tikzstyle{vertex}=[circle,fill=black!10,minimum size=20pt,inner sep=0pt]
\tikzstyle{selected vertex} = [vertex, fill=red!24]
\tikzstyle{edge} = [draw,thick,->,bend left]
\tikzstyle{weight} = [font=\small]
\tikzstyle{selected edge} = [draw,line width=5pt,-,red!50]
\tikzstyle{ignored edge} = [draw,line width=5pt,-,black!20]
\begin{tikzpicture}[scale=0.85]
    \foreach \pos/\name in {{(0,12)/f_1}, {(6,12)/x_1}, {(0,11)/f_2}, {(6,11)/x_2},
					{(0,10)/f_3}, {(6,10)/x_3}, {(0,9)/f_4}, {(6,9)/x_4},  
					{(0,8)/f_5}, {(6,8)/x_5}, {(0,7)/f_6}, {(6,7)/x_6}, 
		 			 {(0,6)/f_7}, {(6,6)/x_7}}
        \node[vertex] (\name) at \pos {$\name$};
    \draw[-] (f_1) to (x_1);
    \draw[-] (f_2) to (x_2);
    \draw[-] (f_3) to (x_3);
    \draw[-] (f_4) to (x_4);
    \draw[-] (f_5) to (x_5);
    \draw[-] (f_6) to (x_6);
    \draw[-] (f_7) to (x_7);
\end{tikzpicture}
\end{center}
\vspace{-5pt}
\caption{Complete matching $M$ for structure $S$ from Example \ref{ex:struct-matrix}.}\label{fig:matching}
\end{figure}

\noindent
Corollary \ref{cor:tcm} summarizes the results we have so far.

\begin{mycor}\label{cor:tcm}
Let $\mathcal S(\mathcal E, \mathcal V)$ be a complete structure. Then a total causal mapping $\varphi\!:\, \mathcal E \to \mathcal V$ over $\mathcal S$ can be found by (Alg. \ref{alg:tcm}) \textsf{TCM} in time that is bounded by $O(\sqrt{|\mathcal E|}\,|\mathcal S|)$. 
\end{mycor}
\begin{proof}
Let $B = (V_1 \cup V_2, E)$ be the bipartite graph corresponding to complete structure $\mathcal S$ given to \textsf{TCM}, where $V_1 \mapsto \mathcal E$, $V_2 \mapsto \mathcal V$, and $E \mapsto \mathcal S$. 
The translation of $\mathcal S$ into $B$ is done by a scan over it. This scan is of length $|\mathcal S| = |E|$. Note that number $|E|$ of edges rendered is precisely the length $|\mathcal S|$ of structure, where the denser the structure, the greater $|\mathcal S|$ is. The re-translation of the matching computed by internal procedure \textsf{Hopcroft-Karp}, in turn, is done at expense of $|\mathcal E| = |\mathcal V| \leq |\mathcal S|$. Thus, it is easy to see that \textsf{TCM} is dominated by the maximum matching algorithm \textsf{Hopcroft-Karp}, which is known to be $O(\sqrt{|V_1|+|V_2|}\,|E|)$, i.e., $O(\sqrt{|\mathcal E|+|\mathcal V|}\,|\mathcal S|)$. Since $\mathcal S$ is assumed complete, we have $|\mathcal E| \!=\! |\mathcal V|$ then $\sqrt{|\mathcal E|+|\mathcal V|} = \sqrt{2}\,\sqrt{|\mathcal E|}$. Therefore, \textsf{TCM} must have running time at most $O(\sqrt{|\mathcal E|}\,|\mathcal S|)$.
$\Box$
\end{proof}

\begin{myremark}\label{rmk:correct-mapping}
Let $\mathcal S(\mathcal E, \mathcal V)$ be a complete structure. Then we know (cf. Proposition \ref{prop:mapping-existence}) that a total causal mapping over $\mathcal S$ exists. Let it be defined $\varphi \triangleq \textsf{TCM}(\mathcal S)$. Then the causal ordering implicit in $\mathcal S$ shall be correctly extracted (cf. Proposition \ref{prop:causal-ordering}) by processing the causal dependencies induced by $\varphi$, as we show in Chapter \ref{ch:reasoning}. 
$\Box$
\end{myremark}

Now we are ready to accomplish the hypothesis encoding into fd's, as we show next.

\section{The Encoding Scheme}\label{sec:scheme}

We shall encode variables as relational attributes and map equations onto fd's through total causal mappings. 
Let $Z$ be a set of attribute symbols such that $Z \!\simeq\! \mathcal V$, where $\mathcal S(\mathcal E, \mathcal V)$ is a complete structure; and let $\phi,\, \upsilon \notin Z$ be two special attribute symbols kept to identify (resp.) phenomena and hypotheses. We are explicitly distinguishing symbols in $Z$, assigned by the user into structure $\mathcal S$, from epistemological symbols $\phi$ and $\upsilon$. 
Then we consider a sense of Simon's into the nature of scientific modeling and interventions \cite{simon1953}, summarized in Def. \ref{def:exogenous}.

\begin{mydef}\label{def:exogenous}
Let $\mathcal S(\mathcal E, \mathcal V)$ be a structure and $x \in \mathcal V$ be a variable. We say that $x$ is \textbf{exogenous} if there exists an equation $f \!\in \mathcal E$ such that $Vars(f) = \{x\}$. In this case $f$ can be written $f(x) \!= 0$, and must be mapped to $x$ in any total causal mapping $\varphi$ over $\mathcal S$. We say that $x$ is \textbf{endogenous} otherwise.
\end{mydef}

\noindent
Remark \ref{rmk:modeling} introduces an interpretation of Def. $\!$\ref{def:exogenous} with a data dependency flavor.

\begin{myremark}
The values of exogenous variables (attri\-butes) are to be determined empirically, outside of the system (proposed structure $\mathcal S$). Such values are, therefore, dependent on the phenomenon id $\phi$ only. The values of endogenous variables (attributes) are in turn to be determined theoretically, within the system. They are dependent on the hypothesis id $\upsilon$ and shall be dependent on the phenomenon id $\phi$ as well indirectly. 
$\Box$
\label{rmk:modeling}
\end{myremark}

As introduced in \S\ref{sec:encoding}, the encoding scheme we are presenting here is not obvious. It goes beyond Simon's structural equations to abstract the data-level semantics of mathematical deterministic hypotheses. Whereas Simon's structural equations are able to represent only linear equations, our encoding scheme can represent non-linear equations and arbitrarily complex mathematical operators by means of its data representation of deterministic hypotheses.

For instance, take non-linear equation $y = a\, x^2$ and suppose that, considering the context of its complete system of equations, (Alg. \ref{alg:tcm}) \textsf{TCM} maps it onto variable $y$. Then, by an abstraction of the equation semantics, we shall encode it into fd $a \,x\, \upsilon \to y$. That is, the hypothesis identifier $\upsilon$ captures the data-level semantics of the hypothesis equation.\footnote{Note that, without the hypothesis id, infinitely many equations fit the pattern $a \,x \to y$.}

We encode complete structures into fd sets by means of (Alg. $\!$\ref{alg:h-encode}) \textsf{h-encode}. $\!$Fig. \ref{fig:h-fdschema} presents an fd set defined $\Sigma \triangleq$ \textsf{h-encode}($\mathcal S$), encoding the same structure $\mathcal S$ from Example \ref{ex:struct-matrix}.

\begin{spacing}{1.1}
\begin{algorithm}[H]
\caption{Hypothesis encoding.}
\label{alg:h-encode}
\begin{algorithmic}[1]
\Procedure{h-encode}{$\mathcal S\!:\, \text{structure over}\; \mathcal E \;\text{and}\; \mathcal V$, $\;\mathcal D\!:\,$ domain variables}
\Require $\mathcal S$ given is a complete structure, i.e., $|\mathcal E|=|\mathcal V|$
\Ensure Returns a non-redundant fd set $\Sigma$
\State $\Sigma \gets \varnothing$
\State $\varphi \gets \textsf{TCM}(\mathcal S)$\vspace{1pt}
\ForAll{$\langle f,\, x \rangle \in \varphi_t$}\vspace{1pt}
\State $Z \gets X \setminus \{x\}$, where $\langle f, X\rangle \in \mathcal S$
\If{$Z = \varnothing\;$ or $\;Z \subseteq \mathcal D$} \Comment{$x$ is exogenous}
\If{$x \notin \mathcal D$} \Comment{supress $\phi$-fd for dimensions like time $t$}
\State $\Sigma \gets \Sigma \cup \langle Z \cup \{\phi\}, x \rangle$
\EndIf
\Else \Comment{$x$ is endogenous}
\State $\Sigma \gets \Sigma \cup \langle Z \cup \{\upsilon\},\; x \rangle$
\EndIf
\EndFor
\State \Return $\Sigma$
\EndProcedure
\end{algorithmic}
\end{algorithm}
\end{spacing}

\begin{figure}[t]
\begin{framed}
\vspace{-15pt}
\begin{eqnarray*}
\Sigma = \{\quad 
\phi &\to& x_1,\\ 
\phi &\to& x_2,\\ 
\phi &\to& x_3,\\ 
x_1\,x_2\,x_3\,x_5\,\upsilon &\to& x_4,\\
x_1\,x_3\,x_4\,\upsilon &\to& x_5,\\
x_4\,\upsilon &\to& x_6,\\
x_5\,\upsilon &\to& x_7 \quad\}.
\end{eqnarray*}
\vspace{-20pt}
\end{framed}
\vspace{-5pt}
\caption{Encoded fd set $\Sigma$ (cf. $\!$Alg. $\!$\ref{alg:h-encode}) for the structure from Example \ref{ex:struct-matrix}.}
\label{fig:h-fdschema}
\end{figure}


Now we study the design-theoretic properties of the encoded fd sets. We shall make use of the concept of `canonical' fd sets (also called `minimal' \cite[p.$\!$ 390]{ullman1988}), see Def. $\!$\ref{def:minimal}.
\begin{mydef}
Let $\Sigma$ be an fd set. We say that $\Sigma$ is \textbf{canonical} if:
\begin{itemize}
\vspace{-2pt}
\item[(a)] each fd in $\Sigma$ has the form $X \!\to A$, where $|A|=1$;\vspace{-2pt}
\item[(b)] For no $\langle X, A\rangle \in \Sigma$ we have $(\Sigma - \{\langle X, A\rangle\})^+ = \Sigma^+$;\vspace{-2pt}
\item[(c)] for each fd $X \!\to A$ in $\Sigma$, there is no $Y \subset X$ such that $(\Sigma \setminus \{X \!\to A\} \cup \{Y \!\to A \})^+ = \Sigma^+$.
\end{itemize}
\label{def:minimal}
\end{mydef}
For an fd set satisfying such properties (Def. $\!$\ref{def:minimal}) individually, we say that it is (a) \emph{singleton-rhs}, (b) \emph{non-redundant} and (c) \emph{left-reduced}. It is said to have an attribute $A$ in $X$ that is `extraneous' w.r.t. $\!\Sigma$ if it is \emph{not} left-reduced (Def. \ref{def:minimal}-c) \cite[p. $\!$74]{maier1983}. Finally, an fd $X \!\to Y$ in $\Sigma$ is said \emph{trivial} if $Y \subseteq X$. Note that the presence of a trivial fd in a an fd set is sufficient to make it redundant.

\begin{mythm}\label{thm:non-redundant}
Let $\Sigma$ be an fd set defined $\Sigma \triangleq$ \textsf{h-encode}($\mathcal S$) for some complete structure $\mathcal S$. Then $\Sigma$ is non-redundant and singleton-rhs but may not be left-reduced (then may not be canonical). 
\end{mythm}
\begin{proof}
We show that properties (a-b) of Def. \ref{def:minimal} must hold for $\Sigma$ produced by (Alg. \ref{alg:h-encode}) \textsf{h-encode}, but property (c) may not hold (i.e., encoded fd set $\Sigma$ may not be left-reduced). See Appendix, \S\ref{a:non-redundant}. 
$\Box$
\end{proof}

We draw attention to the significance of Theorem \ref{thm:non-redundant}, as it sheds light on a connection between Simon's complete structures \cite{simon1953} and fd sets \cite{ullman1988}. In fact, we continue to elaborate on that connection in next chapter to handle causal ordering processing symbolically by causal reasoning over fd's.

\section{Experiments}\label{sec:tcm-experiments}

\noindent
Fig. \ref{fig:tcm-experiments} shows the results of experiments we have carried out in order to study how effective the procedure of hypothesis encoding is in practice, in particular its behavior for hypotheses whose structure $\mathcal S$ has been randomly generated over orders of magnitude $|\mathcal S| \approx 2^k$, to have length up to $|\mathcal S| \approx 2^{20} \lesssim 1M$. The largest structure considered, with $|\mathcal S|\approx 1M$, has been generated to have exactly $|\mathcal E| = 2.5K$. 

We have executed ten runs for each tested order of magnitude, and then taken its mean running time in $ms$.\footnote{The experiments were performed on a 2.3GHZ/4GB Intel Core i5 running Mac OS X 10.6.8.} The plot is shown in Fig. \ref{fig:tcm-experiments} in logscale base 2. In fact a a near-, sub-quadratic slope is expected for the curve structure length $|\mathcal S| \times \text{time}$.

These scalability results are compatible with the computational complexity of \textsf{h-encode}, which is (cf. Corollary \ref{cor:tcm}) bounded by $O(\sqrt{|\mathcal E|}\,|\mathcal S|)$.\footnote{Note that, for any arbitrary structure $\mathcal S(\mathcal E, \mathcal V)$, we have $|\mathcal E| \leq |\mathcal S| \leq |\mathcal E|^2$. So, in worst case (the densest structure possible) we have $|\mathcal S| = |\mathcal E|^2$ and then can establish a time bound function of $|\mathcal E|$ only, viz., $O(|\mathcal E|^2\sqrt{|\mathcal E|})$.}

\begin{figure}[H]
\advance\leftskip 0.6cm
\begin{tikzpicture}[y=.015cm, x=2cm,font=\sffamily,scale=1.0]
    \begin{axis}[
        height=0.4\textwidth,
        width=0.75\textwidth,
        xlabel={structure length $|\mathcal S|$},
        ylabel={time [ms]},
        xmode=log,
       log basis x={2},
        ymode=log,
       log basis y={2}
]
      \addplot[mark=*, mark size=1.2pt, mark options={fill=black}] 
		file {./data/tcm.data}; 
     \end{axis}
\end{tikzpicture}
\caption{Performance of hypothesis encoding (in logscale).}\label{fig:tcm-experiments}
\end{figure}
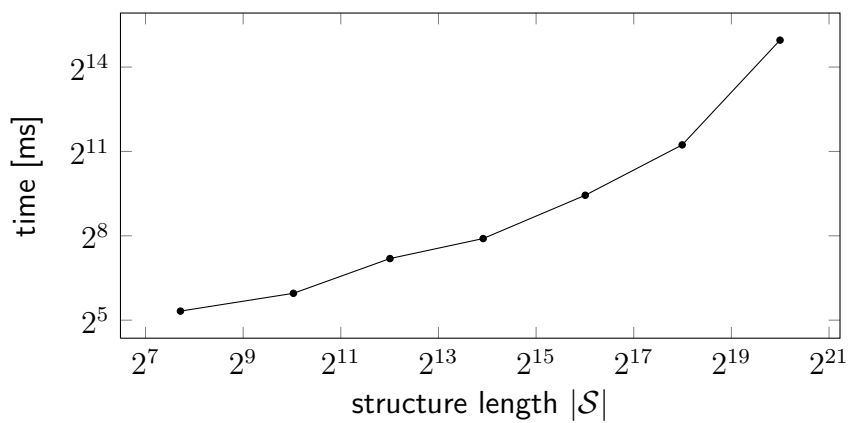

\section{Related Work}\label{sec:related-work-encoding}
\noindent
Modeling physical and socio-economical systems as a set of equations is a traditional modeling approach, and a very large bulk of models exist up to date. 
Simon's early work on structural equations and causal ordering comprises a specific notion of causality aimed at further contributing to the potential of such modeling approach (cf. \cite{simon1952}). It is meant for \emph{identifying} influences among variables (or their values) that are implicit in the system model for enabling \emph{informed} interventions. These may apply either to the system (phenomenon) under study, or to the model itself (say, when its predictions are not approximating observations very well).

Significant research effort has been devoted to causal modeling and reasoning in the past decades in both statistics and AI (cf. \cite{pearl2000,darwiche2010}). 
The notion of causality used can be traced back to the early work of Simon's and others in Econometrics. Nonetheless, there are two important differences to be emphasized:

\begin{itemize}
\item Such work is majorly devoted to deal with (statistical) qualitative hypotheses, not (deterministic) quantitative hypotheses;

\item The causal model is assumed as given or is derived from data, instead of being converted or synthesized from a set of equations. 
\end{itemize}

These are both core differences that also apply to our work in comparison to the bulk of existing work in probabilistic DB's. 
Our main point here, though, is to clarify the technical context and state of the art of the problem of causal ordering. A few works have been concerned with extracting a causal model out of some previous existing formal specification such as a set of equations. This is a reason why causal ordering has been an yet barely studied problem from the computational point of view.

Dash and Druzdzel revisit the problem and re-motivate it in light of modern applications \cite{druzdzel2008}. First, they provide a formal description of how Simon's \textsf{COA} gives a summary of the causal dependencies implicit in a given SEM. That is, in clustering the strongly coupled variables into a causal graph, \textsf{COA} provides a condensed representation of the causal model implicit in the given SEM. They show then that any valid total causal mapping produced for a given SEM must be consistent with \textsf{COA}'s partial causal mapping. 

Yet, the serious problem is that the algorithm turns out to be intractable. In fact, no formal study of \textsf{COA}'s computational properties can yet be found in the literature. 
In this thesis we have obtained the (negative) hardness result that it is intractable, which turns out to be compatible with Nayak's intuition (sic.) that it is a worst-case exponential time algorithm (cf. \cite[p. 37]{nayak1996}). 

Inspired on Serrano and Gossard's work on constraint modeling and reasoning \cite{serrano1987}, Nayak reports an approach that is provably quite effective to process the causal ordering: extract a total causal mapping and then compute the transitive closure of the direct causal dependencies. In this thesis we build upon it to perform causal reasoning in terms of a form of transitive reasoning. Such approach fits very well to our use case, viz., the synthesis of p-DB's from fd's. As we show in Chapter \ref{ch:reasoning}, we process the causal ordering of a hypothesis structure (abstracted as a SEM) in terms of acyclic causal reasoning over fd's and prove its correctness. This is enabled by the encoding scheme presented in this chapter.

\section{Summary of Results }\label{sec:encoding-conclusions}

In this chapter we have studied and developed an encoding scheme to process the mathematical structure of a deterministic hypothesis into a set of fd's towards the encoding of hypotheses `as data.' Then we have studied the design-theoretic properties held by such an encoded fd set. We list the results achieved as follows.

\begin{itemize}
\item By Theorem \ref{thm:coa-hardness}, we know (an original hardness result) that Simon's approach to process the causal ordering of a structure is intractable; 

\item By building upon on the work of Simon \cite{simon1953} and Nayak \cite{nayak1994} (cf. Propositions \ref{prop:causal-ordering} and \ref{prop:mapping-existence}), we have framed an approach to efficiently extract the basic information (a total causal mapping) for processing the causal ordering implicit in the mathematical structure of a deterministic hypothesis; 

\item By Corollary \ref{cor:tcm}, we know how to process the complete structure of a hypothesis into a total causal mapping in time that is bounded by $O(\sqrt{|\mathcal E|}\,|\mathcal S|)$. That is, the machinery of hypothesis encoding is provably suitable for very large hypothesis structures.

\item By Theorem \ref{thm:non-redundant}, which studies the design-theoretic properties of the encoded fd sets, we have unraveled the connection between Simon's complete structures and fd sets to further explore it in next chapter. 

\item We have performed experiments (cf. Fig. \ref{fig:tcm-experiments}) to study how effective the approach is in practice, or how it scales for hypotheses whose structure $\mathcal S$ is randomly generated to have length up to the order of $|\mathcal S| \lesssim 1M$.\footnote{The tests were up to this order only because of the hardware limitations of our experimental settings. In theory (cf. complexity time bounds), larger structures can be handled very efficiently.}
\end{itemize}


\chapter{Causal Reasoning over FD's}\label{ch:reasoning}
\noindent
In this chapter we present a technique to address the problem of causal ordering processing in order to enable the synthesis of U-relational DB's. In \S\ref{sec:armstrong} we introduce Armstrong's classical inference system to reason over fd's. In \S\ref{sec:ptc} we develop the core concept and algorithm of the folding of an fd set, as a method for acyclic causal reasoning over fd's. 
In \S\ref{sec:connections} we show its connections (equivalence) with causal reasoning. In \S\ref{sec:folding-experiments} we present experiments on how the method behaves in practice. \S\ref{sec:related-work-reasoning} we discuss related work. In \S\ref{sec:reasoning-conclusions} we conclude the chapter.

\section{Preliminaries: Armstrong's Inference Rules}\label{sec:armstrong}

\noindent
As usual notational conventions from the DB literature \cite{ullman1988,abiteboul1995}, we write $X, Y, Z$ to denote sets of relational attributes and $A, B, C$ to denote singleton attribute sets. Also, we write $XY$ as shorthand for $X \cup Y$. 

Functional dependency theory relies on Armstrong's inference rules (or axioms) of (R0) reflexivity, (R1) augmentation and (R2) transitivity, which forms a sound and complete inference system for reasoning over fd's \cite{ullman1988}. From R0-R2 one can derive additional rules, viz., (R3) decomposition, (R4) union and (R5) pseudo-transi\-tivity. 
 \vspace{3pt}\\ 
\indent
\textbf{R0}. If $Y \!\subseteq X$, then $X \!\to Y$;\vspace{1pt}\\
\indent
\textbf{R1}. If $X \!\to Y$, then $XZ \!\to YZ$;\vspace{1pt}\\
\indent
\textbf{R2}. $\!$If $X \!\to Y\!$ and $Y \!\to W\!$, then $X \!\to W$;\vspace{1pt}\\
\indent
\textbf{R3}. If $X \!\to YZ$, then $X \!\to Y$ and $X \!\to Z$;\vspace{1pt}\\
\indent
\textbf{R4}. If $X \!\to Y$ and $X \!\to Z$, then $X \!\to YZ$;\vspace{1pt}\\
\indent
\textbf{R5}. $\!$If $X \!\to Y\!$ and $YZ \!\to W\!$, then $XZ \!\to W$.
\vspace{5pt}

\noindent
Given an fd set $\Sigma$, one can obtain $\Sigma^+$, the closure of $\Sigma$, by a finite application of rules R0-R5. We are concerned with reasoning over an fd set in order to process its implicit causal ordering. The latter, as we shall see in \S\ref{sec:connections}, can be performed in terms of \mbox{(pseudo-)}transitive reasoning. Note that R2 is a particular case of R5 when $Z\!=\!\varnothing$, then we shall refer to R5 reasoning and understand R2 included. The next definition opens up a way to compute $\Sigma^+$ very efficiently.

Let $\Sigma$ be an fd set on attributes $U$, with $X \subseteq U$. Then $X^+$, the \emph{attribute closure} of $X$ w.r.t. $\Sigma$, is the set of attributes $A$ such that $\langle X, A\rangle \in \Sigma^+$.
%
\noindent
Bernstein has long given algorithm (Alg. \ref{alg:xclosure}) \textsf{XClosure} to compute $X^+$. It is polynomial time in $|\Sigma| \cdot |U|$ (cf. \cite{bernstein1976}), where $\Sigma$ and $U$ are (resp.) the given fd set and the attribute set over which it is defined. A tighter time bound (linear time in $|\Sigma| \cdot |U|$) is achievable as discussed further in Remark \ref{rmk:linear}.

\vspace{10pt}
\begin{spacing}{1.1}
\begin{algorithm}[H]
\caption{Attribute closure $X^+$ (cf. \cite[p. $\!$388]{ullman1988}).}
\label{alg:xclosure}
\begin{algorithmic}[1]
\Procedure{XClosure}{$\Sigma\!: \text{fd set},\; X\!: \text{attribute set}$}
\Require $\Sigma$ is an fd set, $X$ is a non-empty attribute set
\Ensure $X^+$ is the attribute closure of $X$ w.r.t. $\Sigma$
\State size $\gets 0$
\State $\Lambda \gets \varnothing$
\State $X^+ \gets X$
\While{size $< |X^+|$ }
\State size $\gets |X^+|$
\State $\Sigma \gets \Sigma \setminus \Delta$
\ForAll{$\langle Y,\, Z \rangle \in \Sigma$}
\If{$Y \subseteq X^+$}
\State $\Delta \gets \Delta \cup \{\langle Y,\, Z \rangle\}$ \Comment{consumes fd}
\State $X^+ \gets\, X^+ \cup Z$
\EndIf
\EndFor
\EndWhile
\State \Return $X^+$
\EndProcedure
\end{algorithmic}
\end{algorithm}
\end{spacing}
\vspace{-2pt}

\section{Acyclic Pseudo-Transitive Reasoning}\label{sec:ptc}
\noindent
As discussed in the previous chapter, we shall process the causal ordering in terms of computing the transitive closure of each endogenous variable (predictive attribute). Before we proceed to that, we shall develop some machinery to reason over fd's in terms of Armstrong's rule R5 (pseudo-transitivity). We shall then demonstrate the correspondence between this kind of reasoning with causal reasoning shortly in the sequel.

\vspace{-2pt}
\begin{mydef}\label{def:ptc}
Let $\Sigma$ be a set of fd's on attributes $U\!$. Then $\Sigma^{\triangleright}$, the \textbf{pseudo-transitive closure} of $\Sigma$, is the minimal set $\Sigma^{\triangleright} \supseteq \Sigma$ such that $X \!\to Y$ is in $\Sigma^{\vartriangleright}\!$, with $XY\! \subseteq U$, iff it can be derived from a finite (possibly empty) application of rule R5 over fd's in $\Sigma$. In\vspace{-2pt} that case, we may write $X \!\xrightarrow{\triangleright} Y$ and omit `w.r.t. $\!\Sigma$' if it can be understood from the context.
\end{mydef}

We are in fact interested in a very specific proper subset of $\Sigma^\vartriangleright\!$, say, a kernel of fd's in $\Sigma^\vartriangleright$ that gives a ``compact'' representation of the causal ordering implicit in $\Sigma$. Note that, to characterize such special subset we shall need to be careful w.r.t. the presence of cycles in the causal ordering.

\begin{framed}
\vspace{-13pt}
\begin{mydef}\label{def:folded}
Let $\Sigma$ be a set of fd's on attributes $U\!$, and $\langle X, \!A\rangle \in \Sigma^\vartriangleright\!$ with $X\!A \subseteq U$. $\!$We say that $X \!\to\! A$ is \textbf{folded} (w.r.t.$\!$ $\Sigma$), and write $X \xrightarrow{\looparrowright} A$, if it is non-trivial and for no $Y \subset U$ with $Y \!\nsupseteq X$, we have $Y \!\to X$ and $X \!\not\to Y$ in $\Sigma^+$.
\end{mydef}
\vspace{-13pt}
\end{framed}

The intuition of Def. \ref{def:folded} is that an fd is folded when there is no sense in going on with pseudo-transitive reasoning over it anymore (nothing new is discovered). Given an fd $X \!\to A$ in fd set $\Sigma$, we shall be able to find some folded fd $Z \!\to A$ by applying (R5) pseudo-transitivity as much as possible while ruling out cyclic or trivial fd's in some clever way.

\begin{mydef}
$\!$Let $\Sigma$ be an fd set on attributes $U$, and $\langle X, A\rangle \in \Sigma$ be an fd with $XA \subseteq U$. Then,
\begin{itemize}
\item[(a)] $A^\looparrowright\!$, the \textbf{(attribute) folding} of $A$ (w.r.t. $\!\!\Sigma$) is an\vspace{-2pt} attribute set $Z \!\subset U$ such that $Z \xrightarrow{\looparrowright} A$;

\item[(b)] Accordingly, $\Sigma^\looparrowright\!$, the \textbf{folding} of $\Sigma$, is a proper subset $\Sigma^\looparrowright\! \subset \Sigma^\vartriangleright$ such that an fd $\langle Z, A\rangle \in \Sigma^\vartriangleright\!$ is in $\Sigma^\looparrowright$\vspace{-2pt} iff $X \xrightarrow{\looparrowright} A$ for some $Z \subset U$. 
\end{itemize}
\label{def:folding}
\end{mydef}

\begin{myex} (continued). 
Fig. \ref{fig:h-fdschema} shows an fd set $\Sigma$ (left) and its folding $\Sigma^\looparrowright\!\!$ (right). Note that the folding can be obtai\-ned by computing the attribute folding for $A$ in each fd $X \!\to A$ in $\Sigma$. We illustrate below some reasoning steps to partially compute an attribute folding considering the subset of fd's in $\Sigma$ with $\phi \notin X$. 

\begin{spacing}{1}
\begin{eqnarray*}
1.\,\;\;\;\;\;\;\;\;\;\;\;x_1\,x_2\,x_3\,x_5\,\upsilon &\to& x_4 \quad \text{\emph{[given]}}\\
2.\,\;\;\;\;\;\;\;\;\;\;\;\;\;\;\;x_1\,x_3\,x_4\,\upsilon &\to& x_5 \quad \text{\emph{[given]}}\\
3.\,\;\;\;\;\;\;\;\;\;\;\;\;\;\;\;\;\;\;\;\;\;\;\;x_5\,\upsilon &\to& x_7 \quad \text{\emph{[given]}}\\
4.\,\;\;\;\;\;\;\;\;\;\;\;\;\;\;\;x_1\,x_3\,x_4\,\upsilon &\to& x_7 \quad \text{\emph{[R5 over (2), (3)]}}\\
5.\,\;\;\;\;\;\;\;\;\;\;\;x_1\,x_2\,x_3\,x_5\,\upsilon &\to& x_7 \quad \text{\emph{[R5 over (1), (4)]}}\\
\line(1,0){70}&\!\!\!\!\! \line(1,0){35}\!\!\!\!\!&\line(1,0){120}\\\vspace{-2pt}
6.\;\;\therefore\;\;\;\;x_1\,x_2\,\,x_3\,x_4\,\upsilon &\to& x_7 \quad \text{\emph{[R5 over (2), (5)]}}.
\end{eqnarray*}
\end{spacing}
\vspace{9pt}

Note that (6) is still amenable to further application of R5, say over (1), to derive (7) $x_1\,x_2\,x_3\,x_5\,\upsilon \!\to x_7$. However, even though (1) and (6) have (resp.) the form (1) $X \!\to A$ and (6) $Y \!\to B$ with $Y \xrightarrow{\triangleright} X$, we have $X \xrightarrow{\triangleright} Y$ as well which characterizes a cycle that fetches nothing into $Y$.\footnote{Note, when at the step deriving (6) by R5 over (2), (5), that such cycle was not yet formed.} In fact, if we consider only the fd's 1-3 given, then (6) itself satisfies Def. \ref{def:folded} and then is folded. The same holds, e.g., for (1) by an empty application of R5. 
$\Box$ 
\end{myex}

\begin{figure}[t]
\begin{framed}
\vspace{-16pt}
\begin{subfigure}{0.45\columnwidth}
\begin{eqnarray*}
\Sigma = \{\quad 
\phi &\to& x_1,\\ 
\phi &\to& x_2,\\ 
\phi &\to& x_3,\\ 
x_1\,x_2\,x_3\,x_5\,\upsilon &\to& x_4,\\
x_1\,x_3\,x_4\,\upsilon &\to& x_5,\\
x_4\,\upsilon &\to& x_6,\\
x_5\,\upsilon &\to& x_7 \;\}.
\end{eqnarray*}
\end{subfigure}
\hspace{-8pt}
\begin{subfigure}{0.45\columnwidth}
\begin{eqnarray*}
\Sigma^\looparrowright \!= \{\;
\phi &\to& x_1,\\ 
\phi &\to& x_2,\\ 
\phi &\to& x_3,\\ 
\phi\,\upsilon\,x_5 &\to& x_4,\\
\phi\,\upsilon\,x_4 &\to& x_5,\\
\phi\,\upsilon\,x_5 &\to& x_6,\\
\phi\,\upsilon\,x_4 &\to& x_7 \;\}.
\end{eqnarray*}
\end{subfigure}
\end{framed}
\vspace{-11pt}
\caption{Fd set $\Sigma$ encoding (cf. Alg. \ref{alg:h-encode}) the structure of Fig. \ref{fig:coa-a} and its folding $\Sigma^\looparrowright$ derived by Alg. \ref{alg:folding}.}
\label{fig:h-fdschema}
\vspace{-6pt}
\end{figure}

\begin{mylemma}
Let $\mathcal S(\mathcal E, \mathcal V)$ be a complete structure, $\varphi$ a total causal mapping over $\mathcal S$ and $\Sigma$ an fd set encoded through $\varphi$ given $\mathcal S$. 
If $\langle X, A\rangle \in \Sigma$, then $A^\looparrowright\!$, the attribute folding of $A$ (w.r.t. $\!\Sigma$) exists and is unique.
\label{lemma:folding-unique}
\end{mylemma}
\begin{proof}
See Appendix, \S\ref{a:folding-unique}.
$\Box$
\end{proof}

\noindent
We give an original algorithm (Alg. \ref{alg:folding}) to com\-pute the folding of an fd set. At its core there lies (Alg. \ref{alg:afolding}) \textsf{AFolding}, which can be understood as a non-obvious variant of \textsf{XClosure} (cf. Alg. \ref{alg:xclosure}) designed for acyclic pseudo-transitivity reasoning. 
In order to compute the folding of attribute $A$ in fd $\langle X, A\rangle \in \Sigma$, algorithm \textsf{AFolding} can be seen as backtracing the causal ordering implicit in $\Sigma$ towards $A$. Analogously, in terms of the directed graph $G_{\varphi}$ induced by the causal ordering (see Fig. \ref{fig:causal-graph}), that would comprise graph traversal to identify the nodes $x \in \mathcal V$ that have $x_a \in \mathcal V$ in their reachability, i.e., $x \rightsquigarrow x_a$. Rather, \textsf{AFolding}'s processing of the causal ordering is fully symbolic based on Armstrong's rewrite rule R5.

\begin{myex}
Cyclicity in an fd set $\Sigma$ may have the effect of making its folding $\Sigma^\looparrowright\!$ to degenerate to $\Sigma$ itself. For instance, consider $\Sigma\!=\!\{A \!\to B,\, B \!\to A\}$. Note that $\Sigma$ is canonical, and \textsf{AFolding} (w.r.t. $\Sigma$) is $B$ given $A$, and $A$ given $B$. That is, $\Sigma^\looparrowright \!= \Sigma$. 
$\Box$
\label{ex:degenerate}
\end{myex}

\begin{spacing}{1.1}
\begin{algorithm}[H]
\caption{Folding of an fd set.}
\label{alg:folding}
\begin{algorithmic}[1]
\Procedure{folding}{$\Sigma\!: \text{fd set}$} 
\Require $\Sigma$ given encodes complete structure $\mathcal S$
\Ensure Returns fd set $\Sigma^\looparrowright$, the folding of $\Sigma$
\State $\Sigma^\looparrowright \gets \varnothing$

\ForAll{$\langle X,\, A \rangle \in \Sigma$}
\State $Z \gets\, \textsf{AFolding}(\Sigma,\, A)$
\State $\Sigma^\looparrowright \gets \Sigma^\looparrowright \cup \langle Z,\, A \rangle$
\EndFor
\State \Return $\Sigma^\looparrowright$
\EndProcedure
\end{algorithmic}
\end{algorithm}
\end{spacing}
%
\begin{spacing}{1.1}
\begin{algorithm}[H]
\caption{Folding of an attribute w.r.t. an fd set.}
\label{alg:afolding}
\begin{algorithmic}[1]
\Procedure{AFolding}{$\Sigma\!: \text{fd set},\; A\!: \text{attribute}$} 
\Require $\Sigma$ is parsimonious
\Ensure Returns $A^\looparrowright\!$, the attribute folding of $A$ (w.r.t. $\Sigma$)
\State $\Delta \gets \varnothing$ \Comment{consumed fd's}
\State $\Lambda \gets \varnothing$  \Comment{consumed attrs.}
\State $A^\star \gets A$ \Comment{stores attrs. $A$ is found to be `causally dependent' on (cf. \S\ref{sec:connections})}
\State \emph{size} $\gets 0$
\While{\emph{size} $< |A^\star|$ } \Comment{halts when $A^{(i+1)}\!=\!A^{(i)}$}
\State \emph{size} $\gets |A^\star|$
\State $\Sigma \gets \Sigma \setminus \Delta$
\ForAll{$\langle Y,\, B \rangle \in \Sigma$}
\If{$B \in A^\star$}
\State $\Delta \gets \Delta \cup \{\langle Y,\, B \rangle\}$ \Comment{consumes fd}
\State $A^\star \gets\, A^\star \cup Y$
\State $\Lambda \gets\, \Lambda \cup B$ \Comment{consumes attr.}
\ForAll{$C \in Y$}
\If{$C \in \Lambda$ and $B \in X$ for $\langle X, C\rangle \in \Delta$}\Comment{cyclic fd}
\State $\Lambda \gets\, \Lambda \setminus B$ \Comment{reingests it to simulate cyclic app. of R5}
\EndIf
\EndFor
\EndIf
\EndFor
\EndWhile
\State \Return $A^\star \setminus \Lambda$
\EndProcedure
\end{algorithmic}
\end{algorithm}
\end{spacing}

\begin{mythm}\label{thm:afolding}
Let $\mathcal S(\mathcal E, \mathcal V)$ be a complete structure, and $\Sigma$ an fd set encoded given $\mathcal S$. 
Now, let $\langle X, A\rangle \in \Sigma$. Then \textsf{AFolding}($\Sigma, A$) correctly computes $A^\looparrowright\!$, the attribute folding of $A$ (w.r.t. $\!\Sigma$) in time $O(|\mathcal S|^2)$.
\end{mythm}
\begin{proof}
For the proof roadmap, note that \textsf{AFolding} 
is monotone (size of $A^\star$ can only increases) and terminates precisely when $A^{(i+1)}\!=\!A^{(i)}$, where $A^{(i)}$ denotes the attributes in $A^\star$ at step $i$ of the outer loop. The folding $A^\looparrowright\!$ of $A$ at step $i$ is $A^{(i)} \setminus \Lambda^{(i)}$. We shall prove by induction, given attribute $A$ from fd $X \!\to A$ in parsimonious $\Sigma$, that $A^\star \!\setminus \Lambda$ returned by \textsf{AFolding}($\Sigma, A$) is the unique attribute folding $A^\looparrowright\!$ of $A$. 
See Appendix, \S\ref{a:afolding}.
$\Box$
\end{proof}

\begin{myremark}
\label{rmk:linear}
Let $\Sigma$ be an arbitrary fd set on attribute set $U$. Beeri and Bernstein gave a straightforward optimization to (Alg. \ref{alg:xclosure}) \textsf{XClosure} to make it linear in $|\Sigma| \cdot |U|$ (cf. $\!$\cite[p. $\!\!$43-5]{beeri1979}), where $|\Sigma| \cdot |U|$ is the maximum length for a string encoding all the fd's. Note that the actual length of such string in our case is exactly $|\mathcal S|$. The optimization mentioned applies likewise to (Alg. \ref{alg:afolding}) \textsf{AFolding}.\footnote{We omit its tedious exposure here. In short, it shall require one more auxiliary data structure to keep track, for each fd not yet consumed, of how many attributes not yet consumed appear in its rhs.}
That is, \textsf{AFolding} can be implemented to run in linear time in $|\mathcal S|$. 
$\Box$
\end{myremark}

\begin{mycor}
Let $\mathcal S(\mathcal E, \mathcal V)$ be a complete structure, and $\Sigma$ an fd set encoded given $\mathcal S$. 
Then algorithm \textsf{folding}($\Sigma$) correctly computes $\Sigma^\looparrowright\!$, the folding of $\Sigma$ in time that is $f(|\mathcal S|)\,\Theta(|\mathcal E|)$, where $f(|\mathcal S|)$ is the time complexity of (Alg. \ref{alg:afolding}) \textsf{AFolding}.
\label{cor:folding}
\end{mycor}
\begin{proof}
See Appendix, \S\ref{a:folding}. 
$\Box$
\end{proof}

Finally, it shall be convenient to come with a notion of \emph{parsimonious} fd sets (see Def. \ref{def:parsimonious}), which is suggestive of a distinguishing feature of such mathematical information systems in comparison with arbitrary information systems. 

\begin{mydef}
Let $\Sigma$ be set of fd's on attributes $U$. Then, we say that $\Sigma$ is \textbf{parsimonious} if it is canonical and, for all fd's $\langle X, A\rangle \in \Sigma$ with $XA \subseteq U$, there is no $Y \subset U$ such that $Y \neq X$ and $\langle Y, A\rangle \in \Sigma$.
\label{def:parsimonious}
\end{mydef}

Proposition \ref{prop:folding-parsimonious} then shall be useful further in connection with the concept of the folding.

\begin{myprop}
Let $\mathcal S(\mathcal E, \mathcal V)$ be a complete structure, $\varphi$ a total causal mapping over $\mathcal S$ and $\Sigma$ an fd set encoded through $\varphi$ given $\mathcal S$. Let $\Sigma^\looparrowright\!$ be the folding of $\Sigma$, then $\Sigma^\looparrowright\!$ is parsimonious.
\label{prop:folding-parsimonious}
\end{myprop}
\begin{proof}
See Appendix, \S\ref{a:folding-parsimonious}. 
$\Box$
\end{proof}

\section{Equivalence with Causal Ordering}\label{sec:connections}
\noindent
Now we show the equivalence of acyclic pseudo-transitive reasoning with causal ordering processing. We start with Theorem \ref{thm:connections1}, which establishes the equivalence between the notion of causal dependency and the fd encoding scheme presented in Chapter \ref{ch:encoding}.

\begin{mythm}\label{thm:connections1}
Let $\mathcal S(\mathcal E, \mathcal V)$ be a complete structure, $\varphi$ a total causal mapping over $\mathcal S$ and $\Sigma$ an fd set encoded through $\varphi$ given $\mathcal S$. Then $x_a, x_b \in \mathcal V$ are such that $x_b$ is causally dependent on $x_a$, i.e., $(x_a, x_b) \in C^+_\varphi$ iff there is some non-trivial fd $\langle X, B\rangle \in \Sigma^\vartriangleright$ with $A \in X$, where $B \mapsto x_b$ and $A \mapsto x_a$. 
\end{mythm}
\begin{proof}
We prove the statement by induction. We consider first the `if' direction, and then its `only if' converse. See Appendix \S\ref{a:connections1}.
$\Box$
\end{proof}

Def. \ref{def:first-cause} then gives useful terminology for a neat concept towards our goal in this chapter.

\begin{mydef}\label{def:first-cause}
Let $\mathcal S(\mathcal E, \mathcal V)$ be a structure with variables $x_a, x_b \in \mathcal V$, and $\varphi$ a total causal mapping over $\mathcal S$ inducing set of direct causal dependencies $C_{\varphi}$ and its transitive closure $C^+_{\varphi}$. We say that $x_a$ is a \textbf{first cause} of $x_b$ in $\mathcal S$ if $(x_a, x_b) \in C^+_{\varphi}$ and, for no $x \in \mathcal V$, we have $(x, x_a) \in C^+_{\varphi}$.
\end{mydef}

Proposition \ref{prop:exogenous} connects the notion of first cause with those of exogenous and endogenous variables introduced in Chapter \ref{ch:encoding}.

\begin{myprop}\label{prop:exogenous}
Let $\mathcal S(\mathcal E, \mathcal V)$ be a structure with variable $x  \in \mathcal V$. Then $x$ can only be a first cause of some $y \in \mathcal V$ if $x$ is exogenous. Accordingly, any variable $y \in \mathcal V$ can only have some first cause $x \in \mathcal V$ if it is endogenous.
\end{myprop}
\begin{proof}
Straightforward from definitions, see Appendix \S\ref{a:exogenous}. 
$\Box$
\end{proof}

Note that exogenous variables are encoded into fd's $X \to A$ with $\phi \in X$. Since the values of such variables are assigned ``outside'' the system (cf. Remark \ref{rmk:modeling}), they are devoid of indirect causal dependencies and then have no uncertainty except for their own. Thus, we have not to be concerned at all with processing the causal (uncertainty) chaining towards them. Our goal is rather find the first causes of the endogenous variables (predictive attributes). 

We shall need then the terminology of Def. \ref{def:y-projection}, and then we introduce Lemma \ref{lemma:connections} paving the way to our goal.

\begin{mydef}\label{def:y-projection}
Let $\mathcal S$ be a structure, and $\Sigma$ be a set of fd's encoded over it. Then $\Upsilon(\Sigma)$, the \textbf{$\upsilon$-projection of $\Sigma$}, is the subset of fd's $X \to A$ such that $\upsilon \in X$. Accordingly, $\Phi(\Sigma)$, the \textbf{$\phi$-projection of $\Sigma$}, is the subset of fd's $X \to A$ such that $\upsilon \notin X$.
\end{mydef}

Fig. \ref{fig:y-folding-projection} illustrates the $\upsilon$-projection of an fd set and the folding applied over such fd subset in order to compute the first causes of endogenous variables.

\begin{mylemma}\label{lemma:connections}
Let $\mathcal S(\mathcal E, \mathcal V)$ be a complete structure, $\varphi$ a total causal mapping over $\mathcal S$ and $\Sigma$ an fd set encoded through $\varphi$ given $\mathcal S$. Then a variable $x_a \in \mathcal V$ can only be a first cause of some variable $x_b \in \mathcal V$, where $\langle X, B\rangle \in \Sigma$, and $B \mapsto x_b$, $A \mapsto x_a$, if either (i) $A \in X$ or (ii) $A \notin X$ but there is $\langle Z, C\rangle \in \Sigma^\vartriangleright\!$ with $A \in Z$ and $C \in X$. 
\end{mylemma}
\begin{proof}
We prove the statement by construction out of Theorem \ref{thm:connections1}, see Appendix \S\ref{a:lemma-connections}. 
$\Box$
\end{proof}

\begin{figure}[t]
\begin{framed}
\vspace{-11pt}
\begin{subfigure}{0.45\columnwidth}
\begin{eqnarray*}
\Sigma = \{\quad 
\phi &\to& x_1,\\ 
\phi &\to& x_2,\\ 
\phi &\to& x_3,\\ 
x_1\,x_2\,x_3\,x_5\,\upsilon &\to& x_4,\\
x_1\,x_3\,x_4\,\upsilon &\to& x_5,\\
x_4\,\upsilon &\to& x_6,\\
x_5\,\upsilon &\to& x_7 \;\}.
\end{eqnarray*}
\end{subfigure}
\hspace{-10pt}
\begin{subfigure}{0.45\columnwidth}
\begin{eqnarray*}
\,\Upsilon(\Sigma)^\looparrowright \!= \{\; 
x_1\;x_2\;x_3\,\upsilon\,x_5 &\to& x_4,\\
x_1\;x_2\;x_3\,\upsilon\,x_4 &\to& x_5,\\
x_1\;x_2\;x_3\,\upsilon\,x_5 &\to& x_6,\\
x_1\;x_2\;x_3\,\upsilon\,x_4 &\to& x_7 \;\}.
\end{eqnarray*}
\end{subfigure}
\end{framed}
\vspace{-11pt}
\caption{Fd set $\Sigma$ encoding the structure of Fig. \ref{fig:coa-a} and the folding $\Upsilon(\Sigma)^\looparrowright$ of its $\upsilon$-projection.}
\label{fig:y-folding-projection}
\vspace{-9pt}
\end{figure}

Finally, Theorem \ref{thm:connections2} further clarifies the purpose of the folding and its meaning in terms of causal ordering.

\begin{mythm}\label{thm:connections2}
Let $\mathcal S(\mathcal E, \mathcal V)$ be a complete structure, $\varphi$ a total causal mapping over $\mathcal S$ and $\Sigma$ an fd set encoded through $\varphi$ given $\mathcal S$. Now, let $B$ be an attribute that encodes some variable $x_b \in \mathcal V$.
If $\langle X, B\rangle \in \Upsilon(\Sigma)^\looparrowright\!$,\footnote{Note that the folding is taken w.r.t. the $\upsilon$-projection of $\Sigma$, then $x_b$ where $B \mapsto x_b$ is an endogenous variable.} then every first cause $x_a$ of $x_b$ (if any) is encoded by some attribute $A \in X$.
\end{mythm}

\begin{proof}
We show that the existance of a missing first cause $x_c$ of $x_b$ for folded $X \xrightarrow{\looparrowright} B$, where $B \mapsto x_b$ and $C \mapsto x_c$ but $C \notin X$ leads to a contradiction. See Appendix \S\ref{a:thm-connections2}. 
$\Box$
\end{proof}

\begin{myremark}\label{rmk:ptc-folding}
Observe that, on the one hand, the goal of computing the transitive closure $C^+_\varphi$ of a set of induced causal dependencies $C_\varphi$ is to derive the entire causal ordering of a given structure. The goal of folding, on the other hand, is not to discover all variables (attributes) a given variable (attribute) is causally dependent on, but only all of its first causes (if any). 
$\Box$
\end{myremark}

In particular, the results just shown comprise a method to compute, for each endogenous variable (predictive attribute), all of its first causes. This is a core goal of the reasoning device developed in this chapter in order to enable the automatic synthesis of hypotheses as uncertain and probabilistic (U-relational) data.

\section{Experiments}\label{sec:folding-experiments}

\noindent
Fig. \ref{fig:folding-experiments} shows the results of experiments we have carried out in order to study how effective the causal reasoning over fd's is in practice, in particular its behavior for hypotheses whose structure $\mathcal S$ has been randomly generated over orders of magnitude $|\mathcal S| \approx 2^k$, to have length up to $|\mathcal S| \approx 2^{20} \lesssim 1M$. The largest structure considered, with $|\mathcal S|\approx 1M$, has been generated to have exactly $|\mathcal E| = 2.5K$. 
Like in the experiments of the previous chapter, we have executed ten runs for each tested order of magnitude, and then taken its mean running time in $ms$.\footnote{The experiments were performed on a 2.3GHZ/4GB Intel Core i5 running Mac OS X 10.6.8.} 

The plot is shown in Fig. \ref{fig:tcm-experiments} is in logscale base 2. Notice the linear rate of growth across orders of magnitude (base 2) from $1K$ to $1M$ sized structures. For a growth factor of 2 in structure length (doubled), the time required by causal reasoning grows a factor of 2 (doubled as well). 
These scalability results are compatible with the computational complexity of \textsf{folding}, which is bounded by $O(|\mathcal S|^2)$. Yet, that is a bit overestimated time bound as we see in the plot of Fig. \ref{fig:folding-experiments}.

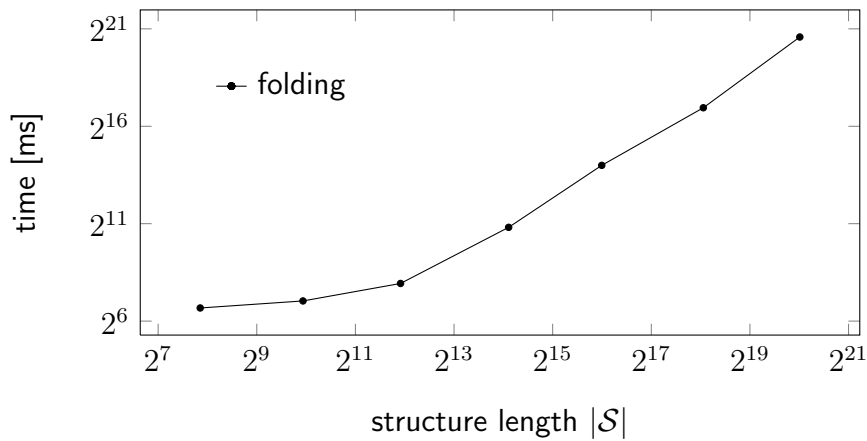
\begin{figure}[t]
\advance\leftskip 0.6cm
\begin{tikzpicture}[y=.015cm, x=2cm,font=\sffamily,scale=1.0]
    \begin{axis}[
        height=0.4\textwidth,
        width=0.75\textwidth,
        x label style={at={(axis description cs:0.5,-0.23)},anchor=south},
        xlabel={structure length $|\mathcal S|$},
        y label style={at={(axis description cs:0.01,.5)},anchor=south},
        ylabel={time [ms]},
        xmode=log,
       log basis x={2},
        ymode=log,
       log basis y={2}
]
      \addplot[mark=*, mark size=1.2pt, mark options={fill=black}] 
		file {./data/folding.data}; 
     \end{axis}
	\begin{scope}[shift={(0.5,220)}] 
	\draw (0,0) -- 
		plot[mark=*, mark size=1.2pt, mark options={fill=black}] (0.1,0) -- (0.2,0)
		node[right]{\textsf{folding}};
	\end{scope}
\end{tikzpicture}
\caption{Performance of acyclic causal reasoning over fd's (logscale).}\label{fig:folding-experiments}
\vspace{-8pt}
\end{figure}

\section{Related Work}\label{sec:related-work-reasoning}

The concept of fd set folding and the design of (Alg. $\!$\ref{alg:afolding}) \textsf{AFolding} as a not quite obvious variant of \textsf{XClosure}, is an original approach to the problem of processing the causal ordering of a hypothesis via \emph{acyclic pseudo-transitive reasoning over fd's}. To the best of our knowledge, such a specific form of fd reasoning was an yet unexplored problem in the database research literature (reasoning over fd's is extensively covered in Maier \cite{maier1983}).

Recent years have seen the emergence of some foundational work in causality in databases \cite{meliou2010}. It is motivated for improving DB usability in terms of providing users with explanations to query answers (and non-answers). Essentially, the idea is borrowed from AI work on causality (cf. \S\ref{sec:related-work-encoding}) to identify causal ordering between tuples. Given a query and its result set, the system should be able to explain to the user what tuples `caused' that answer, or why possibly expected tuples are missing. That requires causal chain of tuples for a given query, which can be computationally expensive as the database instance can be very large. For conjunctive queries, the causality is said to be computed very efficiently \cite{meliou2010b}. 

A more specific problem addressed by Kanagal et. al is the so-called sensitivity analysis \cite{kanagal2011}, which is aimed at establishing a more refined connection between the query answer (output) and elements of the DB instance (input) for supporting user interventions. Instead of providing the user with causes and non-causes, the goal is to enable the user to know how changes in the input affect the output. This line of work is strongly related to the vision of `reverse data management' \cite{meliou2011}.

Causal reasoning in the presence of constraints (viz., fd's) is an yet unexplored topic, though called for as worth of future work by Meliou et al. \cite[p. 3]{meliou2010}. The fd's are rich information that can be exploited for the sake of explanation and sensitivity analysis. Once they are available, it is intuitive that the search space of such problems shall be significantly reduced. 

In fact, our encoding of equations into fd's captures the causal chain from exogenous (input) to endogenous (output) tuples in at schema level. Nonetheless our form of causal reasoning over fd's is geared for hypothesis management and analytics, from a uncertainty management point of view. A concrete connection to causality in DB's is not yet established.

\section{Summary of Results }\label{sec:reasoning-conclusions}

In this chapter we have studied and developed a technique for acyclic causal reasoning over fd's. We list the results achieved as follows.

\begin{itemize}
\item We have developed principled concepts and a core algorithm, viz., (Alg. \ref{alg:folding}) the \textsf{folding}, in order to perform acyclic pseudo-transitive reasoning over fd's. This is towards an efficient method for causal reasoning, yet elegant as a database formalism for the systematic construction of hypothesis probabilistic DB's.

\item We have given a reasonably tight time bound for the behavior of such reasoning device in terms of the structure given as input. We have established (cf. Theorem \ref{thm:afolding}, Corollary \ref{cor:folding}) the time bound of $O(|\mathcal S|^2)$ for the \textsf{folding} algorithm.

\item We have shown the correctness of the \textsf{folding} algorithm in connection with causal reasoning (cf. Theorem \ref{thm:connections1}, Theorem \ref{thm:connections2}). 

\item We have defined the core notion of first causes (cf. Def. \ref{def:first-cause}, Proposition \ref{prop:exogenous}), which is meant to guide the procedure of \textsf{U-intro} (Chapter \ref{ch:synthesis4u}) by precisely capturing the uncertainty factors on endogenous variables (predictive attributes). This is similar, yet markedly different from computing the transitive closure of causal dependencies (cf. Remark \ref{rmk:ptc-folding}).

\item We have performed experiments (cf. Fig. \ref{fig:folding-experiments}) to study how effective the approach of causal reasoning over fd's is in practice, or how it scales for hypotheses whose structure $\mathcal S$ is randomly generated to have length up to the order of $|\mathcal S| \lesssim 1M$. The experiments show that the time bound, though already effective for very large structures, are a bit overestimated.
\end{itemize}


\chapter{Probabilistic Database Synthesis}\label{ch:synthesis4u}
\noindent
In this chapter we present a technique to synthesize hypothesis U-relations. At this stage of the pipeline, relational schema $\boldsymbol H$ is loaded with datasets computed from the hypotheses under alternative trials (input settings). The challenge is how to model or design its probabilistic version (i.e., render the U-relations $\boldsymbol Y$) so that it is suitable for data-driven hypothesis management and analytics. 

In \S\ref{sec:u-relations} we introduce U-relational DB's. In \S\ref{sec:u-intro} we present a running example to illustrate the uncertainty introduction procedure (\textsf{U-intro} in the pipeline, cf. Fig. \ref{fig:pipeline}). 
In \S\ref{sec:u-factorization} we present the technique to factorize the uncertainty present in the `big' fact table in terms of the well-defined uncertainty factors. In \S\ref{sec:u-propagation} then we show how to propagate such uncertainty into the predictive attributes properly, i.e., based on their first causes detected as shown in Chapter \ref{ch:reasoning}. 
In \S\ref{sec:related-work-synthesis4u}, we discuss related work. 
Finally, in \S\ref{sec:synthesis4u-conclusions} we conclude the chapter.
\footnote{We postpone the presentation of experiments on p-DB synthesis (\textsf{U-intro} as a whole) to \S\ref{sec:experiments-app}.}

\section{Preliminaries: U-Relations and Probabilistic WSA}\label{sec:u-relations}
\noindent
Three remarkable features of U-relations are: \emph{expressiveness} (being closed under positive relational algebra queries); \emph{succinctness} (efficient storage of a very large number of possible worlds through vertical decompositions to support attribute-level uncertainty); and \emph{efficient query processing} (including confidence computation) \cite{koch2009}. 

A U-relational database or U-database is a finite set of structures,
\begin{center}
$\boldsymbol W = \{\,\langle R_1^1,\, ...,\, R_m^1,\, p^{[1]}\rangle,\, ...,\, \langle\, R_1^n,\, ..., R_m^n,\, p^{[n]}\,\rangle\,\}$,
\end{center}
 of relations $R_1^i,\, ...,\, R_m^i$ and numbers $0 < p^{[i]} \leq 1$ such that
$\sum_{\,1\leq\, i \,\leq n} p^{[i]} = 1$. An element $R_1^i,\, ...,\, R_m^i,\, p^{[i]} \in \boldsymbol W$ is a \emph{possible world}, with $p^{[i]}$ being its probability \cite{koch2009}.

Probabilistic world-set algebra (p-$\!$WSA) consists of the operations of relational algebra, an operation for computing tuple confidence \textsf{conf}, and the \textsf{repair-key} operation for introducing uncertainty --- by giving rise to alternative worlds as maximal-subset repairs of an argument key \cite{koch2009}.

Let $R_\ell[U]$ be a relation, and $XA \subseteq U$. For each possible world $\langle R_1, ..., R_m,$ $p\rangle \in \boldsymbol W$, let $A \in U$ contain only numerical values greater than zero and let $R_\ell$ satisfy the fd $(U \setminus A) \to U$. Then, \textsf{repair-key} is:

\vspace{12pt}
$\llbracket \textsf{repair-key}_{X@A}(R_\ell) \rrbracket (\boldsymbol W\!) := \left\{ \,\langle\, R_1,  ..., R_\ell, R_m, \hat{R_\ell}\,[U \setminus A], \,\hat{p}\,\rangle\, \right\}$,
\vspace{12pt}

\noindent
where $\langle R_1, ..., R_\ell,\, R_m,\, p \rangle \in \boldsymbol W$, $\hat{R}_\ell$ is a maximal repair of fd $X \to U$ in $R_\ell$, and $\hat{p} = p\, \cdot \displaystyle\prod_{t \in \hat{R}_\ell} \frac{t.B}{\sum_{s \,\in\, R_\ell \,:\, s.X\,=\,t.X} s.B}$.

\vspace{10pt}
\noindent
U-relations (cf. Fig. \ref{fig:maybms}) have in their schema a set of pairs $(V_i, D_i)$ of \emph{condition columns} (cf.$\!$ \cite{koch2009}) to map each discrete random variable $\textsf{x}_i$ to one of its possible values \mbox{(e.g., $\textsf{x}_1 \!\mapsto\! 1$)}. The world table $W\!$ stores their mar\-ginal probabilities (cf. $\!$ the notion of \emph{pc-tables} \cite[Ch. $\!$2]{suciu2011}). 
For an illustration of the data transformation from certain to uncertain relations, consider query (\ref{eq:explanation}) in p-WSA's extension of relational algebra, whose result set is materialized into U-relation \textsf{Y}$_0$ as shown in (Fig. \ref{fig:maybms}).
\vspace{-10pt}
\begin{eqnarray}
\! Y_0 \,:=\, \pi_{\phi,\upsilon}( \textsf{repair-key}_{\phi @\textsf{Conf}} (H_0)\,).
\label{eq:explanation}
\end{eqnarray}
\noindent
Also, let $R[\,\overline{V_i\,D_i} \;|\, sch(R)\,],\, S[\,\overline{V_j\,D_j} \;|\, sch(S)\,]$ be two U-relations, where $R.\,\overline{V_i\,D_i}$ is the union of all pairs of condition columns $V_i\,D_i$ in $R$, then operations of selection $\llbracket\, \sigma_\psi(R)\,\rrbracket$, projection $\llbracket\, \pi_Z(R) \,\rrbracket$, and product $\llbracket R \times S \rrbracket$ issued in relational algebra are rewritten in positive relational algebra on U-relations:\vspace{-4pt}\\

$\!\!\!\!\!\llbracket \sigma_\psi(R) \rrbracket := \sigma_\psi(R[\overline{V_i\,D_i}\,\,|\, sch(R)])$;\vspace{3pt}

$\!\!\!\!\!\llbracket\, \pi_Z(R) \,\rrbracket := \pi_{\,\overline{V_i\,D_i}\,Z}(R)$;\vspace{3pt}

$\!\!\!\!\!\llbracket R \times S \rrbracket := \pi_{(R.\overline{V_i\,D_i} \;\cup\; S.\overline{V_i\,D_i}) \to \overline{V\,D} \;\cup\; sch(R) \,\cup\, sch(S)} ( R \bowtie_{R.\overline{V_i\,D_i} \;\textsf{is consistent with}\; S.\overline{V_j\,D_j}} S )$.
\vspace{1pt}

If $R$ and $S$ have $k$ and $\ell$ pairs of condition columns each, then $\llbracket R \times S \rrbracket$ returns a U-relation with $k + \ell$ such pairs. If $k=0$ or $\ell=0$ (or both), then $R$ or $S$ (or both) are classical relations, but the rewrite rules above apply accordingly. All that rewriting is parsimonious translation (sic. \cite{koch2009}): the number of algebraic operations does not increase and each of the operations selection, projection and product/join remains of the same kind. Query plans are hardly more complicated than the input queries. In fact, it has been verified hat off-the-shelf relational database query optimizers do well in practice.

For a comprehensive overview of U-relations and p-WSA we refer the reader to \cite{koch2009}. In this thesis we look at U-relations from the point of view of p-DB design, for which no methodology has yet been proposed. We are concerned in particular with hypothesis management applications \cite{goncalves2014}.

\begin{figure}[t]
\centering
\begingroup\setlength{\fboxsep}{2pt}
\colorbox{blue!7}{%
   \begin{tabular}{c|c|c|c}
  \textsf{H}$_0$ & $\phi$ & $\upsilon$ & \textsf{Conf}\\
      \hline    
   & $1$ & $1$ & 2\\
   & $1$ & $2$ & 2\\
   & $1$ & $3$ & 1\\
   \end{tabular}
}\endgroup\vspace{2pt}\\
\begingroup\setlength{\fboxsep}{2pt}
\colorbox{yellow!15}{%
   \begin{tabular}{c|>{\columncolor[gray]{0.92}}c||c|c}
  \textsf{Y}$_0$ & $V \mapsto D$ & $\phi$ & $\upsilon$\\
      \hline    
   & $\textsf{x}_0 \mapsto 1$ & $1$ & $1$\\
   & $\textsf{x}_0 \mapsto 2$ & $1$ & $2$\\
   & $\textsf{x}_0 \mapsto 3$ & $1$ & $3$\\
   \end{tabular}
}\endgroup
\begingroup\setlength{\fboxsep}{2pt}
\colorbox{yellow!15}{%
   \begin{tabular}{c|>{\columncolor[gray]{0.92}}c||c}
  \textsf{W} & $V \mapsto D$ & \textsf{Pr}\\
      \hline    
   & $\textsf{x}_0 \mapsto 1$ & $.4$\\
   & $\textsf{x}_0 \mapsto 2$ & $.4$\\
   & $\textsf{x}_0 \mapsto 3$ & $.2$\\
   \end{tabular}
}\endgroup
\vspace{-3pt}
\caption{U-relation generated by the repair-key operation.}
\label{fig:maybms}
\vspace{-9pt}
\end{figure}

\section{Running Example}\label{sec:u-intro}

Before proceeding, we consider Example \ref{ex:population}, which is 
fairly representative to illustrate how to deal with correlations in the predictive data of deterministic hypotheses for the sake of suitable data-driven analytics.

\begin{myex}
We explore three slightly different theoretical models in population dynamics with applications in Ecology, Epidemics, Economics, etc: (\ref{eq:malthus}) Malthus' model, (\ref{eq:logistic})$\!$ the logistic equation and (\ref{eq:lotka-volterra}) the Lotka-$\!$Volterra model. In practice, such equations are meant to be extracted from \textsf{MathML}-compliant XML files (cf. Chapter \ref{ch:applicability}). For now, consider that the ordinary differential equation notation `$\dot{x}$' is read `variable $x$ is a function of time $t$ given initial condition $x_0$.$\!$' 
\begin{eqnarray}
\dot{x}=bx
\label{eq:malthus}
\end{eqnarray}
\vspace{-38pt}
\begin{eqnarray}
\dot{x}=b(1-x/K)x
\label{eq:logistic}
\end{eqnarray}
\vspace{-36pt}
\begin{eqnarray}
\left\{ 
  \begin{array}{lll}
\dot{x} &=& x(b - py)\\
\dot{y} &=& y(rx - d)
\end{array} \right.
\label{eq:lotka-volterra}
\end{eqnarray}
%
The models are completed (by the user) with additional equations to provide the values of exogenous variables (or `input parameters'),\footnote{Given $\mathcal S(\mathcal E, \mathcal V)$, it is actually a task of the encoding algorithm (viz., sub-procedure \textsf{TCM}) to infer, for each variable $x \in \mathcal V$, whether it is exogenous or endogenous by means of the total causal mapping.} e.g., $x_0\!=200,\, b=10$, such that we have SEM's (resp.) $\mathcal S_k(\mathcal E_k, \mathcal V_k)$ for $k=1..3$,

\begin{itemize}
\item $\mathcal E_1 \!=\! \{\, f_1(t),\; f_2(x_0),\; f_3(b),\; f_4(x, t, x_0, b) \,\}$;
\item $\mathcal E_2 \!=\! \{ f_1(t),\, f_2(x_0),\, f_3(K),\, f_4(b),\, f_5(x, t, x_0, K, b) \}$;
\item $\mathcal E_3 \!=\! \{\, f_1(t),\; f_2(x_0),\; f_3(b),\; f_4(p),\; f_5(y_0),\; f_6(d),\;f_7(r),$\vspace{1pt}\\ 
			$\quad f_8(x, t, x_0, b, p, y),\; f_9(y, t, y_0, d, r, x) \,\}$. 
\end{itemize}

\noindent
Fig. \ref{fig:population-fds} shows the fd sets encoded from structures $\mathcal S_k$ given above.\footnote{Recall that domain variables like time $t$ are informed to the encoding algorithm to suppress an fd $\phi \!\to t$.}
We also consider trial datasets for hypothesis $\upsilon\!=\!3$ (viz., the Lotka-Volterra model), which are loaded into the `big' fact table relation $H_3 \in \boldsymbol H$ as shown in Fig. \ref{fig:population-table}. We admit a special attribute `trial id' \textsf{tid} to keep hypothesis trials identified until uncertainty is introduced in a controlled way by p-DB synthesis (\textsf{U-intro} stage, cf. Fig. \ref{fig:pipeline}).  $\Box$
\label{ex:population}
\end{myex}

\begin{figure}[t]
\begin{framed}
\vspace{-6pt}
\begin{subfigure}{0.45\textwidth}
\begin{eqnarray*}
\!\Sigma_1 = \{\;\; \phi &\!\to\!& x_0,\\ 
\phi &\!\to\!& b,\\ 
\!\!x_0\,b\,t\,\upsilon &\!\to\!& x \,\,\}.
\end{eqnarray*}
\begin{eqnarray*}
\Sigma_2 = \{\;\; \phi &\!\to\!& x_0,\\ 
\phi &\!\to\!& K,\\ 
\phi &\!\to\!& b,\\ 
x_0\,K\,b\,t\,\upsilon &\!\to\!& x \;\,\}.
\end{eqnarray*}
\end{subfigure}
\hspace{15pt}
\begin{subfigure}{0.45\textwidth}
\begin{eqnarray*}
\Sigma_3 = \{\;\;\; \phi &\!\to\!& x_0,\\ 
\phi &\!\to\!& b,\\
\phi &\!\to\!& p,\\
\phi &\!\to\!& y_0,\\ 
\phi &\!\to\!& d,\\ 
\phi &\!\to\!& r,\\
x_0\,b\,p\,t\,\upsilon\,y &\!\to\!& x,\\ 
y_0\,d\,r\,t\,\upsilon\,x &\!\to\!& y \;\,\}.
\end{eqnarray*}
\end{subfigure}
\vspace{3pt}
\end{framed}
\vspace{-8pt}
\caption{Fd sets encoded from the given structures $\mathcal S_k(\mathcal E_k,\, \mathcal V_k)$ for hypotheses $k=1..3$ from Example \ref{ex:population}.}
\label{fig:population-fds}
\vspace{-3pt}
\end{figure}

\begin{spacing}{1.1}
\begin{figure}[H]
\centering
\begingroup\setlength{\fboxsep}{3pt}
\colorbox{blue!7}{%
   \begin{tabular}{c|>{\columncolor[gray]{0.92}}c||c|c|c|c|c|c|c|c|c|c|c}
  $\!\!$\textsf{H}$_3$ & $\!$\textsf{tid}$\!$ & $\phi$ & $\upsilon$ & $\!t\!$ & $x_0$ & $b$ & $p$ & $y_0$ & $d$ & $r$ & $x$ & $y$\\
      \hline
   & $\!\!1\!\!$ & $1$ & $3$ & $\!0\!$ & $30$ & $\!\!.5\!\!$ & $\!\!.02\!\!$ & $\!\!4\!\!$ & $\!\!.75\!\!$ & $\!\!.02\!\!$ & $\!\!30\!\!$ & $\!\!4\!\!$\\
   & $\!\!1\!\!$ & $1$ & $\!3\!$ & $\!...\!$ & $30$ & $\!.5\!$ & $\!.02\!$ & $\!\!4\!\!$ & $\!.75\!$ & $\!.02\!$ & ... & ...\\
   \cline{2-13}
   
   & $\!\!2\!\!$ & $1$ & $3$ & $\!0\!$ & $30$ & $\!\!.5\!\!$ & $\!\!.018\!\!$ & $\!\!4\!\!$ & $\!\!.75\!\!$ & $\!\!.023\!\!$ & $\!\!30\!\!$ & $\!\!4\!\!$\\
   & $\!\!2\!\!$ & $1$ & $\!3\!$ & $\!...\!$ & $30$ & $\!.5\!$ & $\!.018\!$ & $\!\!4\!\!$ & $\!.75\!$ & $\!.023\!$ & ... & ...\\
   \cline{2-13}
   
   & $\!\!3\!\!$ & $1$ & $3$ & $\!0\!$ & $30$ & $\!\!.4\!\!$ & $\!\!.02\!\!$ & $\!\!4\!\!$ & $\!\!.8\!\!$ & $\!\!.02\!\!$ & $\!\!30\!\!$ & $\!\!4\!\!$\\
   & $\!\!3\!\!$ & $1$ & $\!3\!$ & $\!...\!$ & $30$ & $\!.4\!$ & $\!.02\!$ & $\!\!4\!\!$ & $\!.8\!$ & $\!.02\!$ & ... & ...\\
   \cline{2-13}
   
   & $\!\!4\!\!$ & $1$ & $3$ & $\!0\!$ & $30$ & $\!\!.4\!\!$ & $\!\!.018\!\!$ & $\!\!4\!\!$ & $\!\!.8\!\!$ & $\!\!.023\!\!$ & $\!\!30\!\!$ & $\!\!4\!\!$\\
   & $\!\!4\!\!$ & $1$ & $\!3\!$ & $\!...\!$ & $30$ & $\!.4\!$ & $\!.018\!$ & $\!\!4\!\!$ & $\!.8\!$ & $\!.023\!$ & ... & ...\\
   \cline{2-13}
   
   & $\!\!5\!\!$ & $1$ & $3$ & $\!0\!$ & $30$ & $\!\!.397\!\!$ & $\!.02\!$ & $\!\!4\!\!$ & $\!\!.786\!\!$ & $\!.02\!$ & $\!\!30\!\!$ & $\!\!4\!\!$\\
   & $\!\!5\!\!$ & $1$ & $3$ & $\!...\!$ & $30$ & $\!.397\!$ & $\!.02\!$ & $\!\!4\!\!$ & $\!.786\!$ & $\!.02\!$ & ... & ...\\
   \cline{2-13}
   
   & $\!\!6\!\!$ & $1$ & $3$ & $\!0\!$ & $30$ & $\!.397\!$ & $\!.018\!$ & $\!\!4\!\!$ & $\!.786\!$ & $\!.023\!$ & $\!30\!$ & $4$\\   
   & $\!\!6\!\!$ & $1$ & $\!3\!$ & $5\!$ & $30$ & $\!.397\!$ & $\!.018\!$ & $\!\!4\!\!$ & $\!.786\!$ & $\!.023\!$ & $\!50.1\!$ & $\!62.9\!\!$\\
   & $\!\!6\!\!$ & $1$ & $\!3\!$ & $10$ & $30$ & $\!.397\!$ & $\!.018\!$ & $\!\!4\!\!$ & $\!.786\!$ & $\!.023\!$ & $\!13.8\!$ & $\!8.65\!\!$\\
   & $\!\!6\!\!$ & $1$ & $\!3\!$ & $15$ & $30$ & $\!.397\!$ & $\!.018\!$ & $\!\!4\!\!$ & $\!.786\!$ & $\!.023\!$ & $\!79.3\!$ & $\!8.23\!\!$\\
   & $\!\!6\!\!$ & $1$ & $\!3\!$ & $20$ & $30$ & $\!.397\!$ & $\!.018\!$ & $\!\!4\!\!$ & $\!.786\!$ & $\!.023\!$ & $\!12.6\!$ & $\!30.7\!\!$\\
   & $\!\!6\!\!$ & $1$ & $\!3\!$ & $\!...\!$ & $30$ & $\!.397\!$ & $\!.018\!$ & $\!\!4\!\!$ & $\!.786\!$ & $\!.023\!$ & ... & ...\\
   \end{tabular}
}\endgroup
\vspace{-2pt}
\caption{`Big' fact table $H_3$ of hypothesis $k\!=\!3$ from Example \ref{ex:population}  loaded with trial datasets identified by special attribute \textsf{tid}.}
\label{fig:population-table}
\vspace{-3pt}
\end{figure}
\end{spacing}

Given the `big' fact table $H_3$, p-DB synthesis has two main parts: process the `empirical' uncertainty present in the `big' fact table and synthesize it out (decompose it) into independent u-factors (\emph{u-factorization}); and then propagate it precisely into the predictive data (\emph{u-propagation}).

\section{U-Factorization}\label{sec:u-factorization}
\noindent
As we have seen in \S\ref{sec:u-relations}, the \textsf{repair-key} operation allows one to create a discrete random variable in order to repair an argument key in a given relation. Our goal here is to devise a technique to perform such operation in a principled way for hypothesis management. It is a basic design principle to have exactly one random variable for each distinct uncertainty factor (`u-factor' for short), which requires carefully identifying the actual sources of uncertainty present in relations $\boldsymbol H$. 

The multiplicity of (competing) hypotheses is itself a standard one, viz., the \emph{theoretical} u-factor. Consider an `explanation' table like $H_0$ in Fig. $\!$\ref{fig:maybms}, which stores (as foreign keys) all hypotheses available and their target phenomena. We can take such $H_0$ as explanation table for the three hypotheses of Example \ref{ex:population}. Then a discrete random variable $V_0$ 
is defined into $Y_0\,[\,V_0\,D_0\,|\,\phi\, \upsilon\,]$ by query formula (\ref{eq:explanation}). U-relation $Y_0$ is considered standard in p-DB synthesis, as the repair of $\phi$ as a key in (standard) $H_0$. 

Hypotheses, nonetheless, are (abstract) `universal state\-ments' \cite{losee2001}. In order to produce a (concrete) valuation over their endogenous attributes (predictions), one has to inquire into some particular `situated' phenomenon $\phi$ and tentatively assign a valuation over the exogenous attributes, which can be eventually tuned for a target $\phi$. The multiplicity of such (competing) empirical estimations for a hypothesis $k$ leads to Problem \ref{prob:u-learning}, viz., learning \emph{empirical} u-factors for each  $H_k \in \boldsymbol H$.

\begin{myprob}
Let $\Sigma_k$ be an fd set encoded given hypothesis structure $\mathcal S_k$, and $H_k$ its `big' fact table relation loaded with trial data. Now, let $Z$ be the set of attributes encoding exogenous variables in $H_k$, then the problem of \textbf{u-factor learning} is:
\begin{enumerate}
\item to \textbf{infer} in $H_k$ `casual' fd's $\phi\,B_i \!\to B_j$, $\phi\,B_j \!\to B_i$ not in $\Sigma_k\!$ (strong input correlations), 
where $B_i,\, B_j \in Z$; 
\item to form maximal \textbf{groups} $G_1,\,...,\,G_n \subseteq Z$ of attributes such that for all $B_i,\, B_j \in G_a$, the casual fd's $\phi\,B_i \!\to B_j$ and $\phi\,B_j \!\to B_i$ hold in $H_k$;
\item to pick, for each group $G_a$, any $A \in G_a$ as a \textbf{pivot} representative and insert $\phi\,A \!\to B$ into an fd set $\Omega_k$ for all $B \in (G_a \setminus A)$.
\end{enumerate}
\label{prob:u-learning}
\vspace{-6pt}
\end{myprob}

\newcolumntype{g}{>{\columncolor{green!15}}c}
\newcolumntype{r}{>{\columncolor{red!20}}c}
\begin{spacing}{1.1}
\begin{figure}[H]
\centering
\begingroup\setlength{\fboxsep}{3pt}
\colorbox{blue!7}{%
   \begin{tabular}{c|>{\columncolor[gray]{0.92}}c||c|c|c|c|g|r|c|g|r|c|c}
  $\!\!$\textsf{H}$_3$ & $\!$\textsf{tid}$\!$ & $\phi$ & $\upsilon$ & $\!t\!$ & $x_0$ & $b$ & $p$ & $y_0$ & $d$ & $r$ & $x$ & $y$\\
      \hline
   & $\!\!1\!\!$ & $1$ & $3$ & $\!0\!$ & $30$ & $\!\!.5\!\!$ & $\!\!.02\!\!$ & $\!\!4\!\!$ & $\!\!.75\!\!$ & $\!\!.02\!\!$ & $\!\!30\!\!$ & $\!\!4\!\!$\\
   & $\!\!1\!\!$ & $1$ & $\!3\!$ & $\!...\!$ & $30$ & $\!.5\!$ & $\!.02\!$ & $\!\!4\!\!$ & $\!.75\!$ & $\!.02\!$ & ... & ...\\
   \cline{2-13}
   
   & $\!\!2\!\!$ & $1$ & $3$ & $\!0\!$ & $30$ & $\!\!.5\!\!$ & $\!\!.018\!\!$ & $\!\!4\!\!$ & $\!\!.75\!\!$ & $\!\!.023\!\!$ & $\!\!30\!\!$ & $\!\!4\!\!$\\
   & $\!\!2\!\!$ & $1$ & $\!3\!$ & $\!...\!$ & $30$ & $\!.5\!$ & $\!.018\!$ & $\!\!4\!\!$ & $\!.75\!$ & $\!.023\!$ & ... & ...\\
   \cline{2-13}
   
   & $\!\!3\!\!$ & $1$ & $3$ & $\!0\!$ & $30$ & $\!\!.4\!\!$ & $\!\!.02\!\!$ & $\!\!4\!\!$ & $\!\!.8\!\!$ & $\!\!.02\!\!$ & $\!\!30\!\!$ & $\!\!4\!\!$\\
   & $\!\!3\!\!$ & $1$ & $\!3\!$ & $\!...\!$ & $30$ & $\!.4\!$ & $\!.02\!$ & $\!\!4\!\!$ & $\!.8\!$ & $\!.02\!$ & ... & ...\\
   \cline{2-13}
   
   & $\!\!4\!\!$ & $1$ & $3$ & $\!0\!$ & $30$ & $\!\!.4\!\!$ & $\!\!.018\!\!$ & $\!\!4\!\!$ & $\!\!.8\!\!$ & $\!\!.023\!\!$ & $\!\!30\!\!$ & $\!\!4\!\!$\\
   & $\!\!4\!\!$ & $1$ & $\!3\!$ & $\!...\!$ & $30$ & $\!.4\!$ & $\!.018\!$ & $\!\!4\!\!$ & $\!.8\!$ & $\!.023\!$ & ... & ...\\
   \cline{2-13}
   
   & $\!\!5\!\!$ & $1$ & $3$ & $\!0\!$ & $30$ & $\!\!.397\!\!$ & $\!.02\!$ & $\!\!4\!\!$ & $\!\!.786\!\!$ & $\!.02\!$ & $\!\!30\!\!$ & $\!\!4\!\!$\\
   & $\!\!5\!\!$ & $1$ & $3$ & $\!...\!$ & $30$ & $\!.397\!$ & $\!.02\!$ & $\!\!4\!\!$ & $\!.786\!$ & $\!.02\!$ & ... & ...\\
   \cline{2-13}
   
   & $\!\!6\!\!$ & $1$ & $3$ & $\!0\!$ & $30$ & $\!.397\!$ & $\!.018\!$ & $\!\!4\!\!$ & $\!.786\!$ & $\!.023\!$ & $\!30\!$ & $4$\\   
   & $\!\!6\!\!$ & $1$ & $\!3\!$ & $5\!$ & $30$ & $\!.397\!$ & $\!.018\!$ & $\!\!4\!\!$ & $\!.786\!$ & $\!.023\!$ & $\!50.1\!$ & $\!62.9\!\!$\\
   & $\!\!6\!\!$ & $1$ & $\!3\!$ & $10$ & $30$ & $\!.397\!$ & $\!.018\!$ & $\!\!4\!\!$ & $\!.786\!$ & $\!.023\!$ & $\!13.8\!$ & $\!8.65\!\!$\\
   & $\!\!6\!\!$ & $1$ & $\!3\!$ & $15$ & $30$ & $\!.397\!$ & $\!.018\!$ & $\!\!4\!\!$ & $\!.786\!$ & $\!.023\!$ & $\!79.3\!$ & $\!8.23\!\!$\\
   & $\!\!6\!\!$ & $1$ & $\!3\!$ & $20$ & $30$ & $\!.397\!$ & $\!.018\!$ & $\!\!4\!\!$ & $\!.786\!$ & $\!.023\!$ & $\!12.6\!$ & $\!30.7\!\!$\\
   & $\!\!6\!\!$ & $1$ & $\!3\!$ & $\!...\!$ & $30$ & $\!.397\!$ & $\!.018\!$ & $\!\!4\!\!$ & $\!.786\!$ & $\!.023\!$ & ... & ...\\
   \end{tabular}
}\endgroup
\caption{`Big' fact table $H_3$ of hypothesis $k\!=\!3$ from Example \ref{ex:population} with u-factors $\{b,\, d\}$ and $\{p,\, r\}$ emphasized (resp.) in colors green and red.}
\label{fig:exogenous}
\vspace{-10pt}
\end{figure}
\end{spacing}

\noindent
U-factor learning is meant to process only the attributes $Z \subset U$ from $H_k[U]$ that are inferred exogenous in the given hypothesis, i.e., for all $A \in Z$, there is an fd $\langle X, A\rangle \in \Phi(\Sigma_k)$, where the latter is the $\phi$-projection of $\Sigma_k$. Such attributes are then `officially' unrelated. In fact, by `casual' fd's we mean correlations that, for a set of experimental trials, may occasionally show up in the trial input data; e.g., $x_0 \leftrightarrow y_0$ hold in $H_3$, but not because $x_0$ and $y_0$ are related in principle (theory). 

Fig. \ref{fig:exogenous} helps to illustrate Problem \ref{prob:u-learning} through the `big' fact table. We emphasize u-factors $\{b,\, d\}$ and $\{p,\, r\}$ in colors green and red. Observe that values of $b$ are strongly correlated (one-to-one) with values of $\,d\,$ for $\phi=1$, just like $p$ and $r$. Note also that $\{x_0, y_0\}$ can be seen as a certain factor. From the user point of view, this is a record that reflects a common practice in computational science known as (parameter) sensibility analysis.

Problem \ref{prob:u-learning} is dominated by the (problem of) discovery of fd's in a relation, which is not really a new problem (e.g., see \cite{huhtala1999}). We then keep focus on the synthesis method as a whole and omit our detailed \textsf{u-factor-learning} algorithm in particular.\footnote{In short, we make use of relational algebra \textsf{group-by} operation and build a pruned lattice of attribute groups having the same number of rows under the grouping (similarly to \cite{huhtala1999}).} Its output, fd set $\Omega_k$, is then filled in (completed) with the $\upsilon$-projection $\Upsilon(\Sigma_k)$. 
For illustration consider hypothesis $\upsilon\!=\!3$ and its trial input data recorded in $H_3$ in Fig. \ref{fig:population-table}. We show its resulting fd set $\Omega_3$ in Fig. \ref{fig:y-fdschema} (left), together with its folding $\Omega_k^\looparrowright$ (right). The latter is then input to (Alg. \ref{alg:merge}) \textsf{merge} to get the final information necessary for the actual synthesis of U-relations, as captured in Def. \ref{def:u-factorization-fd-set}. 
For an illustration of the merging of fd's with equivalent left-hand sides, note in Fig. \ref{fig:y-fdschema} (right) that $\,\phi\,x_0\,b\,p\,t\,\upsilon\,y \,\leftrightarrow\, \phi\,x_0\,b\,p\,t\,\upsilon\,x$ holds in $(\Omega_3^\looparrowright)^+$.

\begin{mydef}\label{def:u-factorization-fd-set}
Let $\mathcal S_k$ and $H_k$ be the complete structure and `big' fact table of hypothesis $k$, and $\Sigma_k$ an fd set defined $\Sigma_k \triangleq \textsf{h-encode}(\mathcal S_k)$. Now, let $\Omega_k \triangleq \textsf{u-factor-learning}(\,H_k,$ $\Phi(\Sigma_k)\,) \;\cup\; \Upsilon(\Sigma_k)$, and $\Omega_k^\looparrowright$ be the folding of $\Omega_k$. Finally, define $\Gamma_k \triangleq \textsf{merge}(\,\Omega_k^\looparrowright\,)$. Then we say that $\Gamma_k$ is the \textbf{u-factorization} of $\mathcal S_k$ over $H_k$.
\end{mydef}

\begin{figure}[t]
\begin{framed}
\vspace{-16pt}
\begin{subfigure}{0.45\columnwidth}
\begin{eqnarray*}
\Omega_3 = \{\;\;\; \phi\;x_0 &\to& y_0,\\ 
\phi\;b &\to& d,\\ 
\phi\;p &\to& r,\\ 
x_0\;b\;p\;t\;\upsilon\;y &\!\to\!& x,\\ 
y_0\;d\;r\;t\;\upsilon\;x &\!\to\!& y \;\;\;\}.
\end{eqnarray*}
\end{subfigure}
\hspace{12pt}
\begin{subfigure}{0.45\columnwidth}
\begin{eqnarray*}
\Omega_3^\looparrowright  = \{\;\;\; \phi\;x_0 &\to& y_0,\\ 
\phi\;b &\to& d,\\ 
\phi\;p &\to& r,\\ 
\phi\;x_0\;b\;p\;t\;\upsilon\;y &\!\to\!& x,\\ 
\phi\;x_0\;b\;p\;t\;\upsilon\;x &\!\to\!& y \;\;\;\}.
\end{eqnarray*}
\end{subfigure}
\end{framed}
\vspace{-5pt}
\caption{Fd set $\Omega_3$ (compare with $\Sigma_3$) and its folding $\Omega_3^\looparrowright\!$.}
\label{fig:y-fdschema}
\vspace{-6pt}
\end{figure}

\begin{spacing}{1.1}
\begin{algorithm}[H]
\caption{Merge fd's with equivalent left-hand sides.}
\label{alg:merge}
\begin{algorithmic}[1]
\Procedure{merge}{$\Sigma: \text{fd set}$}
\vspace{1pt}
\State $\Omega \gets \varnothing$
\ForAll{$\langle X, C\rangle \in \Sigma$}
\If{there is $\langle Z, W\rangle \in \Omega$ such that $X \leftrightarrow Z$ holds in $\Sigma^+$}
\State $\Omega \gets \Omega \setminus \langle Z, W\rangle$
\State $S \gets X \setminus Z$
\State $\Omega \gets \Omega \cup \langle Z, WSC\rangle$ \Comment{merges equivalent keys}
\Else
\State $\Omega \gets \Omega \cup \langle X, C\rangle$
\EndIf
\EndFor
\State \Return $\Omega$
\EndProcedure
\end{algorithmic}
\end{algorithm}
\end{spacing}

\begin{myremark}\label{rmk:claim}
Let $\Gamma_k$ be the u-factoriza\-tion of structure $\mathcal S_k$ over `big' fact table $H_k$.  Then every fd in $\Gamma_k$ encodes a clear-cut \emph{claim}, either empirical, in $\Phi(\Gamma_k)$, 
or theoretical, in $\Upsilon(\Gamma_k)$. 
That is ensured by the \textsf{merge} algorithm, which groups fd's in $\Omega_k^\looparrowright$ with equivalent left-hand sides.
$\Box$
\end{myremark}

We are now able to employ a notion of u-factor decomposition formulated in Def. \ref{def:u-factor} into query formula (\ref{eq:u-factor}) in p-WSA's extension of relational algebra.

\begin{mydef}
\label{def:u-factor}
Let $\mathcal S_k$ be the complete structure of hypothesis $k$, and $H_k[U]$ its `big' fact table such that $\Gamma_k$ is the u-factorization of $\mathcal S_k$ over $H_k$. 
Now, let $G_a \subset U$ be a set of attributes $G_a=A\,G$ such that, for all $B \in G$, an fd $\phi\,A \!\to B$ exists in $\Phi(\Gamma_k)$. 
Then we define U-relation $Y_k^i\,[V_i\,D_i\,|\,\phi\, A\,G\,]$ by query formula (\ref{eq:u-factor}), and say that $Y_k^i$ is a \textbf{u-factor projection} of $H_k$.
\begin{eqnarray}
\label{eq:u-factor}
\! Y_k^i := \pi_{\,\phi\, A\,G}\,( \textsf{repair-key}_{\,\phi\, @\textsf{\,\emph{count}\,}} (\,\gamma_{\,\phi,\, A,\,G,\,\textsf{\emph{count}}(*)}(H_k)\,)\,)
\end{eqnarray}
where $\gamma$ is relational algebra's grouping operator. 
\end{mydef}

\begin{spacing}{1.1}
\begin{figure}[H]
\advance\leftskip-0.2cm
\begingroup\setlength{\fboxsep}{2pt}
\colorbox{yellow!10}{%
   \begin{tabular}{c|>{\columncolor[gray]{0.92}}c||c|c|c}
  $\!$\textsf{Y}$_3^1\!$ & $\!\!V \!\mapsto\! D\!\!$ & $\!\phi\!$ & $\!x_0\!$ & $y_0\!\!\!$\\
      \hline    
   & $\textsf{x}_2 \!\mapsto\! 1$ & $\!2\!$ & $\!30\!$ & $4\!$\\
   \end{tabular}
}\endgroup
\begingroup\setlength{\fboxsep}{2pt}
\colorbox{yellow!10}{%
   \begin{tabular}{c|>{\columncolor[gray]{0.92}}c||c|c|c}
 $\!$\textsf{Y}$_3^2\!$ & $\!\!V \!\mapsto\! D\!\!$ & $\!\phi\!$ & $b$ & $d\!\!\!$\\
      \hline    
   & $\textsf{x}_3 \!\mapsto\! 1$ & $\!2\!$ & $.5$ & $.5\!$\\
   & $\textsf{x}_3 \!\mapsto\! 2$ & $\!2\!$ & $.4$ & $.8\!$\\
   & $\textsf{x}_3 \!\mapsto\! 3$ & $\!2\!$ & $\!.397\!$ & $\!.786\!\!\!\!$\\
   \end{tabular}
}\endgroup
\begingroup\setlength{\fboxsep}{2pt}
\colorbox{yellow!10}{%
   \begin{tabular}{c|>{\columncolor[gray]{0.92}}c||c|c|c}
  $\!$\textsf{Y}$_3^3\!$ & $\!\!V \!\mapsto\! D\!\!$ & $\!\phi\!$ & $p$ & $r\!\!\!\!$\\
      \hline    
   & $\textsf{x}_4 \!\mapsto\! 1$ & $\!2\!$ & $\!.020\!$ & $\!.020\!\!\!$\\
   & $\textsf{x}_4 \!\mapsto\! 2$ & $\!2\!$ & $\!.018\!$ & $\!.023\!\!\!$\\
   \end{tabular}
}\endgroup
\caption{U-factor projections rendered for hypothesis $\upsilon=3$.}
\label{fig:u-factors}
\vspace{-2pt}
\end{figure}
\end{spacing}

\noindent
The synthesis of u-factor projections, in particular the application of \textsf{repair-key} (cf. Eq. \ref{eq:u-factor}), has an important consequence for the u-factorization $\Gamma_k$ of $H_k$, viz., the introduction of new fd's into $\Gamma_k^\prime$ defined as follows (see Def. \ref{def:repaired-fd-set}). We shall consider it (rather than $\Gamma_k$) to study design-theoretic properties of synthesized $\boldsymbol Y_k$ in \S\ref{sec:properties}.

\begin{mydef}\label{def:repaired-fd-set}
Let $\mathcal S_k$ and $H_k[U]$ be (resp.) the complete structure and `big' fact table of hypothesis $k$, and $\Gamma_k$ be the u-factorization of $\mathcal S_k$ over $H_k$. Now, let $\Gamma_k^\prime \,\triangleq \bigcup_{\,i \,\in\, I}\{\phi \!\to A_i\,G_i\} \,\cup\, \Gamma_k$, where $I$ indexes all u-factor projections $Y_k^i\,[\,V_i\,D_i\,|\,\phi\,A_i\,G_i\,]$ of $H_k$. We say that $\Gamma_k^\prime$ is the \textbf{repaired factorization} of $\mathcal S_k$ over $H_k$.
\end{mydef}

\section{U-Propagation}\label{sec:u-propagation}
\noindent
U-propagation is a central part of \textsf{U-intro} and the pipeline itself. 
Recall that all the machinery developed so far, from hypothesis encoding to causal reasoning to u-factorization is for enabling predictive analytics. Let us briefly reconstruct it.

For hypothesis structure $\mathcal S_k(\mathcal E, \mathcal V)$, take any endogenous variable \mbox{$x_c \in \mathcal V$} encoded by attribute $C \!\mapsto\! x_c$. There should be exactly one fd $\langle X, C\rangle \in \Upsilon(\Sigma_k)^\looparrowright\!$. By Theorem \ref{thm:connections2}, for every first cause $x_b$ of $x_c$ there is $B \in X$ where $B \mapsto x_b$ with $x_b \in \mathcal V$. 
Now, observe that when $\Omega_k$ is rendered by u-factor learning, it is filled partly with fd's from $\Upsilon(\Sigma_k)$, and partly with fd's processed from $\Phi(\Sigma_k)$. This is to summarize exogenous variables into clear-cut independent u-factors. It means that, after u-factor learning, each first cause $x_b$ encoded by $B$ shall be represented by some pivot attribute $A_i$ which is, if not $A_i=B$ itself, then occasionally strongly correlated to it (i.e., $\phi\,A_i \to B$ and $\phi\,B \!\to A_i$ hold in $H_k$). 

Further then $\Omega_k$ is subject to folding such that, if $B \in X$ for $\langle X, C\rangle \in \Upsilon(\Omega_k)$, now we have $A_i \in Z$ for $\langle Z, C\rangle \in \Upsilon(\Omega_k^\looparrowright)$. 
This processing from fd $X \!\to C$ (where $X$ contains the first causes) into $Z \!\to C$ (where $Z$ contains only their pivot representatives) is meant for enabling an economical representation of uncertainty. Our running example is small, but such a principle is quite relevant for large-scale hypotheses (say, when $|\mathcal S| \approx 1M$). The correctness of such u-factor summarization shall be ensured by Proposition \ref{prop:folding-parsimonious}, which let us know that $\Omega_k^\looparrowright$ is parsimonious then (by Def. \ref{def:parsimonious}) canonical, therefore (by Def. \ref{def:minimal}) left-reduced. 

Now, all fd's in $\Upsilon(\Omega_k^\looparrowright)$ have form $Z \!\to C$ and we are almost ready for u-propagation. 
Note that, as a result of u-factorization, each pivot attribute $A_i \in Z$ is associated with random variable $V_i$ from U-relation $Y_k^i\,[\,V_i\,D_i \,|\,\phi\, A_i\,G_i\,]$. Then we shall use each $A_i \in Z$ (from $Z \!\to C$) as a surrogate to $Y_k^i$ in order to propagate factorized uncertainty into `predictive' U-relations $Y_k^j\,[\,\overline{V_j\,D_j} \,|\,S\,T\,]$ by a join formula. Attribute sets $S$ and $T$ are defined after merging fd's in $\Omega_k^\looparrowright\!$ with equivalent lhs to get $\Gamma_k$ and pass it as argument for synthesis. We let $S=Z \setminus W$ such that $S$ contains the domain variables only (e.g., $\phi,\, \upsilon,\, t$). The pivot attributes in $Z$ shall not be included in the data columns of $Y_k^j$, but leave their trace through the condition columns $\overline{V_j\,D_j}$ that annotate $sch(Y_k^j)$ as a repair of the key $S \!\to T$.

All that (cf. Def. \ref{def:u-propagation}) is abstracted into general p-WSA query formula (\ref{eq:u-propagation}), and employed in (Alg. \ref{alg:synthesize}\,) \textsf{synthesize} to accomplish u-propagation (Part II).

\begin{mydef}
\label{def:u-propagation}
Let $\mathcal S_k$ be the complete structure of hypothesis $k$, and $H_k[U]$ its `big' fact table such that $\Gamma_k$ is the u-factorization of $\mathcal S_k$ over $H_k$. 
Now, let $Z \!\to T$ be an fd in $\Upsilon(\Gamma_k)$.
Then we define U-relation $Y_k^j\,[\,\overline{V_j\,D_j}\,|\,S\,T]$ by query formula (\ref{eq:u-propagation}), and say that $Y_k^{j}$ is a \textbf{predictive projection} of $H_k$, where:
\vspace{-8pt}
\begin{eqnarray}
\label{eq:u-propagation}
Y_k^{j} :=\, \pi_{\,S,\,T\,}(\, \sigma_{\upsilon=k}(Y_0) \bowtie (\bowtie_{\,i\,\in\, I} Y_k^i\,) \bowtie \,H_k\,)
\end{eqnarray}
\begin{itemize}
\vspace{-11pt}
\item[(a)] $\;Y_k^i\,[\,V_i\,D_i \,|\, \phi\, A_i\, G_i\,]$ is a u-factor projection of $H_k$;
\vspace{-8pt}
\item[(a)] we have $i \in I$ if $A_i \in Z$;
\vspace{-8pt}
\item[(c)] we take $S = Z \setminus W$, where $W$ is the set of all pivot attributes representing first causes in $Z$.
\end{itemize}
\end{mydef}

\begin{spacing}{1.1}
\begin{algorithm}[H]
\caption{p-DB synthesis applied over folding fd set.}
\label{alg:synthesize}
\begin{algorithmic}[1]
\Procedure{synthesize}{$\mathcal S_k\!: \text{structure}$, $H_k\!: \text{`big' table}$, $Y_0\!: \text{explan. table}$}
\vspace{2pt}
\Require $\mathcal S_k$ is complete \vspace{1pt}
\Ensure U-relations $\boldsymbol Y_k$ returned are a BCNF, lossless decomposition of $H_k$
\vspace{-1pt}
\State $\Sigma_k \gets \textsf{h-encode}(\mathcal S_k)$
\State $\Omega_k \gets \textsf{u-factor-learning}(\,\Phi(\Sigma_k),\, H_k\,) \;\cup\; \Upsilon(\Sigma_k)$
\State $\Gamma_k \gets \textsf{merge}(\, \textsf{folding}(\Omega_k)\,)$
\vspace{-2pt}
\Statex{}
\hrulefill
\Statex{Part I: \textbf{U-factorization}}
\vspace{1pt}
\ForAll{$\langle \phi\,A, G\rangle \in \Phi(\Gamma_k)$} \Comment{scans over the u-factors of hypothesis $k$}
\State $Y_k^i \gets \pi_{\,\phi,\, A,\,G}\,( \textsf{repair-key}_{\phi @\textsf{\emph{count}}} (\,\gamma_{\,\phi,\, A,\,G,\, \textsf{\emph{count}}(*)}(H_k)\,)\,)$
\State $\boldsymbol Y_k \gets \boldsymbol Y_k \cup Y_k^i$
\EndFor
\vspace{-6pt}
\Statex{}
\hrulefill
\Statex{Part II: \textbf{U-propagation}}
\vspace{1pt}
\ForAll{$\langle Z, T\rangle \in \Upsilon(\Gamma_k)$} \Comment{scans over the claims of hypothesis $k$}
\State $W \gets \varnothing$ \Comment{prepares to keep track of u-factor pivot attributes}
\ForAll{$Y_k^i\,[\,\phi\,A_i\,G_i\,] \in \boldsymbol Y_k$}\vspace{2pt}
\If{$A \in Z$} \Comment{$A$ is a first cause of all $B \in T$}
\State $I = I \cup \{i\}$ \Comment{indexes the u-factor projection}
\State $W \gets W \cup A$ \Comment{keeps track of u-factor's pivot attribute}
\EndIf
\EndFor
\State $S \gets Z \setminus W$ \Comment{removes u-factor pivot attributes}
\State $Y_k^j \gets \pi_{\,S,\,T\,}(\, \sigma_{\upsilon=k}(Y_0) \bowtie (\bowtie_{\,\,i\,\in\, I} Y_k^i) \bowtie H_k\,) \!\!\!$\vspace{2pt}
\State $\boldsymbol Y_k \gets \boldsymbol Y_k \cup Y_k^j$
\EndFor
\State \Return $\boldsymbol Y_k$
\EndProcedure
\end{algorithmic}
\end{algorithm}
\end{spacing}

Fig. \ref{fig:u-relations} shows the rendered U-relations for hypothesis $k\!=\!3$ whose `big' fact table is shown in Fig. \ref{fig:population-table}. Note that $\textsf{tid}\!=\!6$ in $H_3$ corresponds now to $\theta = \{\,\textsf{x}_1 \!\mapsto\! 3,\, \textsf{x}_2 \!\mapsto\! 1,$ $\textsf{x}_3 \!\mapsto\! 3,\, \textsf{x}_4 \!\mapsto\! 2 \,\}$, where $\theta$ defines a particular world in $\boldsymbol W$ whose probability is $\textsf{Pr}(\theta) \!\approx\! .055$. This value is derived from the marginal probabilities stored in world table $W$ (see Fig. \ref{fig:u-relations}) as a result of the application of formulas Eq. \ref{eq:explanation} and Eq. \ref{eq:u-factor}.

\begin{myremark}\label{rmk:u-intro}
Observe that, although (Alg. \ref{alg:synthesize}) \textsf{synthesize} operates locally for each hypothesis $k$, the effects of p-DB synthesis (\textsf{U-intro}) in the pipeline are global on account of the (global) `explanation' relation $H_0$ (then U-relation $Y_0$), e.g., see Fig. \ref{fig:u-relations}. In fact, the probability of each tuple (row), say, in U-relation $Y_k^j$ with $\phi=p$ for hypothesis $\upsilon=k$, is distributed among all the hypotheses $\ell\neq k$ that are keyed in $Y_0$ under $\phi=p$, i.e., all hypotheses that compete at $\phi=p$. $\Box$
\end{myremark}

\begin{spacing}{1}
\begin{figure}[H]
\advance\leftskip-0.2cm
\begin{center}
\begingroup\setlength{\fboxsep}{2pt}
\colorbox{yellow!10}{%
   \begin{tabular}{c|>{\columncolor[gray]{0.92}}c||c|c}
  $\!$\textsf{Y}$_0\!$ & $\!\!V \!\mapsto\! D\!\!$ & $\!\phi\!$ & $\upsilon$\\
      \hline    
   & $\textsf{x}_1 \!\mapsto\! 1$ & $\!1\!$ & $\!1\!$\\
   & $\textsf{x}_1 \!\mapsto\! 2$ & $\!1\!$ & $\!2\!$\\
   & $\textsf{x}_1 \!\mapsto\! 3$ & $\!1\!$ & $\!3\!$\\
   \end{tabular}
}\endgroup
\hspace{3pt}
\begingroup\setlength{\fboxsep}{2pt}
\colorbox{yellow!15}{%
   \begin{tabular}{c|>{\columncolor[gray]{0.92}}c||c}
  \textsf{W} & $V \mapsto D$ & \textsf{Pr}\\
      \hline    
   & $\textsf{x}_1 \mapsto 1$ & $.4$\\
   & $\textsf{x}_1 \mapsto 2$ & $.4$\\
   & $\textsf{x}_1 \mapsto 3$ & $.2$\\
   \cline{2-3}
   & $\textsf{x}_2 \mapsto 1$ & $1$\\
   \cline{2-3}
   & $\textsf{x}_3 \mapsto 1$ & $.33$\\
   & $\textsf{x}_3 \mapsto 2$ & $.33$\\
   & $\textsf{x}_3 \mapsto 3$ & $.33$\\
   \cline{2-3}
   & $\textsf{x}_4 \mapsto 1$ & $.5$\\
   & $\textsf{x}_4 \mapsto 2$ & $.5$\\
   \end{tabular}
}\endgroup
\vspace{-6pt}\\
\end{center}
\begingroup\setlength{\fboxsep}{2pt}
\colorbox{yellow!10}{%
   \begin{tabular}{c|>{\columncolor[gray]{0.92}}c||c|c|c}
  $\!$\textsf{Y}$_3^1\!$ & $\!\!V \!\mapsto\! D\!\!$ & $\!\phi\!$ & $\!x_0\!$ & $y_0\!\!\!$\\
      \hline    
   & $\textsf{x}_2 \!\mapsto\! 1$ & $\!1\!$ & $\!30\!$ & $4\!$\\
   \end{tabular}
}\endgroup
\begingroup\setlength{\fboxsep}{2pt}
\colorbox{yellow!10}{%
   \begin{tabular}{c|>{\columncolor[gray]{0.92}}c||c|c|c}
 $\!$\textsf{Y}$_3^2\!$ & $\!\!V \!\mapsto\! D\!\!$ & $\!\phi\!$ & $b$ & $d\!\!\!$\\
      \hline    
   & $\textsf{x}_3 \!\mapsto\! 1$ & $\!1\!$ & $.5$ & $.5\!$\\
   & $\textsf{x}_3 \!\mapsto\! 2$ & $\!1\!$ & $.4$ & $.8\!$\\
   & $\textsf{x}_3 \!\mapsto\! 3$ & $\!1\!$ & $\!.397\!$ & $\!.786\!\!\!\!$\\
   \end{tabular}
}\endgroup
\begingroup\setlength{\fboxsep}{2pt}
\colorbox{yellow!10}{%
   \begin{tabular}{c|>{\columncolor[gray]{0.92}}c||c|c|c}
  $\!$\textsf{Y}$_3^3\!$ & $\!\!V \!\mapsto\! D\!\!$ & $\!\phi\!$ & $p$ & $r\!\!\!\!$\\
      \hline    
   & $\textsf{x}_4 \!\mapsto\! 1$ & $\!1\!$ & $\!.020\!$ & $\!.020\!\!\!$\\
   & $\textsf{x}_4 \!\mapsto\! 2$ & $\!1\!$ & $\!.018\!$ & $\!.023\!\!\!$\\
   \end{tabular}
}\endgroup
\vspace{6pt}\\
\begingroup\setlength{\fboxsep}{3pt}
\colorbox{yellow!10}{%
   \begin{tabular}{c|>{\columncolor[gray]{0.95}}c|>{\columncolor[gray]{0.95}}c|>{\columncolor[gray]{0.95}}c|>{\columncolor[gray]{0.95}}c||c|c|c|c|c}
  \textsf{Y}$_3^4$ & $V_1 \!\mapsto\! D_1$ & $V_2 \!\mapsto\! D_2$ & $V_3 \!\mapsto\! D_3$ & $V_4 \!\mapsto\! D_4$ & $\phi$ & $\upsilon$ & $t$ & $y$ & $x$\\
      \hline    
   & $\textsf{x}_1 \mapsto 3$ & $\textsf{x}_2 \mapsto 1$ & $\textsf{x}_3 \mapsto 1$ & $\textsf{x}_4 \mapsto 1$ & $\!1\!$ & $3$ & $1900$ & $4$ & $30$\\
   & $\textsf{x}_1 \mapsto 3$ & $\textsf{x}_2 \mapsto 1$ & $\textsf{x}_3 \mapsto 1$ & $\textsf{x}_4 \mapsto 1$ & $1$ & $3$ & $\!...\!$ & ... & ...\\
   \cline{2-10}
   & $\!\!...\!\!$ & $\!\!...\!\!$ & $\!\!...\!\!$ & $\!\!...\!\!$ & $\!1\!$ & $3$ & $\!...\!$ & ... & ...\\
   \cline{2-10}
   & $\textsf{x}_1 \mapsto 3$ & $\textsf{x}_2 \mapsto 1$ & $\textsf{x}_3 \mapsto 3$ & $\textsf{x}_4 \mapsto 2$ & $1$ & $3$ & $1900$ & $4$ & $30$\\
   & $\textsf{x}_1 \mapsto 3$ & $\textsf{x}_2 \mapsto 1$ & $\textsf{x}_3 \mapsto 3$ & $\textsf{x}_4 \mapsto 2$ & $1$ & $3$ & $1901$ & $4.12$ & $41.5\!\!$\\
   & $\textsf{x}_1 \mapsto 3$ & $\textsf{x}_2 \mapsto 1$ & $\textsf{x}_3 \mapsto 3$ & $\textsf{x}_4 \mapsto 2$ & $1$ & $3$ & $1902$ & $5.78$ & $56.7$\\
   & $\textsf{x}_1 \mapsto 3$ & $\textsf{x}_2 \mapsto 1$ & $\textsf{x}_3 \mapsto 3$ & $\textsf{x}_4 \mapsto 2$ & $1$ & $3$ & $1903$ & $11.7$ & $72.8$\\
   & $\textsf{x}_1 \mapsto 3$ & $\textsf{x}_2 \mapsto 1$ & $\textsf{x}_3 \mapsto 3$ & $\textsf{x}_4 \mapsto 2$ & $1$ & $3\!$ & $1904$ & $31.1$ & $75.9$\\
   & $\textsf{x}_1 \mapsto 3$ & $\textsf{x}_2 \mapsto 1$ & $\textsf{x}_3 \mapsto 3$ & $\textsf{x}_4 \mapsto 2$ & $1$ & $3$ & $\!...\!$ & ... & ...\\
   \end{tabular}
}\endgroup
\caption{U-relations rendered for hypothesis $\upsilon=3$.}
\label{fig:u-relations}
\vspace{-2pt}
\end{figure}
\end{spacing}

U-relations rendered by p-DB synthesis are ready for querying. Typical queries comprise the \textsf{conf()} aggregate operation, inquiring the probability (or confidence) for each tuple to `true' in the probability space captured by the hypothesis competition. We illustrate queries in Chapter \ref{ch:applicability}.

\section{Design-Theoretic Properties}\label{sec:properties}

\noindent
For the \textsf{U-intro} procedure to be meaningful, we have to study design-theoretic properties of u-factor projections and prediction projections synthesized out of `big' fact table $H_k$ for the sake of predictive analytics. In particular, for the projections to be claim-centered, we submit that they should satisfy Boyce-Codd normal form (BCNF, cf. Def. \ref{def:nf}) w.r.t. the repaired factorization $\Gamma_k^\prime$ of $H_k$; 
and for them to be a correct decomposition of the uncertainty present in $H_k$, their join should be lossless (preserve the data in $H_k$, cf. Def. \ref{def:lossless-join}) w.r.t. $\Gamma_k^\prime$. 

Note that in this study we consider repaired factorization $\Gamma_k^\prime$ (not u-factoriza\-tion $\Gamma_k$), since it is the one which actually holds in $\boldsymbol Y_k$ after key repairing.

\subsection{Claim-Centered Decomposition}

\noindent
As emphasized through Remark \ref{rmk:claim}, every fd in u-factorization $\Gamma_k$ is a claim (cf. Remark \ref{rmk:claim}), then the same holds for repaired factorization $\Gamma_k^\prime$. Thus, for a claim-centered decomposition of `big' fact table $H_k$, it is desirable that U-relational schema $\boldsymbol Y_k$ that it satisfies BCNF w.r.t. $\Gamma_k^\prime$. BCNF (`do not represent the same fact twice' \cite[p. 251]{abiteboul1995}) is our notion of `good design' for uncertainty decomposition in view of predictive analytics. This is to avoid the uncertainty of one claim to be undesirably mixed with the uncertainty of another claim.

\begin{mydef}\label{def:fd-projection}
Let $R[U]$ be a relation scheme over set $U\!$ of attributes, and $\Sigma$ a set of fd's. Then the \textbf{projection} of $\Sigma$ onto $R[U]$, written $\pi_U(\Sigma)$, is the subset $\Sigma^\prime \subseteq \Sigma$ of fd's $X \!\to Z$ such that $XZ \subseteq U$.
\end{mydef}

\begin{mydef}\label{def:nf}
Let $R[U]$ be a relation scheme over set $U\!$ of attributes, and $\Sigma$ a set of fd's on $U$. We say that:
\begin{itemize}
\item[(a)] $R$ is in \textbf{BCNF} if, for all $\langle X, A\rangle \in \Sigma^+\!$ with $A \!\not\in X$ and $XA \!\subseteq\! U$, we have $X \!\to U$ (i.e., $\!X$ is a superkey for $R$);\vspace{-2pt}
\item[(b)] A schema $\boldsymbol R$ is in BCNF if all of its schemes $R_1, ..., R_n \in \boldsymbol R$ are in BCNF. 
\end{itemize}
\end{mydef}

\begin{myex}
To illustrate the concept of BCNF, let us consider canonical fd set $\Sigma\!=\!\{{A \!\to B,}$ $B \!\to C\,\}$ over attributes $U\!=\!\{A, B, C\}$, and a tentative schema containing a single relation $R[A B C]$. This relation is not in BCNF because, for one, $B \!\to C$ violates it ($C \nsubseteq B$ but $B$ is not a superkey for $R$). 
$\Box$
\label{ex:folding}
\end{myex}

\noindent

Observe also that an overdecomposed schema may (trivially) satisfy BCNF. 
For example, let $\Sigma=\{A \!\to B,\, A \!\to C\}$ then by Def. \ref{def:nf} both schemas $\boldsymbol R=R_1[ABC]$ and $\boldsymbol R^\prime=\{R_1[AB],\; R_2[AC]\}$ are in BCNF w.r.t. $\Sigma$. The second, however, breaks data into two tables making their access more difficult than necessary since both $B$ and $C$ brings in information about $A$. That is, if the schema were to be synthesized over the fd's in $\Sigma$, then it would be desirable to apply (R4) union or merge them before. Our point is that, if we target at a BCNF-satisfying schema, then it is also desirable for it to be the minimal-cardinality schema in BCNF.

Theorem \ref{thm:bcnf} guarantees the BCNF property w.r.t. $\Gamma_k^\prime$ by design for every schema $\boldsymbol Y_k$ rendered by (Alg. \ref{alg:synthesize}) \textsf{synthesize} over $\Gamma_k$.

\begin{mythm}\label{thm:bcnf}
Let $\mathcal S_k$ and $H_k$ be (resp.) the complete structure and `big' fact table of hypothesis $k$, and let $\;\Gamma_k^\prime$ be the repaired factorization of $\mathcal S_k$ over $H_k$, and $Y_0$ the `explanation' table where hypothesis $k$ is recorded. Now, let $\boldsymbol Y_k$ be a U-relational schema defined $\boldsymbol Y_k \triangleq \textsf{synthesize}(\mathcal S_k, H_k, Y_0)$. 
Then $\boldsymbol Y_k$ is in BCNF w.r.t. $\Gamma_k^\prime$ and is minimal-cardinality. 
\end{mythm}
\begin{proof}
We exploit the fact that the projection of $(\Gamma_k^\prime)^+$ onto u-factor projections and predictive projections define a disjoint partition of $(\Gamma_k^\prime)^+$ into its $\phi$-projection $\Phi(\Gamma_k^\prime)^+$ and $\upsilon$-projection $\Upsilon(\Gamma_k^\prime)^+$. Since we know the form of fd's in each of them, the search space for BCNF violations is significantly reduced. The minimality of $|\boldsymbol Y_k|$ in turn comes from (Alg. \ref{alg:merge}) \textsf{merge}. 
See Appendix, \S\ref{a:bcnf}.
$\Box$
\end{proof}

\subsection{Correctness of Uncertainty Decomposition}

\noindent
Recall from the preliminaries (cf. \S\ref{sec:u-relations}) that the U-relational equivalent of the relational product operation (main sub-operation of the join operation) has been introduced. Now, we provide the classical definition of a lossless join \cite{ullman1988}, i.e., when a decomposition of data from a relation into two or more relations is known to preserve the data in its original form by an application of the join. 
\begin{mydef}\label{def:lossless-join}
Let $R[U]$ be a (U-)relational schema synthesized into collection $\boldsymbol R = \bigcup_{i=1}^n R_i$ and let $\Sigma$ be an fd set on attributes $U$. We say that $\boldsymbol R$ has a \textbf{lossless join} w.r.t. $\Sigma$ if for every instance $r$ of $R[U]$ satisfying $\Sigma$, we have $r =\; \bowtie_{\,i=1}^{\,n} \pi_{R_i} (r)$.
\end{mydef}

The lossless join property is of interest to ensure that our decomposition of the data from the `big' fact table into u-factor projections preserves the data so that their join to `annotate' the predictive projections when propagated by means of the U-relational join operation is correct. Theorem \ref{thm:lossless} guarantees that is the case.

\begin{mythm}\label{thm:lossless}
Let $\mathcal S_k$ be the complete structure of hypothesis $k$, and $H_k[U]$ its `big' fact table such that $\Gamma_k^\prime$ is the repaired factorization of $\mathcal S_k$ over $H_k$ and $Y_0$ is the `explanation' table where hypothesis $k$ is recorded. Now, let $\boldsymbol Y_k$ be a U-relational schema defined $\boldsymbol Y_k \triangleq \textsf{synthesize}(\mathcal S_k, H_k, Y_0)$. 
Then,
\begin{itemize}
\item[(a)] the join $\bowtie_{\,i=1}^{\,m}  Y_k^i\,[\,V_i\,D_i\,|\,\phi\, A_i\,G_i\,]$ of any subset of the u-factor projections  of $H_k$ is lossless w.r.t. $\Gamma_k^\prime$. 
\item[(b)] any predictive projection $Y_k^j\,[\,\overline{V_j\,D_j}\,|\,S\,T\,]$, result of a join of the theoretical u-factor $Y_0\,[\,V_0\,D_0\,|\,\phi\,\upsilon\,]$ with the `big' fact table $H_k[U]$ and in turn with u-factor projections $Y_k^i\,[\,\overline{V_i\,D_i}\,|\,\phi\,A_i\,G_i\,]$, is lossless w.r.t. $\Gamma_k^\prime$.
\end{itemize}
\end{mythm}
\begin{proof}
We make use of a lemma from Ullman \cite[p. 397]{ullman1988}, and then the proof comes straightforwardly. 
See Appendix, \S\ref{a:lossless}.
$\Box$
\end{proof}

\begin{myremark}\label{rmk:lossless}
The significance of Theorem \ref{thm:bcnf} lies in that it guarantees the decomposition of uncertainty based on the causal ordering processing is in fact claim-centered as desirable for predictive analytics. 
Theorem \ref{thm:lossless} in turn is significant as it ensures that all the \emph{empirical} uncertainty implicit in a hypothesis `big' fact table $H_k$ can be decomposed into u-factor projections that are (a) independent (not strongly correlated, cf. Problem \ref{prob:u-learning}), and (b) can be fully recovered by a join that is lossless w.r.t. repaired factorization $\Gamma_k^\prime$ of structure $\mathcal S_k$ over $H_k$. This is essential to make sure that in u-propagation the composition of the required u-factors recovers the uncertainty associated with the predictive data. Since repaired factorization $\Gamma_k^\prime$ is known to be a correct processing of the causal ordering (cf. results of Chapter \ref{ch:reasoning}), altogether Theorem \ref{thm:lossless} guarantees that the first causes are joined together correctly towards the predictive variables influenced by them.
$\Box$
\end{myremark}

As we have seen, the p-DB synthesis technique presented here is essentially targeted at design-theoretic properties. It is also motivated by computational performance, as uncertainty decomposition is desirable also to speed up probabilistic inference \cite[p. 30-1]{suciu2011}. In fact, the \textsf{U-intro} procedure is fully grounded in U-relations and p-WSA as implemented in the \textsf{MayBMS} system. Its computational performance is dominated by U-relational query processing. We present experimental studies on the \textsf{U-intro} procedure in \S\ref{sec:experiments-app}, as they are designed from a applicability point of view. The goal is to provide some reference computational measures for prospective users.

\section{Related Work}\label{sec:related-work-synthesis4u}

Informed on research on Graphical Models (GM) \cite{darwiche2009}, Suciu et al. provide a striking motivation for work on \emph{probabilistic database design} \cite[p.30-1]{suciu2011}. In GM design, probability distributions on large sets of random variables are decomposed into factors of simpler probability functions, over small sets of these variables. The factors can be identified, e.g., by using a set of axioms (the so-called `graphoids') for reasoning about the probabilistic independence of variables \cite{verma1988}. 
The same design principle (sic.) applies to p-DB's \cite{suciu2011}: the data should be decomposed into its simplest components so that only key constraints hold in a table (i.e., it is in BCNF). Attribute- and tuple-level correlations should guide the table decomposition into simpler tables. Ideally, the original table with its probability distribution can be recovered as a query (a view) from the decomposed tables. We have followed such principle in our claim-centered decomposition for predictive analytics.

In fact, a connection between database normalization theory and factor decomposition in Graphical Models (GM) has been discussed by Verma and Pearl \cite{verma1988}, but has not been explored since then. To date, there is no formal design theory for p-DB's \cite{suciu2011}. 
A step in that direction is taken by Sarma et al. \cite{sarma2007}. Their initiative revisits dependency theory in view of reformulating fd's for uncertain schema design \cite{sarma2007}. Our work takes a different direction. We refer to classical dependency theory and U-relational operations (viz, its uncertainty-introduction operator) to construct p-DB's systematically from scratch. We have focused on the extraction and processing of fd's towards a factorized U-relational schema. The synthesized schema is ensured to be in BCNF and have a lossless join.

Despite some major differences, our synthesis method builds upon the classical theory of relational schema design by synthesis \cite{bernstein1976}. Classical design by synthesis \cite{bernstein1976} was once criticized due to its too strong `uniqueness' of fd's assumption \cite[p. $\!$443]{fagin1977}, as it reduces the problem of design to symbolic reasoning on fd's, arguably neglecting semantic issues. Probabilistic design, however, has roots in statistical design so that the problem is less amenable to human factors. As we extract the dependencies from a formal specification, design by synthesis is doing nothing but translating seamlessly (into fd's) the reduction made by the user herself in her tentative model for the studied phenomenon.

The last decade has seen significant research effort to make DB systems really usable \cite{jagadish2007}. Our design-by-synthesis framework can also be understood as a technique for \emph{user-friendly p-DB design}. For instance, in comparison, the \textsf{CRIUS} system supports another kind of user-friendly DB design approach that provides users with a spreadsheet-like direct manipulation interface to increasingly add structure to their data \cite{qian2010}. Our dependency extraction and processing, instead, completely alleviates the user from the burden of data organization.

Also related to probabilistic DB design is the topic of conditioning a p-DB. It has been firstly addressed by Koch and Olteanu motivated by data cleaning applications \cite{koch2008}. They have introduced the \texttt{assert} operation to implement, as in AI, a kind of knowledge compilation, viz., world elimination in face of constraints (e.g., FDs). For hypothesis management, nonetheless, we need to apply \emph{Bayes' conditioning} by asserting observed data, not constraints. In \S\ref{sec:analytics} we have presented an example that settles the kind of conditioning problem that is relevant to the \mbox{$\Upsilon$-DB} vision. In Chapter \ref{ch:applicability} we present realistic use cases. We have addressed the problem at application level only in order to complete the realization of the vision in a real prototype system. The formulation of Bayes' conditioning as an extension of, say, the U-relational data model is open to future work (cf. \S\ref{sec:future-work}).

\section{Summary of Results }\label{sec:synthesis4u-conclusions}

In this chapter we have studied and developed our end-purpose technique for the synthesis of a probabilistic DB geared for predictive analytics.
It completes the pipeline (Fig. \ref{fig:pipeline}) so that conditioning can then be performed iteratively.

\begin{itemize}
\item Algorithm \ref{alg:synthesize} \textsf{synthesize} gives a general formulation of how to perform uncertainty introduction from causal dependencies given in the form of fd's.

\item By Problem \ref{prob:u-learning}, we have given a definition of uncertainty factor learning from data available in a given relation.

\item Remark \ref{rmk:u-intro} provides an example of the u-factor and predictive projections resulting from p-DB synthesis and their corresponding probability distributions stored in the world table; 

\item By Remark \ref{rmk:claim} and Theorem \ref{thm:bcnf}, we have shown that U-relational schema $\boldsymbol Y_k$ synthesized over the fd's processed by causal reasoning is in BCNF. That is, it is in fact a claim-centered decomposition as desirable for predictive analytics.

\item Theorem \ref{thm:lossless} ensures that such (uncertainty) decomposition is correct, as the original (probability distribution) `big' fact table is fully recoverable by a lossless join.
\end{itemize}


\chapter{Applicability}\label{ch:applicability}

In this chapter we show the applicability of $\Upsilon$-DB in real-world scenarios. We present use cases in Computational Physiology extracted from the Physiome project.\footnote{\url{http://physiome.org}.} 
In \S\ref{sec:physiome} we introduce the Physiome project as providing a testbed for \mbox{$\Upsilon$-DB}. 
Then in \S\ref{sec:case-studies} we go through some Physiome case studies to show the construction of \mbox{$\Upsilon$-DB} and its application for data-driven hypothesis management and analytics. 
In \S\ref{sec:demo} we present a prototype of the $\Upsilon$-DB system and demonstrate it through the running example introduced in \S\ref{sec:u-intro}. 
In \S\ref{sec:experiments-app} we present experiments on Physiome hypotheses. In \S\ref{sec:discussion-app} we provide a general discussion on the applicability of \mbox{$\Upsilon$-DB}, its assumptions and scope. In \S\ref{sec:conclusions-app} we conclude the chapter.

\section{The Physiome Project as a Testbed}\label{sec:physiome}

\noindent
The Physiome project is an initiative to seriously address the problems of reproducibility, model integration and sharing in Computational Physiology \cite{bassingthwaighte2000,hunter2003}. It essentially comprises:

\begin{itemize}
\item a curated \emph{repository} of 380+ computational physiology models available online for researchers;\footnote{The Physiome model repository is expanded to over $73K+$ models by including models extracted from other sources (such as the EBML-EBI BioModels DB, the CellML Archive, and the Kegg Pathways DB) and converted to MML automatically.}

\item the \emph{Mathematical Modeling Language} (MML) to allow models to be written in declarative form and then exported into a number of XML-compliant interoperable formats;\footnote{\url{http://www.physiome.org/jsim/docs/MML_Intro.html}.}

\item a \emph{problem-solving environment} called \textsf{JSim} to allow researchers to code their MML models straightforwardly, run them under different parameter and solver settings and build customized data plots to see the results.
\end{itemize}

From the point of view of \mbox{$\Upsilon$-DB}, Physiome is an external data source that provides a very interesting testbed with realistic scenarios. We extract Physiome models into \mbox{$\Upsilon$-DB} by means of a wrapper we have implemented to read XMML files (\textsf{JSim}'s XML encoding of MML models). Simulation trial datasets are rendered by a parametrized UNIX script we have developed to invoke \textsf{JSim} automatically (in batch mode). Currently, $\Upsilon$-DB's Physiome wrapper is designed to read MAT files to load both the model input (parameter settings data) and its associated model output (computed predictive data) for each simulation trial.

Physiome does not keep records of phenomena in a repository, but it does have observational data attached to some of the entries of the model repository. Such models appear in the filter `models with data,' meaning that they have one or more observational datasets and plots showing how the model data fits to observations. We shall make use of model entries containing observational data in the realistic scenarios presented in this paper.

\section{Case Studies}\label{sec:case-studies}
\noindent
In this section we present use cases extracted from the Physiome model repository.\footnote{\url{http://www.physiome.org/Models/modelDB/}.}

\subsection{Case: Hemoglobin Oxygen Saturation}\label{subsec:case-hemoglobin}
\noindent
In this case we stress the potential of data-driven hypothesis analytics in comparison to handcrafted curve fitting (visual) analysis. We study three different hypotheses that perform ``closely'' visually when compared to their target phenomenon dataset, see Fig. \ref{fig:hemoglobin}. All of them have been empirically set as fit as possible to the observations (`R1s1' dataset) in their local view (in separate), and are now compared together in a global view.

\begin{figure}[H]
\centering
\includegraphics[keepaspectratio,width=.855\textwidth]{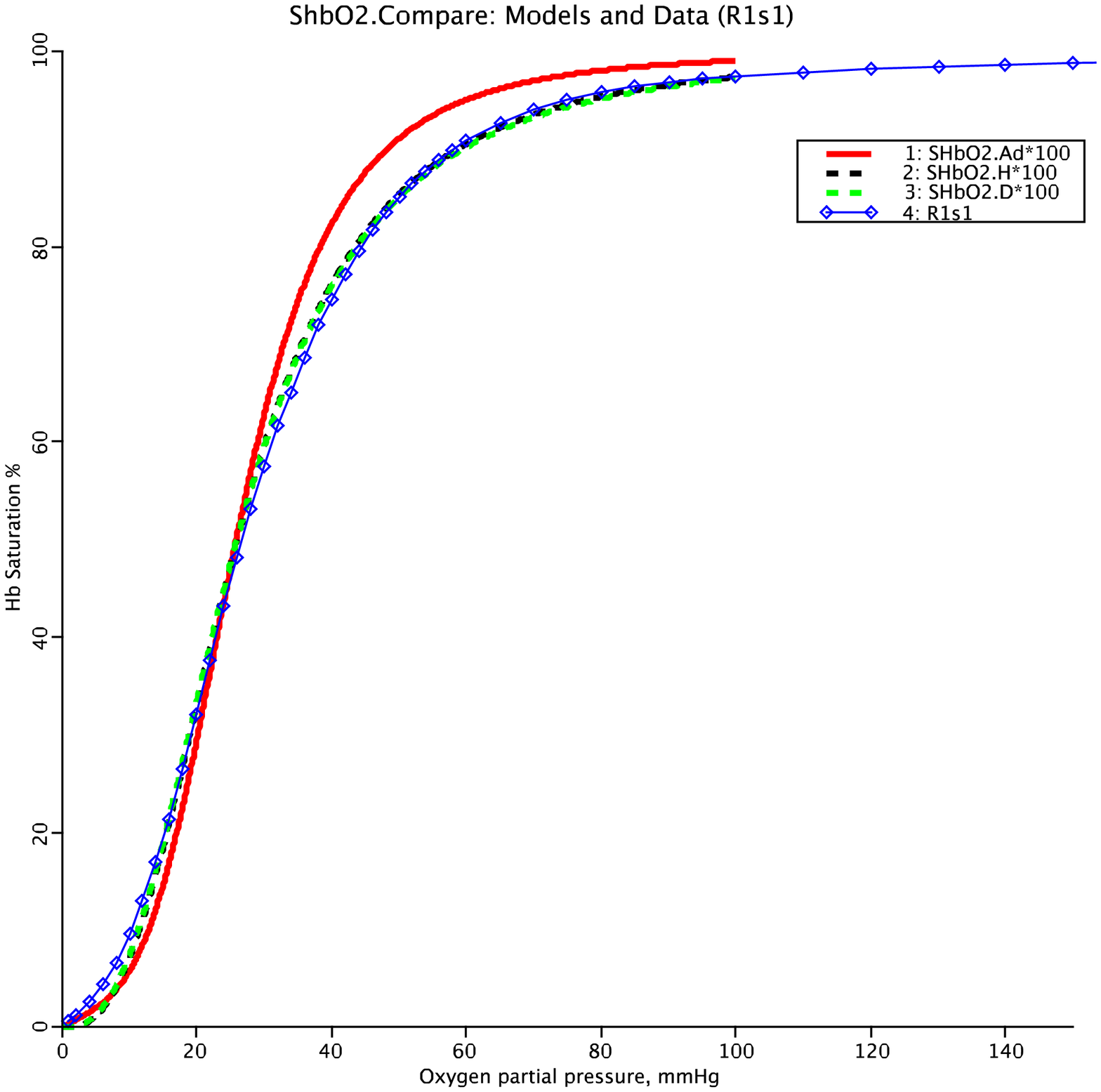}
\caption[Plot of hemoglobin oxygen saturation hypotheses]{Plot of hemoglobin oxygen saturation hypotheses (SHbO2.\{Ad, H, D\} curves) and their target observations (`R1s1' dataset). (source: Physiome).}
\label{fig:hemoglobin}
\end{figure}

\begin{myex}\label{ex:hemoglobin}
The resources of this example are shown in Fig. \ref{fig:case-hemoglobin}. We consider the Physiome model entries described in relation \begin{footnotesize}\textsf{HYPOTHESIS}\end{footnotesize}, associated to the phenomenon described in relation \begin{footnotesize}\textsf{PHENOMENON}\end{footnotesize} (cf. explanation relation $H_0$). One single hypothesis trial (its best fit) is considered for each hypothesis.
 $\Box$
\end{myex}

\begin{figure}[H]\footnotesize
\advance\leftskip-0.2cm
\begin{spacing}{0.9}
\begingroup\setlength{\fboxsep}{4pt}
\colorbox{yellow!15}{%
   \begin{tabular}{c|>{\columncolor[gray]{0.92}}c||l|p{0.58\linewidth}}
  \textsf{HYPOTHESIS} & $\upsilon$ & \textsf{Name} & \textsf{Description}\\
      \hline    
   & $28$ & \textsf{HbO.Hill} & Hill Equation for O2 binding to hemoglobin.\vspace{6pt}\\
   & $31$ & \textsf{HbO.Adair} & Hemoglobin O2 saturation curve using Adair's 4-site equation.\vspace{6pt}\\
   & $32$ & \textsf{HbO.Dash} & Hemoglobin O2 saturation curve at varied levels of PCO2 and pH.\vspace{6pt}\\
   \end{tabular}
}\endgroup
\vspace{5pt}
\begingroup\setlength{\fboxsep}{3pt}
\colorbox{blue!4}{%
   \begin{tabular}{c|>{\columncolor[gray]{0.92}}c||p{0.51\linewidth}}
  \textsf{PHENOMENON} & $\phi$ & \textsf{Description}\\
      \hline    
   & $1$ & Hemoglobin oxygen saturation with observational dataset from Sevenringhaus 1979.\vspace{6pt}\\ 
   \end{tabular}
}\endgroup
\hspace{0pt}
\begingroup\setlength{\fboxsep}{3pt}
\colorbox{blue!4}{%
   \begin{tabular}{c|>{\columncolor[gray]{0.92}}c|>{\columncolor[gray]{0.92}}c}
  $H_0$ & $\phi$ & $\upsilon$\\
      \hline    
   & $1$ & $28$\\ 
   & $1$ & $31$\\ 
   & $1$ & $32$\\ 
   \end{tabular}
}\endgroup
\end{spacing}
\caption[Descriptive (textual) data of Example \ref{ex:hemoglobin}.]{Descriptive (textual) data of Example \ref{ex:hemoglobin}, with ids $\upsilon$ from Physiome's model repository (\footnotesize{\url{http://www.physiome.org/Models/modelDB/}}).}
\label{fig:case-hemoglobin}
\end{figure}

\noindent
\textbf{Encoding}. The fd encoding of hypotheses $\upsilon \in \{28, 31, 32\}$ is shown (resp.) in Fig. \ref{fig:case-hemoglobin-encoding28}, Fig. \ref{fig:case-hemoglobin-encoding31} and Fig. \ref{fig:case-hemoglobin-encoding32}.

\begin{figure}[h!]
\begin{footnotesize}
\begin{framed}
\vspace{-10pt}
\begin{eqnarray*}
\Sigma_{28} = \{\quad
\textsf{KO2\;\;n\;\;pO2} \;\upsilon \quad\to\quad \textsf{SHbO2\_H},\\
\textsf{n\;\;p50} \;\upsilon \quad\to\quad \textsf{KO2},\\
\phi \quad\to\quad \textsf{n\, p50\, pO2\_delta\, pO2\_max\, pO2\_min} \quad\}.
\end{eqnarray*}
\vspace{-15pt}
\end{framed}
\end{footnotesize}
\vspace{-5pt}
\caption{Fd set $\Sigma_{28}$ of hypothesis $\upsilon\!=\!28$.}
\label{fig:case-hemoglobin-encoding28}
\end{figure}

\begin{figure}[h!]
\begin{footnotesize}
\begin{framed}
\vspace{-10pt}
\begin{eqnarray*}
\Sigma_{31} = \{\quad
\textsf{A1\;\;A2\;\;A3\;\;A4\;\;pO2} \;\upsilon \quad\to\quad \textsf{SHbO2\_Ad},\\
\textsf{p50} \;\upsilon \quad\to\quad \textsf{A1\;\;A2\;\;A4},\\
\phi \quad\to\quad \textsf{A3\, p50\, pO2\_delta\, pO2\_max\, pO2\_min} \quad\}.
\end{eqnarray*}
\vspace{-15pt}
\end{framed}
\end{footnotesize}
\vspace{-5pt}
\caption{Fd set $\Sigma_{31}$ of hypothesis $\upsilon\!=\!31$.}
\label{fig:case-hemoglobin-encoding31}
\end{figure}

\begin{figure}[h!]
\begin{footnotesize}
\begin{framed}
\vspace{-10pt}
\begin{eqnarray*}
\Sigma_{32} = \{\quad
\textsf{KO2\;\;pO2} \;\upsilon \quad\to\quad \textsf{SHbO2\_D},\\
\textsf{term1\;\;term2\;\;term3} \;\upsilon \quad\to\quad \textsf{KO2},\\
\textsf{HpRBC\, K6p\, alphaO2} \;\upsilon \quad\to\quad \textsf{term1},\\
\textsf{HpRBC\, K3\, K5\, alphaCO2\, pCO2} \;\upsilon \quad\to\quad \textsf{term2},\\
\textsf{HpRBC\, K2\, K4\, alphaCO2\, pCO2} \;\upsilon \quad\to\quad \textsf{term3},\\
\textsf{pH} \;\upsilon \quad\to\quad \textsf{HpRBC},\\
\phi \;\;\to\;\; \textsf{K2\;\;K3\;\;K4\;\;K5\;\;K6p\;\;alphaCO2\;\;alphaO2\;\;pCO2\;\;pH\;\;pO2\_delta, pO2\_max\, pO2\_min} \;\}.
\end{eqnarray*}
\vspace{-15pt}
\end{framed}
\end{footnotesize}
\vspace{-5pt}
\caption{Fd set $\Sigma_{32}$ of hypothesis $\upsilon\!=\!32$.}
\label{fig:case-hemoglobin-encoding32}
\end{figure}

\noindent
\textbf{Symbol Mappings}. As we have seen, the insertion of hypothesis trial datasets requires users to specify a target phenomenon and the corresponding mappings from the hypothesis symbols to the target phenomenon symbols. In this use case, we have:

\begin{itemize}
\item ${\mathcal M}_{28 \mapsto 1} = \{\; \textsf{pO2} \mapsto \textsf{pO2},\; \textsf{SHbO2\_H} \mapsto \textsf{SHbO2} \;\}$;\vspace{-3pt}
\item ${\mathcal M}_{31 \mapsto 1} = \{\; \textsf{pO2} \mapsto \textsf{pO2},\; \textsf{SHbO2\_Ad} \mapsto \textsf{SHbO2} \;\}$;\vspace{-3pt}
\item ${\mathcal M}_{32 \mapsto 1} = \{\; \textsf{pO2} \mapsto \textsf{pO2},\; \textsf{SHbO2\_D} \mapsto \textsf{SHbO2} \;\}$;
\end{itemize}

\noindent
\textbf{Hypothesis Management}. Query Q1 illustrates the feature of hypothesis management for this case. We consider the user is interested in all \textsf{SHbO2} predictions over a subset of the \textsf{pO2} domain. Its result set is shown in Fig. \ref{fig:case-hemoglobin-management}.

\vspace{12pt}
\noindent
$\;\;$\textsf{Q1. (\textbf{select} phi, upsilon, tid, ``pO2'', ``SHbO2\_H'' \textbf{as} SHbO2 \textbf{from} Y28\_claim1}\\\indent \textsf{\textbf{where} phi=1 \textbf{and} ``pO2''>=20 \textbf{and} ``pO2''<=40) \textbf{union all}}\\
\indent\textsf{(\textbf{select} phi, upsilon, tid, ``pO2'', ``SHbO2\_Ad'' \textbf{as} SHbO2 \textbf{from} Y31\_claim1}\\\indent \textsf{\textbf{where} phi=1 \textbf{and} ``pO2''>=20 \textbf{and} ``pO2''<=40) \textbf{union all}}\\
\indent\textsf{(\textbf{select} phi, upsilon, tid, ``pO2'', ``SHbO2\_D'' \textbf{as} SHbO2 \textbf{from} Y32\_claim1}\\\indent \textsf{\textbf{where} phi=1 \textbf{and} ``pO2''>=20 \textbf{and} ``pO2''<=40) \textbf{order by} ``pO2'', upsilon, tid;}
\vspace{9pt}

\begin{figure}[h]\footnotesize
\centering
\begingroup\setlength{\fboxsep}{3pt}
\colorbox{blue!5}{%
   \begin{tabular}{c|>{\columncolor[gray]{0.92}}c|>{\columncolor[gray]{0.92}}c|>{\columncolor[gray]{0.92}}c||c|c}
  \textsf{Q1} & $\phi$ & $\upsilon$ & \textsf{tid} & \textsf{pO2} & \textsf{SHbO2}\\
      \hline    
   & $1$ & $28$ & $1$ & $20$ & $0.329956122828398$\\
   & $1$ & $31$ & $1$ & $20$ & $0.294443723056007$\\
   & $1$ & $32$ & $1$ & $20$ & $0.334165672301096$\\
   \cline{2-6}
   & $\cdots$ &  $\cdots$ & $\cdots$ & $\cdots$ & $\cdots$\\
   \cline{2-6}
   & $1$ & $28$ & $1$ & $40$ & $0.761898061367189$\\
   & $1$ & $31$ & $1$ & $40$ & $0.823463829100424$\\
   & $1$ & $32$ & $1$ & $40$ & $0.759042287799556$\\
   \end{tabular}
}\endgroup
\caption{Result set of hypothesis management query Q1.}
\label{fig:case-hemoglobin-management}
\end{figure}

\vspace{-3pt}
\noindent
\textbf{Hypothesis Analytics}. Fig. \ref{fig:case-hemoglobin-analytics} shows the results of analytics after conditoning the probability distribution in the presence of observations (`R1s1' dataset). 
The fact that hypothesis $\upsilon=31$ provides the best explanation for the studied phenomenon is enabled by the application of Bayesian inference as implemented within the $\Upsilon$-DB system. The contribution of the $\Upsilon$-DB methodology is to equip users with a tool for large-scale, data-driven hypothesis management and analytics.

\vspace{5pt}
\begin{figure}[h!]\footnotesize
\centering
\begingroup\setlength{\fboxsep}{3pt}
\colorbox{yellow!15}{%
   \begin{tabular}{c|>{\columncolor[gray]{0.92}}c|>{\columncolor[gray]{0.92}}c|>{\columncolor[gray]{0.92}}c||c|c|c|c}
  \textsf{PH1\_CONF} & $\phi$ & $\upsilon$ & \textsf{tid} & \textsf{pO2} & \textsf{SHbO2} & \textsf{Prior} & \textsf{Posterior}\\
      \hline    
   & $1$ & $28$ & $1$ & $1$ & $0.000151184162020125$ & $.333$ & $.326$\\
   & $1$ & $31$ & $1$ & $1$ & $0.003789100566457180$ & $.333$ & \cellcolor{red!25} $.349$\\
   & $1$ & $32$ & $1$ & $1$ & $0.000178973375779681$ & $.333$ & $.325$\\
   \cline{2-8}   
   & $\cdots$ &  $\cdots$ & $\cdots$ & $\cdots$ & $\cdots$ & $\cdots$\\   
   \cline{2-8}   
   & $1$ &  $28$ & $1$ & $100$ & $0.974346796798538$ & $.333$ & $.326$\\
   & $1$ &  $31$ & $1$ & $100$ & $0.990781330988763$ & $.333$ & \cellcolor{red!25} $.349$\\
   & $1$ &  $32$ & $1$ & $100$ & $0.972764121981342$ & $.333$ & $.325$\\
   \end{tabular}
}\endgroup
\vspace{3pt}
\caption{Results of analytical study on the hemoglobin phenomenon.}
\label{fig:case-hemoglobin-analytics}
\vspace{-15pt}
\end{figure}


\subsection{Case: Baroreflex Dysfunction in Dahl SS Rat}\label{subsec:case-baroreflex}
\noindent
This case is extracted from the Virtual Physiological Rat project,\footnote{\url{http://virtualrat.org/computational-models/}.} Here we show the potential of data-driven hypothesis management and analytics for model tuning. Fig. \ref{fig:baroreflex} shows the best fit of a baroreflex model for an observational dataset acquired by experiment on Dahl SS rat \cite{beard2010}. We in turn use $\Upsilon$-DB to carry out such hypothesis management and analytics. We generate by a parameter sweep script $1K$ trials and insert them into the database. A best fit is then selected automatically by Bayesian inference.

\vspace{15pt}
\begin{figure}[H]
\centering
\includegraphics[keepaspectratio,width=.65\textwidth]{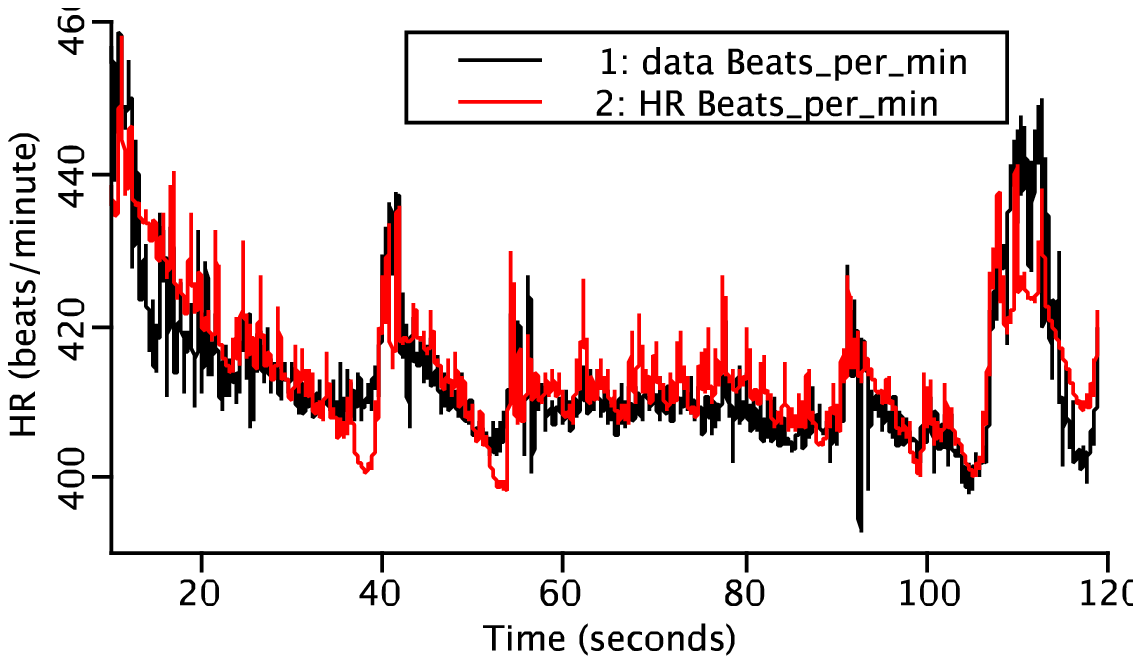}
\vspace{-10pt}
\caption[Plot of baroreflex hypothesis for Dahl SS Rat]{Plot of baroreflex hypothesis (`HR') for Dahl SS Rat and its target observations (`data'). (source: \cite{beard2010}).}
\label{fig:baroreflex}
\vspace{-4pt}
\end{figure}

\begin{myex}\label{ex:baroreflex}
The resources of this example are shown in Fig. \ref{fig:case-baroreflex}. We consider the single hypothesis entry described in relation \begin{footnotesize}\textsf{HYPOTHESIS}\end{footnotesize}, and the phenomenon described in relation \begin{footnotesize}\textsf{PHENOMENON}\end{footnotesize}. By parameter sweep, $1K$ trials are inserted into $\Upsilon$-DB for management and analytics.
 $\Box$
\end{myex}

\begin{figure}[H]\footnotesize
\advance\leftskip-0.2cm
\begin{spacing}{0.9}
\begingroup\setlength{\fboxsep}{4pt}
\colorbox{yellow!15}{%
   \begin{tabular}{c|>{\columncolor[gray]{0.92}}c||l|p{0.58\linewidth}}
  \textsf{HYPOTHESIS} & $\upsilon$ & \textsf{Name} & \textsf{Description}\\
      \hline
   & $1001$ & \textsf{Baroreflex\_SB\_CT} &  Physiological model of the full baroreflex heart control system based on experimental measurements.\vspace{6pt}\\
   \end{tabular}
}\endgroup
\vspace{5pt}
\begingroup\setlength{\fboxsep}{3pt}
\colorbox{blue!4}{%
   \begin{tabular}{c|>{\columncolor[gray]{0.92}}c||p{0.5\linewidth}}
  \textsf{PHENOMENON} & $\phi$ & \textsf{Description}\\
      \hline    
   & $2$ & Baroreflex dysfunction in Dahl SS Rat.\vspace{6pt}\\ 
   \end{tabular}
}\endgroup
\hspace{0pt}
\begingroup\setlength{\fboxsep}{3pt}
\colorbox{blue!4}{%
   \begin{tabular}{c|>{\columncolor[gray]{0.92}}c|>{\columncolor[gray]{0.92}}c}
  $H_0$ & $\phi$ & $\upsilon$\\
      \hline    
   & $2$ & $1001$\\ 
   \end{tabular}
}\endgroup
\end{spacing}
\caption[Descriptive (textual) data of Example \ref{ex:baroreflex}.]{Descriptive (textual) data of Example \ref{ex:baroreflex}.}
\label{fig:case-baroreflex}
\vspace{-20pt}
\end{figure}

\vspace{12pt}
\noindent
\textbf{Encoding}. The fd encoding of hypothesis $\upsilon \in \{1001\}$ is shown in Fig. \ref{fig:case-baroreflex-encoding}.

\noindent
\textbf{Symbol Mappings}. We consider that the user provides symbol mappings:
\begin{itemize}
\item ${\mathcal M}_{1001 \mapsto 2} = \{\; \textsf{Time} \mapsto \textsf{Time},\; \textsf{HR} \mapsto \textsf{HR} \;\}$;
\end{itemize}

\vspace{8pt}
\noindent
\textbf{Hypothesis Management}. In query Q2 we consider that the user is interested in time instants where the heart rate is higher than a threshold, say, 300 beats/min. The result set is shown in Fig. \ref{fig:case-baroreflex-management}.

\vspace{15pt}
\indent\indent$\;\;$\textsf{Q2. \textbf{select} phi, upsilon, tid, ``Time'', ``HR'' \textbf{from} Y1001\_claim1}\\ \indent\indent\indent \textsf{\textbf{where} phi=2 \textbf{and} ``HR''>=300 
\textbf{order by} ``Time'', tid;}
\vspace{9pt}

\begin{figure}[h]\footnotesize
\vspace{-3pt}
\centering
\begingroup\setlength{\fboxsep}{3pt}
\colorbox{blue!5}{%
   \begin{tabular}{c|>{\columncolor[gray]{0.92}}c|>{\columncolor[gray]{0.92}}c|>{\columncolor[gray]{0.92}}c||c|c}
  \textsf{Q2} & $\phi$ & $\upsilon$ & \textsf{tid} & \textsf{Time} & \textsf{HR}\\
      \hline    
   & $2$ & $1001$ & $1$ & $0.61$ & $300.013659905941$\\
   & $2$ & $1001$ & $2$ & $0.61$ & $300.011268345391$\\
   & $\cdots$ &  $\cdots$ & $\cdots$ & $\cdots$ & $\cdots$\\
   & $2$ & $1001$ & $96$ & $0.61$ & $300.001934440349$\\
   \cline{2-6}
   & $2$ & $1001$ & $1$ & $0.62$ & $300.607671377207$\\   
   & $\cdots$ &  $\cdots$ & $\cdots$ & $\cdots$ & $\cdots$\\
   \end{tabular}
}\endgroup
\caption{Result set of hypothesis management query Q2.}
\label{fig:case-baroreflex-management}
\vspace{-12pt}
\end{figure}

\vspace{8pt}
\noindent
\textbf{Hypothesis Analytics}. Fig. \ref{fig:case-baroreflex-analytics} shows the results of analytics on phenomenon $\phi\!=\!2$ after conditioning the probability distribution in the presence of observations (`SSBN9\_HR' dataset). Since this case deals with model tuning, viz., $1K$ slightly different parameter settings, the trial ranking is decided by small differences in the posterior probability distribution (cf. Fig \ref{fig:case-baroreflex-analytics}).

\vspace{8pt}
\begin{figure}[h]\footnotesize
\vspace{3pt}
\centering
\begingroup\setlength{\fboxsep}{3pt}
\colorbox{yellow!15}{%
   \begin{tabular}{c|>{\columncolor[gray]{0.92}}c|>{\columncolor[gray]{0.92}}c|>{\columncolor[gray]{0.92}}c||c|c|c|c}
  $\!\!$\textsf{PH2\_CONF} & $\phi$ & $\upsilon$ & \textsf{tid} & \textsf{Time} & \textsf{HR} & \textsf{Prior} & \textsf{Posterior}$\!\!$\\
      \hline    
   & $2$ & $1001$ & $491$ & $0.4$ & $\!286.556506432110\!$ & $\!.001000\!$ & \cellcolor{red!25} $\!.00103159\!$\\
   & $2$ & $1001$ & $591$ & $0.4$ & $\!286.525209765226\!$ & $\!.001000\!$ & $\!.00103144\!$\\
   & $2$ & $1001$ & $492$ & $0.4$ & $\!286.555565558231\!$ & $\!.001000\!$ & $\!.00103023\!$\\
   & $\cdots$ &  $\cdots$ & $\cdots$ & $\cdots$ & $\cdots$ & $\cdots$\\   
   \cline{2-8}
   & $2$ & $1001$ & $491$ & $\!118.9\!$ & $\!421.251853783050\!$ & $\!.001000\!$ &  \cellcolor{red!25} $\!.00103159\!$\\
   & $2$ & $1001$ & $591$ & $\!118.9\!$ & $\!421.110425308905\!$ & $\!.001000\!$ & $\!.00103144\!$\\\
   & $2$ & $1001$ & $492$ & $\!118.9\!$ & $\!421.297317710657\!$ & $\!.001000\!$ & $\!.00103023\!$\\
   & $\cdots$ &  $\cdots$ & $\cdots$ & $\cdots$ & $\cdots$ & $\cdots$\\   
   \end{tabular}
}\endgroup
\caption{Results of analytical study on the baroreflex phenomenon.}
\label{fig:case-baroreflex-analytics}
\vspace{-15pt}
\end{figure}

\begin{figure}[t!]\scriptsize
\begin{framed}
\vspace{-17pt}
\begin{eqnarray*}
\Sigma_{1001} = \{\quad
\textsf{HR}\;\upsilon \quad\to\quad \textsf{Period},\vspace{-3pt}\\
\textsf{Beta\;\;HR\_p\;\;HR\_s\;\;HRmin\;\;HRo}\;\upsilon \quad\to\quad \textsf{HR},\vspace{-3pt}\\
\textsf{HRo\;\;delta\_HR\_p}\;\upsilon \quad\to\quad \textsf{HR\_p},\\
\textsf{delta\_HR\_pfast\;\;delta\_HR\_pslow}\;\upsilon \quad\to\quad \textsf{delta\_HR\_p},\\
\textsf{HRo\;\;delta\_HR\_s}\;\upsilon \quad\to\quad \textsf{HR\_s},\\
\textsf{Time\;\;delta\_HR\_s\_\_Time\_min\;\;delta\_HR\_ss\;\;tau\_HR\_nor}\;\upsilon \quad\to\quad \textsf{delta\_HR\_s},\\
\textsf{Gamma\;\;delta\_HR\_ps}\;\upsilon \quad\to\quad \textsf{delta\_HR\_pfast},\\
\!\!\textsf{Gamma\;\;Time\;\;delta\_HR\_ps\;\;delta\_HR\_pslow\_\_Time\_min\;\;tau\_HR\_ach}\;\upsilon \;\;\to\;\; \textsf{delta\_HR\_pslow},\\
\textsf{K\_nor\;\;c\_nor\;\;delta\_HR\_smax}\;\upsilon \quad\to\quad \textsf{delta\_HR\_ss},\\
\textsf{C\_ach\;\;K\_ach\;\;delta\_HR\_pmax}\;\upsilon \quad\to\quad \textsf{delta\_HR\_ps},\\
\textsf{Time\;\;Ts\;\;c\_nor\_\_Time\_min\;\;q\_nor, tau\_nor}\;\upsilon \quad\to\quad \textsf{c\_nor},\\
\textsf{C\_ach\_\_Time\_min\;\;Time\;\;Tp\;\;q\_ach\;\;tau\_ach}\;\upsilon \quad\to\quad \textsf{C\_ach},\\
\textsf{Gs\;\;Tsmax\;\;Tsmin\;\;alpha\_cns\;\;alpha\_s0}\;\upsilon \quad\to\quad \textsf{Ts},\\
\textsf{Gp\;\;Tpmax\;\;Tpmin\;\;alpha\_cns\;\;alpha\_p0}\;\upsilon \quad\to\quad \textsf{Tp},\\
\textsf{Gcns\;\;n}\;\upsilon \quad\to\quad \textsf{alpha\_cns},\\
\textsf{S\;\;Zeta\;\;delta\;\;delta\_th}\;\upsilon \quad\to\quad \textsf{n},\\
\textsf{Eps\_1\;\;Eps\_wall}\;\upsilon \quad\to\quad \textsf{delta},\\
\textsf{B1\;\;Eps\_1\_\_Time\_min\;\;Eps\_2\;\;Eps\_2\_\_Time\_min\;\;Eps\_wall\;\;K1\;\;Kne\;\;Time}\;\upsilon \quad\to\quad \textsf{Eps\_1},\\
\textsf{B1\;\;B2\;\;Eps\_1\;\;Eps\_1\_\_Time\_min\;\;Eps\_2\_\_Time\_min\;\;Eps\_3\;\;Eps\_3\_\_Time\_min\;\;K1\;\;K2\;\;Time}\;\upsilon \;\to\; \textsf{Eps\_2},\\
\textsf{B2\;\;B3\;\;Eps\_2\;\;Eps\_2\_\_Time\_min\;\;Eps\_3\_\_Time\_min\;\;K2\;\;K3\;\;Time}\;\upsilon \quad\to\quad \textsf{Eps\_3},\\
\textsf{R\;\;R0}\;\upsilon \quad\to\quad \textsf{Eps\_wall},\\
\textsf{A}\;\upsilon \quad\to\quad \textsf{R},\\
\textsf{A\_\_Time\_min\;\;Bwall\;\;Cwall\;\;P\;\;R0\;\;Time}\;\upsilon \quad\to\quad \textsf{A},\\
\textsf{HRmax\;\;HRo}\;\upsilon \quad\to\quad \textsf{delta\_HR\_smax},\\
\textsf{HRmin\;\;HRo}\;\upsilon \quad\to\quad \textsf{delta\_HR\_pmax},\\
\textsf{Time}\;\phi \quad\to\quad \textsf{data},\\
\phi \;\;\to\;\; \textsf{A\_\_Time\_min\;\;B1\;\;B2\;\;B3\;\;Beta\;\;Bwall\;\;C\_ach\_\_Time\_min\;\;Cwall\;\;Eps\_1\_\_Time\_min}\\ \textsf{Eps\_2\_\_Time\_min\;\;Eps\_3\_\_Time\_min\;\;Gamma\;\;Gcns\;\;Gp\;\;Gs\;\;HRmax\;\;HRmin\;\;HRo\;\;K1\;\;K2\;\;K3}\\
\textsf{K\_ach\;\;K\_nor\;\;Kne\;\;P\;\;R0\;\;S\;\;Time\_delta\;\;Time\_max\;\;Time\_min\;\;Tpmax\;\;Tpmin\;\;Tsmax\;\;Tsmin}\\
\textsf{Zeta\;\;alpha\_p0\;\;alpha\_s0\;\;c\_nor\_\_Time\_min\;\;delta\_HR\_pslow\_\_Time\_min\;\;delta\_HR\_s\_\_Time\_min}\\
\textsf{delta\_th\;\;q\_ach\;\;q\_nor\;\;tau\_HR\_ach\;\;tau\_HR\_nor\;\;tau\_ach\;\;tau\_nor} \quad\}.
\end{eqnarray*}
\vspace{-20pt}
\end{framed}
\vspace{-10pt}
\caption{Fd set $\Sigma_{1001}$ of hypothesis $\upsilon\!=\!1001$.}
\label{fig:case-baroreflex-encoding}
\end{figure}

\subsection{Case: Myogenic Behavior of a Blood Vessel}\label{subsec:case-blood-vessel}
\noindent
Computational models of physiology may account for diverse effects that take place at different levels of biological organization from the organ to the cellular and molecular levels \cite{hunter2003}. Typically, a sophisticate model is developed incrementally by, say, adding detail into some previously existing model or extending its dimensionality (e.g., extending it from a stationary to a dynamic account of phenomena). In this case study (cf. Example \ref{ex:-blood-vessel}) we consider alternative models of the myogenic behavior of a reference human blood vessel.

\begin{figure}[H]
\advance\leftskip-0.5cm
\includegraphics[keepaspectratio,width=1.05\textwidth]{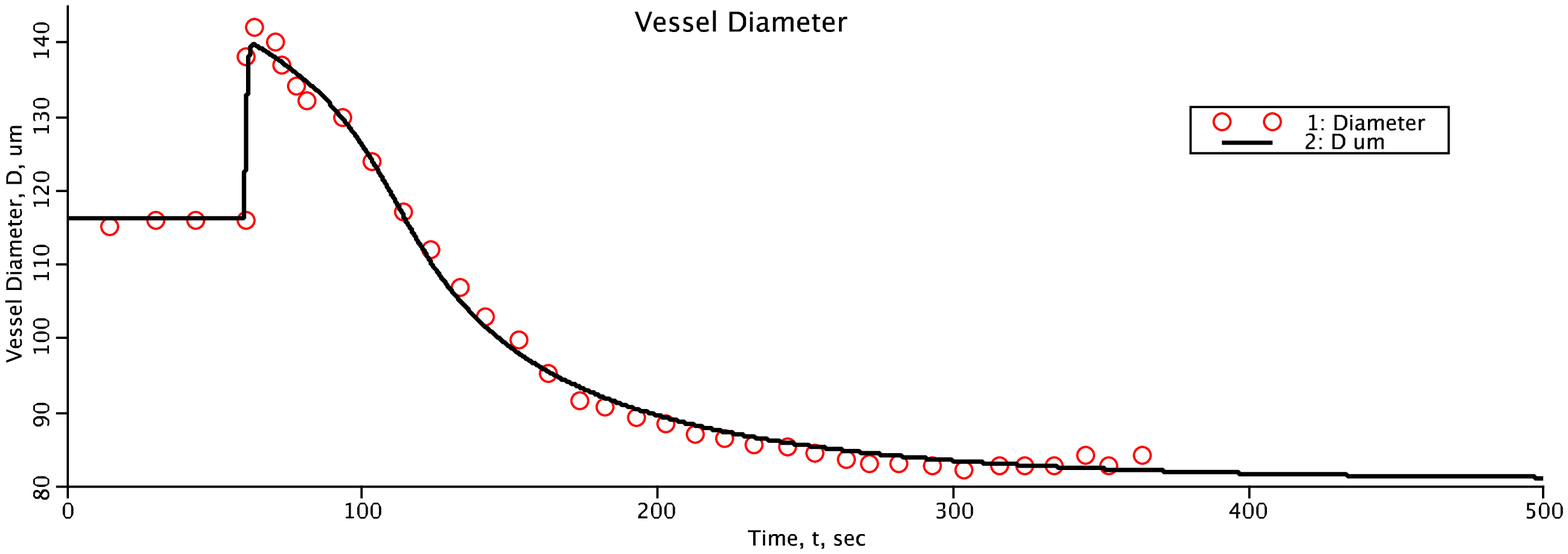}
\caption[Plot of myogenic behavior hypothesis]{Plot of myogenic behavior hypothesis (`D') according to $\upsilon\!=\!89$, trial \textsf{tid=2}, and its target observations (`Diameter').}
\label{fig:blood-vessel}
\end{figure}

\begin{myex}\label{ex:-blood-vessel}
(See Fig. \ref{fig:case-blood-vessel}). We consider the Physiome model entries displayed in relation \begin{footnotesize}\textsf{HYPOTHESIS}\end{footnotesize}, and two phenomena (see relation \begin{footnotesize}\textsf{PHENOMENON}\end{footnotesize}). One trial is considered for hypothesis $\upsilon=60$, and two for hypothesis $\upsilon=89$.
 $\Box$
\end{myex}

\vspace{15pt}
\begin{figure}[H]\footnotesize
\advance\leftskip-0.2cm
\begin{spacing}{0.9}
\begingroup\setlength{\fboxsep}{4pt}
\colorbox{yellow!15}{%
   \begin{tabular}{c|>{\columncolor[gray]{0.92}}c||p{0.20\linewidth}|p{0.525\linewidth}}
  \textsf{HYPOTHESIS} & $\!\!\upsilon\!\!$ & \textsf{Name} & \textsf{Description}\\
      \hline    
   & $\!\!60\!\!$ & \textsf{Myogenic\_Compliant} \textsf{\_Vessel} & This model simulates the flow through a passive and actively responding vessel driven by a sinusoidal pressure input.$\!\!$\vspace{6pt}\\
   & $\!\!89\!\!$ & \textsf{Myo\_Dyn\_Resp\_wFit} & This model describes the dynamic response of a vessel after a step increase in intraluminal pressure.$\!\!$\vspace{6pt}\\
   \end{tabular}
}\endgroup
\vspace{5pt}
\begingroup\setlength{\fboxsep}{3pt}
\colorbox{blue!4}{%
   \begin{tabular}{c|>{\columncolor[gray]{0.92}}c||p{0.5\linewidth}}
  \textsf{PHENOMENON} & $\phi$ & \textsf{Description}\\
      \hline    
   & $3$ & Dynamics of vessel diameter in response to pulsatile intraluminal pressure.\vspace{6pt}\\ 
   \end{tabular}
}\endgroup
\hspace{0pt}
\begingroup\setlength{\fboxsep}{3pt}
\colorbox{blue!4}{%
   \begin{tabular}{c|>{\columncolor[gray]{0.92}}c||>{\columncolor[gray]{0.92}}c}
  $H_0$ & $\phi$ & $\upsilon$\\
      \hline    
   & $3$ & $60$\\
   & $3$ & $89$\\
   \end{tabular}
}\endgroup
\caption[Descriptive (textual) data of Example \ref{ex:-blood-vessel}.]{Descriptive (textual) data of Example \ref{ex:-blood-vessel}.}
\label{fig:case-blood-vessel}
\end{spacing}
\end{figure}

\noindent
\textbf{Encoding}. The fd encoding of hypotheses $\upsilon \in \{60, 89\}$ is shown (resp.) in Fig. \ref{fig:case-blood-vessel-encoding60} and Fig. \ref{fig:case-blood-vessel-encoding89}.

\vspace{8pt}
\begin{figure}[h!]\footnotesize
\begin{framed}
\vspace{-12pt}
\begin{eqnarray*}
\Sigma_{60} = \{\quad
\textsf{Fcomp\;\;Fout} \;\upsilon \quad\to\quad \textsf{Fin},\\
\textsf{Pin\;\;Pout\;\;R} \;\upsilon \quad\to\quad \textsf{Fout},\\
\textsf{V\;\;V\_\_t\_min\;\;t} \;\upsilon \quad\to\quad \textsf{Fcomp},\\
\textsf{D\;\;L\;\;mu} \;\upsilon \quad\to\quad \textsf{R},\\
\textsf{D\;\;L} \;\upsilon \quad\to\quad \textsf{V},\\
\textsf{D\;\;Pin\;\;T} \;\upsilon \quad\to\quad \textsf{A},\\
\textsf{A\;\;C1a\;\;C1p\;\;C2a\;\;C2p\;\;C3a\;\;Dp100\;\;Pin\;\;T\;\;Ttarget} \;\upsilon \quad\to\quad \textsf{Atarget},\\
\textsf{Atarget\;\;Cglobal\;\;Cmyo\;\;Pin} \;\upsilon \quad\to\quad \textsf{D},\\
\textsf{D\;\;D\_\_t\_min\;\;Dc\;\;Tc\;\;Ttarget\;\;t\;\;taud} \;\upsilon \quad\to\quad \textsf{T},\\
\textsf{A\;\;A\_\_t\_min\;\;Atarget\;\;t\;\;taua} \;\upsilon \quad\to\quad \textsf{Ttarget},\\
\textsf{Ac} \;\upsilon \quad\to\quad \textsf{A\_\_t\_min},\\
\textsf{Cglobal\;\;Cmyo\;\;Tc} \;\upsilon \quad\to\quad \textsf{Ac},\\
\textsf{Dc} \;\upsilon \quad\to\quad \textsf{D\_\_t\_min},\\
\textsf{Dc\;\;Pmean} \;\upsilon \quad\to\quad \textsf{Tc},\\
\textsf{Pamp\;\;Pmean\;\;t\;\;tnorm} \;\upsilon \quad\to\quad \textsf{Pin},\\
\textsf{C1a\;\;C1p\;\;C2a\;\;C2p\;\;C3a\;\;Cglobal\;\;Cmyo\;\;Dp100\;\;Pc} \;\upsilon \quad\to\quad \textsf{Dc},\\
\phi \quad\to\quad \textsf{C1a\;\;C1p\;\;C2a\;\;C2p\;\;C3a\;\;Cglobal\;\;Cmyo\;\;Dp100\;\;L\;\;Pamp\;\;Pc\;\;Pext}\\
\textsf{Pmean\;\;Pout\;\;V\_\_t\_min\;\;mu\;\;t\_delta\;\;t\_max\;\;t\_min\;\;taua\;\;taud\;\;tnorm} \quad\}.
\end{eqnarray*}
\vspace{-12pt}
\end{framed}
\caption{Fd set $\Sigma_{60}$ of hypothesis $\upsilon\!=\!60$.}
\label{fig:case-blood-vessel-encoding60}
\end{figure}

\begin{figure}[h!]\footnotesize
\begin{framed}
\vspace{-12pt}
\begin{eqnarray*}
\Sigma_{89} = \{\quad
\textsf{D\, P\, T} \;\upsilon \quad\to\quad \textsf{A},\\
\textsf{A\, C1a\, C1p\, C2a\, C2p\, C3a\, Dp100\, P\, T\, Ttarget} \;\upsilon \quad\to\quad \textsf{Atarget},\\
\textsf{Atarget\, Cglobal\, Cmyo\, P} \;\upsilon \quad\to\quad \textsf{D},\\
\textsf{D\, D\_\_t\_min\, Dc\, Tc\, Ttarget\, t\, taud} \;\upsilon \quad\to\quad \textsf{T},\\
\textsf{A\, A\_\_t\_min\, Atarget\, t\, taua} \;\upsilon \quad\to\quad \textsf{Ttarget},\\
\textsf{Ac} \;\upsilon \quad\to\quad \textsf{A\_\_t\_min},\\
\textsf{Dc\, Pc} \;\upsilon \quad\to\quad \textsf{Tc},\\
\textsf{Cglobal\, Cmyo\, Dc\, Pc} \;\upsilon \quad\to\quad \textsf{Ac},\\
\textsf{Dc} \;\upsilon \quad\to\quad \textsf{D\_\_t\_min},\\
\textsf{DelP\, Pc} \;\upsilon \quad\to\quad \textsf{P},\\
\textsf{C1a\, C1p\, C2a\, C2p\, C3a\, Cglobal\, Cmyo\, Dp100\, Pc} \;\upsilon \quad\to\quad \textsf{Dc},\\
\textsf{t} \;\phi \quad\to\quad \textsf{DelP},\\
\phi \;\to\; \textsf{C1a\, C1p\, C2a\, C2p\, C3a\, Cglobal\, Cmyo\, Dp100\, Pc\, t\_delta\, t\_max\, t\_min\, taua\, taud} \;\}.
\end{eqnarray*}
\vspace{-12pt}
\end{framed}
\caption{Fd set $\Sigma_{89}$ of hypothesis $\upsilon\!=\!89$.}
\label{fig:case-blood-vessel-encoding89}
\end{figure}

\vspace{8pt}
\noindent
\textbf{Symbol Mappings}. We consider that the user provides symbol mappings:

\begin{itemize}
\item ${\mathcal M}_{60 \mapsto 1} = \{\; \textsf{t} \mapsto \textsf{Time},\; \textsf{D} \mapsto \textsf{Diameter} \;\}$;
\item ${\mathcal M}_{89 \mapsto 1} = \{\; \textsf{t} \mapsto \textsf{Time},\; \textsf{D} \mapsto \textsf{Diameter} \;\}$;
\end{itemize}

\noindent
\textbf{Hypothesis Management}. Query Q3 illustrates the feature of hypothesis management for this case. The user selects all diameter predictions within the time interval $t \in [100, 300]$ (cf. plot in Fig. \ref{fig:blood-vessel}). Its result set is shown in Fig. \ref{fig:case-blood-vessel-management}.

\indent$\;\;$\textsf{Q3. \textbf{select} phi, upsilon, tid, ``t'', ``D'' \textbf{from} Y60\_claim1}\\\indent\indent \textsf{\textbf{where} phi=3 \textbf{and} ``t''>=100 \textbf{and} ``t''<=300 \textbf{union all}}\\
\indent\indent\textsf{\textbf{select} phi, upsilon, tid, ``t'', ``D'' \textbf{from} Y89\_claim1}\\\indent\indent \textsf{\textbf{where} phi=3 \textbf{and} ``t''>=100 \textbf{and} ``t''<=300 \textbf{order by} ``t'', upsilon, tid;}
\vspace{9pt}

\begin{figure}[h!]\footnotesize
\centering
\begingroup\setlength{\fboxsep}{3pt}
\colorbox{blue!5}{%
   \begin{tabular}{c|>{\columncolor[gray]{0.92}}c|>{\columncolor[gray]{0.92}}c|>{\columncolor[gray]{0.92}}c||c|c}
  \textsf{Q3} & $\phi$ & $\upsilon$ & \textsf{tid} & \textsf{t} & \textsf{D}\\
      \hline    
   & $3$ & $89$ & $1$ & $100.00$ & $194.622865847211$\\
   & $3$ & $89$ & $1$ & $100.00$ & $97.3787340059609$\\
   & $3$ & $89$ & $2$ & $100.00$ & $126.167727083098$\\
   & $3$ & $89$ & $1$ & $100.01$ & $194.626017703936$\\
   & $3$ & $89$ & $1$ & $100.01$ & $98.0174705905828$\\
   & $3$ & $89$ & $2$ & $100.01$ & $126.161751822302$\\
   & $\cdots$ &  $\cdots$ & $\cdots$ & $\cdots$ & $\cdots$\\
   \end{tabular}
}\endgroup
\caption{Result set of hypothesis management query Q3.}
\label{fig:case-blood-vessel-management}
\end{figure}

\noindent
\textbf{Hypothesis Analytics}. Fig. \ref{fig:case-blood-vessel-analytics} shows the results of analytics on phenomenon $\phi=3$ after conditoning the probability distribution in the presence of observations, viz., `Davis\_Sikes\_Fig3\_Myo\_DigData' dataset.

\vspace{15pt}
\begin{figure}[h!]\footnotesize
\centering
\begingroup\setlength{\fboxsep}{3pt}
\colorbox{yellow!15}{%
   \begin{tabular}{c|>{\columncolor[gray]{0.92}}c|>{\columncolor[gray]{0.92}}c|>{\columncolor[gray]{0.92}}c||c|c|c|c}
  \textsf{PH3\_CONF} & $\phi$ & $\upsilon$ & \textsf{tid} & \textsf{Time} & \textsf{Diameter} & \textsf{Prior} & \textsf{Posterior}\\
      \hline    
   & $3$ & $60$ & $1$ & $14.8$ & $194.996792066637$ & $.5$ & $.000$\\
   & $3$ & $89$ & $1$ & $14.8$ & $97.0568250956827$ & $.25$ & $.269$\\
   & $3$ & $89$ & $2$ & $14.8$ & $116.327813203282$ & $.25$ & \cellcolor{red!25} $.731$\\
   \cline{2-8}   
   & $\cdots$ &  $\cdots$ & $\cdots$ & $\cdots$ & $\cdots$ & $\cdots$\\   
   \cline{2-8}   
   & $3$ & $60$ & $1$ & $30.5$ & $195.684170988267$ & $.5$ & $.000$\\
   & $3$ & $89$ & $1$ & $30.5$ & $97.0568250767575$ & $.25$ & $.269$\\
   & $3$ & $89$ & $2$ & $30.5$ & $116.327813337087$ & $.25$ & \cellcolor{red!25} $.731$\\
   \end{tabular}
}\endgroup
\vspace{3pt}
\caption{Results of analytics on the vessel's myogenic behavior phenomenon.}
\label{fig:case-blood-vessel-analytics}
\end{figure}

\noindent
In this case study two tentative models have been considered under a uniform prior probability distribution which has been updated to a posterior distribution. Note that, even though hypothesis $\upsilon\!=\!60$ has its probability weight concentrated in a single trial, the Bayesian inference is able to indicate $\upsilon=89$ as the best explanation for $\phi=3$ and $\textsf{tid}=2$, in particular, its best fit.

\section{System Prototype}\label{sec:demo}

\noindent
A first prototype of the $\Upsilon$-DB system has been implemented as a Java web application, with the pipeline component in the server side on top of \textsf{MayBMS} (a backend extension of \textsf{PostgreSQL}). 
We have developed a demonstration of this prototype (cf. \cite{goncalves2015b}), in which we go through the whole design-by-synthesis pipeline (Fig. $\!$\ref{fig:pipeline}) exploring use case scenarios. In this section we provide a brief demonstration of the system in the population dynamics scenario previously introduced in this thesis.

The demonstration unfolds in three phases. In the first phase, we show the ETL process to give a sense of what the user has to do in terms of simple phenomena description, hypothesis naming and file upload to get her phenomena and hypotheses available in the system to be managed as data. In the second phase, we reproduce some typical queries of hypothesis management (like those shown in the previous section). 
In the third phase, we enter the hypothesis analytics module. The user chooses a phenomenon for a hypothesis evaluation study, and the system lists all the predictions with their probabilities under some selectivity criteria (e.g., population at year 1920). The predictions are ranked according to their probabilities, which are conditioned on the observational data available for the chosen phenomenon.

\subsection{Demo Screenshots}\label{subsec:demo-screenshots}
\noindent

Fig. $\!$\ref{fig:demo} shows screenshots of the system. Fig. $\!$\ref{fig:demo}(a) shows the research projects currently available for a user. Figs. $\!$\ref{fig:demo}(b, c) show the ETL interfaces for phenomenon and hypothesis data definition (by synthesis), and then the insertion of hypothesis trial datasets, i.e., explanations of a hypothesis towards a target phenomenon. Fig. $\!$\ref{fig:demo}(d) shows the interface for basic hypothesis management by listing the predictions of a given simulation trial. Figs. $\!$\ref{fig:demo}(e, f) show two tabs of the hypothesis analytics module, viz., selection of observations and then viewing the corresponding alternative predictions ranked by their conditioned probabilities.


\subsection{Demo Case: Population Dynamics}\label{subsec:case-population-dynamics}
\noindent
In this case we refer to a well-known problem in Computational Science, viz., population dynamics scenarios, to demonstrate the $\Upsilon$-DB system prototype. Fig. \ref{fig:us-population} shows census data collected from in the US from 1790 to 1990.\footnote{Cf. \url{https://www.census.gov/population/censusdata/table-4.pdf}.} Fig. \ref{fig:hudson-population} shows observational data collected from Hudson's Bay from 1900 to 1920 on the Lynx-Hare population \cite{elton1942}.

\vspace{15pt}
\begin{figure}[H]
\centering
\includegraphics[keepaspectratio,width=.8\textwidth]{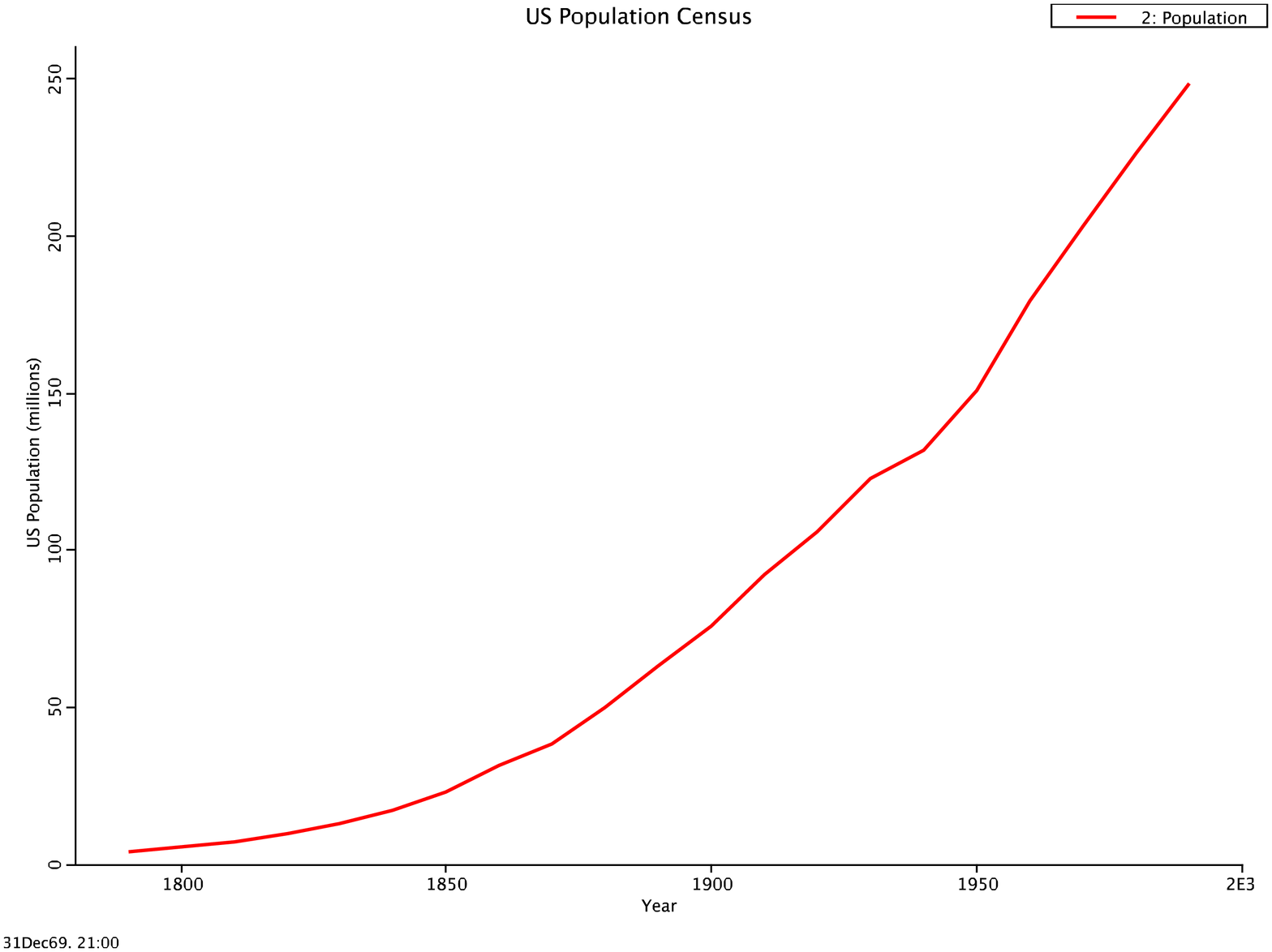}
\vspace{-5pt}
\caption[Lynx-Hare population observed]{Census US population from 1790 to 1990.}
\label{fig:us-population}
\vspace{-19pt}
\end{figure}

\vspace{30pt}
\begin{figure}[H]
\centering
\includegraphics[keepaspectratio,width=.8\textwidth]{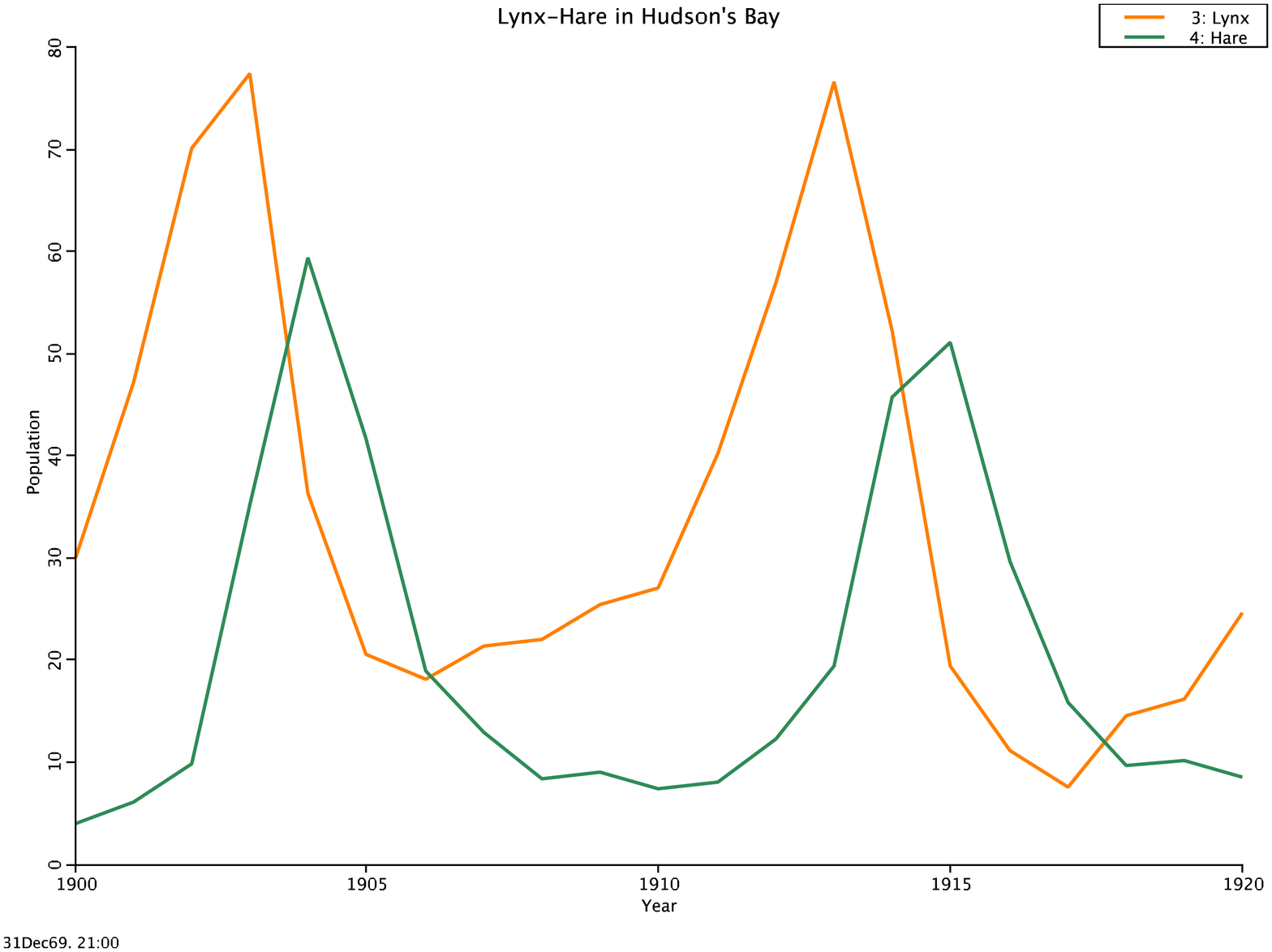}
\vspace{-5pt}
\caption[Lynx-Hare population observed]{Lynx-Hare population observed in Hudon's Bay from 1900 to 1920.}
\label{fig:hudson-population}
\vspace{-19pt}
\end{figure}

\begin{figure}[H]
\vspace{-12pt}
\centering
\begin{subfigure}[t]{0.49\textwidth}
\advance\leftskip1cm
\includegraphics[width=.69\textwidth]{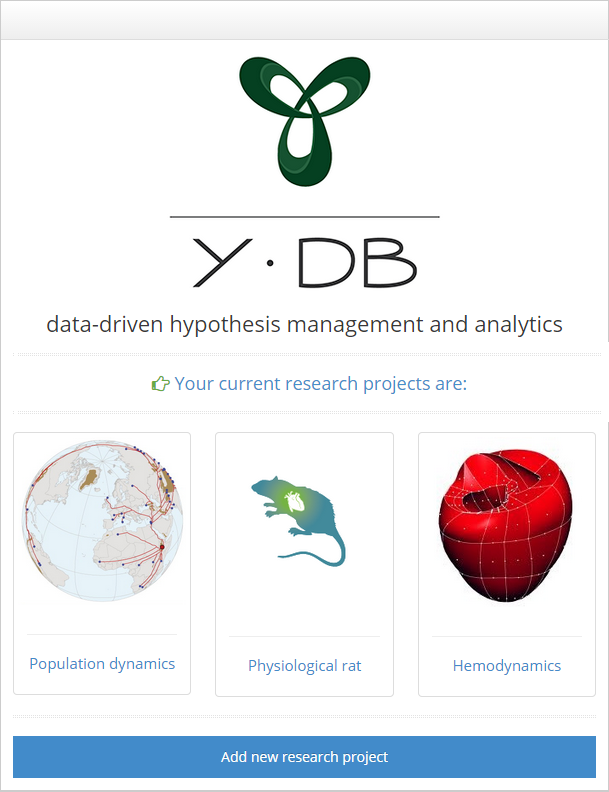}
\vspace{-1pt}
\caption{Research dashboard after login.}
\label{fig:logo}
\vspace{-6pt}
\end{subfigure}
\hspace{-10pt}
\begin{subfigure}[t]{0.49\textwidth}
\advance\leftskip0.75cm
\includegraphics[width=.78\textwidth]{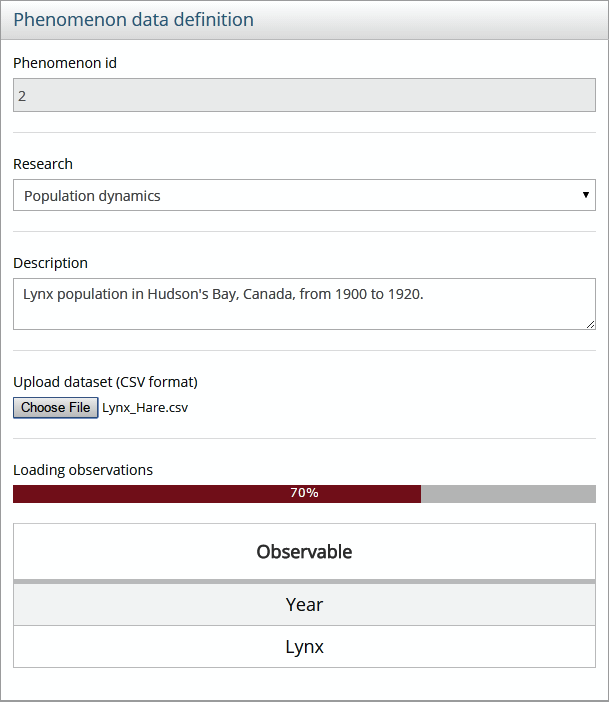}
\caption{Phenomenon data definition.}
\label{fig:etl-phenomenon}
\vspace{-5pt}
\end{subfigure}\vspace{4pt}\\
\begin{subfigure}[t]{0.49\textwidth}
\advance\leftskip0.8cm
\includegraphics[width=.78\textwidth]{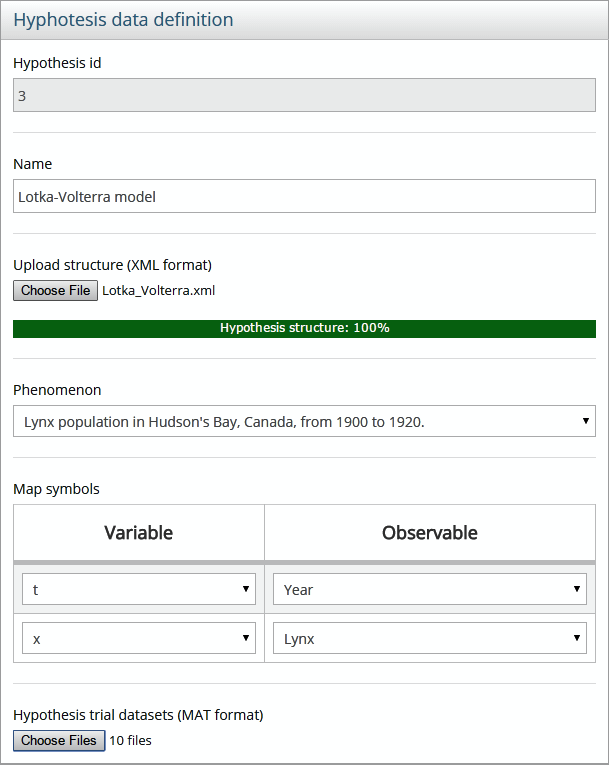}
\caption{Hypothesis data definition.}
\label{fig:etl-hypothesis}
\vspace{-4pt}
\end{subfigure}
\begin{subfigure}[t]{0.49\textwidth}
\advance\leftskip1cm
\includegraphics[width=.68\textwidth]{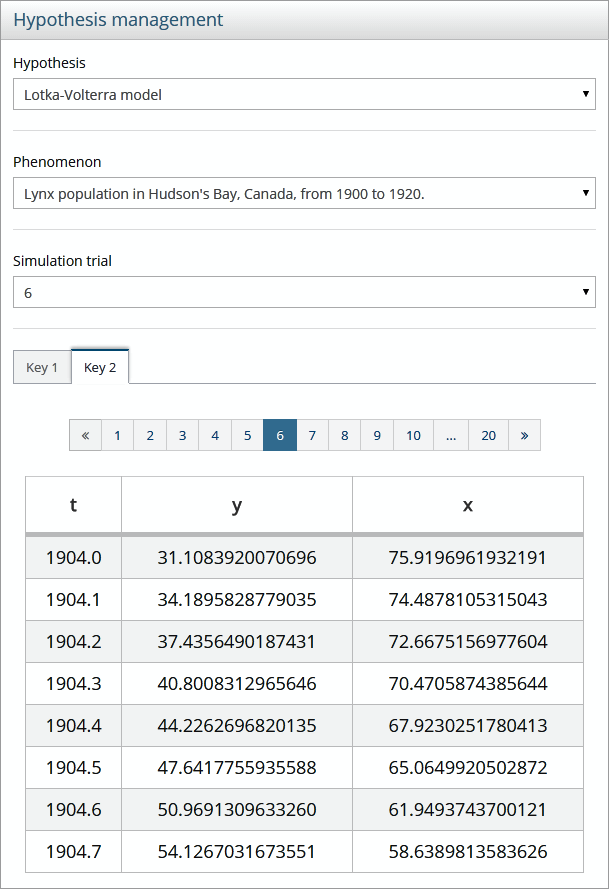}
\caption{Hypothesis management.}
\label{fig:management}
\vspace{-4pt}
\end{subfigure}\vspace{4pt}\\
\begin{subfigure}[t]{0.49\textwidth}
\advance\leftskip0.75cm
\includegraphics[width=.78\textwidth]{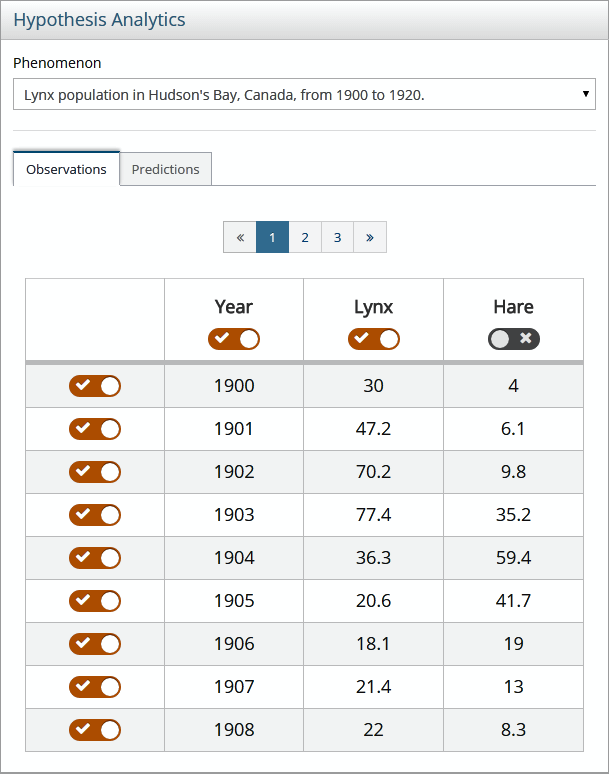}
\caption{Analytics: selected observations tab.}
\label{fig:analytics1}
\end{subfigure}
\hspace{-5pt}
\begin{subfigure}[t]{0.49\textwidth}
\advance\leftskip0.75cm
\includegraphics[width=.76\textwidth]{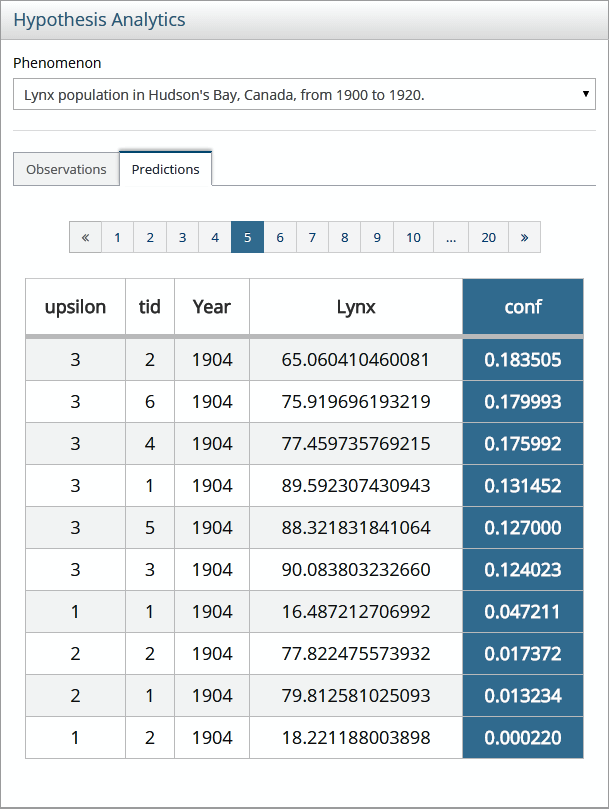}
\caption{Analytics: ranked predictions tab.}
\label{fig:analytics2}
\end{subfigure}
\vspace{-12pt}
\caption{Screenshots of this first prototype of the $\Upsilon$-DB system.}
\label{fig:demo}
\end{figure}

\begin{myex}\label{ex:-population}
(See Fig. \ref{fig:case-population}). We consider the model entries displayed in relation \begin{footnotesize}\textsf{HYPOTHESIS}\end{footnotesize}, and two phenomena (see relation \begin{footnotesize}\textsf{PHENOMENON}\end{footnotesize}). For $\phi=1$, three trials are considered
for hypothesis $\upsilon=1$ and six for hypothesis $\upsilon=2$. For $\phi=2$, in turn, two trials are considered for hypothesis $\upsilon=1$ and $\upsilon=2$, and six trials for hypothesis $\upsilon=3$. Note the data definition interfaces in Figs. $\!$\ref{fig:demo}(b, c).
 $\Box$
\end{myex}

\begin{figure}[H]\footnotesize
\advance\leftskip-0.2cm
\begin{spacing}{0.9}
\begingroup\setlength{\fboxsep}{4pt}
\colorbox{yellow!15}{%
   \begin{tabular}{c|>{\columncolor[gray]{0.92}}c||p{0.20\linewidth}|p{0.525\linewidth}}
  \textsf{HYPOTHESIS} & $\!\!\upsilon\!\!$ & \textsf{Name} & \textsf{Description}\\
      \hline    
   & $\!\!1\!\!$ & \textsf{Malthusian growth model} & Exponential growth model `growth in population is proportional to its size' is considered the first principle of population dynamics.$\!\!$\vspace{6pt}\\
   & $\!\!2\!\!$ & \textsf{Logistic equation} & This model introduces growth saturation to the Malthusian model due to the limitation of resources.$\!\!$\vspace{6pt}\\
   & $\!\!3\!\!$ & \textsf{Lotka-Volterra model} & This model describes predator-prey interactions.$\!\!$\vspace{6pt}\\
   \end{tabular}
}\endgroup
\vspace{5pt}
\begingroup\setlength{\fboxsep}{3pt}
\colorbox{blue!4}{%
   \begin{tabular}{c|>{\columncolor[gray]{0.92}}c||p{0.5\linewidth}}
  \textsf{PHENOMENON} & $\phi$ & \textsf{Description}\\
      \hline    
   & $1$ & US population from 1790 to 1990.\\
   & $2$ & Lynx population in Hudson's Bay, Canada, from 1900 to 1920.\\
   \end{tabular}
}\endgroup
\hspace{0pt}
\begingroup\setlength{\fboxsep}{3pt}
\colorbox{blue!4}{%
   \begin{tabular}{c|>{\columncolor[gray]{0.92}}c||>{\columncolor[gray]{0.92}}c}
  $H_0$ & $\phi$ & $\upsilon$\\
      \hline    
   & $1$ & $1$\\
   & $1$ & $2$\\
   & $2$ & $1$\\
   & $2$ & $2$\\
   & $2$ & $3$\\
   \end{tabular}
}\endgroup
\caption[Descriptive (textual) data of Example \ref{ex:-blood-vessel}.]{Descriptive (textual) data of Example \ref{ex:-population}.}
\label{fig:case-population}
\end{spacing}
\end{figure}

\vspace{5pt}
\noindent
\textbf{Encoding}. The fd encoding of hypotheses $\upsilon \in \{1,\, 2,\, 3\}$ is shown (resp.) in Fig. \ref{fig:case-population-encoding1}, Fig. \ref{fig:case-population-encoding2} and Fig. \ref{fig:case-population-encoding3}. See hypothesis structure processing in Fig. $\!$\ref{fig:demo}(c).
\vspace{15pt}

\begin{figure}[h!]\footnotesize
\begin{framed}
\vspace{-12pt}
\begin{eqnarray*}
\Sigma_{1} = \{\quad
\textsf{b\;\;t\;\;x\_\_t\_min} \;\upsilon \quad\to\quad \textsf{x},\\
\phi \quad\to\quad \textsf{b\;\;t\_delta\;\;t\_max\;\;t\_min\;\;x\_\_t\_min} \quad\}.
\end{eqnarray*}
\vspace{-20pt}
\end{framed}
\vspace{-10pt}
\caption{Fd set $\Sigma_{1}$ of hypothesis $\upsilon\!=\!1$.}
\label{fig:case-population-encoding1}
\vspace{-10pt}
\end{figure}
\begin{figure}[h!]\footnotesize
\begin{framed}
\vspace{-12pt}
\begin{eqnarray*}
\Sigma_{2} = \{\quad
\textsf{K\;\;b\;\;t\;\;x\_\_t\_min} \;\upsilon \quad\to\quad \textsf{x},\\
\phi \quad\to\quad \textsf{K\;\;b\;\;t\_delta\;\;t\_max\;\;t\_min\;\;x\_\_t\_min} \quad\}.
\end{eqnarray*}
\vspace{-20pt}
\end{framed}
\vspace{-10pt}
\caption{Fd set $\Sigma_{2}$ of hypothesis $\upsilon\!=\!2$.}
\label{fig:case-population-encoding2}
\end{figure}
\begin{figure}[h!]\footnotesize
\begin{framed}
\vspace{-12pt}
\begin{eqnarray*}
\Sigma_{3} = \{\quad
\textsf{b\;p\;t\;\;y} \;\upsilon \quad\to\quad \textsf{x},\\
\textsf{d\;r\;t\;\;x} \;\upsilon \quad\to\quad \textsf{y},\\
\phi \quad\to\quad \textsf{b\;\;d\;\;p\;\;r\;\;t\_delta\;\;t\_max\;\;t\_min\;\;x\_\_t\_min\;\;y\_\_t\_min} \quad\}.
\end{eqnarray*}
\vspace{-20pt}
\end{framed}
\vspace{-10pt}
\caption{Fd set $\Sigma_{3}$ of hypothesis $\upsilon\!=\!3$.}
\label{fig:case-population-encoding3}
\end{figure}

\noindent
\textbf{Symbol Mappings}. We consider that the user provides the following symbol mappings for (resp.) phenomena $\phi=1$ and $\phi=2$, see the interface for mapping symbols in Fig. $\!$\ref{fig:demo}(c).

\begin{itemize}
\item ${\mathcal M}_{1 \mapsto 1} = \{\; \textsf{t} \mapsto \textsf{Year},\; \textsf{x} \mapsto \textsf{Population} \;\}$;
\item ${\mathcal M}_{2 \mapsto 1} = \{\; \textsf{t} \mapsto \textsf{Year},\; \textsf{x} \mapsto \textsf{Population} \;\}$;
\item ${\mathcal M}_{1 \mapsto 2} = \{\; \textsf{t} \mapsto \textsf{Year},\; \textsf{x} \mapsto \textsf{Lynx} \;\}$;
\item ${\mathcal M}_{2 \mapsto 2} = \{\; \textsf{t} \mapsto \textsf{Year},\; \textsf{x} \mapsto \textsf{Lynx} \;\}$;
\item ${\mathcal M}_{3 \mapsto 2} = \{\; \textsf{t} \mapsto \textsf{Year},\; \textsf{x} \mapsto \textsf{Lynx} \;\}$;
\end{itemize}

\vspace{12pt}
\noindent
\textbf{Hypothesis Management}. Query Q4 illustrates the feature of hypothesis management for this case. The user selects hypothesis $\upsilon\!=\!3$ (the Lotka-Volterra model), and filters its available data for trial $\textsf{tid}\!=\!6$ on phenomenon $\phi=2$. Both the form-based query set-up and its result set are shown in Fig. $\!$\ref{fig:demo}(d).

\vspace{12pt}
\indent$\;\;$\textsf{Q4. \textbf{select} ``t'', ``y'', ``x'' \textbf{from} Y3\_claim1}\\\indent\indent \textsf{\textbf{where} upsilon=3 \textbf{and} phi=2 \textbf{and} tid=6 \textbf{order by} ``t'';}
\vspace{9pt}

\vspace{12pt}
\noindent
\textbf{Hypothesis Analytics}. Fig. \ref{fig:case-us-population-analytics} and Fig. \ref{fig:case-lynx-population-analytics} show the results of analytics on (resp.) phenomena $\phi=1$ and $\phi=2$ after conditoning the probability distribution in the presence of (resp.) observational datasets `US-census' and `Lynx-Hare.' In the first one, the user verifies that hypothesis $\upsilon\!=\!1$ (the Malthusian model) is unlikely  to be competitive with hypothesis $\upsilon\!=\!2$ (the Logistic equation) as an approximation of the US population dynamics from 1790 to 1990. That is, if the user knows her current trials are reasonable, then more trials on the Malthusian model hardly could outperform trials on the Logistic equation for the studied phenomenon.

\begin{figure}[t!]\footnotesize
\centering
\begingroup\setlength{\fboxsep}{3pt}
\colorbox{yellow!15}{%
   \begin{tabular}{c|>{\columncolor[gray]{0.92}}c|>{\columncolor[gray]{0.92}}c|>{\columncolor[gray]{0.92}}c||c|c|c|c}
  \textsf{PH1\_CONF} & $\phi$ & $\upsilon$ & \textsf{tid} & \textsf{Year} & \textsf{Population} & \textsf{Prior} & \textsf{Posterior}\\
      \hline    
   & $\cdots$ &  $\cdots$ & $\cdots$ & $\cdots$ & $\cdots$ & $\cdots$\\   
   \cline{2-8}   
   & $1$ & $1$ & $1$ & $1920$ & $194.102222534948$ & $.250$ & $.000000$\\
   & $1$ & $1$ & $2$ & $1920$ & $140.244165184248$ & $.250$ & $.000000$\\
   & $1$ & $2$ & $1$ & $1920$ & $82.3031951115155$ & $.125$ & $.133038$\\
   & $1$ & $2$ & $2$ & $1920$ & $108.251924734215$ & $.125$ & $.239684$\\
   & $1$ & $2$ & $3$ & $1920$ & $105.918217777077$ & $.125$ & $.290026$\\
   & $1$ & $2$ & $4$ & $1920$ & $105.988231944275$ & $.125$ & \cellcolor{red!25} $.337251$\\
   \cline{2-8}   
   & $\cdots$ &  $\cdots$ & $\cdots$ & $\cdots$ & $\cdots$ & $\cdots$\\   
   \end{tabular}
}\endgroup
\caption{Results of analytics on the US population phenomenon.}
\label{fig:case-us-population-analytics}
\end{figure}
%

\vspace{20pt}
\begin{figure}[t!]\footnotesize
\centering
\begingroup\setlength{\fboxsep}{3pt}
\colorbox{yellow!15}{%
   \begin{tabular}{c|>{\columncolor[gray]{0.92}}c|>{\columncolor[gray]{0.92}}c|>{\columncolor[gray]{0.92}}c||c|c|c|c}
  \textsf{PH2\_CONF} & $\phi$ & $\upsilon$ & \textsf{tid} & \textsf{Year} & \textsf{Lynx} & \textsf{Prior} & \textsf{Posterior}\\
      \hline    
   & $\cdots$ &  $\cdots$ & $\cdots$ & $\cdots$ & $\cdots$ & $\cdots$\\   
   \cline{2-8}   
   & $2$ & $1$ & $1$ & $1904$ & $16.49$ & $.167$ & $.047$\\
   & $2$ & $1$ & $2$ & $1904$ & $18.22$ & $.167$ & $.000$\\
   & $2$ & $2$ & $1$ & $1904$ & $79.81$ & $.167$ & $.013$\\
   & $2$ & $2$ & $2$ & $1904$ & $77.82$ & $.167$ & $.017$\\
   & $2$ & $3$ & $1$ & $1904$ & $89.59$ & $.055$ & $.131$\\
   & $2$ & $3$ & $2$ & $1904$ & $65.06$ & $.055$ & \cellcolor{red!25} $.184$\\
   & $2$ & $3$ & $3$ & $1904$ & $90.08$ & $.055$ & $.124$\\
   & $2$ & $3$ & $4$ & $1904$ & $77.46$ & $.055$ & $.176$\\
   & $2$ & $3$ & $5$ & $1904$ & $88.32$ & $.055$ & $.127$\\
   & $2$ & $3$ & $6$ & $1904$ & $75.92$ & $.055$ & $.180$\\
   \cline{2-8}
   & $\cdots$ &  $\cdots$ & $\cdots$ & $\cdots$ & $\cdots$ & $\cdots$\\   
   \end{tabular}
}\endgroup
\caption[Results of analytics on the Hudson's Bay lynx population phenomenon]{Results of analytics on the Hudson's Bay lynx population pheno- menon, see interfaces in Figs. $\!$\ref{fig:demo}(e, f).}
\label{fig:case-lynx-population-analytics}
\end{figure}

\section{Experiments}\label{sec:experiments-app}

The efficiency and scalability of the U-relational representation system and its p-WSA query algebra have been extensively demonstrated \cite{antova2008}. \mbox{$\Upsilon$-DB}'s, as U-relational hypothesis DB's, must therefore be as efficient and scalable as any arbitrary U-relational DB. 

In these experiments (see Fig. \ref{fig:experiments}) we provide some measures of performance of the method of {$\Upsilon$-DB} in the particular context of our real-world Physiome testbed. Our purpose here is to provide a concrete feel on how efficient the {$\Upsilon$-DB} methodology can be. However, most of these tests (the four graphs on the bottom in Fig. \ref{fig:experiments}) involve the data level and then require more of the hardware. Our current experimental setup (personal computer)\footnote{These experiments were performed on a 2.3GHZ/4GB Intel Core i5 running Mac OS X 10.6.8 and MayBMS (a PostgreSQL 8.3.3 extension).} allows us to reach a scale in which the uncertain data being processed in synthesis `4U' is sized up to $1\,GB$. 
For the two first graphs (\emph{XML extraction} and \emph{encoding}), we have collected the response time on the measure of interest over different structure lengths. Each one corresponds to a real Physiome hypothesis from the table of Fig. \ref{fig:physiome-hypotheses}. The last hypothesis in that table, $\upsilon=379$, is used for the tests of the four last graphs in Fig. \ref{fig:experiments}, viz., \emph{u-learning}, \emph{u-factorization}, \emph{u-propagation} and \emph{conditioning}. We have set different number of trials (\textsf{ntrials}) over it, each one having $1\,MB$. The last test in each of such four graphs, with $1K$ trials, is processing $1\,GB$ of uncertain data at once and then fits the machine's main memory. 
We interpret the performance results shown in these graphs as follows for each measure of performance.

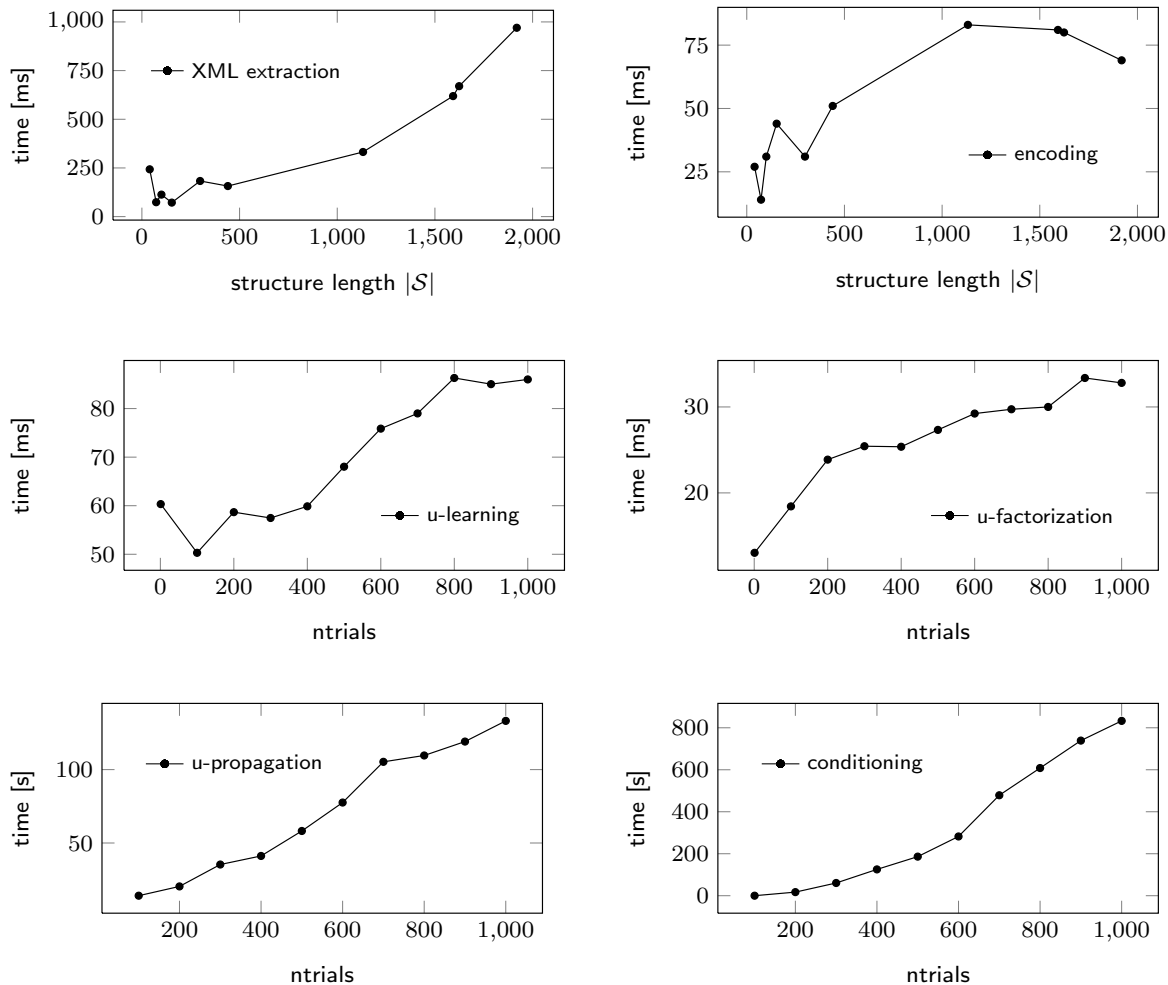
\begin{figure}[t!]\scriptsize
\advance\leftskip-0.4cm
\begin{subfigure}{0.465\textwidth}
\begin{tikzpicture}[y=3cm, x=.013cm,font=\sffamily,scale=1.1]
    \begin{axis}[
        height=0.6\textwidth,
        width=1\textwidth,
        xlabel={structure length $|\mathcal S|$},
        y label style={at={(axis description cs:0.075,.5)},anchor=south},
        ylabel={time [ms]},
        xtick={0,500,1000,1500,2000},
        ytick={0,250,500,750,1000},
]
      \addplot[mark=*, mark size=1.2pt, mark options={fill=black}] 
		file {./data/extraction.data}; 
     \end{axis}
	\begin{scope}[shift={(35,0.6)}] 
	\draw (0,0) -- 
		plot[mark=*, mark size=1.2pt, mark options={fill=black}] (15,0) -- (30,0)
		node[right]{\textsf{XML extraction}};
	\end{scope}
\end{tikzpicture}
\end{subfigure}
\hspace{30pt}
\begin{subfigure}{0.465\textwidth}
\begin{tikzpicture}[y=3cm, x=.015cm,font=\sffamily,scale=1.1]
    \begin{axis}[
        height=0.6\textwidth,
        width=1\textwidth,
        xlabel={structure length $|\mathcal S|$},
        y label style={at={(axis description cs:0.1,.5)},anchor=south},
        ylabel={time [ms]},
        xtick={0,500,1000,1500,2000},
        ytick={0,25,50,75,100},
]
      \addplot[mark=*, mark size=1.2pt, mark options={fill=black}] 
		file {./data/encoding.data}; 
     \end{axis}
	\begin{scope}[shift={(200,0.25)}] 
	\draw (0,0) -- 
		plot[mark=*, mark size=1.2pt, mark options={fill=black}] (15,0) -- (30,0)
		node[right]{\textsf{encoding}};
	\end{scope}
\end{tikzpicture}
\end{subfigure}\vspace{20pt}\\
\begin{subfigure}{0.465\textwidth}
\begin{tikzpicture}[y=3cm, x=.015cm,font=\sffamily,scale=1.1]
    \begin{axis}[
        height=0.6\textwidth,
        width=1\textwidth,
        xlabel={ntrials},
        y label style={at={(axis description cs:0.05,.5)},anchor=south},
        ylabel={time [ms]},
]
      \addplot[mark=*, mark size=1.2pt, mark options={fill=black}] 
		file {./data/ulearning.data}; 
     \end{axis}
	\begin{scope}[shift={(205,0.22)}] 
	\draw (0,0) -- 
		plot[mark=*, mark size=1.2pt, mark options={fill=black}] (15,0) -- (30,0)
		node[right]{\textsf{u-learning}};
	\end{scope}
\end{tikzpicture}
\end{subfigure}
\hspace{30pt}
\begin{subfigure}{0.465\textwidth}
\begin{tikzpicture}[y=3cm, x=.015cm,font=\sffamily,scale=1.1]
    \begin{axis}[
        height=0.6\textwidth,
        width=1\textwidth,
        xlabel={ntrials},
        y label style={at={(axis description cs:0.1,.5)},anchor=south},
        ylabel={time [ms]},
]
      \addplot[mark=*, mark size=1.2pt, mark options={fill=black}] 
		file {./data/ufactorization.data}; 
     \end{axis}
	\begin{scope}[shift={(170,0.22)}] 
	\draw (0,0) -- 
		plot[mark=*, mark size=1.2pt, mark options={fill=black}] (15,0) -- (30,0)
		node[right]{\textsf{u-factorization}};
	\end{scope}
\end{tikzpicture}
\end{subfigure}
\vspace{20pt}\\
\begin{subfigure}{0.465\textwidth}
\begin{tikzpicture}[y=3cm, x=.015cm,font=\sffamily,scale=1.1]
    \begin{axis}[
        height=0.6\textwidth,
        width=1\textwidth,
        xlabel={ntrials},
        y label style={at={(axis description cs:0.1,.5)},anchor=south},
        ylabel={time [s]},
        xtick={0, 200, 400, 600, 800, 1000},
        ytick={0, 50, 100},
]
      \addplot[mark=*, mark size=1.2pt, mark options={fill=black}] 
		file {./data/upropagation.data}; 
     \end{axis}
	\begin{scope}[shift={(35,0.6)}] 
	\draw (0,0) -- 
		plot[mark=*, mark size=1.2pt, mark options={fill=black}] (15,0) -- (30,0)
		node[right]{\textsf{u-propagation}};
	\end{scope}
\end{tikzpicture}
\end{subfigure}
\hspace{30pt}
\begin{subfigure}{0.465\textwidth}
\begin{tikzpicture}[y=3cm, x=.015cm,font=\sffamily,scale=1.1]
    \begin{axis}[
        height=0.6\textwidth,
        width=1\textwidth,
        xlabel={ntrials},
        y label style={at={(axis description cs:0.1,.5)},anchor=south},
        ylabel={time [s]},
        xtick={0, 200, 400, 600, 800, 1000},
]
      \addplot[mark=*, mark size=1.2pt, mark options={fill=black}] 
		file {./data/conditioning.data}; 
     \end{axis}
	\begin{scope}[shift={(35,0.6)}] 
	\draw (0,0) -- 
		plot[mark=*, mark size=1.2pt, mark options={fill=black}] (15,0) -- (30,0)
		node[right]{\textsf{conditioning}};
	\end{scope}
\end{tikzpicture}
\end{subfigure}
\caption{Performance behavior of $\Upsilon$-DB on a Physiome testbed scenario.}\label{fig:experiments}
\end{figure}

\begin{figure}[t]\footnotesize
   \centering
\begingroup\setlength{\fboxsep}{5pt}
\colorbox{blue!5}{%
   \begin{tabular}{c|>{\columncolor[gray]{0.92}}c| l |c | c}
  \textsf{HYPOTHESIS} & $\upsilon$ & \textsf{name} & $|\mathcal S|$ & $|\mathcal E|$\\
      \hline  
      & $186$ & \textsf{Regulatory\_Vessel} & $40$ & $20$\\
      & $89$ & \textsf{Myo\_Dyn\_Resp\_wFit} & $73$ & $28$\\
      & $60$ & \textsf{Myogenic\_Compliant\_Vessel} & $100$ & $38$\\
      & $75$ & \textsf{Baroreceptor\_Lu\_et\_al\_2001} & $153$ & $74$\\
      & $70$ & \textsf{4-State\_Sarcomere\_Energetics} & $298$ & $91$\\
      & $120$ & \textsf{Comp\_four\_gen\_weibel\_lung} & $440$ &$186$\\
      & $91$ & \textsf{CardiopulmonaryMechanics} & $1132$ & $412$\\
      & $93$ & \textsf{CardiopulmonMechGasBloodExch} & $1593$ & $525$\\
      & $153$ & \textsf{HighlyIntegHuman} & $1624$ & $538$\\
      & $154$ & \textsf{HighlyIntHuman\_wIntervention} & $1919$ & $634$\\
      \cline{2-5}
      & $379$ & \textsf{Baroreflex\_SB\_CT} & $171$ & $74$\\      
   \end{tabular}
}\endgroup
\caption{Physiome hypotheses used in the experiments.}
\label{fig:physiome-hypotheses}
\end{figure}

\begin{itemize}
\item \emph{Extraction}. Some fluctuation may be due to practicalities of XML DOM access methods. The point of this performance study is to have practical measures of the amount of time taken to process representative hypothesis structures. Note that even for structures of size $|\mathcal S|=2\,K$ the amount of time required to extract a hypothesis is kept at subsecond order of magnitude (interactive response time) in a personal machine.

\item \emph{Encoding}. 
Some fluctuation is expected due to varying degrees of coupling between variables in the hypothesis structures. Note that, although $|\mathcal S|$ provides a very good measure of their size and complexity, the extent to which they are intricate should cause impact on the encoding procedure, which is any case kept $O(\sqrt{|\mathcal E|} |\mathcal S|)$. Again, the point here is to provide a notion of the amount of time required to encode representative hypotheses. For a scalability test on the encoding procedure, cf. Fig. \ref{fig:tcm-experiments}.

\item \emph{U-intro}. The \textsf{U-intro} procedure is composed of \emph{u-factor learning}, \emph{u-factori}- \emph{zation} and \emph{u-propagation}. We observed in previous tests that it was dominated by the learning component, viz., the discovery of occasional fd's in the `big' fact table. However, this is no longer the case once we implemented the workaround of keeping (in addition to the `big' table) a table containing only the exogenous (input parameter) variables, as it has negligible size w.r.t. the data of the endogenous (predictive output) variables. Then u-learning became subsecond again and the \textsf{U-intro} procedure became dominated by \emph{u-propagation}. In fact, the procedure of u-factorization, carried out once the fd's are discovered by u-learning, is also sub-second then has negligible processing time w.r.t. u-propagation. The latter is the most expensive sub-procedure of the synthesis method.

\item \emph{Conditioning}. The conditioning procedure is run for a selected phenomenon. It is composed of four main parts. First, by operation \textsf{conf()} it performs a probabilistic inference sub-query on the proper predictive projection of the `big' fact table of each hypothesis associated with the phenomenon. Second, it combines the results of each such sub-query through a union all query whose result set is a multi-hypothesis predictive table. Third, it loads the phenomenon observation sample data and the predictive data from the multi-hypothesis table into memory to apply Bayesian inference. Finally, the prior probability distribution of the predictive table is updated with the posterior and all the corresponding marginal probabilities are updated in their original U-relational tables. In our tests, this procedure is carried out over varying number of trials (\textsf{ntrials}). The total response times are shown in the last plot of Fig. \ref{fig:experiments}. \end{itemize}

This performance behavior is to be interpreted in the context of ETL in DW's. Loading and setting up an $\Upsilon$-DB has an overhead that shall be, though, much lower in high-performance machines. Such overhead is nonetheless justified for the use case of hypothesis management and analytics as opposed to simulation data management and exploratory analytics (cf. \S\ref{subsec:simulation-data}).

\section{Discussion}\label{sec:discussion-app}

We have verified that the hypothesis ratings/rankings shown in \S\ref{sec:case-studies} coincide with the results (e.g., of model tuning) described in the Physiome model entries and their related publications. 
That validates the applicability of the \mbox{$\Upsilon$-DB} methodology as a tool for data-driven analysis in such realistic scenarios.

The current practice in Computational Science for model evaluation and comparison in the presence of observational data is somewhat handcrafted: model agreement is assessed either qualitatively by referring to curve shapes in data plots or quantitatively by means of ad-hoc scripts. The \mbox{$\Upsilon$-DB} methodology offers a tool to perform data-driven hypothesis analytics semi-automatically directly in the database under the support of its querying capabilities. It has, therefore, potential to be a step towards higher standards of reproducibility and scalability.

\textbf{Realistic assumptions}. The core assumption of our framework is that the hypotheses are given in a formal specification which is encodable into a SEM that is complete (satisfies Defs. \ref{def:structure}, \ref{def:complete}). Also, as a semantic assumption which is standard in scientific modeling, we consider a one-to-one correspondence between real-world entities and variable/at\-tribute symbols within a structure, and that all of them must appear in some of its equations/fd's. For most science use cases involving deterministic models (if not all), such assumptions are quite reasonable. It can be a topic of future work (cf. \S\ref{sec:future-work}) to explore business use cases as well.

\textbf{Hypothesis learning}. $\!$The (user) method for hypothesis formation is irrelevant to our framework, as long as the resulting hypothesis is encodable into a SEM. So, a promising use case is to incorporate machine learning methods into our framework to scale up the formation/extraction of hypotheses and evaluate them under the querying capabilities of a p-DB. Con\-sider, e.g., learning the equations, say, from \textsf{Eureqa} \cite{schmidt2009}.\footnote{\url{http://creativemachines.cornell.edu/Eureqa}.}

\textbf{Qualitative hypotheses}. The $\Upsilon$-DB methodology is primarely motivated by computational science (usually involving differential equations). It is, however, applicable to qualitative deterministic models as well. Boolean Networks, e.g., consist in sets of functions $f(x_1, x_2, \!.., x_n)$, where $f$ is a Boolean expression. For instance, Fig. \ref{fig:boolean-hypothesis} presents the system of Boolean equations of a tentative Boolean Network model for a plant hormone (Fig. \ref{fig:boolean}) published in \cite{boolean2006}.\footnote{Cf. \url{http://atlas.bx.psu.edu/booleannet/booleannet.html}.} The notation, e.g., \textsf{SphK*}, is read (just like an ordinary differential equation), `the next state value of variable \textsf{SphK} is given by the state value of variable \textsf{ABA}. The parameters in this kind of model are the variable initial conditions.

\begin{figure}[H]\scriptsize
\begin{framed}
\vspace{-12pt}
\begin{eqnarray*}
\mathcal H = \{\quad
\textsf{SphK*} &=& \textsf{ABA},\\
\textsf{S1P*} &=&  \textsf{SphK},\\
\textsf{GPA1*} &=& \textsf{(S1P or not GCR1) and AGB1},\\
\textsf{PLD*} &=& \textsf{GPA1},\\
\textsf{PA*} &=& \textsf{PLD},\\
\textsf{pHc*} &=& \textsf{ABA},\\
\textsf{OST1*} &=& \textsf{ABA},\\
\textsf{ROP2*} &=& \textsf{PA},\\
\textsf{Atrboh*} &=& \textsf{pHc and OST1 and ROP2 and not ABI1},\\
\textsf{ROS*} &=& \textsf{Atrboh}\\
\textsf{H+ATPase*} &=& \textsf{not ROS and not pHc and not Ca2}^+\textsf{c}\\
\textsf{ABI1*} &=& \textsf{pHc and not PA and not ROS},\\
\textsf{RCN1*} &=& \textsf{ABA},\\
\textsf{NIA12*} &=& \textsf{RCN1},\\
\textsf{NOS*} &=& \textsf{Ca2}^+\textsf{c},\\
\textsf{NO*} &=& \textsf{NIA12 and NOS},\\
\textsf{GC*} &=& \textsf{NO},\\
\textsf{ADPRc*} &=&  \textsf{NO},\\
\textsf{cADPR*} &=& \textsf{ADPRc},\\
\textsf{cGMP*} &=& \textsf{GC},\\
\textsf{PLC*} &=& \textsf{ABA and Ca2}^+\textsf{c},\\
\textsf{InsP3*} &=& \textsf{PLC},\\
\textsf{InsPK*} &=& \textsf{ABA},\\
\textsf{InsP6*} &=& \textsf{InsPK},\\
\textsf{CIS*} &=& \textsf{(cGMP and cADPR) or (InsP3 and InsP6)},\\
\textsf{Ca2}^+\textsf{ATPase*} &=& \textsf{Ca2}^+\textsf{c},\\
\textsf{Ca2}^+\textsf{c *} &=& \textsf{(CaIM or CIS) and (not Ca2}^+\textsf{ATPase)},\\
\textsf{AnionEM*} &=& \textsf{((Ca2}^+\textsf{c or pHc) and not ABI1 ) or (Ca2}^+\textsf{c and pHc)},\\
\textsf{Depolar*} &=& \textsf{KEV or AnionEM  or (not H+ATPase) or (not KOUT) or Ca2}^+\textsf{c},\\
\textsf{CaIM*} &=& \textsf{(ROS or not ERA1 or not ABH1) and (not Depolar)},\\
\textsf{KOUT*} &=& \textsf{(pHc or not ROS or not NO) and Depolar},\\
\textsf{KAP*} &=& \textsf{(not pHc or not Ca2}^+\textsf{c) and Depolar},\\
\textsf{KEV*} &=& \textsf{Ca2}^+\textsf{c},\\
\textsf{PEPC*} &=& \textsf{not ABA},\\
\textsf{Malate*} &=& \textsf{PEPC and not ABA and not AnionEM},\\
\textsf{RAC1*} &=& \textsf{not ABA and not ABI1},\\
\textsf{Actin*} &=& \textsf{Ca2}^+\textsf{c  or not RAC1},\\
\textsf{Closure*} &=& \textsf{(KOUT or KAP ) and AnionEM  and Actin and not Malate} \quad\}.
\end{eqnarray*}
\vspace{-20pt}
\end{framed}
\vspace{-10pt}
\caption{Example of Boolean Network hypothesis.}
\label{fig:boolean-hypothesis}
\end{figure}

\begin{figure}[H]
\centering
\includegraphics[keepaspectratio,width=.99\textwidth]{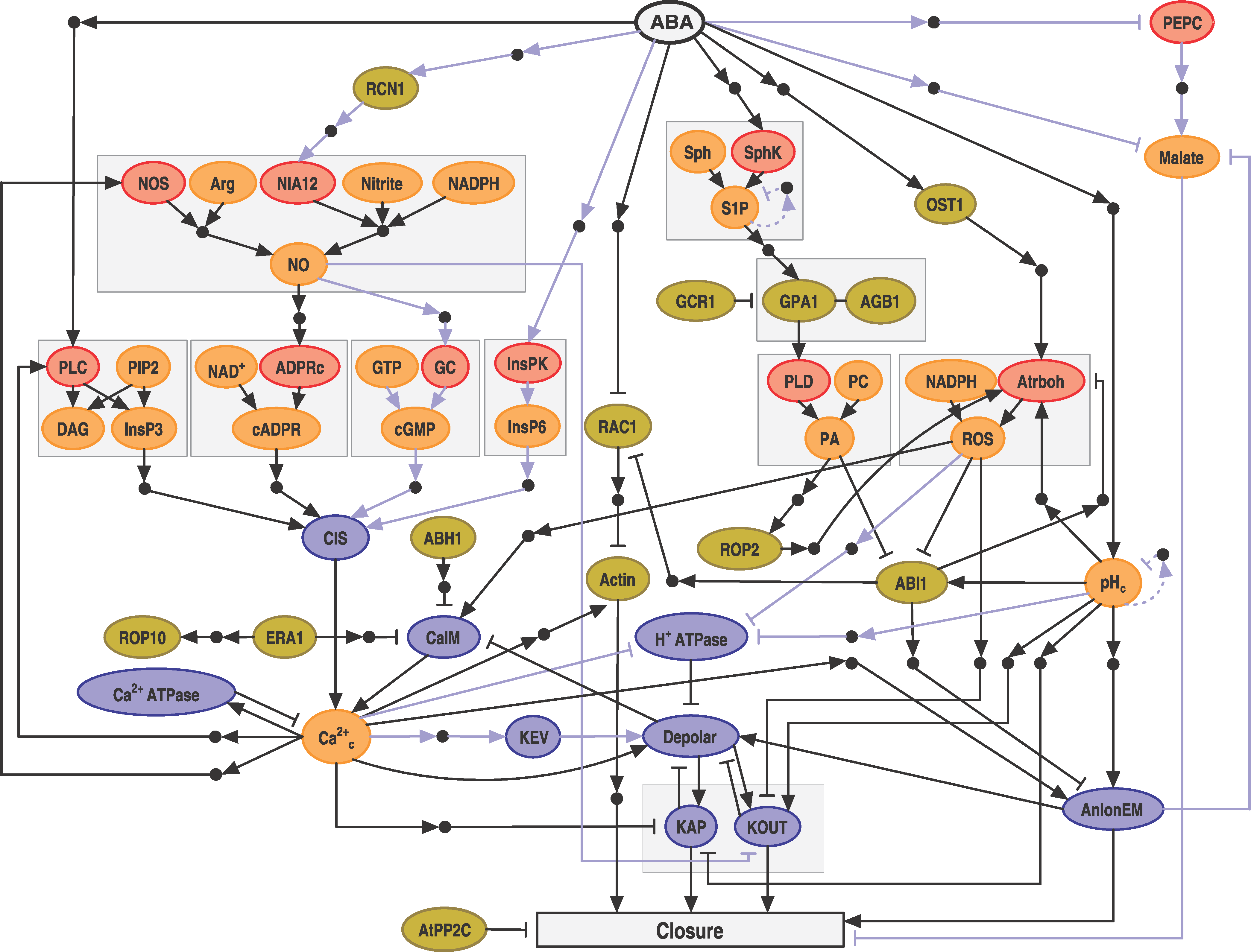}
\vspace{-5pt}
\caption[Example of Boolean Network model]{Example of Boolean Network model (source: \cite{boolean2006}).}
\label{fig:boolean}
\vspace{-7pt}
\end{figure}

Several kinds of dynamical system can be modeled in this formalism. Applications have grown out of gene regulatory network to social network and stock market predictive analytics. Even if richer semantics is considered (e.g., fuzzy logics), our encoding method is applicable likewise, as long as the equations are still deterministic.

\section{Conclusions}\label{sec:conclusions-app}

\noindent
In this chapter we have demonstrated and discussed the applicability of the \mbox{$\Upsilon$-DB} methodology. We have referred to real-world use case scenarios derived from the Physiome research project. We have shown in some detail the process of building an \mbox{$\Upsilon$-DB} with representative models from Physiome's model repository. That qualitative assessment is followed by experiments that provide some concrete feel on the performance behavior of \mbox{$\Upsilon$-DB} for models with up to 600+ mathematical variables.


\chapter{Conclusions}\label{ch:conclusions}

In this chapter we (\S\ref{sec:questions}) revisit the research questions addressed by this thesis, (\S\ref{sec:significance}) point out its significance and limitations, (\S\ref{sec:future-work}) list open problems and topics for future work, and (\S\ref{sec:final-considerations}) conclude with final considerations.

\section{Revisiting the Research Questions}\label{sec:questions}

Let us now revisit the conceptual (RQ1-4) and technical (RQ5-9) research questions.

\begin{itemize}
\item[\textbf{RQ1}.]\label{rq1}
\emph{How to define and encode hypotheses `as data'? What are the sources of uncertainty that may be present and should be considered?}

In Chapter \ref{ch:vision} we have provided core abstractions that compose the vision of hypotheses `as data,' or the $\Upsilon$-DB vision. The problem of hypothesis encoding has been defined and addressed further in Chapter \ref{ch:encoding}. We have distinguished two main sources of uncertainty in our model of uncertainty for hypothesis management, viz., (i) \emph{theoretical} uncertainty, as arising from competing hypotheses; and (ii) \emph{empirical} uncertainty, as arising from the alternative trial datasets available for each hypothesis.

\item[\textbf{RQ2}.]\label{rq2} 
\emph{How does hypotheses `as data' relate with observational data or, likewise, phenomena `as data' from a database perspective?}

Also in Chapter \ref{ch:vision}, we have presented a conceptual framework in which we have defined hypotheses `as data' and shown how it can be compared against phenomena `as data.' In fact, hypothesis management is really significant when it is possible to rate/rank hypotheses in the presence of (some partial piece of) evidence. 

\item[\textbf{RQ3}.]\label{rq3}
\emph{Does every piece of simulated data qualify as a scientific hypothesis? What is the difference between managing `simulation' data from managing `hypotheses' as data?}

Early in Table \ref{tab:hypothesis}, we provided a comparison between simulation data management and hypothesis data management. 
Furthermore, the scientific research process is abstracted in Chapter \ref{ch:vision} as a well-defined problem of data cleaning. Hypotheses are seen from an applied science point of view, and then are reduced into data such that a piece of simulation data is considered a hypothesis whenever it is assigned to explain some specific phenomenon.

\item[\textbf{RQ4}.]\label{rq4} 
\emph{Is there available a proper (machine-readable) data format we can use to automatically extract mathematically-expressed hypotheses from?}

We anticipated in Chapter \ref{ch:vision} the adoption of the XML data model as the general data format for extracting hypothesis specifications from. In particular, since we deal here with mathematical hypotheses, we refer to \textsf{MathML} as a standard for hypothesis specification. Concretely, in Chapter \ref{ch:applicability} we present use case demonstration scenarios for which we have developed a specific wrapper, viz., for the extraction of hypotheses specified in \textsf{MML} (Mathematical Modeling Language).

\item[\textbf{RQ5}.]\label{rq5}
\emph{Is there an algorithm to, given a SEM, efficiently extract its causal ordering? What are the computational properties of this problem?}

As shown in Chapter \ref{ch:encoding}, Simon's treatment of the problem of causal ordering given a SEM $\mathcal S(\mathcal E, \mathcal V)$ is NP-Hard. In the same chapter, we have discussed this problem in detail and presented an effective, efficient algorithmic approach to the problem. The computational cost for the whole process of hypothesis encoding is bounded by $O(\sqrt{|\mathcal S|} |\mathcal E|)$. Experiments show that the approach performs well in practice for large hypotheses.

\item[\textbf{RQ6}.]\label{rq6} 
\emph{What is the connection between SEM's and fd's? Can we devise an encoding scheme to `orient equations' and then effectively transform one into the other with guarantees? Once we do it, what design-theoretic properties have such a set of fd's?}

Also in Chapter \ref{ch:encoding}, we have presented an algorithmic encoding scheme to transform a SEM into a set of fd's with guarantees in terms of preserving the hypothesis causal structure. Our study of this problem has revealed some interesting properties of the resulting fd sets, in particular, that they are always `non-redundant' and, in comparison with arbitrary information systems, more precise and economical in the sense that, for any given attribute, there is exacly one fd with it in its rhs.

\item[\textbf{RQ7}.]\label{rq7} 
\emph{Is such fd set ready to be used for p-DB schema synthesis as an encoding of the hypothesis causal structure? If not, what kind of further processing we have to do? Can we perform it efficiently by reasoning directly on the fd's? How does it relate to the SEM's causal ordering?}

As we discuss in Chapter \ref{ch:reasoning}, the encoded fd set must be further processed to find the `first causes' for each of its predictive variable.
For addressing that, in Chapter \ref{ch:reasoning} we have presented the concept of the folding of an fd set and an efficient algorithm to compute it. Also, we have shown the equivalence of such fd reasoning with causal ordering processing.

\item[\textbf{RQ8}.]\label{rq8}
\emph{Is the uncertainty decomposition required for predictive analytics reducible to the structure level (fd processing), or do we need to process the simulated data to identify additional uncertainty factors? Finally, what properties are desirable for a p-DB schema targeted at hypothesis management? Are they ensured by this synthesis method?}

In Chapter \ref{ch:synthesis4u} we have presented a conceptual framework to address synthesis for uncertainty `4U.' In particular, we have introduced the need to process, for each hypothesis, its trial datasets available, and presented an efficient algorithm to factorize and propagate the overall uncertainty present in a given hypothesis (as a competing explanation for a target phenomenon). Then we have motivated BCNF as a notion of ``good'' design w.r.t. the factorized fd set based on the folding concept, and the lossless join property as required for the correctness of uncertainty decomposition. We have shown that the synthesized p-DB schema bears both properties.

\item[\textbf{RQ9}.]\label{rq9} 
\emph{Given all such a design-theoretic machinery to process hypotheses into \mbox{(U-)relational} DB's, what properties can we detect on the hypotheses back at the conceptual level? Do we have now technical means to speak of hypotheses that are ``good'' in terms of principles of the philosophy of science?}

Equipped with the design-theoretic machinery proposed in this thesis, we are able to, given a SEM, to automatically \textbf{(1)} extract its causal ordering, \textbf{(2)} detect its strongly coupled components and decide, for a given predictive projection, what are its associated u-factor projections (if any), and shall be able as well to \textbf{(3)} query the hypothesis ranking for a phenomenon of interest. All these are technical means to \cite{losee2001}: \textbf{(1$^\prime$)} extract the hypothesis `empirical content' and `predictive power;' \textbf{(2$^\prime$)} unravel its cohesiveness and how parsimonious it is in terms of the number of different claims or epistemological units carried within it, as well as its empirical grounding (`first causes'); and finally, we shall be able to \textbf{(3$^\prime$)} appraise it in face of competing or alternative explanations.

\end{itemize}

\section{Significance and Limitations}\label{sec:significance}

This thesis addresses the pressing call for large-scale, data-driven hypothesis management and analytics \cite{haas2011,dhar2013,goncalves2014}. Some reasons that contribute for its significance are listed (cf. \cite{goncalves2014,goncalves2015a,goncalves2015b}).

\begin{itemize}
\item Structured deterministic hypotheses are now shown to be \emph{encodable} as uncertain and probabilistic (U-relational) data based on p-DB principles;
\begin{itemize}
\item Study of the connection between SEM's and fd's, with contribution both to computational properties of the causal ordering problem, and to causal reasoning over fd's;
\item First \emph{synthesis} method for the construction of p-DB's from some previous existing formal specification. 
\end{itemize}
\item Definition of a concrete \emph{use case} of data-driven hypothesis management and analytics;
\begin{itemize}
\item New class of \emph{applications} introduced for p-DB's;
\item Settled the problem of Bayes' \emph{conditioning} in p-DB's.
\end{itemize}
\end{itemize}

Now some limitations of the thesis are listed.

\begin{itemize}
\item The Bayesian inference is implemented at application level, yet not formulated as a principled technical solution within research in p-DB's.

\item The encoding scheme to transform the mathematical system of a hypothesis into a set of fd's enabling the synthesis of the p-DB is applicable to structured \emph{deterministic} models only, not stochastic ones.
\end{itemize}

\section{Open Problems and Future Work}\label{sec:future-work}

Open problems and topics of future work are listed (no particular order). 

\begin{enumerate}
\item The design of a dedicated algebraic operation for Bayes' conditioning in p-WSA.

\item Investigation of other data dependency formalisms (e.g., multi-valued dependencies \cite{ullman1988}), approximate fd's \cite{huhtala1999}, conditional fd's \cite{fan2011}) to extend the scope of $\Upsilon$-DB towards structured stochastic models.

\item Development of techniques for systematic hypothesis extraction as a well-defined problem of (web) information extraction;

\item Investigation of business use case scenarios for data-driven decision making on top of $\Upsilon$-DB;

\item Definition of a machine learning use case scenario to industrialize hypothesis formation and assess $\Upsilon$-DB's performance feasibility in such a scenario;

\item Development of automatic 
data sampling techniques to leverage the data definition of both hypotheses and phenomena in $\Upsilon$-DB from a statistical point of view.
\end{enumerate}

\section{Final Considerations}\label{sec:final-considerations}

In this thesis we have developed the vision of $\Upsilon$-DB, which is essentially an abstraction of hypotheses as uncertain and probabilistic data. It comprises a design-theoretic methodology for the systematic construction and management of U-relational hypothesis DB's. It is meant to provide a principled approach to enable scientists and engineers to manage and evaluate (rate/rank) large-scale scientific hypotheses as theoretical data. 
We have addressed some core technical challenges over the $\Upsilon$-DB vision in order to properly encode deterministic hypotheses as uncertain and probabilistic data.

As envisioned by Jim Gray \cite{hey2009}, the scientific method has been shifting towards being operated as a data-driven discipline which is rapidly gaining ground \cite{dhar2013}. In this thesis we have strived for proposing some core principles and techniques for enabling data-driven hypothesis management and analytics, opening a promising line of research in both probabilistic databases and simulation data management. 



\bibliographystyle{bibliography/abnt}
\bibliography{bibliography/thesis}

\appendix 


\begin{center}
\chapter{Detailed Proofs}\label{ch:proofs} 
\end{center}


\section{Proofs of Hypothesis Encoding}\label{sec:proofs-encoding}

\subsection{Proof of Theorem \ref{thm:coa-hardness}}\label{a:coa-hardness}
\noindent
\emph{``Let $\mathcal S(\mathcal E, \mathcal V)$ be a complete structure. Then the extraction of its causal ordering by Simon's $\textsf{COA}(\mathcal S)$ is intractable (NP-Hard).''}
\begin{proof}
We show that, at each recursive step of \textsf{COA}, to find all non-trivial minimal subsets (i.e., $|\mathcal E^\prime| \geq 2$) translates into an optimization problem associated with the decision problem BPBP, which we know by Lemma \ref{lemma:pseudo-biclique} to be NP-Complete.

First, recall (Def. \ref{def:complete}) that a structure $\mathcal S(\mathcal E,\, \mathcal V)$ is complete if $|\mathcal E|=|\mathcal V|$; e.g., for the structure given in Fig. \ref{fig:hardness} (left), note (Def. \ref{def:minimal}) the minimal structure $\mathcal S^\prime(\mathcal E^\prime,\, \mathcal V^\prime)$, where $\mathcal E^\prime=\{\, f_1,\, f_2,\, f_3\,\}$.  
For non-trivial minimal structures, i.e., when $|\mathcal E^\prime|=K \geq 2$, it is easy to see that its corresponding bipartite graph $G=(V_1^\prime \cup V_2^\prime,\, E^\prime)$, where $V_1 \mapsto \mathcal E$, $V_2 \mapsto \mathcal V$ and $E \mapsto \mathcal S$ must have number of edges $|E^\prime|\geq 2K$ and, for all its vertices $u \in V_1^\prime \cup V_2^\prime$, $u$ must have $deg(u) \geq 2$, i.e., $G$ is a pseudo-biclique in accordance with Def. \ref{def:pseudo-biclique}. That intuition is elaborated as follows.

The point is that, no matter how big is such structure $\mathcal S^\prime$, its equations $f \in \mathcal E^\prime$ are such that $|Vars(f)|\geq 2$ (as $\mathcal S^\prime$ is non-trivial) and its variables can be grouped in local patterns from the sparsest kind to the densest. To construct an instance of the sparsest case, let $\mathcal S^\prime$ be built by setting a first equation where its entry in the structure matrix $A_{\mathcal S}$ has form $(1, 1, 0^+)$ and then, for the next $|\mathcal E^\prime|-2$ equations, 
shift such pair of 1's one position right w.r.t. the previous one. Then complete it with a last equation whose form is form $(1, 0^+, 1)$. That is, the structure is built with unique pairs of 1's spread all over the structure. 
Then, deciding whether there is a minimal structure of size $|\mathcal E^\prime|=K$ corresponds exactly to BPBP. It is a special case (BBP), when such minimal structure is the densest possible, i.e., when $A_{\mathcal S}$ has only 1's and then its corresponding bipartite graph is a $K$-balanced biclique with $|E^\prime| = K^2$, and $deg(u)=K$ for all vertices $u \in V_1^\prime \cup V_2^\prime$. For instance, see the minimal structure with $\mathcal E^\prime=\{\, f_4,\, f_5\,\}$ found at the second recursive step of \textsf{COA} in Fig. \ref{fig:coa}. 
$\Box$
\end{proof}

\subsection{Proof of Proposition \ref{prop:causal-ordering}}\label{a:causal-ordering}
\noindent
\emph{``Let $\mathcal S(\mathcal E, \mathcal V)$ be a structure, and $\varphi_1\!:\, \mathcal E \to \mathcal V$ and $\varphi_2\!:\, \mathcal E \to \mathcal V$ be any two total causal mappings over $\mathcal S$. Then $C^+_1$ = $C^+_2$.''} 
\begin{proof}
The proof is based on an argument from Nayak \cite{nayak1994}, which we present here arguably much clearer. Intuitively, it shows that if $\varphi_1$ and $\varphi_2$ differ on the variable an equation $f$ is mapped to, then such variables, viz., $\varphi_1(f)$ and $\varphi_2(f)$, must be causally dependent on each other (strongly coupled). 
 
To show $C^+_1$ = $C^+_2$ reduces to $C^+_1 \subseteq C^+_2$ and $C^+_2 \subseteq C^+_1$. We show the first containment, with the second being understood as following by symmetry. Closure operators are extensive, $X \subseteq cl(X)$, and 
idempotent, $cl(cl(X)) = cl(X)$. That is, if we have $C_1 \subseteq C_2^+$, then we shall have $C_1^+ \subseteq (C_2^+)^+$ and, by idempotence, $C_1^+ \subseteq C_2^+$. 

Then it suffices to show that $C_1 \subseteq C_2^+$, i.e., for any $(x^\prime,\, x) \in C_1$, we must show that $(x^\prime,\, x) \in C_2^+$ as well. Observe by Def. \ref{def:tcm} that both $\varphi_1$ and $\varphi_2$ are bijections, then, invertible functions. If $\varphi_1^{-1}(x) = \varphi_2^{-1}(x)$, then we have $(x^\prime,\, x) \in C_2$ and thus, trivially, $(x^\prime,\, x) \in C_2^+$. Else, $\varphi_1$ and $\varphi_2$ disagree in which equations they map onto $x$. But we show next, in any case, that we shall have $(x^\prime,\, x) \in C_2^+$. 

Take all equations $g \in \mathcal E^\prime \subseteq \mathcal E$ such that $\varphi_1(g) \neq \varphi_2(g)$, and let $n \leq |\mathcal E|$ be the number of such `disagreed' equations. Now, let $f \in \mathcal E^\prime$ be such that its mapped variable is $x = \varphi_1(f)$. 
Construct a sequence of length $2n$ such that, $s_0 = \varphi_1(f) = x$ and, for $1 \leq i \leq 2n$, element $s_i$ is defined $s_i = \varphi_2(\varphi_1^{-1}(s_{i-1}))$. That is, we are defining the sequence such that, for each equation $g \in \mathcal E^\prime$, its disagreed mappings $\varphi_1(g)=x_a$ and $\varphi_2(g)=x_b$ are such that $\varphi_1(g)$ is immediately followed by $\varphi_2(g)$. As $x_a,\, x_b \in Vars(g)$, we have $(x_a,\, x_b) \in C_2$ and, symmetrically, $(x_b,\, x_a) \in C_1$. The sequence is of form $s=\langle \underbrace{x,\, x_f}_{f}, \hdots, \underbrace{x_a,\, x_b}_{g}, \hdots, \underbrace{x_{2n-1},\, x_{2n}}_{h} \rangle$.

Since $x$ must be in the codomain of $\varphi_2$, we must have a repetition of $x$ at some point $2 \leq k \leq 2n$ in the sequence index, with $s_k=x$ and $s_{k-1}=x^{\prime\prime}$ such that $(x^{\prime\prime},\, x) \in C_2$. If $x^{\prime\prime}=x^\prime$, then $(x^\prime,\, x) \in C_2$ and obviously $(x^\prime,\, x) \in C_2^+$. Else, note that $x_f$ must also be in the codomain of $\varphi_1$, while $x^{\prime\prime}$ in the codomain of $\varphi_2$. Let $\ell$ be the point in the sequence, $3 \leq \ell \leq 2n\!-\!1$, at which $s_\ell=x_f=x_a$ and $s_{\ell+1}=x_b$ for some $x_b$ such that $(x_f,\, x_b) \in C_2$. It is easy to see that, either we have $x_b=x^{\prime\prime}$ or $x_b \neq x^{\prime\prime}$ but $(x_b,\, x^{\prime\prime}) \in C_2^+$. Thus, by transitivity on such a causal chain, we must have $(x_f,\, x^{\prime\prime}) \in C_2^+$ and eventually $(x_f,\, x) \in C_2^+$. Finally, since $x^\prime \in Vars(f)$ and $\varphi_2(f)=x_f$, we have $(x^\prime,\, x_f) \in C_2$ and, by transitivity, $(x^\prime,\, x) \in C_2^+$. 
$\Box$
\end{proof}

\subsection{Proof of Theorem \ref{thm:non-redundant}.}\label{a:non-redundant}
\noindent
\emph{``Let $\Sigma$ be an fd set defined $\Sigma \!\triangleq\!$ \textsf{h-encode}($\mathcal S$) for some complete structure $\mathcal S$. Then $\Sigma$ is non-redundant and singleton-rhs but may not be left-reduced (then may not be canonical).''}

\begin{proof}
We will show that properties (a-b) of Def. \ref{def:minimal} hold for $\Sigma$ produced by (Alg. \ref{alg:h-encode}) \textsf{h-encode}, but property (c) may not hold. 

At initialization, the algorithm sets $\Sigma\!=\!\varnothing$ and then inserts an fd $\langle X, A\rangle \in \Sigma$ for each $\langle f,\, x \rangle \!\in \varphi_t$ scanned, where $x \mapsto A$ and $X\cap A=\varnothing$. At termination, for all fd's in $\Sigma$ we obviously have $|A|=1$ then property (a) holds. Also, note that $\varphi\!: S \!\to Vars(S)$ is, by Def. \ref{def:tcm}, a bijection. 

Now, for property (b) not to hold there must be some fd $\langle X, A\rangle \in \Sigma$ that is redundant and then can be found in the closure of $\Gamma \!=\! \Sigma \setminus \langle X, A\rangle$. 
By Lemma \ref{lemma:singleton-rhs} (below), that can be the case only if $A \subseteq X$ or there is $\langle Y, A\rangle \in\Gamma$ for some $Y$. 
But from $X\cap A=\varnothing$, we have $A \nsubseteq X$; and from $\varphi$ being a bijection it follows that there can be no such fd in $\Gamma$. Thus it must be the case that $\Sigma$ is non-redundant, i.e., property (b) holds.

Finally, property (c) does not hold if there can be some fd $\langle X, A\rangle \in \Sigma$ with $Y \!\subset X$ such that $\Gamma=\Sigma \setminus \langle X, A\rangle \cup \langle Y, A\rangle$ has the same closure as $\Sigma$. $\!$That is, 
if we may find $\langle Y, A\rangle \in \Sigma^+$. 
$\!$Now, pick structure $S$ whose ($3 \times 3$) matrix $A_s$ has rows $(1, 0, 0), (1, 1, 0), (1, 1, 1)$ as an instance. Alg. \ref{alg:h-encode} encodes it into $\Sigma\!=\!\{\phi \!\to x_1,\; x_1 \,\upsilon \!\to x_2,\; x_1 \,x_2 \,\upsilon \!\to x_3 \}$. Let $Y\!=\!\{x_1, \upsilon\}$, and $B\!=\!\{x_2\}$. 
Note that $x_1\,\upsilon \to x_2 \in \Sigma$ can be written as $\langle Y, B\rangle \in \Sigma$, and $x_1 \,x_2 \,\upsilon \!\to x_3 \in \Sigma$ as $\langle YB, A\rangle \in \Sigma$. Now observe that $\langle Y, A\rangle \in \Sigma^+$ can be derived by R5 over $\langle Y, B\rangle, \langle YB, A\rangle \in \Sigma$, which is sufficient to show that property (c) may not hold. That is, $B$ is ``extraneous'' in $\langle YB, A\rangle \in \Sigma$ and can be removed from its lhs without loss of information to $\Sigma$.
$\Box$
\end{proof}

\begin{mylemma}\label{lemma:singleton-rhs}
Let $\,\Sigma$ be a (Def. $\!$\ref{def:minimal}-a) singleton-rhs fd set on attributes $U$. Then $X \!\to A$ can only be in $\Sigma^+$, where $XA \subseteq U$, if $A \subseteq X$ or there is non-trivial $\langle Y, A\rangle \in \Sigma$ for some $Y \subset U$.
\end{mylemma}
\begin{proof}
By Lemma \ref{lemma:ullman} (below), we know that $X \!\to A \in \Sigma^+$ iff $A \!\subseteq X^+$. We need to prove that if $A \!\nsubseteq X$ and there is no $Y \!\to A$ in singleton-rhs $\Sigma$, then $A \!\nsubseteq X^+$. But this is equivalent to show that (Alg. \ref{alg:xclosure}) \textsf{XClosure} gives only correct answers for $X^+$ w.r.t. $\Sigma$, which is known (cf. theorem from Ullman \cite[p. 389]{ullman1988}). Note that \textsf{XClosure}($\Sigma$, $X$) inserts $A$ in $X^+$ only if $A \subseteq X$ or there is some fd  $\langle Y, A\rangle \in \Sigma$. 
$\Box$
\end{proof}

\begin{mylemma}\label{lemma:ullman}
Let $\Sigma$ be an fd set. An fd $X \!\to Y$ is in $\Sigma^+$ iff $Y \subseteq X^+$, where $X^+$ is the attribute closure of $X$ w.r.t. $\Sigma$.
\end{mylemma}
\begin{proof}
This is from Ullman \cite[p. $\!$386]{ullman1988}. Let $Y\!=\!A_1\,...\,A_n$ and suppose $Y \subseteq X^+$. Then for each $A_i$, we have $A_i \in X^+$ and, by the definition of $X^+$, we must have $\langle X, A_i\rangle \in \Sigma^+$. Then it follows by (R4) union that $X \to Y$ is in $\Sigma^+$ as well. Conversely, suppose $\langle X, Y\rangle \in \Sigma^+$. Then, by (R3) decomposition we have $\langle X, A_i\rangle \in \Sigma^+$ for each $A_i \in Y$. 
$\Box$
\end{proof}

\section{Proofs of Causal Reasoning}\label{sec:proofs-causal-reasoning}

\subsection{Proof of Lemma \ref{lemma:folding-unique}}\label{a:folding-unique}
\noindent
\emph{``Let $\mathcal S(\mathcal E, \mathcal V)$ be a complete structure, $\varphi$ a total causal mapping over $\mathcal S$ and $\Sigma$ an fd set encoded through $\varphi$ given $\mathcal S$. 
If $\langle X, A\rangle \in \Sigma$, then $A^\looparrowright\!$, the attribute folding of $A$ (w.r.t. $\!\Sigma$) exists and is unique.''}
\begin{proof}
The existance of $A^\looparrowright$ is ensured by the degenerate case where $X = A^\looparrowright\!$, as $X \!\to A$ is itself in $\Sigma^\vartriangleright$ by an empty application of R5. If $X \!\to A$ is in fact folded w.r.t. $\!\Sigma$, then the folding of $A$ exists. Else, it is not folded yet $X \!\to A$ is non-trivial because by Theorem \ref{thm:non-redundant} $\Sigma$ is non-redundant. Then, by Def. \ref{def:folded} there must be some 
$Y \subseteq U$ with $Y \nsupseteq X$ such that $Y \to X$ is in $\Sigma^+$ and $X \not\to Y$. By Def. \ref{def:ptc}, there is a finite application of R5 over fd's in $\Sigma$ to derive $Y \xrightarrow{\triangleright} X$. 
Then by R2$\,\sim\,$R5 over $X \!\to A$, we have $Y \!\to A$. Although there may be many such (intermediate) attribute sets $Y \subset U$ along the transitive chaining satisfying the conditions above, we claim there is at least one that is a folding of $A$. Suppose not. Then, for all such $Y \subset U$, there is some $Y^\prime \subset U$ with $Y^\prime \nsupseteq Y$ such that $Y^\prime \to Y$ and $Y \not\to Y^\prime$, leading to an infinite regress. Nonetheless, in so far as cycles are ruled out by force of Def. \ref{def:folded}, then $\Sigma^+$ must have an infinite number of fd's. But $\Sigma^+$ is finite, viz., bounded by $2^{2|U|}$ (cf. \cite[p. $\!165$]{abiteboul1995}). $\lightning$. Therefore the folding of $A$ must exist.

Moreover, observe that $\Sigma$ is encoded through $\varphi$, which is by Def. \ref{def:tcm} a bijection. Then we have $\langle X, A\rangle \in \Sigma$ for exactly one attribute set $X$. Then, as a straightfoward follow-up of the rationale that led us to infer the folding existance, note that there must be a single chaining $Y^n \!\xrightarrow{\triangleright} ... \!\xrightarrow{\triangleright} Y^1 \!\xrightarrow{\triangleright} Y^0 \!\xrightarrow{\triangleright} X \!\xrightarrow{\triangleright} A$. Again, as cycles are ruled out by force of Def. \ref{def:folded} and $\Sigma^+$ is finite, then the folding of $A$ is unique. 
$\Box$
\end{proof}

\subsection{Proof of Theorem \ref{thm:afolding}}\label{a:afolding}
\noindent
\emph{``Let $\mathcal S(\mathcal E, \mathcal V)$ be a complete structure, and $\Sigma$ an fd set encoded given $\mathcal S$. 
Now, let $\langle X, A\rangle \in \Sigma$. Then \textsf{AFolding}($\Sigma, A$) correctly computes $A^\looparrowright\!$, the attribute folding of $A$ (w.r.t. $\!\Sigma$) in time $O(|\mathcal S|^2)$.''}

\begin{proof}
For the proof roadmap, note that \textsf{AFolding} 
is monotone (size of $A^\star$ can only increase) and terminates precisely when $A^{(i+1)}\!=\!A^{(i)}$, where $A^{(i)}$ denotes the attributes in $A^\star$ at step $i$ of the outer loop. The folding $A^\looparrowright\!$ of $A$ at step $i$ is $A^{(i)} \setminus \Lambda^{(i)}$. We shall prove by induction, given attribute $A$ from fd $X \!\to A$ in $\Sigma$, that $A^\star \!\setminus \Lambda$ returned by \textsf{AFolding}($\Sigma, A$) is the unique attribute folding $A^\looparrowright\!$ of $A$. 

$\!\!$(\textbf{Base case}). 
By Theorem \ref{thm:non-redundant}, $\Sigma$ is non-redundant with (then) non-trivial $\langle X, A\rangle \!\in \Sigma$ for exactly one attribute set $X$, the algorithm always reaches step $i\!=\!1$, which is our base case. Then $X$ is placed in $A^{(1)}$ and $A$ in $\Lambda^{(1)}$, and we have $A^{(1)}\!=\!XA$ and $\Lambda^{(1)}\!=\!A$. Therefore, $A^{(1)} \setminus \Lambda^{(1)} \!=\! X$, and in fact we have $\langle X, A\rangle \in \Sigma^\vartriangleright$ by an empty application of R5. For it to be specifically in $\Sigma^\looparrowright\!\! \subset \Sigma^\vartriangleright\!$, it must be folded w.r.t. set $\Delta$ of consumed fd's at this step, viz., $\Delta^{(1)} \!=\! \{X\!\to A\}$. In fact, as the only fd in $\Delta^{(1)}$, by Def. \ref{def:folded} it must be folded w.r.t. $\Delta^{(1)}$, and we have $A^\looparrowright\!=X$ at step $i\!=\!1$. 

$\!\!$(\textbf{Induction}). Now, let $i\!=\!k$, for $k\!>\!1$, and assume that $\langle A^{(k)} \!\setminus\Lambda^{(k)},\, A\rangle \in \Sigma^\looparrowright\! \subset \Sigma^\vartriangleright\!$ with $A^{(k)} \!\neq \Lambda^{(k)}$. 
By Lemma \ref{lemma:folding-unique} we know that $A^\looparrowright\!=A^{(k)}\setminus \Lambda^{(k)}$ is the unique folding of $A$ at step $i\!=\!k$. 
For the inductive step, suppose $Y$ is placed in $A^{(k+1)}$ and $B$ in $\Lambda^{(k+1)}$ because $\langle Y, B \rangle \!\in \Sigma \setminus \Delta^{(k)}$ and $B \in A^{(k)}$. 

Since $B \in A^{(k)}$ and $B \notin \Lambda^{(k)}$ (it is yet just be consumed into $\Lambda^{(k+1)}$), we can write $(A^{(k)}\setminus \Lambda^{(k)})=ZB$ for some $Z \neq B$, where $\underbrace{(A^{(k)}\setminus \Lambda^{(k)})}_{ZB}\!\to A$ is assumed in $\Sigma^\looparrowright\!$. 
Now, with the application of R5 consuming $Y \!\to B$ we have 
$\underbrace{(A^{(k)}YB\setminus \Lambda^{(k)}B)}_{ZS}\!\to A$, where $S=Y\setminus \Lambda^{(k)}$. We claim that $ZS \!\to A$ is folded w.r.t. $\Delta^{(k+1)}$. 

Suppose not. Then by Def. \ref{def:folded} there must be some $W \nsupseteq ZS$ such that $W \!\to ZS$ is in $(\Delta^{(k+1)})^+$ but $ZS \not\to W$.  
Since $ZS \neq \varnothing$, there must be some $C \in ZS$, i.e., $C \notin \Lambda^{(k+1)}$. Note that, as $W \!\to ZS$ is in $(\Delta^{(k+1)})^+$, then by (R3) decomposition we have $W \!\to C$ in $(\Delta^{(k+1)})^+$ as well. But by Lemma \ref{lemma:singleton-rhs} that can only be the case if there is some $\langle T, C \rangle \in \Delta^{(k+1)}$, which means $C$ has been already consumed into $\Lambda^{(k+1)}$, though $C \notin \Lambda^{(k+1)}$. $\lightning$.

Finally, as for the time bound, note that 
in worst case, exactly one fd $Y \!\to B$ is consumed from $\Sigma$ into $\Delta$ for each step of the outer loop, where $|\Sigma|=|\mathcal E|$. That is, let $n=|\mathcal E|$, then $n$ is decreased stepwise in arithmetic progression 
such that $n + (n\!-\!1) +  \hdots  + 1 = n\,(n\!-\!1)/2$ scans are required overall, i.e., $O(n^2)$. Note also, however, that $B$ may be the only symbol read at each such fd scan but in worst case at most $|U|=|\mathcal V|$ symbols are read. Thus our measure $n$ should be actually overestimated $n=|\mathcal E|\,|\mathcal V|=|\mathcal S|$. Therefore Alg. \ref{alg:afolding} is bounded by $O(|\mathcal S|^2)$. 
$\Box$
\end{proof}

\subsection{Proof of Corollary \ref{cor:folding}}\label{a:folding}
\noindent
\emph{``Let $\mathcal S(\mathcal E, \mathcal V)$ be a complete structure, and $\Sigma$ an fd set encoded given $\mathcal S$. 
Then algorithm \textsf{folding}($\Sigma$) correctly computes $\Sigma^\looparrowright\!$, the folding of $\Sigma$ in time that is $f(|\mathcal S|)\,\Theta(|\mathcal E|)$, where $f(|\mathcal S|)$ is the time complexity of (Alg. \ref{alg:afolding}) \textsf{AFolding}.''}
\vspace{3pt}
\begin{proof}
By Theorem \ref{thm:afolding}, we know that sub-procedure (Alg. \ref{alg:afolding}) \textsf{AFolding} is correct and terminates. Then (Alg. \vspace{-2.5pt}\ref{alg:folding}) \textsf{folding} necessarily inserts in $\Sigma^\looparrowright\!$ (initialized empty) exactly one fd $Z \!\xrightarrow{\looparrowright} A$ for each fd $X \!\to A$ in $\Sigma$ scanned. Thus, at termination we have $|\Sigma^\looparrowright\!|\!=\!|\Sigma|$. Again, as \textsf{AFolding} is correct, we know $Z$ is the unique folding of $A$. Therefore it must be the case that Alg. \ref{alg:folding} is correct. Finally, for the time bound, the algorithm iterates over each fd in $\Sigma$ without having to read its symbols, and at each such step \textsf{AFolding} takes time that is $f(n)$. Thus \textsf{folding} takes $f(|\mathcal S|)\,\Theta(|\mathcal E|)$. But we know from Theorem 
\ref{thm:afolding} and Remark \ref{rmk:linear} that $f(|\mathcal S|) \in O(|\mathcal S|)$, then it takes $O(|\mathcal S|\,|\mathcal E|)$. 
$\Box$
\end{proof}

\subsection{Proof of Proposition \ref{prop:folding-parsimonious}}\label{a:folding-parsimonious}
\noindent
\emph{``Let $\mathcal S(\mathcal E, \mathcal V)$ be a complete structure, $\varphi$ a total causal mapping over $\mathcal S$ and $\Sigma$ an fd set encoded through $\varphi$ given $\mathcal S$. Let $\Sigma^\looparrowright\!$ be the folding of $\Sigma$, then $\Sigma^\looparrowright\!$ is parsimonious.''}
\begin{proof}
By Lemma \ref{lemma:folding-unique} we know that, for each fd $\langle X, A\rangle \in \Sigma$, the attribute folding $Z$ of $A$ such that $Z \xrightarrow{\looparrowright} A$ exists and is unique. That is, for no $Y \neq Z$ we have $Y \xrightarrow{\looparrowright} A$. Thus $\Sigma^\looparrowright \triangleq \textsf{folding}(\Sigma)$ automatically satisfies Def. \ref{def:parsimonious}, as long as we show it is canonical (cf. Def. \ref{def:minimal}).

Moreover, by Theorem \ref{thm:non-redundant} we know that $\Sigma$ is both non-redundant and singleton-rhs. 
Now, consider by Lemma \ref{lemma:folding-unique} that \textsf{AFolding} builds a bijection mapping each\vspace{-2pt} $\langle X, A\rangle \in \Sigma$ to exactly one $\langle Z, A\rangle \in \Sigma^\looparrowright\!$ such that $Z \xrightarrow{\looparrowright} A$. Since $\Sigma$ is singleton-rhs, it is obvious that $\Sigma^\looparrowright\!$ is as well and covers all attributes in the rhs of fd's in $\Sigma$. Also, the bijection implies $|\Sigma^\looparrowright\!| = |\Sigma|$. Since $\Sigma$ is non-redundant and has exactly one fd with each attribute $A$ in its rhs, then by Lemma \ref{lemma:singleton-rhs} so is $\Sigma^\looparrowright\!$.

Finally we will show that unlike $\Sigma$, its folding $\Sigma^\looparrowright\!$ must be left-reduced. Suppose not. Then for some fd $Z \!\to A$ in $\Sigma^\looparrowright\!$ there is $S \subset Z$ such that non-trivial $S \!\to A$ is in $(\Sigma^\looparrowright\!)^+$. Since $Z \!\to A$ is the only fd in $\Sigma^\looparrowright\!$ with $A$ in its rhs and $S \!\to A$ is non-trivial, we must have $S \xrightarrow{\vartriangle} Z \xrightarrow{\vartriangle} A$. 

Now, suppose $S \!\to A$ is not folded. Then there is $W \nsupseteq S$ such that $W \!\to S$ is in $(\Sigma^\looparrowright\!)^+$ but $S \not\to W$. Note that $W \neq Z$, as $W \nsupseteq S$. Also, $W \!\to S$ and $S \!\to Z $implies $W \!\to Z$ by (R5) transitivity. Note also that $S \not\to W$ and $S \!\to Z$ implies $Z \not\to W$. But $Z \!\to A$ is assumed folded. $\lightning$. That is, $S \!\to A$ must be folded. Then we have both $S \!\to A$ and $Z \!\to A$ folded, though $S \neq Z$. That is, the attribute folding of $A$ is not unique, even though we know by Lemma \ref{lemma:folding-unique} that it must be unique. $\lightning$. Thus $\Sigma^\looparrowright\!$ must be left-reduced, altogether, therefore, parsimonious. $\Box$
\end{proof}

\subsection{Proof of Theorem \ref{thm:connections1}}\label{a:connections1}
\noindent
\emph{``Let $\mathcal S(\mathcal E, \mathcal V)$ be a complete structure, $\varphi$ a total causal mapping over $\mathcal S$ and $\Sigma$ an fd set encoded through $\varphi$ given $\mathcal S$. Then $x_a, x_b \in \mathcal V$ are such that $x_b$ is causally dependent on $x_a$, i.e., $(x_a, x_b) \in C^+_\varphi$ iff there is some non-trivial fd $\langle X, B\rangle \in \Sigma^\vartriangleright$ with $A \in X$, where $B \mapsto x_b$ and $A \mapsto x_a$.''} 

\begin{proof}
We prove the statement by induction. We consider first the `if' direction, and then its `only if' converse.

\textbf{(Base case)}. Let $\langle X, B\rangle \in \Sigma$ be some fd with $A \in X$, where $B \mapsto x_b$ and $A \mapsto x_a$. By Theorem \ref{thm:non-redundant}, it is non-trivial and then by default (i.e., an empty application of R5) it is in $\Sigma^\vartriangleright$. But as $X \to B$ is in $\Sigma$, (Alg. \ref{alg:h-encode}) \textsf{h-encode} ensures that there is exactly one equation $f \in \mathcal E$ such that $\varphi(f)=x_b$ and $x_a \in Vars(f)$ where $B \mapsto x_b$ and $A \mapsto x_a$. 
Then by force of Eq. \ref{eq:direct-causal-dependencies} we must have $(x_a, x_b) \in C_\varphi$. Thus, we obviously have $(x_a, x_b) \in C^+_\varphi$ as well.

\textbf{(Induction)}. Now, recall Armstrong's (R5) pseudo-transitivity rule adapted here for the particular case of singleton-rhs fd sets, viz., if $Y \!\to C$ and $CZ \!\to B$, then $YZ \!\to B$. By the inductive hypothesis, take any two non-trivial fd's $\langle Y, C\rangle, \langle CZ, B\rangle \in \Sigma^\vartriangleright$ with $B \notin Y$ and assume that the causal dependency property holds for their attributes that encode variables. That is, let $D \in Y$ and $E \in Z$, where $D \mapsto x_d$ and $E \mapsto x_e$ for $x_d, x_e \in \mathcal V$ such that $(x_d, x_c), (x_e, x_b), (x_c, x_b) \in C^+_\varphi$.
Note that both $Y \!\to C$ and $CZ \!\to B$ are non-trivial, then $C \notin Y$, $B \notin Z$ and $B \neq C$. Moreover, $B \notin Y$ has been assumed such that the fd $\langle YZ, B\rangle \in \Sigma^\vartriangleright$ to be derived by R5 over $Y \!\to C$ and $CZ \!\to B$ is also non-trivial to satisfy the condition of the theorem. 
Now, it is easy to see that the property holds likewise for non-trivial fd $\langle YZ, B\rangle \in \Sigma^\vartriangleright$. In fact, $(x_d, x_c), (x_c, x_b) \in C^+_\varphi$ implies $(x_d, x_b) \in C^+_\varphi$ and also by the inductive hypothesis we have $(x_e, x_b) \in C^+_\varphi$. That is, for either some $D \in Y$ or some $E \in Z$, we must have $(x_d, x_b), (x_e, x_b) \in C^+_\varphi$.

The converse `only if' direction can be shown by a symmetrical inductive argument. That is, for the base case suppose $(x_a, x_b) \in C_\varphi$. Then, by Eq. \ref{eq:direct-causal-dependencies} we know there is some $f \in \mathcal E$ such that $\varphi(f)=x_b$ and $x_a \in Vars(f)$. Moreover, in that case (Alg. \ref{alg:h-encode}) \textsf{h-encode} ensures there must be some non-trivial fd $\langle X, B\rangle \in \Sigma$ with $A \in X$ where $B \mapsto x_b$ and $A \mapsto x_a$. Thus by an empty application of R5 we have $\langle X, B\rangle \in \Sigma^\vartriangleright$. The inductive step shows the property still holds for arbitrary causal dependencies in $C^+_\varphi$. 
$\Box$
\end{proof}

\subsection{Proof of Proposition \ref{prop:exogenous}}\label{a:exogenous}
\noindent
\emph{``Let $\mathcal S(\mathcal E, \mathcal V)$ be a structure with variable $x  \in \mathcal V$. Then $x$ can only be a first cause of some $y \in \mathcal V$ if $x$ is exogenous. Accordingly, any variable $y \in \mathcal V$ can only have some first cause $x \in \mathcal V$ if it is endogenous.''}
\begin{proof}
The proof is straightforward from definitons.  
For the first statement, suppose by contradiction that $x \in \mathcal V$ is not exogenous but is a first cause of some $y \in \mathcal V$. 
By Def. \ref{def:tcm}, $\varphi$ is bijective then there is some $f \in \mathcal E$ such that $\varphi(f)=x$. Moreover, as $x$ is not exogenous then by Def. \ref{def:exogenous} it must be endogenous. In other words, there must be some $x_a \in \mathcal V$ such that $x_a \neq x \in Vars(f)$ and then by Eq. \ref{eq:direct-causal-dependencies} we have $(x_a, x) \in C_{\varphi}$ hence $(x_a, x) \in C^+_{\varphi}$. However, as $x$ is a first cause, by Def. \ref{def:first-cause} there can be no $y \in \mathcal V$ such that $(y, x) \in C^+_{\varphi}$. $\lightning$.

Now, a symmetrical argument proves the second statement. Also by contradiction take a variable $y \in \mathcal V$ that is not endogenous and suppose it has some first cause $x \in \mathcal V$. 
As variable $y$ is not endogenous then by Def. \ref{def:exogenous} it must be exogenous. In other words, there must be $f \in \mathcal E$ such that $Vars(f)=\{y\}$. Thus for any total causal mapping $\varphi$ over $\mathcal S$, we must have $\varphi(f)=y$ and, for no $x \in \mathcal V$, we have $(x, y) \in C_{\varphi}$. Therefore it is not possible to derive $(x, y) \in C^+_{\varphi}$ for some $x \in \mathcal V$. But as $y$ has some first cause $x \in \mathcal V$ by assumption, we must have $(x, y) \in C^+_{\varphi}$. $\lightning$.
$\Box$
\end{proof}

\subsection{Proof of Lemma \ref{lemma:connections}}\label{a:lemma-connections}
\noindent
\emph{``Let $\mathcal S(\mathcal E, \mathcal V)$ be a complete structure, $\varphi$ a total causal mapping over $\mathcal S$ and $\Sigma$ an fd set encoded through $\varphi$ given $\mathcal S$. Then a variable $x_a \in \mathcal V$ can only be a first cause of some variable $x_b \in \mathcal V$, where $\langle X, B\rangle \in \Sigma$, and $B \mapsto x_b$, $A \mapsto x_a$, if either (i) $A \in X$ or (ii) $A \notin X$ but there is $\langle Z, C\rangle \in \Sigma^\vartriangleright\!$ with $A \in Z$ and $C \in X$.''}

\begin{proof}
We prove the statement by construction out of Theorem \ref{thm:connections1}.

By Def. \ref{def:first-cause}, one of the conditions for $x_a$ to be a first cause of $x_b$ is that $(x_a, x_b) \in C^+_\varphi$. 
Moreover, by Theorem \ref{thm:connections1} we know that $(x_a, x_b) \in C^+_\varphi$ can only hold if there is some non-trivial fd $\langle Z, B\rangle \in \Sigma^\vartriangleright$ with $A \in Z$, where $B \mapsto x_b$ and $A \mapsto x_a$. Now, by Def. \ref{def:tcm} $\varphi$ is bijective then there is $\langle X, B\rangle \in \Sigma$. Moreover, since $\Sigma$ is parsimonious, $X \!\to B$ is the only fd in $\Sigma$ with $B$ in its rhs. So let $A \notin X$. Then we know $X \neq Z$ hence $X \!\to B$ cannot be the fd required by Theorem \ref{thm:connections1}. 

Then such fd $Z \xrightarrow{\triangleright} B$ with $A \in Z$ can only exist if derived by some finite application of R5. That is, there must be some $\langle Z, C\rangle \in \Sigma\,$ with $A \in Z$ such that $X=CW$ for some $W$ and then R5 can be applied over $\langle Z, C\rangle, \langle CW, B\rangle \in \Sigma$ to get non-trivial $\langle ZW, B\rangle \in \Sigma^\vartriangleright$ where $A \in Z$. 

Now, it is easy to see that when such fd $Z \!\to C$ with $A \in Z$ does not exists in $\Sigma^\vartriangleright$ (the second condition of the lemma), then obviously $Z \!\to C$ cannot exist in $\Sigma$ to satisfy the requirement imposed by Theorem \ref{thm:connections1}. That is, $(x_a, x_b) \notin C^+_\varphi$.
$\Box$
\end{proof}

\subsection{Proof of Theorem \ref{thm:connections2}}\label{a:thm-connections2}
\noindent
\emph{``Let $\mathcal S(\mathcal E, \mathcal V)$ be a complete structure, $\varphi$ a total causal mapping over $\mathcal S$ and $\Sigma$ an fd set encoded through $\varphi$ given $\mathcal S$. Now, let $B$ be an attribute that encodes some variable $x_b \in \mathcal V$.
If $\langle X, B\rangle \in \Upsilon(\Sigma)^\looparrowright\!$,\footnote{Note that the folding is taken w.r.t. the $\upsilon$-projection of $\Sigma$, then $x_b$ where $B \mapsto x_b$ is an endogenous variable.} then every first cause $x_a$ of $x_b$ (if any) is encoded by some attribute $A \in X$.''}

\begin{proof}
We show that the existance of a missing first cause $x_c$ of $x_b$ for folded $X \xrightarrow{\looparrowright} B$, where $B \mapsto x_b$ and $C \mapsto x_c$ but $C \notin X$ leads to a contradiction.   

Suppose, by contradiction, that there is some missing first cause $x_c \in \mathcal V$ of $x_b$, where $C \mapsto x_c$ and $C \notin X$. 
Then, by Lemma \ref{lemma:connections}, since variable $x_c$ is a first cause of variable $x_b$, it must be exogenous and, for $\langle Y, B\rangle \in \Upsilon(\Sigma)$ either (i) $C \in Y$ or (ii) $C \notin Y$ but there is $\langle Z, D\rangle \in \Upsilon(\Sigma)^\vartriangleright$ with $C \in Z$ and some $D \in Y$. 

In the first case (i), since $x_c$ is exogenous and $\Sigma$ is parsimonious, we have $\langle \phi, C\rangle \in \Sigma$ but by Def. \ref{def:y-projection} there can be no $W \!\to C$ in the $\upsilon$-projection $\Upsilon(\Sigma)$ of $\Sigma$. That is, $C$ cannot be `consumed' by R5 and then $\langle Y, B\rangle \in \Upsilon(\Sigma)$ with $C \in Y$ implies that, for any $\langle W, B\rangle \in \Upsilon(\Sigma)^\vartriangleright$, we must have $C \in W$. However, by assumption we have $\langle X, B\rangle \in \Upsilon(\Sigma)^\looparrowright\!$ then, by Def. \ref{def:folding}, $\langle X, B\rangle \in \Upsilon(\Sigma)^\vartriangleright$, yet $C \notin X$. $\lightning$.

In the second case (ii), observe that $C \in Z$ and $D \in Y$, and let $Y=DS$. Then by R5 over $Z \!\to D$ and $DS \!\to B$ we get $\langle ZS, B\rangle \in \Upsilon(\Sigma)^\vartriangleright$, where $C \in Z$. 
Well, either $ZS \!\to B$ is folded or it is not, rendering two cases for analysis. If $ZS \!\to B$ is folded, then both $ZS \!\to B$ and $X \!\to B$ are folded. But as $\Upsilon(\Sigma)$ is parsimonious, then by Lemma \ref{lemma:folding-unique} the folding of $B$ must be unique. Therefore we must have $ZS=X$, with $C \in Z$ but $C \notin X$. $\lightning$.

Else, assume $ZS \!\to B$ is not folded. Then by Def. \ref{def:folded} there is some $W$ with $W \nsupseteq ZS$ such that non-trivial $W \!\to ZS$ is in $\Upsilon(\Sigma)^+$ and $ZS \!\not\to W$.

However, as $C \in Z$ and $W \!\to ZS$, by (R3) decomposition we must have $W \!\to C$ in $\Upsilon(\Sigma)^+$, either with $C \in W$ or with $C \notin W$ and then $W \!\to C$ is non-trivial. But we know the latter cannot be the case by the same argument used in the first case (i), viz., $x_c$ is exogenous with $C \mapsto x_c$ and $\Sigma$ is parsimonious. That is, we must have $C \in W$. Furthermore, as $W \!\to ZS$ in $\Upsilon(\Sigma)^+$, then by R5 over $ZS \!\to B$ we get $\langle W, B\rangle \in \Upsilon(\Sigma)^\vartriangleright$ with $C \in W$. Now it is easy to see that the same situation recurs to $W \!\to B$. If it is not folded, eventually for some $T$ we will have $T \xrightarrow{\triangleright} W \xrightarrow{\triangleright} B$ with $C \in T$, where $T \!\to B$ will be folded just like $X \!\to B$. That is, by (Lemma \ref{lemma:folding-unique}) the uniqueness of the folding of $B$, we will have $T=X$ with $C \in T$ and $C \notin X$. $\lightning$. 
$\Box$
\end{proof}

\section{Proofs of Probabilistic DB Synthesis}\label{sec:proofs-synthesis4u}

\subsection{Proof of Theorem \ref{thm:bcnf}}\label{a:bcnf}
\noindent
\emph{``Let $\mathcal S_k$ and $H_k$ be (resp.) the complete structure and `big' fact table of hypothesis $k$, and let $\;\Gamma_k^\prime$ be the repaired factorization of $\mathcal S_k$ over $H_k$, and $Y_0$ the `explanation' table where hypothesis $k$ is recorded. Now, let $\boldsymbol Y_k$ be a U-relational schema defined $\boldsymbol Y_k \triangleq \textsf{synthesize4u}(\mathcal S_k, H_k, Y_0)$. 
Then $\boldsymbol Y_k$ is in BCNF w.r.t. $\Gamma_k^\prime$ and is minimal-cardinality.''}
\begin{proof}

Let $Y_k^i\,[\,V_i\,D_i\,\,|\,\phi\,A_i\,G_i\,]$ and $Y_k^j\,[\,\overline{V_j\,D_j\,}\,|\,S\,T\,]$ be (resp.) any u-factor projection and predictive projection of $H_k$. Note that all fd's in $\Gamma_k^\prime$ are either in $\Phi(\Gamma_k^\prime)$ of form $\phi\,A_i \to B$ or $\phi \to B$, or in $\Upsilon(\Gamma_k^\prime)$ of form $A_1\,A_2\,...\,A_\ell\,S \!\to T$ with $\upsilon \in S$. 
We must show that no fd in $(\Gamma_k^\prime)^+$ can violate $Y_k^i$ or $Y_k^j$. It is easy to see that the projection (cf. Def. \ref{def:fd-projection}) of non-trivial fd's in $\Phi(\Gamma_k^\prime)^+$ onto $Y_k^j$ is empty, just like the projection of non-trivial fd's in $\Upsilon(\Gamma_k^\prime)^+$ onto $Y_k^i$.

For the u-factor projections, note by Def. \ref{def:nf} that for any fd $X \!\to C$ in $(\Gamma_k^\prime)^+$ to violate BCNF in $Y_k^i\,[\,V_i\,D_i\,\,|\,\phi\,A_i\,G_i\,]$, it must be non-trivial ($C \not\in X$) with $XC \subseteq \phi\,A_i\,G_i$ but $X \not\to \phi\,A_i\,G_i$ (that is, $X$ is not a superkey for $Y_k^i$).
Note that we have both $\phi\,A_i \!\to B$ and $\phi \!\to A_i\,B$ in $(\Gamma_k^\prime)^+$ for any $B \in G_i$, but both $\phi\,A_i$ and $\phi$ are superkeys for $Y_k^i$. Also, that there can be no non-trivial fd's $\langle X, C\rangle \in \Phi(\Gamma_k^\prime)^+$ with $\phi \not\in X$, and by the definition of Problem \ref{prob:u-learning} we know that $A_i\,G_i$ is a maximal group. So, for any non-trivial $\langle X, C\rangle \in \Phi(\Gamma_k^\prime)^+$, $X$ must be a superkey for $Y_k^i$. Thus no u-factor projection can be subject of BCNF violation w.r.t. $\Gamma_k^\prime$.

Now, for predictive projection $Y_k^j\,[\,\overline{V_j\,D_j\,}\,|\,S\,T\,]$ let us reconstruct the process towards deriving $\langle S, T\rangle \in \Upsilon(\Gamma_k^\prime)^\vartriangleright$. Note that it is derived by \textsf{synthesize4u} simulating $\ell$ applications of R5 over $\langle \phi, A_i\rangle,\, \langle A_1\, A_2\, ... A_\ell\,S,\, T \rangle \in \Gamma_k^\prime$ for $0 \leq i \leq \ell$. Note also that (i) no cyclic fd's can be involved in such R5 applications, as they must always be over an fd in $\Phi(\Gamma_k^\prime)$; and (ii) as a result of (Alg. \ref{alg:merge}) \textsf{merge}, $A_1\, A_2\, ... A_\ell\,S \!\to T$ was the only non-trivial fd in the projection 
$\pi_{A_1\, A_2\, ... A_\ell\,ST}(\Gamma_k^\prime)$. 
Thus, 
the only non-trivial fd's in the projection $\pi_{ST}(\,(\Gamma_k^\prime)\,)^+$ in addition to $S \!\to T$ itself must be of form $S \!\to C$ rendered out of (R3) decomposition from it for all $C \in T$. In any such fd's, we have $S$ as a superkey for $Y_k^j$. Therefore no predictive projection can be subject of BCNF violation w.r.t. $\Gamma_k^\prime$.

For the minimality note, as a consequence of (Alg. \ref{alg:merge}) \textsf{merge}, any two schemes $Y_k^p[XZ]$, $Y_k^q[VW]$ are rendered by \textsf{synthesize4u} into $\boldsymbol Y_k$ iff we have fd's $\langle X, Z\rangle, \langle V, W\rangle \in (\Gamma_k^\prime)^+$ and $X \not\leftrightarrow V$, i.e., it is not the case that both $X \!\to V$ and $V \!\to X$ hold in $(\Gamma_k^\prime)^+$. Now, to prove that $\boldsymbol Y_k$ is minimal-cardinality, we have to find that merging any such pair of arbitrary schemes shall hinder BCNF in $\boldsymbol Y_k$. In fact, take $\boldsymbol Y_k^\prime := \boldsymbol Y_k \setminus (Y_k^p[XZ] \cup Y_k^q[VW]) \cup Y_k^\ell[XZVW]$. As $X \not\leftrightarrow V$, then neither $X$ nor $V$ can be a superkey for $Y_k^\ell$, which therefore cannot be in BCNF.
$\Box$
\end{proof}

\subsection{Proof of Theorem \ref{thm:lossless}}\label{a:lossless}
\noindent
\emph{``Let $\mathcal S_k$ be the complete structure of hypothesis $k$, and $H_k[U]$ its `big' fact table such that $\Gamma_k^\prime$ is the repaired factorization of $\mathcal S_k$ over $H_k$ and $Y_0$ is the `explanation' table where hypothesis $k$ is recorded. Now, let $\boldsymbol Y_k$ be a U-relational schema defined $\boldsymbol Y_k \triangleq \textsf{synthesize4u}(\mathcal S_k, H_k, Y_0)$. 
Then,
\begin{itemize}
\item[(a)] the join $\bowtie_{\,i=1}^{\,m}  Y_k^i\,[\,V_i\,D_i\,|\,\phi\, A_i\,G_i\,]$ of any subset of the u-factor projections  of $H_k$ is lossless w.r.t. $\Gamma_k^\prime$. 
\item[(b)] any predictive projection $Y_k^j\,[\,\overline{V_j\,D_j}\,|\,S\,T\,]$, result of a join of the theoretical u-factor $Y_0\,[\,V_0\,D_0\,|\,\phi\,\upsilon\,]$ with the `big' fact table $H_k[U]$ and in turn with u-factor projections $Y_k^i\,[\,\overline{V_i\,D_i}\,|\,\phi\,A_i\,G_i\,]$, is lossless w.r.t. $\Gamma_k^\prime$.''
\end{itemize}
}
\begin{proof}
For item (a), by Lemma \ref{lemma:lossless}, we know that any pair $Y_k^i\,[\,V_i\,D_i\,|\,\phi\, A_i\,G_i\,],$ $Y_k^j\,[\,V_j\,D_j\,|\, \phi\, A_j\,G_j\,]$ of u-factor projections of $H_k$ will have a lossless join w.r.t. $\!\Gamma_k^\prime$ iff $(\phi\,A_i\,G_i \,\cap\, \phi\,A_j\,G_j) \to (\phi\,A_i\,G_i \setminus\, \phi\,A_j\,G_j)$ or $(\phi\,A_i\,G_i \,\cap\, \phi\,A_j\,G_j) \to (\phi\,A_j\,G_j \setminus\, \phi\,A_i\,G_i)$ hold in $(\Gamma_k^\prime)^+$. By Def. \ref{def:u-factor}, we know that $(\phi\,A_i\,G_i \cap\, \phi\,A_j\,G_j)=\{\phi\}$, and $\phi\,A_i\,G_i \setminus \phi\,A_j\,G_j = A_i\,G_i$. In fact $\phi \!\to A_i\,G_i$ is a repaired fd in $\Gamma_k^\prime$, therefore $Y_k^i$ and $Y_k^j$ have a lossless join. Now, since the join is an associative operation \cite[p. 62]{ullman1988}, and as we have chosen $Y_k^i$ and $Y_k^j$ arbitrarily, then clearly any subset of the u-factor projections must have a lossless join.

For item (b), 
for any predictive projection $Y_k^j\,[\,\overline{V_j\,D_j}\,|\,S\,T\,]$ take the join $\bowtie Y_k^i\,[\,\overline{V_i\,D_i}\,|\,\phi\,\overline{A_i\,G_i}\,]$ of u-factor projections such that, for all $A_i \in \overline{A_i\,G_i}$, we have $A_i \in W \subset Z$ where $S=Z \setminus W$ and $\langle Z, T\rangle \in \Gamma_k^\prime$. That is, $A_i$ is a pivot attribute representing a first cause of some $C \in T$. By item (a), we know that such join is lossless. 

We must show that the join $\bowtie Y_k^i$ with `big' fact table $H_k[U]$ is also lossless. By Lemma \ref{lemma:lossless}, that is the case iff $(\phi\,\overline{A_i\,G_i} \,\cap\, U) \to (\phi\,\overline{A_i\,G_i} \setminus U)$ or $(\phi\,\overline{A_i\,G_i} \,\cap\, U) \to (U \setminus \phi\,\overline{A_j\,G_j})$ hold in $(\Gamma_k^\prime)^+$. In fact, we have $(\phi\,\overline{A_i\,G_i} \,\cap\, U)=\phi\,\overline{A_i\,G_i}$ and $(\phi\,\overline{A_i\,G_i} \setminus U)=\varnothing$ such that $\phi\,\overline{A_i\,G_i} \to \varnothing$ is trivially in $(\Gamma_k^\prime)^+$.

Finally, the join of theoretical u-factor $Y_0\,[\,V_0\,D_0\,|\,\phi\,\upsilon\,]$ with big fact table $H_k[U]$ must be lossless likewise. In fact, note that $(\,\phi\,\upsilon \,\cap\, U)= \phi\,\upsilon$, and $(\phi\,\upsilon \setminus U)=\varnothing$. Then also trivially we have $\phi\,\upsilon \!\to \varnothing$, which is in $(\Gamma_k^\prime)^+$ as well. 
Since the join is commutative \cite[p. 62]{ullman1988}, the order of application is irrelevant therefore the join of all joins examined above taken together must be lossless.
$\Box$
\end{proof}

\begin{mylemma}\label{lemma:lossless}
Let $\Sigma$ be a set of fd's on attributes $U$, and $R_i[S], R_j[T] \in \boldsymbol R[U]$ be relation schemes with $ST \subseteq U$; and let $\pi_{ST}(\Sigma)$ be the projection of $\Sigma$ onto $ST$. Then $R_i[S]$ and $R_j[T]$ have a lossless join w.r.t. $\pi_{ST}(\Sigma)$ iff $(S \cap\, T) \to (S \setminus T)$ or $(S \cap\, T) \to (T \setminus S)$ hold in $\pi_{ST}(\Sigma)^+$.
\end{mylemma}
\begin{proof}
See Ullman \cite[p. $\!$397]{ullman1988}.
$\Box$
\end{proof}


\end{document}